%% file: main.tex
\documentclass[11pt,twoside,a4paper]{article}

\input{preamble}

\title{Entanglement in von Neumann algebraic \\ quantum information theory}
\author{Lauritz van Luijk}
\date{\today}

\begin{document}

\thispagestyle{empty}
\null 
\vspace{33pt}
\null 
\begin{center}\bf\huge
    \rule{\textwidth}{1pt}\\[11pt]

    Entanglement in von Neumann algebraic \\ quantum information theory
    \rule{\textwidth}{1pt}
\end{center}

\vfill

\begin{center}
    Von der Fakultät für Mathematik und Physik\\
    der Gottfried Wilhelm Leibniz Universität Hannover 
    zur Erlangung des Grades\\[66pt]
    Doktor der Naturwissenschaften\\
    Dr.\ rer.\ nat.
    \\[66pt]
    genehmigte Dissertation von
    \\[11pt]
    {\bf \large
    Lauritz van Luijk
    }
    \\[88pt]
\end{center}


\begin{center}\bf \large
    2025
\end{center}

\pagenumbering{roman}

\clearpage

\section*{Abstract}

In quantum systems with infinitely many degrees of freedom, states can be infinitely entangled across a pair of subsystems. But are there \emph{different} forms of infinite entanglement? To address this, we use a rigorous framework for studying information-theoretic properties in infinite systems, formulated in terms of von Neumann algebras. We show that there are operational tasks that distinguish different forms of infinite entanglement. In fact, we find that the classification of von Neumann factors into types and subtypes is in one-to-one correspondence with a family of operational entanglement properties. For instance, Connes' classification of type III factors corresponds to the smallest achievable error when "embezzling" entanglement from the system. Our findings promote the type classification from algebraic bookkeeping to a classification of infinite quantum systems based on their operational entanglement properties.
\\





\clearpage
\ifthenelse{\boolean{printversion}}{
    \thispagestyle{empty}
    \null 
    \clearpage
}{}

\section*{Preface to the arXiv version}\label{sec:preface}

This version is prepared for a publication on the arXiv.
The scientific content is almost identical with those of the published version \cite{van_luijk_entanglement_2025}.
I have corrected a few typos and smaller errors, updated references, and removed non-scientific materials, which were obligatory in the published version.

This work is partly based on the articles \cite{van_luijk_relativistic_2024,van_luijk_embezzlement_2024,van_luijk_pure_2024} (with Alexander Stott\-meister, Reinhard F.\ Werner, and Henrik Wilming) and \cite{van_luijk_critical_2025,van_luijk_multipartite_2025,van_luijk_large-scale_2025} (with Alexander Stott\-meister, and Henrik Wilming).
The idea of considering embezzlement of entanglement in the setting of von Neumann algebras, from which the whole project departed, is taken from unpublished notes of Uffe Haagerup, Volkher B.\ Scholz and Reinhard F.\ Werner~\cite{mbz_draft} (see dedication in \cite{van_luijk_embezzlement_2024}).
The article \cite{van_luijk_uniqueness_2025} is based on \cref{sec:haag} and was written after this thesis was completed.
The contents of \cref{sec:agents,sec:axioms} are intended for future publication.

When citing results from this thesis that have also appeared in separate articles, please cite those articles as well, so that my coauthors can receive the credit they deserve.


\clearpage
\ifthenelse{\boolean{printversion}}{
    \thispagestyle{empty}
    \null 
    \clearpage
}{}

\section*{Acknowledgements}\label{sec:acks}

I have greatly enjoyed my time as a PhD student in Hanover, and I wish to express my sincere gratitude to the many people who made this possible.

Most importantly, I thank my supervisor Reinhard F.\ Werner for all he taught me, for his guidance and encouragement, and for the unconditional academic freedom I received.
You showed me by example how to be a good scientist, to always be open-minded to new ideas, and that science is best enjoyed when shared with colleagues and friends.
Asking you for a PhD position was, without a doubt, the best decision of my academic career.

I deeply thank Alexander Stottmeister, my second supervisor, and Henrik Wilming, who effectively played the same role.
I have enjoyed our countless scientific discussions and wonderful collaborations to the fullest.
You taught me more than I can mention here.

I thank Niklas Galke for sparking my interest in mathematical physics, and all the fun we've had together---first as students and later as collaborators.
I am grateful to Daniel Burgarth for teaching me a great deal of science, your valuable advice, and for inviting me to visit you in Australia, where I spent six wonderful weeks.

I was very lucky to have been involved in many fun and fruitful collaborations in the past years.
For that I thank my collaborators, many of whom have become good friends: Wolfram Bauer, Simon Becker, Lennart Binkowski, Daniel Burgarth, Robert Fulsche, Niklas Galke, Alexander Hahn, Mattias Johnsson, Gereon Ko\ss mann, Robert Salzmann, Ren\'e Schwonnek, Alexander Stottmeister, Reinhard F.\ Werner, Henrik Wilming, and Timo Ziegler.

I thank the past and present members of the Quantum Information Group in Hanover.
I~enjoyed the warm atmosphere and our many after-work gatherings ever since I joined the group when I started working on my master thesis almost five years ago.
I especially thank Wiebke Möller for her help with bureaucratic matters on countless occasions.

I thank Stefaan Vaes for answering several technical questions on MathOverflow, and I thank him and Amine Marrakchi for reaching out to explain results about ergodic dynamical systems, which allowed us to strengthen our results on embezzlement of entanglement after the first version of \cite{van_luijk_embezzlement_2024} appeared on the arXiv.

I'd like to thank Henrik, Alex, and Reinhard who generously proofread parts of this work.
Moreover, I thank Henrik for helping me prepare the figures for this document.

Finally, I thank my family and friends for their unconditional love and support without which this work would not have been possible.
In particular, I thank my parents, Sibylle and Marinus, and my grandmother, Hildburg, for their never-ending enthusiasm for my work---even if most of it sounds like total gibberish to you.
My deep gratitude goes to my wife Alina.
You never fail to bring a smile to my face, and my love for you has only grown stronger over the years.


\null 

Thank you all.

\null

\null

I acknowledge financial support by the MWK Lower Saxony via the Stay Inspired Program (grant ID: 15-76251-2-Stay-9/22-16583/2022).

I recognize that opportunities in academia are not distributed equally, and that I have been fortunate to benefit from circumstances---including race, gender, and educational background---that have eased my path in ways not available to all.

I thank the many open source software contributors whose selfless efforts I have depended on throughout the last years. 

\clearpage
\ifthenelse{\boolean{printversion}}{
    \thispagestyle{empty}
    \null 
    \clearpage
}{}

\tableofcontents

\etocsettocstyle{}{} 

\ifthenelse{\boolean{printversion}}{
    \thispagestyle{empty}
    \null 
    \cleardoublepage
}{
    \clearpage
}

\pagenumbering{arabic}
\section{Introduction}\label{sec:intro}

Quantum entanglement lies at the heart of quantum information theory, underpinning a range of phenomena from teleportation to quantum cryptography and quantum computation \cite{bennett_teleporting_1993,gisin_cryptography_2002,jozsa_role_2003,horodecki_quantum_2009}.
While the structure and manipulation of entanglement are well understood in finite systems, many foundational and physically relevant settings involve infinitely many degrees of freedom.

In systems with infinitely many degrees of freedom, states can be infinitely entangled across a pair of subsystems. 
Examples range from the vacua of relativistic quantum field theories \cite{srednicki_entropy_1993,summers_vacuum_1985,verch_distillability_2005} to ground states in infinite spin chains \cite{keyl_entanglement_2006}.
But does that mean that these states are all the same from the point of view of entanglement theory?
More specifically:
\begin{center}\it
    Are there different kinds of infinite entanglement?\\
    If so, can we distinguish them operationally?
\end{center}
The natural mathematical framework for modeling systems with infinitely many degrees of freedom is that of von Neumann algebras \cite{gabbiani_operator_1993,haag_local_1996,verch_distillability_2005,naaijkensQuantumSpinSystems2017,buchholz_universal_1987}.
The study of quantum entanglement in the von Neumann algebraic setting was initiated in the seminal works of Summers and Werner on the violation of Bell's inequalities in quantum field theory \cite{summers_vacuum_1985,summers_bells_1987,summers_maximal_1987,summers_maximal_1988,summers_bells_1995}.
While several aspects have been considered since then \cite{verch_distillability_2005,summers_independence_1990,keyl_infinitely_2003,huang_states_2008,huang_dense_2019,crann_state_2020,scholz_tsirelsons_2008,hollands_entanglement_2018}, there is, to this point, no systematic study of entanglement theory in the von Neumann algebraic setting.
This work aims to fill this gap.
In doing so, we will discover tight connections between operational entanglement properties and the classification theory for von Neumann algebras.

\null

This thesis is structured as follows:
\Cref{sec:main-results} summarizes the setup and the main results of our work; \cref{sec:vN-prelims} collects the necessary mathematical prerequisites on von Neumann algebras; \cref{sec:vNQI} explains the operational setup of von Neumann algebraic quantum information theory; \cref{sec:catalytic states} studies catalytic states of von Neumann algebras, the monopartite analog of embezzlement of entanglement; \cref{sec:bipartite systems} collects the basic structure results for bipartite systems.
Finally, with all preparations out of the way, \cref{sec:locc,sec:strong-entanglement} study entanglement theory in von Neumann algebraic bipartite systems.

\section{Main results}\label{sec:main-results}

In this section, we give an overview of the results obtained in this work.
We focus on conceptual ideas, the technical details fill the sections to come.
We start off with the setup of von Neumann algebraic quantum information theory, before turning to the central focus of this work: entanglement theory.

\paragraph{\it Setup.}

Let us consider an agent with partial access to a quantum system, potentially having infinitely many degrees of freedom.
This means that the agent can implement a subset $\O$ of the set of all operations on the full system.
The prototypical example is that of an agent who can solely act on a \emph{subsystem} of the full system.
On purely operational grounds, a subsystem is not more than the set of operations that act on it.


The basic assumption underlying von Neumann algebraic quantum information theory is that the agent's operations are described by a von Neumann algebra $\M$, called the \emph{observable algebra}, on the Hilbert space $\H$ of the full system.
The operations that the agent can implement are those whose Kraus operators are elements of $\M$, or equivalently, those that can be implemented by preparing an ancillary system $\K$, applying a unitary from the amplified observable algebra $\M\barox\B(\K)$, and then discarding/measuring the ancilla.
We refer to this class of operations as the \emph{$\M$-inner} operations, denoted $\O_\M$.

This assumption applies to agents who can implement operations that interact solely with specific subsystems of a full system, even if the latter has infinitely many degrees of freedom.
More precisely, the von Neumann algebraic description applies when one considers subsystems in a fixed sector, e.g., the ground state sector of a local Hamiltonian on a spin system or the vacuum sector of a field theory, of the full system \cite{haag_local_1996,naaijkensQuantumSpinSystems2017,keyl_entanglement_2006,verch_distillability_2005,buchholz_universal_1987,keyl_infinitely_2003,gabbiani_operator_1993}.
We discuss these settings in \cref{sec:examples}.

In order to further justify the assumption, we consider a set of operational axioms.
Apart from the obvious consistency assumption that the composition of implementable operations is implementable, we assume the following:

\begin{enumerate}[({A}1)]
    \item\label{axiom:limits_i}
    Limits: The agent's operations are topologically closed.

    \item\label{axiom:ancilla_i} 
    Ancillas: The agent can freely prepare, manipulate, and discard ancillary systems.

    \item\label{axiom:dilation_i} Dilations: 
    Each of the agent's channels can be implemented by preparing an ancillary system, performing a global unitary, and  discarding the ancilla.

    \item\label{axiom:reverse_i} Reversibility: For every unitary channel, the agent can also implement the inverse.
    
    \item\label{axiom:completeness_i} Postselection: An instrument that could be implemented by postselecting outcomes of implementable instruments is implementable.
\end{enumerate}

If we identify agents with the subsystems they act on, these axioms form an operational approach to the question: \emph{What is a subsystem?}

Axioms \ref{axiom:dilation_i} and \ref{axiom:reverse_i} correspond to the idea that the only source of irreversibility is the loss of information, see, e.g., \cite{bennett_demons_1987,gregoratti_quantum_2003}.
We refer to \cref{sec:vNQI} for a more detailed discussion and mathematically precise versions of these axioms, explaining, e.g., the topology in \ref{axiom:limits_i}.

If $\M$ is a von Neumann algebra on $\H$, then the set $\O_\M$ of $\M$-inner operations satisfies the axioms \ref{axiom:limits_i} to  \ref{axiom:completeness_i}.
In fact, we show the converse:

\begin{introtheorem}[See \cref{sec:axioms}]\label{thm:axioms_i}
    The axioms above hold if and only if there is a von Neumann algebra $\M$ on $\H$ such that the agent's operations are exactly the $\M$-inner ones.
\end{introtheorem}

\Cref{thm:axioms_i}, shows that the mathematical assumption that the agent's operations are the $\M$-inner ones is, in fact, a fully operational assumption in disguise.

Apart from these axioms, there are other reasons to consider the setting of operations determined by a von Neumann algebra $\M$ on a Hilbert space $\H$.
Motivations come from symmetry considerations, e.g., in the setting of (quantum) reference frames.
Moreover, the von Neumann algebraic setting can always be regarded as an idealization: For every collection $\O$ of operations on the full system, there is a unique smallest von Neumann algebra $\M$ such that $\O\subset \O_\M$ (see \cref{sec:vNQI}).

For a good understanding of von Neumann algebraic quantum information theory, we need three main ingredients:
\begin{enumerate}[(1)]
    \item\label{it:vNQI1} extending ideas and results from finite-dimensional quantum information theory to the von Neumann algebraic setting;
    \item\label{it:vNQI2} understanding the relation between mathematical (von Neumann algebraic) properties of the observable algebra $\M$ and operational properties of the agent;
    \item\label{it:vNQI3} knowing the von Neumann algebraic properties of the observable algebras in many concrete physical models.
\end{enumerate}

With a good understanding of each of these points, we can  understand the operational properties of concrete physical models and  recognize the ultimate limits of how quantum information can be manipulated in systems with infinitely many degrees of freedom.
What makes the latter interesting is that there are quantitatively and qualitatively new information-theoretic properties in systems with infinitely many degrees of freedom, e.g., the phenomenon of `embezzlement of entanglement' (see below).
For each of the three aspects, there is important previous work, e.g., \cite{hiai_quantum_2021,verch_distillability_2005,berta_smooth_2015,fawzi_asymptotic_2025,hollands_entanglement_2018,crann_state_2020,keyl_infinitely_2003,kuramochi_accessible_2018} for \ref{it:vNQI1}, \cite{summers_maximal_1988,summers_vacuum_1985,werner_local_1987,keyl_infinitely_2003,naaijkens_subfactors_2018,fiedler_jones_2017} for \ref{it:vNQI2}, and \cite{haag_local_1996,gabbiani_operator_1993,naaijkensAnyonsInfiniteQuantum2012,jones_local_2025} for \ref{it:vNQI3}, to name a few.
The main focus of this work is on item \ref{it:vNQI2} above. We will show that the type classification of von Neumann algebras is in one-to-one correspondence with operational properties in entanglement theory.

A priori, it is not clear that there should be any interesting correspondence between mathematical properties of the observable algebras on the one hand and operational properties on the other hand.
In finite dimensions, interesting operational properties are always properties of individual states, not properties of the whole system.
In a way, this is because all finite-dimensional systems are "the same" up to the dimension of the Hilbert space. 
In the general case, however, the rich mathematical structure of von Neumann algebras creates a vast zoo of non-equivalent setups, and item \ref{it:vNQI2} above essentially asks to what degree they can be classified in terms of their operational properties.
Roughly speaking, the \emph{type classification} of von Neumann algebras takes over the role of the dimension in the usual setup of quantum information theory.
General von Neumann algebras can be decomposed into \emph{factors}, i.e., von Neumann algebras with trivial center $\M \cap \M'=\CC1$, where $\M'=\{a' : [a,a']=0 \ \forall a\in\M\}$ denotes the \emph{commutant}.
Factors are classified into types $\I$, $\II$, and $\III$, each further divided into subtypes, see \cref{fig:types}.
Type $\I$ factors are finite or infinite matrix algebras. 
They arise in cases where only finitely many degrees of freedom contribute to the entanglement across a pair of subsystems, e.g., if the full system consists of finitely many qubits or is a system of finitely many bosonic modes.
Types $\II$ and $\III$ are more exotic and only arise in systems with infinitely many degrees of freedom.

\begin{figure}[ht!]
    \centering
    \def\svgwidth{.55\textwidth}
    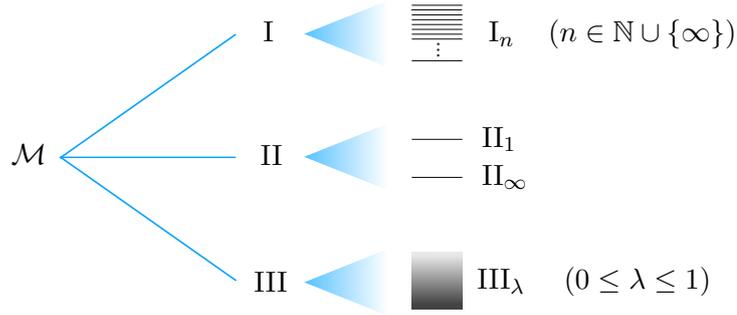
    \caption{The classification of factors into types $\I$, $\II$, $\III$, and their respective subtypes (see \cref{sec:vN-prelims} for details).}
    \label{fig:types}
\end{figure}

The simplest example of a one-to-one correspondence between an operational and an algebraic property is the \emph{no information without disturbance} principle, which states that no information can be gained about a quantum system without disturbing its state.
The principle applies to the agent's subsystem if and only if the observable algebra $\M$ has the algebraic property of being a factor.
To give a less trivial example, it was shown in \cite{scholz_tsirelsons_2008} that counterexamples to Tsirelson's problem in entanglement theory require so-called non-AFD observable algebras (see \cref{sec:vNQI}).

\paragraph{\it Entanglement theory.}

We now turn our attention to entanglement. 
We consider a bipartite system consisting of two agents, conventionally called Alice and Bob, with commuting observable algebras $\M_A$ and $\M_B$.
Let us consider operational axioms for bipartite systems.
\begin{enumerate}[(B1)]
    \item\label{axiom:tomography_i} Tomography: States of the full system can be determined uniquely with correlation experiments.
\end{enumerate}
Independently of the `no information without disturbance' principle, \ref{axiom:tomography_i} is equivalent to the statement that $\M_A$ and $\M_B$ are factors that jointly generated $\B(\H)$.
In general, \ref{axiom:tomography_i} is insufficient to imply the following:
\begin{enumerate}[resume*]
    \item\label{axiom:haag-duality_i} Duality: Alice can implement all operations that commute with all of Bob's operations and vice versa.
    \item\label{axiom:purifications_i} Uniqueness of purifications: If two pure states of the full system have the same $A$-marginal, then they are connected by a partial isometry of $B$. 
\end{enumerate}

\begin{introtheorem}[see \cref{sec:haag}]\label{thm:haag-duality_i}
    The operational axioms \ref{axiom:haag-duality_i} and \ref{axiom:purifications_i} are equivalent. They hold if and only if
    \begin{equation}
        \M_A=\M_B'.
    \end{equation}
\end{introtheorem}

The property $\M_A=\M_B'$, stating that Alice's observable algebra is the commutant of Bob's, is known as \emph{Haag duality} \cite{keyl_entanglement_2006} after a similar property in quantum field theory \cite{haag_local_1996}.
From now on, we consider a bipartite system satisfying the axioms \ref{axiom:tomography_i}, \ref{axiom:haag-duality_i}, and \ref{axiom:purifications_i}.
I.e., we consider a pair $(\M_A,\M_B)$ of commuting factors in Haag duality.
We aim to answer the following question:
\begin{center}\it
    What is the relationship between pure state entanglement theory \\ and the algebraic  properties of the bipartite system?
\end{center}

To warm up, let us consider the finite-dimensional case.
Here, the Hilbert space decomposes as $\H=\H_A\ox\H_B$ and the local observable algebras are $\M_A=\B(\H_A)\ox1$ and $\M_B=1\ox\B(\H_B)$.
Up to unitary equivalence, the setup is fully determined by the dimensions $(n_A,n_B)$ of the local Hilbert spaces.
The minimal local dimension $n=\min(n_A,n_B)$ can be determined from entanglement theory: It is the maximal Schmidt rank that a pure state of the full system can have.
Moreover, the Schmidt decomposition tells us that there is no entanglement property that a pure state can have, which is not shared by a pure state on a bipartite system with local dimensions $(n,n)$.
Therefore, the minimal local dimension completely describes the collection of entanglement properties that the pure states of the full system have.

The finite-dimensional case considered in the previous paragraph is, of course, rather boring.
The reason is that finite-dimensional quantum systems only differ in their local dimensions---a fact that dramatically changes when passing to infinite quantum systems described by von Neumann algebras.
As stated above, the appropriate generalization of the local dimension is the type classification of factors.
The classification, in particular the type $\III$ case \cite{connes_classification_1973,haagerup_uniqueness_2016}, is a deep mathematical result, and we will profit greatly from the tools that have been developed in this line of research.
Generalizing the minimum local dimension, we consider the minimum of the types of $\M_A$ and $\M_B$ with respect to the following naive ordering
\begin{equation}\label{eq:naive type order}
    \I_1 < \I_2 < \ldots < \I_\oo < \II_1 < \II_\oo < \III_\lambda <\III_{\lambda'} , \qquad 0 \le \lambda<\lambda'\le 1.
\end{equation}
We refer to this as the \emph{minimal} type of a bipartite system.
In general, $\M_A$ is of type $\I$, $\II$, or $\III$ if and only if $\M_B$ is. Thus, we can speak of bipartite systems of type $\I$, $\II$, or $\III$.
Moreover, in the type $\III$ case, $\M_A$ and $\M_B$ always have the same subtype. 
We can now state the main result of this work:

\begin{introtheorem}[see \cref{sec:mbz,sec:locc-types}]\label{thm:minimal-type_i}
    The minimal type is in one-to-one correspondence with operational properties in entanglement theory.
\end{introtheorem}

To understand entanglement operationally, we need to extend the concept of local operations and classical communication (LOCC) to the von Neumann algebraic setting \cite{verch_distillability_2005}.
After doing so, we show Nielsen's theorem \cite{nielsen_conditions_1999} for von Neumann algebraic bipartite systems, which was already generalized to the type $\II$ setting in \cite{crann_state_2020}.

In the following, we sketch the operational entanglement properties that are needed for \cref{thm:minimal-type_i} (see \cref{tab:types}).
The easiest is the \emph{one-shot entanglement}.
In the type $\I$ case, every state has finite one-shot entanglement in the sense that only a finite number of Bell pairs can be distilled from it at once.
This is quite different for bipartite systems of types $\II$ and $\III$, where \emph{all} states have infinite one-shot entanglement.
To distinguish the type $\II$ subtypes, we consider the existence of \emph{maximally entangled states}.
Type $\III$ is singled out among the other types by the \emph{trivialization of LOCC}, meaning that all pure states of the full system are approximately LOCC equivalent.
This provides a one-to-one correspondence between operational properties in entanglement theory and the minimal type, except for the subtypes of type $\III$.

\begin{table*}[t]\centering
    \setlength{\tabcolsep}{7pt}
    \renewcommand\arraystretch{1.2}
    \begin{tabular}{@{} l cc cc ccc @{}}
    \toprule
    \multirow{2}{*}[-0.5\dimexpr \aboverulesep + \belowrulesep + \cmidrulewidth]{operational property} & \multicolumn{2}{c}{type $\I$} & \multicolumn{2}{c}{type $\II$} & \multicolumn{3}{c}{type $\III$}\\
    \cmidrule(lr){2-3} \cmidrule(lr){4-5} \cmidrule(l){6-8}
    & $\I_n$       & \ $\I_\infty$   & \ $\II_1$    & \ $\II_\infty$ & \ $\III_0$ & \ $\III_\lambda$ & \ $\III_1$ \\ 
     \midrule
    one-shot entanglement           & \!\!\! $\le$$\log_2n$\!    & \!$<$$\infty$    & $\oo$&$\oo$ & $\oo$&$\oo$&$\oo$                                \\
    all pure states LOCC equivalent & \no&\no     &\no&\no      & \yes&\yes&\yes         \\ 
    maximally entangled state       & \yes &    \no          & \yes        & \no             & \yes&\yes&\yes          \\
    best embezzlement capability $\kappa_{\textit{min}}$  &  2 & 2    & 2&2       & 2 or 0 & 0 & 0 \\ 
    worst embezzlement capability $\kappa_{\textit{max}}$ &  2 & 2    & 2&2       & 2 & \!\!$f(\lambda)$\!\! & 0 \\ 
    \bottomrule
    \end{tabular}
    \caption{One-to-one correspondence between entanglement properties and the minimal type. 
    In the table, we have $n\in\NN$, $\lambda\in(0,1)$, and $f(\lambda) = 2\frac{1-\sqrt\lambda}{1+\sqrt\lambda}$, which is a bijection of $[0,1]$ onto $[0,2]$.
    Some type $\III_0$ systems host embezzling states ($\kappa_{\min}=0$); the others maximally fail doing so ($\kappa_{\min}=2$).  
    }
    \label{tab:types}
\end{table*}

In order to distinguish these, we use the phenomenon of \emph{embezzlement of entanglement}, discovered by van Dam and Hayden in \cite{van_dam_universal_2003}.
Consider two agents, Alice and Bob, sharing an entangled `resource' state $\Phi_{AB}$, and a state $\Omega_{A'B'}$ with little entanglement, which they want to transform into a more entangled state $\Psi_{A'B'}$.
However, they are only allowed to act locally, may not use classical communication, and the resource system has to return to its original state.
I.e., they want to implement the state transition
\begin{equation}\label{eq:mbz_i}
    \Phi_{AB} \ox \Omega_{A'B'} \to \Phi_{AB} \ox \Psi_{A'B'}
\end{equation}
with local unitaries, at least up to arbitrary precision.
As stated, this is, of course, impossible with finite-dimensional resource systems. In fact, it is still impossible if the resource system is a von Neumann algebraic bipartite system of types $\I_\oo$ or $\II$.
In the type $\III$ case, however, embezzlement of entanglement becomes possible.
To quantify this phenomenon, we consider the worst error up to which the resource state $\Phi_{AB}$ allows embezzlement of \emph{arbitrary} entangled states $\Psi_{A'B'}$ given an arbitrary initial state $\Omega_{A'B'}$.
This quantity, denoted $\kappa(\Phi_{AB})$, is zero if and only if arbitrary entangled states can be embezzled from $\Phi_{AB}$. We refer to such states as \emph{embezzling states}.
For resource states in bipartite systems not of type $\III$, $\kappa(\Phi_{AB})$ attains its maximal value.
To discuss system properties instead of state properties, we consider the optimal values
\begin{equation}
    \kappa_{\min}(\M_A,\M_B) := \inf_{\Phi_{AB}} \ \kappa(\Phi_{AB}), \qquad \kappa_{\max}(\M_A,\M_B) := \sup_{\Phi_{AB}} \ \kappa(\Phi_{AB}),
\end{equation}
where the optimization runs over pure states of the full system.
For bipartite systems of type $\III_\lambda$ factors, $0\le \lambda\le 1$, we have
\begin{equation}\label{eq:kappa-max_i}
    \kappa_{\max}(\M_A,\M_B) = 2\;\!\frac{1-\lambda^{1/2}}{1+\lambda^{1/2}}.
\end{equation}
The RHS of \eqref{eq:kappa-max_i} is an invertible continuous function of $\lambda$. 
Its inverse allows us to recover the subtype from the operational quantifier $\kappa_{\max}$.
This completes the one-to-one correspondence mentioned in \cref{thm:minimal-type_i}, which is summarized in \cref{tab:types}, which, in addition, contains the values of $\kappa_{\min}$ showing that embezzling states exist for all subtypes, although they fail to exist in some type $\III_0$ cases.

\paragraph{\it Strong forms of infinite entanglement.}
One of the motivating questions for this work was whether there are interesting operationally distinct forms of infinite entanglement.
We have seen above that all pure states in bipartite systems of types $\II$ and $\III$ have infinite one-shot entanglement (this was first observed in \cite{keyl_infinitely_2003}).
But that does not yet mean that all states in these systems have the same entanglement properties. 
The latter would essentially require all bipartite pure states to be equivalent up to local unitaries (LU), at least up to arbitrary precision.
We call a bipartite system where this is the case \emph{LU transitive}.
These are precisely the type $\III_1$ systems:

\begin{introtheorem}[see \cref{sec:universal-mbz}]\label{thm:univ-mbz_i}
    The following are equivalent:
    \begin{enumerate}[(a)]
        \item the bipartite system is LU transitive;
        \item the bipartite system is a \emph{universal embezzler}: every pure state of the full system is embezzling (equivalently,  $\kappa_{\max}(\M_A,\M_B)=0$);
        \item $\M_A$, and therefore $\M_B$, are of type $\III_1$.
    \end{enumerate}
\end{introtheorem}

Examples of these highly entangled systems can be found in well-studied physical models.
Indeed, the observable algebras in conformal field theories and relativistic quantum field theories are generally of type $\III_1$ \cite{haag_local_1996,gabbiani_operator_1993}.
Moreover, we expect the ground state sectors of critical many-body systems in 1D to yield type $\III_1$ bipartite systems in the left-right bipartition.
This has been shown rigorously for the XY model \cite{keyl_entanglement_2006}, the transverse-field Ising model \cite{van_luijk_critical_2025}, and critical free fermion chains \cite{van_luijk_critical_2025}.

Additionally, we show the (mathematical) existence of multipartite analogs of these systems.
This is remarkable given how hard it is to say anything substantial about multipartite entanglement.
To do this, we can take the multipartite embezzling family of Leung, Toner, and Watrous and form the restricted infinite Hilbert space tensor product.
For any number $N\ge 2$ of parties, this yields a collection $(\M_x)_{x=1}^N$ of pairwise commuting type $\III_1$ factors, which fulfil Haag duality in every bipartition, and satisfy the $N$-partite analog of LU transitivity.
We show that LU transitivity implies that the system is a \emph{multipartite universal embezzler}: For every pure state $\Phi\in \H$ and every pair $\Omega,\Psi\in (\CC^d)^{\ox N}$ of finite-dimensional $N$-partite pure states, the state transition $\Phi\ox\Omega\to \Phi\ox\Psi$ can be realized with local unitaries up to arbitrary precision.
In contrast to the bipartite case, we do not know of any physical realizations of multipartite universal embezzlers.

\paragraph{\it Discussion and outlook}

The results explained above allow us to fully characterize mathematical concepts such as the definition of von Neumann algebras (see \cref{thm:axioms_i}), Haag duality (see \cref{thm:haag-duality_i}), or the type classification of factors (see \cref{thm:minimal-type_i}) in terms of information-theoretic ones, and vice versa.
Since the former are known for the observable algebras in various models, we discover previously unknown operational properties of well-studied models.

For instance, we learn that the entanglement properties of the ground states of, say, the toric code \cite{kitaevFaulttolerantQuantumComputation2003} and a Levin-Wen model \cite{levinStringnetCondensationPhysical2005} with nontrivial quantum dimensions, are qualitatively different.
Indeed, if we partition the lattice into a cone and its complement, the former has type $\II_\oo$ \cite{ogata_type_2024} and, for a suitable fusion category, the latter has $\III_1$ \cite{jones_local_2025}.
Hence, embezzlement of entanglement is possible if the resource system is a suitable Levin-Wen model but impossible if it is the toric code.

It was recently argued that observable algebras in semiclassical quantum gravity can be of type $\II$, in contrast to the type $\III_1$ algebras ubiquitous in relativistic quantum field theory \cite{chandrasekaran_algebra_2023} (see also \cite{jensen_generalized_2023,fewster_quantum_2024}). 
With our results, we can interpret this type reduction from an operational entanglement point of view.
Again, the phenomenon of embezzlement of entanglement distinguishes the two cases.

It would be fascinating to further explore the correspondence between operational properties on the one hand and von Neumann algebraic properties on the other.
Concrete open problems in this direction are mentioned at several points throughout this work.

Moreover, it would be great to have a better understanding of the von Neumann algebraic properties in concrete models (see \cref{sec:many-body}).
Especially, Haag duality is only known to hold for a handful of cases.
For many important models, e.g., the Heisenberg antiferromagnet, neither the type nor Haag duality are known.

\section{Preliminaries on von Neumann algebras}\label{sec:vN-prelims}

\localtableofcontents

\null

This section collects the mathematical prerequisites for this work.
We give an overview of a substantial part of the theory of von Neumann algebras, based on Refs.~\cite{takesaki1,takesaki2,hiai_lectures_2021,haagerup_equivalence_1990}.
This section can be skimmed through on a first read, as we will refer back to it in the later sections.
In particular, only the basic definitions are necessary to understand \cref{sec:vNQI}, where the setup of von Neumann algebraic quantum information theory is explained.
With few exceptions, this section is kept purely mathematical. Physical relevance and interpretation are discussed in the later sections, in particular, in \cref{sec:vNQI}.

Let us begin with a historic remark.
The theory of von Neumann algebras was initiated by von Neumann together with Murray in a series of papers \cite{rings_of_operators1,von_neumann_certain,rings_of_operators2,rings_of_operators3,rings_of_operators4,von_neumann_algebraical}.
Von Neumann's motivation was rooted in applications to quantum mechanics, whose mathematical formulation he had developed earlier \cite{von_neumann_mathematische_1932}. 
He was dissatisfied with the Hilbert space-based formulation of quantum mechanics because it was not singled out among the more general descriptions where a quantum system is described by what we now call a von Neumann algebra \cite{redei}.
Indeed, it is known today that this more general mathematical description of quantum systems becomes necessary for describing systems with infinitely many degrees of freedom, like quantum field theories \cite{haag_local_1996} or quantum many-body systems in the thermodynamic limit \cite{naaijkensQuantumSpinSystems2017}.  
This is precisely why we care about von Neumann algebras in this work.
Von Neumann algebras became one of the most active research areas of pure mathematics in the late 20th century.
Perhaps, the biggest breakthrough in the theory was the discovery of Tomita-Takesaki modular theory in the 1970s \cite{takesaki_tomitas_1970}.


\subsection{The basics}\label{sec:basics}

\subsubsection{Hilbert spaces and bounded operators}\label{sec:hilbert-spaces}

A Hilbert space $\H$ is a complex inner product space that is complete with respect to the induced norm $\norm\Psi = \ip\Psi\Psi^{1/2}$.
All Hilbert spaces in this work are assumed to be separable.
Assuming basic familiarity with Hilbert spaces, we only recall a few essentials to fix notation.
A linear operator $a:\H\to\K$ between Hilbert spaces is continuous if and only if it is bounded in the sense that $\norm{a\Psi}\le C\norm\Psi$ for some $C>0$ and all $\Psi\in\H$. 
The optimal constant in this estimate is the operator norm $\norm a$.
Perhaps the most important fact about Hilbert spaces is the Riesz representation theorem, asserting that $\Psi \mapsto \bra \Psi :=\ip\Psi{\placeholder}$ is an anti-isomorphism between a Hilbert space $\H$ and its dual $\H^*$.
As a consequence, the adjoint of a bounded operator $a:\H\to\K$ is defined as the operator $a^*:\K\to\H$ with
\begin{equation*}
    \ip{\Psi}{a\Phi}=\ip{a^*\Psi}\Phi, \qquad \Psi\in\H,\ \Phi\in\K.
\end{equation*}
A shorter way to write this formula is $ \bra\Psi a = \bra{a^*\Psi}$.
The adjoint is anti-linear and satisfies $\norm{a^*}=\norm a$, $(ab)^*=b^*a^*$.
An operator $a$ is called hermitian if $a^*=a$, a projection if $a=a^*a$, an isometry if $a^*a=1$, and unitary if $a$ and $a^*$ are both isometries. 
Moreover, an operator $a$ is positive, written $a\ge0$, if $a=a^*$ and $\ip\Psi{a\Psi}\ge0$, $\Psi\in\H$, and an operator inequality $a\ge b$ is defined via $a-b\ge0$.

Although we use the `bra' symbol $\bra\Psi$ from Dirac's bra-ket notation for the canonical covector associated with a vector $\Psi\in\H$, we do not use `ket' to denote general vectors.
Instead, we follow the convention of \cite{rfw_lecture} where the symbol $\ket \Psi$, $\Psi\in\H$, is used to denote the operator $\CC\ni\lambda\mapsto \lambda\Psi\in\H$.
This way $\ket \Psi$ and $\bra \Psi$ are adjoints of each other and the operator product $\bra\Psi\ket\Phi$ is simply the scalar $\ip\Psi\Phi$ whereas the reversed product $\ketbra\Psi\Phi$ is the rank 1 operator $\ip\Phi{\placeholder}\Psi$.
In particular, if $\Psi$ is a unit vector, then the orthogonal projection onto its span is $\kettbra\Psi$.
Moreover, we write $\ket j$, $j=0,1,\ldots$, to denote vectors of an orthonormal basis (ONB).

Given a Hilbert space $\H$, we let $\B(\H)$ denote the space of bounded operators on $\H$.
The product, adjoint, and operator norm turn $\B(\H)$ into a unital Banach *-algebra.
The unit is, of course, the identity operator on $\H$, which we simply denote by $1$ or $1_\H$ if we want to emphasize the Hilbert space.
A particularly important property of $\B(\H)$ is the monotone convergence theorem:
A uniformly bounded increasing net $(a_\alpha)$ of hermitian elements has a unique least upper bound $a$ to which it converges $\sigma$-strongly (see below). 

The trace on $\B(\H)$ is denoted $\tr$.
It is initially defined as an $\bar\RR^+$-valued functional on the cone $\B(\H)^+$ of positive bounded operators, where $\bar\RR^+=[0,\oo]$.
The trace class $\T(\H)$ is the subspace of $\B(\H)$ of operators on which the trace norm $\norm\rho_1=\tr|\rho|$ is finite, where $|\rho|=(\rho^*\rho)^{1/2}$.
This norm turns $\T(\H)$ into a Banach space on which the trace extends linearly to a continuous linear functional.
Moreover, $\T(\H)$ is a two-sided ideal in $\B(\H)$, and we have $\tr(a\rho)=\tr(\rho a)$ for $a\in\B(\H)$, $\rho\in\T(\H)$.
In fact, the bilinear map $(a,\rho)\mapsto \tr(a\rho)$ induces an isomorphism
\begin{equation*}
    \T(\H)^* \cong \B(\H),
\end{equation*}
which identifies $\T(\H)$ as the pre-dual $\predualB(\H)$ of $\B(\H)$ \cite[Sec.~II.1]{takesaki1}.
Positive operators with unit trace on $\H$ are called \emph{density operators}.
They span the trace class and, in case the Hilbert space $\H$ describes a quantum system, they describe the (mixed) states of this system (see \cref{sec:vNQI}).

Next, we recall locally convex operator topologies on $\B(\H)$.
We refer to \cite[Sec.~II.1]{takesaki1} for a detailed treatment.
The \emph{weak} operator topology on $\B(\H)$ is the initial topology of the functionals $a \mapsto \ip\Psi{a\Phi}$, $\Psi,\Phi\in\H$, and the \emph{strong} operator topology is the topology induced by the maps $a\mapsto a\Psi\in\H$.
The \emph{$\sigma$-weak} operator topology is the initial topology of the maps $a\mapsto \phi(a)$, $\phi\in\predualB(\H)$, whereas the \emph{$\sigma$-strong} topology is the initial topology of the maps $a\mapsto \phi(a^*a)$.
The \emph{{($\sigma\text-$\nobreak)}\allowbreak strong*} operator topology is defined like the ($\sigma$-)strong operator topology, but the initial topology is formed also with respect to the maps where $a^*$ is used instead of $a$.
Each of these topologies makes $\B(\H)$ a locally-convex topological vector space.
The weak, strong, and strong* topologies coincide with their respective $\sigma$-variants on norm-bounded subsets of $\B(\H)$.
The adjoint operation is continuous for the weak, $\sigma$-weak, strong*, and $\sigma$-strong* topologies, but discontinuous for the strong and $\sigma$-strong topologies.

When discussing linear maps on $\B(\H)$, we need the \emph{point-$\sigma$-weak topology}, i.e., the topology of pointwise convergence in the $\sigma$-weak operator topology, where a net of bounded linear maps $(T_\alpha)$ converges to a linear map $T$ if and only if $T_\alpha(a)$ converges to $T(a)$ in the $\sigma$-weak operator topology. 
By definition of the $\sigma$-weak operator topology, this is equivalent to $\lim_\alpha \phi(T_\alpha(a)) =\phi(T(a))$ for all $a\in\B(\H)$, $\phi\in\predualB(\H)$.

\subsubsection{Von Neumann algebras}\label{sec:von-neumann-algs}

An analogy often used to describe von Neumann algebras is that they represent "noncommutative measure spaces" \cite{connes_noncommutative_1994}. More specifically, the idea is that a von Neumann algebra is the $L^\oo$-space of a noncommutative measure space.
This metaphor is not merely due to the fact that abelian von Neumann algebras are isomorphic to actual $L^\oo$-spaces over measure spaces, but summarizes the general observation that concepts in measure theory transfer to the noncommutative setting. 
We will encounter many examples of this in the following, including monotone convergence theorems, $L^p$-spaces, and Radon-Nikodym derivatives.

A von Neumann algebra on a Hilbert space $\H$ is a unital *-algebra $\M$ of bounded operators on $\H$ that is closed in the strong operator topology.
For instance, $\B(\H)$ is a von Neumann algebra.
If $\R$ is a set of bounded operators that is closed under adjoints, then the \emph{commutant} 
\begin{equation}
    \R' = \{ x\in\B(\H) : [a,x]=0\ \forall a\in\M\}
\end{equation}
is a von Neumann algebra, where $[x,a]:=xa-ax$.
In fact, every von Neumann algebra arises this way. 
This is a consequent of the \emph{bicommutant theorem} asserting that a *-algebra $\M$ is a von Neumann algebra if and only if it is its own bicommutant
\begin{equation}\label{eq:bicommutant-thm}
    \M = \M''.
\end{equation}
The bicommutant theorem tells us that von Neumann algebras always come in pairs $(\M,\M')$.
Another consequence is that a von Neumann algebra contains the spectral projections and the functional calculi of each of its hermitian elements.
Indeed, this follows from \eqref{eq:bicommutant-thm} because the spectral projections and functional calculi of an operator $a=a^*\in\M$ commute with all operators that commute with $a$.

Von Neumann algebras are monotonely closed: The limit of a uniformly bounded increasing net $(a_\alpha)$ of hermitian elements in $\M$ is in $\M$.
Moreover, von Neumann algebras are closed in each of the operator topologies mentioned in the previous subsection.
Each of these topologies is defined with respect to the Hilbert space $\H$. However, the norm, $\sigma$-weak, $\sigma$-strong, and $\sigma$-strong* topologies only depend on the algebraic structure of the von Neumann algebra:
If $\M_j\subset\B(\H_j)$ are von Neumann algebras and if $\pi:\M_1\to\M_2$ is a *-isomorphism, i.e., a linear bijection preserving products and the adjoint operation, then $\pi$ is bicontinuous for the norm, $\sigma$-weak, $\sigma$-strong and $\sigma$-strong* topologies.
Importantly, the positive cone $\M^+$ of a von Neumann algebra is an algebraic property as well. 
Indeed, an operator $a\in \M$ is positive if and only if $a=b^*b$ for some $b\in\M$.

The above definition takes von Neumann algebras as algebras of operators on Hilbert spaces.
In the previous paragraph, we have seen that the most relevant properties of von Neumann algebras depend only on the algebraic structure but not on the Hilbert space on which they act.
In fact, there is an equivalent approach in which von Neumann algebras are defined abstractly:
An \emph{abstract} von Neumann algebra (or W*-algebra) $\M$ is a Banach *-algebra, which satisfies the C*-property $\norm{a^*a}=\norm a^2$ and is itself the dual space of a Banach space $\M_*$, known as the predual \cite[Sec.~III.3]{takesaki1}.
Every abstract von Neumann algebra has a faithful representation, i.e., is *-isomorphic to a von Neumann algebra on a Hilbert space.
The predual of a concrete von Neumann algebra $\M\subset\B(\H)$ is the space $\M_*$ of $\sigma$-weakly continuous linear functionals on $\M$.
Conversely, the $\sigma$-weak operator topology on $\M$ is the weak*-topology induced by the predual.
In addition to the $\sigma$-weak topology, the $\sigma$-strong and $\sigma$-strong* topologies do not depend on the given representation and can be defined for abstract von Neumann algebras.
We will only consider abstract von Neumann algebras with separable preduals.
These are exactly the abstract von Neumann algebras that can be faithfully represented on a separable Hilbert space.

The elements of the predual $\M_*$ are known as \emph{normal} linear functionals.
Among general linear functionals, they are characterized as those that are compatible with monotone limits, i.e., if $(a_\alpha)$ is a bounded increasing net in $\M$ and $\phi\in \M_*$ then $\lim_\alpha \phi(a_\alpha)=\phi(\lim_\alpha a_\alpha)$.
A (normal) \emph{state} on a von Neumann algebra is a (normal) positive linear functional $\phi$ on $\M$ that is normalized $\phi(1)=1$.
If a von Neumann algebra acts on a Hilbert space $\H$, then every unit vector $\Phi\in\H$, induces a normal state on $\M$ via
\begin{equation}\label{eq:vector state}
    \phi(a) = \ip\Phi{a\Phi},\qquad a\in\M.
\end{equation}
The states form a convex set whose extremal points are called \emph{pure states}. 
In general, however, normal pure states may fail to exist.
We denote the normal state space of $\M$ by $\nstates(\M)$.
In the analogy with measure spaces, normal states correspond to absolutely continuous probability measures.

There are a few operations that combine von Neumann algebras into new ones. 
If $\M$ and $\N$ are von Neumann algebras on a Hilbert space $\H$, their intersection $\M\cap\N$ is again a von Neumann algebra.
Moreover, there is a unique smallest von Neumann algebra $\M\vee\N$ that contains both $\M$ and $\N$. It is given by $\M\vee\N=(\M\cup\N)''$.
Together with the commutant, these operations satisfy de Morgan's rule
\begin{equation}
    (\M\vee \N)' = \M'\cap \N'.
\end{equation}
However, the distributivity law fails \cite{rings_of_operators1}. I.e., in general $\M\cap (\N_1\vee \N_2) \ne (\M\cap\N_1)\vee(\M\cap\N_2)$.
Now suppose that $\M$ and $\N$ are von Neumann algebras on Hilbert spaces $\H$ and $\K$, respectively.
The \emph{direct sum} $\M\oplus\N=\{x\oplus y : x\in\M,y\in\N\}$ is a von Neumann algebra on $\H\oplus\K$ whose commutant is $\M'\oplus\N'$.
The more interesting cousin of the direct sum is the \emph{tensor product} $\M\barox\N =\{ x\ox y : x\in\M,y\in\N\}''$ acting on $\H\ox\K$ with commutant $(\M\barox\N)'=\M'\barox\N'$.
Direct sums can be generalized to so-called \emph{direct integrals} $\int^\oplus\M_x\,d\mu(x)$, where the role of the index set is now played by a measure space $(X,\mu)$ and the map $x\mapsto \M_x$ must satisfy certain measurability criteria (see \cite[Sec.~IV.8]{takesaki1} for details).
These operations can be defined on the level of abstract von Neumann algebras \cite{takesaki1}.

\subsubsection{Classification of von Neumann algebras}

We now sketch the basic structure theory of von Neumann algebras.
The \emph{center} of a von Neumann algebra $\M$ is 
\begin{equation}
    Z(\M) = \{z\in \M : [z,x]=0\ \forall x\in\M\} = \M\cap \M'.
\end{equation}
As mentioned above, abelian von Neumann algebras are *-isomorphic with $L^\oo$-spaces.
A von Neumann algebra with trivial center $Z(\M)=\CC1$ is called a \emph{factor}.
Factors are the building blocks of von Neumann algebras:
If $(X,\mu)$ is a standard measure space such that $Z(\M) \cong L^\oo(X,\mu)$, there are measurable fields of factors $\M_x\subset \B(\H_x)$, $x\in X$, such that 
\begin{equation}\label{eq:direct integral}
    \H = \int^\oplus_X\H_x\,d\mu(x),\qquad \M = \int_X^\oplus \M_x \,d\mu(x),
\end{equation}
see \cite[Sec.~IV.8]{takesaki1} for details. 
This decomposition of $(\M,\H)$ is unique up to unitary isomorphism and the usual degeneracies arising from null sets. 
Moreover, the decomposition of $\M$ into factors is intrinsic, i.e., independent of the Hilbert space $\H$ that $\M$ acts on.
This procedure, known as \emph{disintegration}, reduces the structure theory of von Neumann algebras to that of factors.

Before discussing the classification of factors, we briefly collect a few facts about the set $\proj(\M)$ of projections in a von Neumann algebra $\M$.
It follows from the bicommutant theorem that the projection $[V]$ onto a closed subspace $V\subset\H$ is in $\M$ if and only if $V$ is an invariant subspace for $\M'$:
\begin{equation}\label{eq:projections in M}
    [V]\in\M \iff \M'V \subset V.
\end{equation}
For every projection $p\in\proj(\M)$, a von Neumann algebra also contains the complementary projection $p^\perp=1-p$, and we say that a projection $q$ is orthogonal to $p$ if $p+q$ is a projection or, equivalently, if $q\le p^\perp$.
The classification rests on the so-called \emph{Murray-von Neumann comparison theory} of projections: Projections $p$ and $q$ are said to be equivalent, written $p\sim q$, if there exists an operator $v\in\M$ such that $p=v^*v$ and $q=vv^*$, and one says that $p$ dominates $q$, written $q\preceq p$, if $q\le p'\sim p$ for some $p'$.
An operator $v$ such that $vv^*$, $v^*v$ are projections is called a \emph{partial isometry} and the set of partial isometries in $\M$ is denoted $\V(\M)$.
A projection $p$ is \emph{minimal} if $0\ne q\le p \implies q=p$, and it is \emph{finite} if $p\sim q \le p \implies p=q$.
A von Neumann algebra is said to be finite if the unit $1\in\M$ is a finite projection.

The type classification of factors goes as follows \cite{takesaki1} (see \cref{fig:types} for a visualization):
\begin{itemize}
    \item Type $\I$: The factor contains a nonzero minimal projection. Type $\I$ factors are further sub-classified into:
    \begin{itemize}
        \item Type $\I_n$, $n\in\NN\cup\{\oo\}$:  The maximal number of mutually nonzero orthogonal projections is $n$.
        Up to isomorphism, every type $\I_n$ factor is isomorphic to $\B(\H)$ for an $n$-dimensional Hilbert space $\H$.
    \end{itemize}
    \item Type $\II$: The factor contains nonzero finite projections but no nonzero minimal projections. Type $\II$ factors are sub-classified into:
    \begin{itemize}
        \item Type $\II_{1}$\hspace{.55pt} : The unit $1\in\M$ is a finite projection.
        \item Type $\II_\oo$: The unit $1\in\M$ is an infinite, i.e., non-finite, projection.
    \end{itemize}
    \item Type $\III$: The factor contains no nonzero finite projections. Equivalently, $p\sim q$ for all pairs of nonzero projections $p,q\in\M$. Type $\III$ factors are further sub-classified into:
    \begin{itemize}
        \item Type $\III_\lambda$, $0\le\lambda\le1$: The number $\lambda$ is given by 
        \begin{align}
            &\lambda=\left(\frac{2-\diam \nstates(\M)/_\sim}{2+\diam \nstates(\M)/_\sim}\right)^2,\nonumber\\
        \intertext{where} 
            \diam\,& \nstates(\M)/_\sim \ := \adjustlimits\sup_{\psi,\phi\in\nstates(\M)} \inf_{u\in\U(\M)} \norm{\psi-u\phi u^*}
        \label{eq:state space diam}
        \end{align}
        is the \emph{diameter of the state space} modulo unitary equivalence \cite{connes_diameters_1985,haagerup_equivalence_1990} with $\U(\M)$ denoting the group of unitaries in $\M$.
    \end{itemize}    
\end{itemize}
Every factor has exactly one of the specified types and exactly one of the specified subtypes.
A von Neumann algebra $\M$ is said to be of type $\rm X$, if it is a direct integral of type $\rm X$ factors (see \eqref{eq:direct integral}).
The definition of type $\III_\lambda$ factors in terms of the diameter of the state space is taken from \cite{connes_diameters_1985}, where it is shown that $\diam\nstates(\M)/_\sim \, = 2 \frac{1-\lambda^{1/2}}{1+\lambda^{1/2}}$. 
While this is very different from the usual definitions found in textbooks, which we will discuss later on, it has the benefit of being explicitly stated without introducing modular theory.
Moreover, the diameter of the state space will be important in \cref{sec:catalytic states,sec:strong-entanglement}, where we will show that it has an operational meaning in quantum entanglement theory.

The classification of factors into type and subtype is almost a complete classification of \emph{approximately finite-dimensional} (AFD) factors.
A von Neumann algebra is AFD if it is generated by an increasing sequence $\M_1\subset\M_2\subset \ldots\subset\M$ of finite-dimensional subalgebras, i.e., $\M = (\cup_n\M_n)''$.
Apart from type $\III_0$ factors, the isomorphism classes of AFD factors are in one-to-one correspondence with type and subtype \cite[Ch.~XVIII]{takesaki3}.
For approximately finite-dimensional type $\III_0$ factors, a complete invariant is given by the flow of weights (see \cref{sec:spectral states}).
To the best of our knowledge, all von Neumann algebras that appear as observable algebras in models of physical systems are AFD.\footnote{This is particularly relevant in the context of quantum entanglement theory since only non-AFD algebras can give rise to counterexamples to Tsirelson's problem, see \cref{sec:notes-vNQI}.}

\subsubsection{Weights and states}\label{sec:weights and states}

A \emph{normal semifinite weight} on a von Neumann $\M$ is an affine $\sigma$-weakly lower semicontinuous map
\begin{equation*}
    \omega : \M^+ \to \bar\RR^+
\end{equation*}
such that the left-ideal $\mathfrak n_\omega:=\{ a\in \M : \omega(a^*a)<\oo\}$ is $\sigma$-weakly dense.
In the analogy with measure spaces, weights are noncommutative generalizations of $\sigma$-finite measures.
A weight is called \emph{faithful} if $\omega(a^*a)=0\implies a=0$ and \emph{finite} if $\omega(a^*a)<\oo$ for all $a\in\M$.
A normal semifinite weight extends to a densely-defined linear functional with domain $\mathfrak m_\omega:=\lin\{ a^*b : a,b\in\mathfrak n_\omega\}$.
Note that normal states are precisely normal semifinite weights with $\omega(1)=1$.
Like normal states, normal semifinite weights respect monotone limits, which can be seen as a generalization of the monotone convergence theorem.
If $a\in\M$, then $a\omega a^*$ denotes the normal semifinite weight
\begin{equation}
    a\omega a^*(b) = \omega(a^*ba), \qquad b\in\M^+.
\end{equation}
The \emph{support projection} $\supp(\omega)$ of a normal semifinite weight $\omega$ is defined as the smallest projection $p$ such that $\omega=p\omega p$.%
\footnote{For every collection $\{p_\alpha\}$ of projections in $\M$, there exist a unique greatest lower bound $p=:\wedge_\alpha p_\alpha \in\proj(\M)$, which satisfies $q\le p\le p_\alpha$ for all $\alpha$ and all projections that satisfy $q\le p_\alpha$ for all $\alpha$.
Similarly, there is a least upper bound $\vee_\alpha p_\alpha $ and the two are related by $(\vee_\alpha p)^{\perp} = \wedge_{\alpha} p_\alpha^\perp$.
}
If $\omega$ is finite, the definition is equivalent to defining $\supp(\omega)$ as the smallest projection such that $\omega(p)=\omega(1)$.
If $\phi\in\M_*^+$ is implemented by a vector $\Phi\in\H$ (as in \eqref{eq:vector state}), the support projection is given by\footnote{
Since this fact will often be used in this work, we give a short proof: $[\M'\Phi]\in\M''=\M$ holds because  $\M'\Phi$ is an $\M'$-invariant subspace (see \eqref{eq:projections in M}).
The inequality $[\M'\Phi]\ge \supp(\phi)$ follows from $\phi([\M'\Phi])=1$.
If $p\in\proj(\M)$ also satisfies $\phi(p)=1$, then $\Phi=p\Phi\in p\H$, which implies $\M'\Phi\subset \M'p\H = p\M'\H \subset p\H$ and, hence, $p\ge [\M'\Phi]$.
}
\begin{equation}\label{eq:supp proj}
    \supp(\phi) = [\M'\Phi].
\end{equation}
Faithfulness of a normal semifinite weight $\omega$ is equivalent to $\supp(\omega)=1$. 
A vector $\Phi\in\H$ is called \emph{separating} for $\M$ if the normal positive linear functional $\phi$ that it implements is faithful.
Equivalently, if $\norm{a\Phi}=0\implies a=0$ for $a\in\M$.
A vector $\Phi$ is \emph{cyclic} for $\M$ if $\overline{\M\Phi} = \H$.
Being cyclic and being separating are dual notions:

\begin{lemma}[{\cite[Prop.~II.3.17]{takesaki1}}]\label{lem:cyclic separating}
    Let $\M$ be a von Neumann algebra on $\H$ and let $\Phi\in\H$.
    Then $\Phi$ is cyclic for $\M$ if and only if $\Phi$ is separating for $\M'$.
\end{lemma}

A particularly important role in the theory of von Neumann algebras is played by vectors that are both cyclic and separating. \Cref{lem:cyclic separating} implies that a vector is cyclic and separating relative to $\M$ if and only if it has the same property relative to $\M'$.

The von Neumann algebra acts on the predual by left and right multiplication.
If $\phi\in\M_*$, the left multiplication $a\phi$ by an operator $a\in\M$ is defined as
\begin{equation}
    a\phi (b) = \phi(ba),\qquad b\in\M.
\end{equation}
The right multiplication is defined similarly.
The left (resp.\ right) support projection $\supp_l(\phi)$ (resp.\ $\supp_r(\phi)$) of $\phi\in\M_*$ is the smallest projection $p$ such that $p\phi =\phi$ (resp.\ $\phi p=\phi$).
There is a \emph{polar decomposition} theorem for $\M_*$ (see \cite[Sec.~III.4]{takesaki1}), which states that every $\phi\in\M_*$ is of the form 
\begin{equation}\label{eq:polar decomposition for functionals}
    \phi = v\abs \phi
\end{equation}
for a unique pair $(v,\abs\phi)$ of a partial isometry $v\in\M$ and positive linear functional $\abs\phi\in\M_*^+$ with $v^*v=\supp(\abs\phi)=\supp_r(\phi)$ and $vv^*=\supp_l(\phi)$ and a positive functional $\abs\phi\in\M_*^+$.

\subsection{Semifinite von Neumann algebras}\label{sec:semifinite vNas}

A (normal semifinite faithful) \emph{trace} $\tau$ on a von Neumann algebra $\M$ is a normal semifinite faithful weight, which is \emph{tracial} in the sense that
\begin{equation}\label{eq:tracial}
    \tau(a^*a)=\tau(aa^*),\qquad a\in\M,
\end{equation}
where both sides may be infinite.
Equivalently, traciality can be defined as unitary invariance: $\tau=u\tau u^*$ for all unitaries $u\in\M$.



The theory of traces on factors is directly connected to the Murray-von Neumann comparison theory for projections, which can be used to show that there is a map $\tau : \proj(\M)\to \bar\RR^+$ such that
\begin{equation}\label{eq:mvn-ordering-and-trace}
    p\preceq q \quad \iff \quad \tau(p)\le \tau(q),
\end{equation}
which is unique up to positive scalar multiples.
If $\M$ is a type $\I$ factor, we fix the arbitrary scaling by setting $\tau(p)=1$ for some minimal projection $p$. For $\M=\B(\H)$ this implies $\tau(p) = \tr(p)= \dim(p\H)$.
If $\M$ is a type $\II_1$ factor, it is customary to fix the scaling by setting $\tau(1)=1$.
In type $\II_\oo$, there is no canonical way to choose the scaling. 
In type $\III$ factors, where all nonzero projections are equivalent, \eqref{eq:mvn-ordering-and-trace} necessarily implies that $\tau(p)=\oo$ for all nonzero projections.
With these conventions, the classification of type $\I$ and $\II$ factors can be phrased in terms of the set of possible dimensions $D= \{\tau(p) : p\in\proj(\M)\}$ \cite[Sec.~1.4]{hiai_lectures_2021}: 
\begin{itemize}
    \item $D = \{1,\ldots,n\}$ if and only if $\M$ has type $\I_n$, $n\in\NN\cup\{\oo\}$. 
    \item $D=[0,1]$ if and only if $\M$ has type $\II_1$,
    \item $D=\bar\RR^+$ if and only if $\M$ has type $\II_\oo$,
    \item $D=\{0,\oo\}$ if and only if $\M$ has type $\III$.
\end{itemize}
Using the spectral theorem, the map $\tau$ defined on projections can be extended to a \nsf trace, i.e., a normal semifinite faithful weight assigning the same value to $a^*a$ and to $aa^*$ for all $a\in\M$, on $\M$.
A factor is semifinite if and only if it is not of type $\III$, i.e., of type $\I$ or $\II$.
A factor is finite if and only if it is of type $\I_n$ with $n<\oo$ or of type $\II_1$.

A general von Neumann algebra is said to be \emph{(semi-)finite} if it admits a normal faithful tracial state, which is the case if and only if it is a direct integral of (semi-)finite factors.
In the case of an abelian von Neumann algebra $L^\oo(X,\mu)$, the possible \nsf traces correspond to the $\sigma$-finite measures $\nu\sim\mu$ via $\tau(f)=\int_X f\,d\nu$.

\subsubsection{\texorpdfstring{$L^p$}{Lp}-spaces}

As for measure spaces, one can define $L^p$-spaces over von Neumann algebras.
In the following, we will discuss the construction of $L^p$-spaces for semifinite von Neumann algebras, which, in particular, requires treating unbounded operators "in" von Neumann algebras. 

We fix a semifinite von Neumann algebra $\M$ with a \nsf trace $\tau$.
In order to define the $L^p$-spaces properly, we need to include unbounded operators.
A closed operator $a$ on a Hilbert space $\H$ that $\M$ acts on is \emph{affiliated with $\M$}, written 
\begin{equation*}
    a\aff\M,
\end{equation*}
if the spectral projections of the positive self-adjoint operator $\abs a=(a^*a)^{1/2}$ and the partial isometry $v$ appearing in the polar decomposition $a=v\abs a$ are elements of $\M$ \cite[Sec.~4.1]{hiai_lectures_2021}.
Equivalently, $a$ is affiliated with $\M$ if and only if $u'^*au'=a$ for all unitaries $u'\in\M'$.
An operator $a\aff\M$ is \emph{$\tau$-measurable} if $\lim_{t\to\oo} \tau(1_{(t,\oo)}(|a|)) \to 0$ \cite[Sec.~4.1]{hiai_lectures_2021}, where $1_A$ denotes the indicator function of a set $A$.
The set of $\tau$-measurable operators will be denoted 
\begin{equation*}
    L^0(\M,\tau)
\end{equation*}
Clearly, $a^*\aff\M$ if and only if $a\aff\M$, and the same holds for $\tau$-measurability.
If $a,b$ are $\tau$-measurable, then $a+b$ and $ab$ are closable and their closures are $\tau$-measurable, which turns $L^0(\M,\tau)$ into a *-algebra \cite[Sec.~4.1]{hiai_lectures_2021}.
The trace extends to positive operators $0\le a\in L^0(\M,\tau)$ via $\tau(a) = \lim_{t\to\oo} \tau(a 1_{[0,t)}(a))$.
The $L^p$-norm for $1\le p <\oo$ is defined on $L^0(\M,\tau)$ as
\begin{equation}
    \norm{a}_{L^p(\M,\tau)} = \tau(\abs a^p)^{1/p}  \in \bar\RR^+.
\end{equation}
The $L^p$-space $L^p(\M,\tau)$ is now simply the subspace of $L^0(\M,\tau)$ on which the $L^p$-norm is finite.
For $p=\oo$, we set $L^\oo(\M,\tau)=\M$.
The $L^p$-spaces enjoy the usual properties. For instance, they are complete, and Hölder's inequality holds, i.e., if $p^{-1}+q^{-1}=r^{-1}$ for $1\le p,q,r\le \oo$, then \cite[Prop.~4.43]{hiai_lectures_2021}
\begin{equation}
    \norm{ab}_{L^r(\M,\tau)}\le \norm{a}_{L^p(\M,\tau)}\norm{b}_{L^q(\M,\tau)},\qquad a,b\in L^0(\M,\tau).
\end{equation}
In particular, the $L^p$-spaces are closed under left and right multiplication with elements of $\M$.
The trace $\tau$ has a unique extension to a bounded linear function on $L^1(\M,\tau)$ such that $\abs{\tau(\rho)}\le \tau(\abs \rho)= \norm \rho_{L^1(\M,\tau)}$, $\rho\in L^1(\M,\tau)$.
For $\rho\in L^1(\M,\tau)$, the linear functional $\tau(\rho\;\!\placeholder):\M\to\CC$ is normal, i.e., an element of $\M_*$, and the map
\begin{equation}\label{eq:predual semifinite}
    L^1(\M,\tau) \ni \rho \mapsto \tau(\rho\;\!\placeholder) \in \M_*
\end{equation}
is an isomorphism of ordered Banach spaces.
Its inverse is the \emph{Radon-Nikodym derivative}:
\begin{equation}
    \phi = \tau(\rho_\phi(\placeholder)), \qquad \rho_\phi = \frac{d\phi}{d\tau}\in L^1(\M,\tau), \ \phi\in\M_*.
\end{equation}
This bijection between $\M_*^+$ and $L^1(\M,\tau)^+$ extends to weights: If $\phi$ is a normal semifinite weight, there is an operator $0\le \rho_\phi=d\phi/d\tau \aff \M$ such that $\phi(a) = \tau(\rho_\phi^{1/2}a \rho_\phi^{1/2})$, $a\ge0$.

Let us also mention that the $L^2$-space $L^2(\M,\tau)$ is a Hilbert space with the inner product
\begin{equation}
    \ip \xi\eta_{L^2(\M,\tau)} = \tau(\xi^*\eta), \qquad \xi,\eta\in L^2(\M,\tau),
\end{equation}
on which $\M$ acts by left multiplication.

We briefly discuss $L^p$-spaces for two basic examples of abelian von Neumann algebras and type $\I$ factors.
The $L^p$-spaces of an abelian von Neumann algebra $\M= L^\oo(X,\mu)$ relative to the trace $\tau = \int\placeholder\,d\mu$ are simply the standard $L^p$-spaces $L^p(X,\mu)$.
The $L^p$-spaces of a type $\I$ factor $\M = \B(\H)$ relative to the standard trace $\tr$ are the Schatten classes $L^p(\B(\H),\tr) = \T^p(\H)$ \cite[Ex.~4.34]{hiai_lectures_2021}.
In particular, $L^1(\B(\H),\tr)$ is the trace class $\T(\H)$.

\subsubsection{Spectral scales and majorization}\label{sec:spectral scales}

In semifinite factors, we can use the identification $\M_*\cong L^1(\M,\tau)$ to define spectral notions for normal states relative to a choice of \nsf trace $\tau$.
We begin by discussing spectral notions of $\tau$-measurable operators.

For positive operators $0\le a\in L^0(\M,\tau)$, the \emph{distribution function} $D_a$ and the \emph{spectral scale} $\lambda_a$ are given by
\begin{equation}
    \begin{aligned}
          D_a(t) &= \tau(1_{(t,\oo)}(a)),\\
    \lambda_a(t) &= \inf\{s>0: D_a(s)\le t\},
    \end{aligned}\qquad\qquad t>0.
\end{equation}
For general $a\in L^0(\M,\tau)$, $\mu_a(t)=\lambda_{\abs a}(t)$, $t\in\RR$, are called the generalized singular numbers of $a$ \cite[Sec.~4.2]{hiai_lectures_2021}.\footnote{These generalize the concept of singular values of a matrix in the same way that the spectral scales are a generalization of the eigenvalues of a matrix, see below.}
The functions $D_a,\lambda_a:(0,\oo)\to\RR^+$ are non-increasing and right-continuous.
If $f:\RR^+\to \RR^+$ is a continuous non-decreasing function with $f(0)=0$, then $0\le f(a)\in L^0(\M,\tau)$ and
\begin{equation}\label{eq:integration with spectral scale}
    \tau(f(a)) = \int_0^\oo f(\lambda_a(t))\,dt.
\end{equation}
For instance, this implies $\norm{a}_{L^p(\M,\tau)} = \norm{\lambda_{\abs a}}_{L^p(\RR)}$ and $\tau(1_A(a)) = \abs{\lambda_a^{-1}(A)}$, where $\abs{\placeholder}$ denotes the Lebesgue measure.%
\footnote{Thus, the push-forward measure of the Lebesgue measure $dt$ by the map $\lambda_a:\RR^+\to\RR^+$ is precisely the spectral measure of $a$, evaluated with the trace $\tau$.}
We consider the spectral scale in the two basic examples:
\begin{example}\label{exa:spectral scales}
    \begin{enumerate}
        \item 
            Let $\M = L^\oo(X,\mu)$ with $\tau=\int\placeholder\,d\mu$, and let $f\in L^1(X,\mu)^+$.
            Then $\lambda_f = f^\downarrow\in L^1(\RR^+)$ is the $\mu$-equimeasurable decreasing rearrangement of $f$ \cite[Rem.~2.3.1]{fack1986generalized}.
        \item 
            Let $\M = \B(\H)$ and $\tau=\tr$.
            Let $a$ be a positive trace class operator and let $(\alpha_n)$ be its eigenvalues (ordered non-increasingly and repeated according to their multiplicity).
            Then 
            \begin{equation}
                \lambda_a(t) = \sum_n \alpha_n\, 1_{[n-1,n)}(t),\qquad t>0.
            \end{equation}
    \end{enumerate}
\end{example}

Following \cite{hiai_nakamura_maj,hiai_majorization_1987}, \emph{majorization} $\rho\preceq \sigma$ is defined for $\tau$-measurable positive operations $\rho,\sigma\in L^0(\M,\tau)^+$ via majorization of the spectral scales $\lambda_\rho \preceq\lambda_\sigma$, i.e., $\rho\preceq \sigma$ if and only if 
\begin{equation}
    \int_0^t \lambda_\rho(s)\,ds\ge \int_0^t \lambda_\sigma(s)\,ds, \qquad t>0,
\end{equation}
with equality for $t=\oo$. The latter condition is equivalent to $\tau(\rho)=\tau(\sigma)$.
Note that the definition depends on the choice of \nsf trace $\tau$.
For factors, however, majorization is intrinsically defined because there is a unique \nsf trace up to scalar multiples.

Of the many things that can be said about majorization theory in von Neumann algebras (see \cite{petz1985scale,hiai_majorization_1987,hiai_nakamura_maj,hiai1989distance,fack1986generalized,hiai_closed_1991}), we only need the structure of majorization for states on semifinite factors.
Let $\M$ be a semifinite factor with \nsf trace $\tau$.
For normal states $\psi$ on $\M$, we define 
\begin{equation}
    \lambda_\psi:=\lambda_{\rho_\psi}, \quad D_\psi = D_{\rho_\psi}, \qquad \rho_\psi = \frac{d\psi}{d\tau}.
\end{equation}
The distribution function and spectral scale relative to a scaled trace $\tau'=c\tau$, $c>0$, are dilations $D_\psi'(t)=cD_\psi(ct)$ and $\lambda'_\psi(t) = c^{-1}\lambda_\psi(c^{-1}t)$.\footnote{
Indeed, $\rho'_\psi = d\psi/d\tau' = c^{-1} d\psi/d\tau= c^{-1}\rho_\psi$, $D'_\psi(t) = \tau'(1_{(t,\oo)}(\rho_\psi')= c\tau(1_{(t,\oo)}(c^{-1}\rho_\psi)) = c D_\psi(ct)$, and $\lambda_\psi'(t) = \inf\{s>0 : D'_\psi(s)\le t\} = \inf\{c^{-1}s>0: D_\psi(s)\le c^{-1}t\} = c^{-1} \lambda_\psi(c^{-1}t)$.}
If $\phi$ is another normal state on $\M$, we define majorization $\psi\preceq\phi$ via majorization of the density operators $\rho_\psi\preceq\rho_\phi$ or, equivalently, majorization of the spectral scales $\lambda_\psi\preceq\lambda_\phi$.
Note that the majorization of states is independent from the choice of \nsf trace.

\begin{theorem}[{\cite[Thm.~2.5]{hiai_majorization_1987}}]\label{thm:majorization}
    Let $\M$ be a semifinite factor with \nsf trace $\tau$.
    Let $\psi,\phi$ be normal states on $\M$.
    The following are equivalent:
    \begin{enumerate}[(a)]
        \item $\psi \succeq \phi$, i.e., $\lambda_\psi \succeq\lambda_\phi$; 
        \item $\tau(f(\rho_\psi))\ge \tau(f(\rho_\phi))$ for all continuous non-decreasing convex functions $f:\RR^+\to\RR^+$;
        \item $\psi \in \overline\conv\,\{\, u\phi u^* : u\in \U(\M)\,\}$.
    \end{enumerate}
    Moreover, the distance of the unitary orbits of $\psi$ and $\phi$ is given by:
    \begin{equation}
        \inf_{u\in\U(\M)}\, \norm{\psi-u\phi u^*} = \norm{\lambda_\psi-\lambda_\phi}_{L^1(\RR^+)} = \norm{D_\psi-D_\phi}_{L^1(\RR^+)}.
    \end{equation}
\end{theorem}

We conclude our discussion with a construction that shows that every possible spectral scale occurs as the spectral scale of a normal state.
Since this clear in type $\I$ (see \cref{exa:spectral scales}), we consider a type $\II$ factor $\M$ with a \nsf trace $\tau$.
Then we can find an increasing family of projections $e_t$, $0\le t \le \tau(1)$ with 
\begin{equation}\label{eq:typeII family of projections}
    \tau(e_t) = t, \qquad \vee_{s<t} e_s = e_t, \qquad e_{\tau(1)} =1
\end{equation}
(see \cite[Sec.~III.1.7]{blackadar_operator_2006} or \cite[Lem.~4.1]{haagerup_equivalence_1990} for a construction).
Now if $\lambda:(0,\oo)\to\RR^+$ is a non-increasing right-continuous function, the Lebesgue-Stieltjes integral 
\begin{equation}\label{eq:operator with given spectral scale}
    \rho= \int_0^\oo \lambda(t)\,de_t
\end{equation}
defines a positive operator $\rho\in L^0(\M,\tau)$ with $\lambda_a(t) = \lambda(t)$. 
In particular, if $\lambda$ is a probability density, then $\rho\in L^1(\M,\tau)$ and the normal state $\psi=\tr(\rho(\placeholder))$ has spectral scale $\lambda$.

\subsection{Modular theory}\label{sec:modular}

Modular theory was initiated by Tomita and Takesaki \cite{takesaki_tomitas_1970} and led to many breakthroughs in the theory of von Neumann algebras in the early 1970s.
Perhaps its most important application is the classification of type $\III$ factors, due to Connes \cite{connes_classification_1973}.
In the following, we give a brief exposition roughly following \cite{hiai_lectures_2021}.
For a complete account, we refer to  \cite{takesaki2} or \cite{stratila}.

Let $\M$ be a von Neumann algebra on $\H$ and let $\Omega$ be a cyclic separating vector.
The operator $S_\Omega x\Omega= x^*\Omega$ defines a closable anti-linear operator on $D(S_\Omega)=\M\Omega$.
The positive self-adjoint operator $\Delta_\Omega$ and the conjugation\footnote{\label{foot:conjugation}
A conjugation on a Hilbert space is an anti-unitary involution, i.e., an anti-linear operator $J:\H\to\H$ such that $J^2=1$ and $\ip{J\Psi}{J\Phi}=\overline{\ip\Psi\Phi}$ for all $\Psi,\Phi\in\H$.}
$J_\Omega$ such that $\bar S_\Omega = J_\Omega \Delta_\Omega^{1/2}$ is the polar decomposition of $\bar S_\Omega$ are called the \emph{modular operator} and the \emph{modular conjugation} associated with $\Omega$.
It follows that $J_\Omega\Omega =\Delta_\Omega\Omega=\Omega$.
Tomita's Theorem (see \cite[Sec.~2.1]{hiai_lectures_2021}) asserts that
\begin{equation}\label{eq:tomitas thm}
    J_\Omega \M J_\Omega = \M',\qandq \Delta_\Omega^{it} \M \Delta_{\Omega}^{-it}=\M,\quad t\in\RR.
\end{equation}
The one-parameter group $\sigma^\omega$ of automorphisms given by
\begin{equation}\label{eq:modular group}
    \sigma^\omega_t(a) = \Delta_\Omega^{it}a\Delta_\Omega^{-it}, \qquad a\in\M,\ t\in\RR,
\end{equation}
which only depends on the marginal state $\omega$ that $\Omega$ induces on $\M$, is the so-called \emph{modular automorphism group}.
The state $\omega$ is invariant under its own modular automorphism group, i.e., $\omega\circ \sigma^\omega_t=\omega$ for all $t\in\RR$ (this follows from $\Delta_\Omega\Omega=\Omega$).%
\footnote{The modular automorphism group is uniquely determined by the properties that (a) it is point-$\sigma$-weakly continuous (see \cref{sec:hilbert-spaces}), and (b) it satisfies the $\omega$-KMS condition at inverse temperature $\beta=-1$ \cite{takesaki2}.}

Connes showed that the modular automorphism groups of different states are cocycle equivalent (see \cite[Sec.~6.2]{hiai_lectures_2021}).
To understand this, we need a few definitions.
The modular flow is a one-parameter group of automorphisms that is point-$\sigma$-weakly continuous (see \cref{sec:hilbert-spaces}).
Equivalently, we can view it as an action of the group $\RR$ on the von Neumann algebra.\footnote{An action $\theta: G\acts\M$ on a von Neumann algebra $\M$ is a group homomorphism $G\ni \theta\mapsto \theta_g\in\Aut(\M)$.}
A \emph{cocycle} for an action $\theta:\RR\acts\M$ (a $\theta$-cocyle for short) is a $\sigma$-strongly continuous one-parameter family of unitaries $(u_t)$ such that
\begin{equation}\label{eq:cocycle}
    u_{t+s} = u_t \theta_t(u_s), \qquad t,s\in\RR.
\end{equation}
If $\varphi$ is another faithful state, then the $\RR$-actions $\sigma^\varphi$ and $\sigma^\omega$ are cocycle equivalent in the sense that there exists a $\sigma^\varphi$-cocycle $(u_t)$ such that
\begin{equation}\label{eq:connes cocyle}
    u_t\sigma^\varphi_t(a)u_t^* = \sigma^\omega_t(a), \qquad a\in\M,\ t\in\RR.
\end{equation}
The so-called \emph{Connes cocycle} $(D\omega:D\varphi)_t:=u_t$ is uniquely determined by these properties.
Therefore, the equivalence class $\delta_t:=[\sigma^\omega_t]$ in the outer automorphism group $\Out(\M) = \Aut(\M) / \Inn(\M)$ does not depend on $\omega$.
This defines, for every von Neumann algebra $\M$, a unique continuous group homomorphism, 
\begin{equation}\label{eq:modular hom}
    \delta_\M: \RR \to \Out\M,
\end{equation}
called the \emph{modular homomorphism}. It can be used to define the classification of type $\III$ factors (see \cite[Sec.~I.3.1]{stratila}).
We briefly look at the concepts above in a familiar example:

\begin{example}\label{exa:modular-thry-typeI}
    Consider the Hilbert space $\HS(\K)$ of Hilbert-Schmidt operators on a Hilbert space $\K$.
    Let $\M=\B(\K)$, which acts on $\HS(\K)$ via left multiplication.
    If $\rho$ is a faithful density operator on $\K$, then its square root $\rho^{1/2}$ is a unit vector in $\HS(\K)$, which is cyclic separating for $\M$.
    The resulting modular objects are
    \begin{equation}
        J_{\rho^{1/2}}\xi = \xi^*,\qquad \Delta_{\rho^{1/2}}\xi = \rho \xi \rho^{-1}, \qquad \xi\in \HS(\K).
    \end{equation}
    The modular automorphism group acts on $\M=\B(\K)$ as $\sigma^\rho_t(a) = \rho^{it}a\rho^{-it}$, where $\omega=\tr\rho(\placeholder)$.
    If $\sigma$ is another faithful density operator, then the Connes cocycle is $(D\sigma:D\rho)_t = \sigma^{it}\rho^{-it}$.
\end{example}

\subsubsection{The standard form}\label{sec:std form}

The theory of standard forms is due to Haagerup \cite{haagerup_standard_1975}.
The standard form of a von Neumann algebra is a canonical representation, which for type $\I$ factors $\M=\B(\H)$ reduces to the left multiplication on the Hilbert space of Hilbert-Schmidt operators (see \cref{exa:modular-thry-typeI}).


\begin{theorem}[Haagerup \cite{haagerup_standard_1975}\footnote{In \cite{haagerup_standard_1975}, uniqueness is shown assuming \ref{it:std form4}, but in \cite[Lem.~4.9]{ando_ultraproducts_2014}, \ref{it:std form4} is shown to follow from \ref{it:std form1} to \ref{it:std form3}.}]\label{thm:std form}
    For every abstract von Neumann algebra $\M$, there is a Hilbert space $\H$ carrying a faithful representation $\M\subset\B(\H)$, a self-dual positive cone $\P\subset\H$ and a conjugation $J$ such that:
    \begin{enumerate}[(i)]
        \item\label{it:std form1} $J\M J=\M'$,
        \item\label{it:std form2} $aJa\P \subset \P$ for all $a\in \M$,
        \item\label{it:std form3} $J\Omega = \Omega$ for all $\Omega\in\P$.
    \end{enumerate}
    These properties determine $(\H,J,\P)$ up to unitary equivalence.
    Moreover, the following hold:
    \begin{enumerate}[resume*]
        \item\label{it:std form4} $Jz J = z^*$ for all $z\in Z(\M)$.
        \item\label{it:std form5} For $\phi\in\M_*^+$, there exists a unique vector $\Phi\in\P$ such that $\phi(a)=\ip\Phi{a\Phi}$ for all $a\in\M$. 
        $\Phi$ is cyclic separating if and only if $\phi$ is faithful.
        If $\Psi,\Phi\in\P$ correspond to $\psi,\phi\in\M_*^+$ then
        \begin{equation}\label{eq:state-vector-estimate}
            \norm{\Psi-\Phi}^2 \le \norm{\psi-\phi}\le \norm{\Psi-\Phi}\norm{\Psi+\Phi}.
        \end{equation}
        \item\label{it:std form6} For every automorphism $\alpha$ on $\M$, there exists a unique unitary $u_\alpha$ such that $u_\alpha J = J u_\alpha$, $u_\alpha\P=\P$ and $u_\alpha a u_\alpha^* = \alpha(a)$.
        \item\label{it:std form7} If $\Omega\in\P$ is cyclic separating for $\M$, then the modular conjugation is $J_\Omega=J$.
    \end{enumerate}
\end{theorem}

A triple $(\H,J,\P)$ satisfying the specified properties is called a \emph{standard form} of $\M$.
We will later see that there is a canonical way to construct a standard form, which we refer to as \emph{the} standard form and which we denote
\begin{equation*}
    (L^2(\M),J,L^2(\M)^+).
\end{equation*}
Note that items \ref{it:std form1} to \ref{it:std form3} are invariant under swapping $\M$ and $\M'$. Thus, $(\H,J,\P)$ is a standard form of $\M$ if and only if it is a standard form of $\M'$.
A representation of $\M$ on a Hilbert space $\H$ is a \emph{standard representation} if there exists a self-dual positive cone $\P$ and a conjugation $J$ such that $(\H,J,\P)$ is a standard form of $\M$.

\begin{example}\label{exa:std form typeI}
    A standard form of $\B(\H)$ is given by the action of left multiplication on the Hilbert-Schmidt class $\HS(\H)$ with the modular conjugation given by the adjoint operation $J\xi=\xi^*$ and the self-dual positive cone is given by $\P=\HS(\H)^+$, the cone of positive Hilbert-Schmidt operators.
    The vector implementing a normal state $\rho$ is the square root $\rho^{1/2}$ (cp.\ \cref{exa:modular-thry-typeI}).
\end{example}

This generalizes to semifinite von Neumann algebras:

\begin{example}\label{exa:std form semifinite}
    A standard form of a semifinite von Neumann algebra $\M$ with \nsf trace $\tau$ is given by the action of left multiplication on $L^2(\M,\tau)$ with the modular conjugation given by the adjoint operation $J\xi =\xi^*$ and the self-dual positive cone $\P = L^2(\M,\tau)^+$.
    The vector implementing a normal state $\omega$ is the square root $\rho^{1/2}$ of the density operator $\rho=d\omega/d\tau\in L^1(\M,\tau)$.
\end{example}

\begin{lemma}\label{lem:std rep}
    A von Neumann algebra $\M$ on a Hilbert space $\H$ is in standard representation if and only if it admits a cyclic separating vector.%
    \footnote{This Lemma needs our standing convention that Hilbert spaces are separable. It holds for arbitrary Hilbert spaces if the von Neumann algebra is $\sigma$-finite, which holds automatically if the Hilbert space is separable.}
\end{lemma}

\begin{proof}[Proof sketch.]
    If $\Omega$ is a cyclic separating vector, then a standard form $(\H,J,\P)$ is given by 
    \begin{equation}
        J=J_\Omega,\qquad \P = \overline{\{aJa\Omega : a\in \M\}}.
    \end{equation}
    To see this, one has to check that $\P$ is a self-dual positive cone (this requires some work, see \cite[Thm.~1.1]{haagerup_standard_1975}) and that items \ref{it:std form1} to \ref{it:std form3} of \cref{thm:std form} hold.
    Item \ref{it:std form1} follows from Tomita's theorem (see \eqref{eq:tomitas thm}), and the other two items are clear.
    Conversely, if $(\H,J,\P)$ is a standard form, consider a unit vector $\Omega\in\P$ that implements a faithful state $\omega$ on $\M$.
    Then $\Omega$ is faithful for $\M$, which is equivalent to $\Omega$ being cyclic for $\M'$.
    Since $\Omega$ is in the positive cone, it is $J$-invariant which implies that it is cyclic for $\M$: $\overline{\M\Omega} = \overline{J\M'J\Omega} = \overline{J\M'\Omega} = J\H=\H$.
\end{proof}

\subsubsection{The core of a von Neumann algebra}\label{sec:core}

Roughly speaking, the core of a von Neumann algebra $\M$ is a canonical semifinite extension $\Tilde\M\supset\M$ equipped with a canonical trace and a measure scaling $\RR$-action.
This structure, which can be constructed from modular theory, allows for the construction of canonical $L^p$-spaces associated with an arbitrary von Neumann algebra (see \cref{sec:haagerup Lp}).
The latter is of great significance, in particular, for applications to quantum information theory, since the $L^p$-spaces allow for direct generalizations for expressions such as $F(\rho,\sigma) = \norm{\rho^{1/2}\sigma^{1/2}}_1$ (the fidelity of density operators $\rho$ and $\sigma$, see \cite{nielsen_quantum_2010} or \cref{sec:fidelity}) even if $\rho$ and $\sigma$ are states on a general von Neumann algebra \cite{hiai_quantum_2021}.
What is even more important for us is that the core of a von Neumann algebra gives rise to the so-called flow of weights (see \cref{sec:spectral states}), which will be an indispensable tool for understanding embezzlement of entanglement (see \cref{sec:strong-entanglement}).

\begin{theorem}[{\cite{falcone_non-commutative_2001}}]\label{thm:core}
    For every von Neumann algebra $\M$, there is a semifinite extension\footnote{Here, semifinite extension simply means that $\M\subset\Tilde\M$ is a subalgebra and that $\Tilde\M$ is semifinite.} $\Tilde\M\supset\M$, a \nsf trace $\tau$ on $\Tilde\M$, and a point-$\sigma$-weakly continuous action $\tilde\theta:\RR\acts\Tilde\M$ that scales the trace
    \begin{equation}\label{eq:core1}
        \tau\circ \tilde\theta_s = e^{-s}\tau,\qquad s\in\RR,
    \end{equation}
    and has $\M$ as its invariant subalgebra
    \begin{equation}\label{eq:core2}
        \M = \Tilde\M^{\tilde\theta} :=\{ y\in \Tilde\M : \tilde\theta_t(y)=y\ \forall t\in\RR\}.
    \end{equation}
    The triple $c(\M)=(\Tilde\M,\tilde\theta,\tau)$ can be constructed canonically and is uniquely determined by the properties above up to isomorphism and cocycle equivalence (see \eqref{eq:connes cocyle}).
    Moreover:
    \begin{enumerate}[(i)]
        \item The relative commutant of $\M$ in $\Tilde\M$ is the center: $\M'\cap\Tilde\M=Z(\Tilde\M)$
        \item $\M$ is a factor if and only if $\Tilde\theta$ is centrally ergodic, i.e., $Z(\Tilde\M)^{\theta}=\CC1$.
    \end{enumerate}
\end{theorem}

Apart from a canonical way of constructing the triple $(\Tilde\M,\tilde\theta,\tau)$ in \cite{falcone_non-commutative_2001}, which involves many high-level concepts, there is the traditional, more pedestrian, way:
To start off, we have to pick a faithful normal state $\omega$ (or weight) on a von Neumann algebra $\M$.
One then considers the \emph{crossed product} $\Tilde\M=\M\rtimes_{\sigma^\omega}\RR$ relative to the modular flow $\sigma^\omega$.
In general, if $\sigma:G\acts\M$ is an action of a locally abelian group $G$, e.g., $G=\RR$, then the crossed product is a canonical extension $\M\rtimes_\sigma G\supset \M$ in which the action  $\sigma$ becomes inner, i.e., there are unitaries $(u_g)_{g\in G}$ such that $\sigma_g(a) = u_g^*a u_g$ (see \cite[Sec.~8.1]{hiai_lectures_2021}).
The crossed product naturally comes with a dual action $\hat\sigma$ of the dual group $\hat G$ whose fixed points are precisely elements of $\M$.
The modular structure admits the construction of a \nsf trace $\tau$ on $\M\rtimes_{\sigma^\omega}\RR$, which is scaled by the dual action $\tilde\theta = \hat\sigma^\omega$ of $\hat\RR=\RR$.
Moreover, the fact that all modular automorphism groups are cocycle equivalent implies that the triple $(\M\rtimes_{\sigma^\omega},\tilde\theta,\tau)$ is independent of $\omega$ up to an appropriate notion of equivalence \cite[Thm.~X.1.7]{takesaki2}.

In the following, we refer to a triple $c(\M)=(\Tilde\M,\tilde\theta,\tau)$ with the properties in \cref{thm:core} as a \emph{core} for $\M$.
The name "core" for the triple $(\Tilde\M,\tilde\theta,\tau)$ is due to the following: 
A type $\III$ von Neumann algebra $\M$ is, up to isomorphism, given as the crossed product of the  semifinite algebra $\Tilde\M$ by the action $\tilde\theta$:%
\footnote{The isomorphism in \eqref{eq:takesaki duality} is not canonical. There is, however, a canonical isomorphism $\Tilde\M\rtimes_{\tilde\theta}\RR\cong \Tilde\M \ox \B(L^2(\RR))$, which gives rise to \eqref{eq:takesaki duality} since type $\III$ von Neumann algebras absorb type $\I_\oo$ factors.}
\begin{equation}\label{eq:takesaki duality}
    \M \cong \Tilde\M \rtimes_{\tilde\theta}\RR,
\end{equation}
The core of the trivial von Neumann algebra $\M=\CC$ is
\begin{equation}
    c(\CC)= (L^\oo(\RR),\theta,\int_{-\oo}^\oo \placeholder\, e^{-s}ds),
\end{equation}
where $\theta:\RR\acts L^\oo(\RR)$ is the action by translations $\theta_t(f)(s)=f(s-t)$.
Indeed, \eqref{eq:core1} and \eqref{eq:core2} hold trivially for $\M=\CC$ so that the claim follows from the uniqueness of the core. 
The core respects direct integration:

\begin{lemma}[{\cite[Prop.~8.1]{haagerup_equivalence_1990}}]\label{lem:core direct int}
    Let $(X,\mu)$ be a $\sigma$-finite measure space and let $x\mapsto \M_x$ be a measurable field of von Neumann algebras.
    Then $x\mapsto c(\M_x)$ is appropriately measurable in each entry. 
    If $\M = \int^\oplus_X \M_x \,d\mu(x)$ is a direct integral of von Neumann algebras, then 
    \begin{equation}\label{eq:core direct int}
         c(\M)=\int^\oplus_X c(\M_x)\,d\mu(x)
    \end{equation}
    is a core for $\M$, where the direct integral is entry-wise.
\end{lemma}

In contrast to direct sums, tensor products are, in general, not preserved by the core.
However, if one of the tensor factors is semifinite, we have:

\begin{proposition}\label{prop:semifinite amp fow}
    Let $\N$ be a semifinite von Neumann algebra with \nsf trace $\tau_0$ and let $\M$ be another von Neumann algebra. 
    Then a core of $\M\barox\N$ is
    \begin{equation}\label{eq:semifinite amp fow}
        c(\M\barox\N) = c(\M) \ox (\N,\id,\tau_0),
    \end{equation}
    where the tensor product is component-wise, e.g., ${(\M\barox\N)}^\sim = \Tilde\M\barox\N$.
\end{proposition}

\begin{proof}
    Clearly, $\Tilde\M\barox\N$ is a semifinite extension of $\M\ox\N$ with $\tau\ox\tau_0$-scaling action $\tilde\theta \ox\id$.
    Moreover, the fixed-point algebra is
    \begin{equation}
        (\Tilde\M\barox\N)^{\tilde\theta\ox \id} = \Tilde\M^{\tilde\theta}\barox \N^{\id} = \M\barox\N.
    \end{equation}
    Thus, \cref{thm:core} implies that $c(\M) \ox (\N,\id,\tau_0)$ is a core for $\M\barox\N$.
\end{proof}

Since $\N = \N\ox \CC$, this lets us construct a core of a semifinite von Neumann algebra $\N$ with \nsf trace $\tau_0$ via
\begin{equation}\label{eq:semifinite core}
    c(\N) = \big (\N \barox L^\oo(\RR), \id\ox\,\theta, \tau_0 \ox \int_{-\oo}^\oo \placeholder\,e^{-s} ds \big).
\end{equation}
Note that this concrete realization of the core depends on the choice of $\tau_0$ (the dependence is only up to isomorphism and cocycle-equivalence, see \cref{thm:core}).

\subsubsection{Haagerup \texorpdfstring{$L^p$}{Lp}-spaces}\label{sec:haagerup Lp}

Following \cite[Sec.~3]{haagerup_equivalence_1990}, we associate with every weight $\psi$ on $\M$ a \emph{dual weight} $\tilde\psi$ on the core $\Tilde\M$ by:
\begin{equation}
    \tilde\psi(y) = \int_{-\oo}^\oo \psi(\tilde\theta_s(y))\,ds,\qquad y\in\N^+.
\end{equation}
Relative to the trace $\tau$, the dual weight has a density $0\le h_\psi= \frac{d\tilde\psi}{d\tau}\aff\N$ (see \cref{sec:semifinite vNas}).
These constructions are covariant and affine in the sense that $(a\psi a^*+\phi)^\sim = a\tilde\psi a^*+\tilde\phi$ and $h_{a\psi a^*+\phi}=ah_\psi a^*+h_\phi$ for $a\in\M$.
Since, by definition, the dual weight $\tilde\psi$ is $\tilde\theta$-invariant, the trace-scaling property implies
\begin{equation}\label{eq:hpsi scaling}
    \tilde\theta_s(h_\psi) =e^{-s}h_\psi.
\end{equation}

An operator $h\aff\Tilde\M$ is said to have \emph{grade $\alpha>0$} if $\tilde\theta_t(h) = e^{\alpha t} h$ for all $t\in\RR$.
The space of $\tau$-measurable operators of grade $p^{-1}$ is known as the \emph{Haagerup $L^p$-space} $L^p(\M)$.
It is clear from the definition that these are vector spaces and that the products of operators in $L^p(\M)$ with operators in $L^q(\M)$ are in $L^r(\M)$ with $r^{-1}=p^{-1}+q^{-1}$ (simply because the grade is additive under multiplication).
The structure of the $L^1$-space is described by the following:

\begin{theorem}[The Haagerup $L^1$-space {\cite[Sec.~9.1]{hiai_lectures_2021}}]\label{thm:haagerup L1}
    The operator $h_\psi$ is $\tau$-measurable if and only if $\psi$ is finite, i.e., in $\M_*^+$.
    The map $\psi\mapsto h_\psi$ is affine. Its linear extension gives a bijection:
    \begin{equation}\label{eq:haagerup L1}
        \M_* \ni \psi \mapsto h_\psi \in L^1(\M).
    \end{equation}
    The following properties hold:
    \begin{enumerate}[(i)]
        \item\label{it:haagerup L11} $h_{\psi}^* = h_{\psi^*}$, where $\psi^*\in\M_*^+$ is the linear functional $a\mapsto \overline{\psi(a^*)}$,
        \item\label{it:haagerup L12} if $a,b\in\M$, then $h_{a\psi b}=ah_\psi b$,
        \item\label{it:haagerup L13} if $\psi=v \abs\psi$ is the polar decomposition of $\psi$, then $h_\psi = v h_{\abs \psi}$ is the polar decomposition of $h_\psi$. In particular, $\abs{h_\psi}=h_{\abs\psi}$, $h_\psi \ge0 \iff \psi\in\M_*^+$, and $\supp(h_\psi)=\supp(\psi)$.
    \end{enumerate}
    Thus, if we set $\Tr h_\psi = \psi(1)$ for $\psi\in\M_*$, then 
    \begin{equation}
        \norm h_{L^1(\M)} = \tr |h|, \qquad h\in L^1(\M),
    \end{equation}
    turns $L^1(\M)$ into a Banach space, such that \eqref{eq:haagerup L1} is an isometric isomorphism.
\end{theorem}

The map $\tr: L^1(\M)\to \CC$ is called the \emph{Haagerup trace}.
The Haagerup $L^p$-spaces are Banach spaces with the norm
\begin{equation}\label{eq:haagerup Lp norm}
    \norm{h}_{L^p(\M)} = \big(\tr|h|^p\big)^{1/p},
\end{equation}
which is justified by $\abs h^p\in L^1(\M) \iff h\in L^p(\M)$ \cite[Sec.~9.2]{hiai_lectures_2021}.
More generally, an element $h\in L^0(\M,\tau)$ is in $L^p(\M)$ if and only if its polar decomposition is $h= v\abs h$ with $v\in\M$ and $\abs h^p\in L^1(\M)$.
We refer to \cite[Sec.~9]{hiai_lectures_2021} for an overview of the properties of the Haagerup $L^p$-spaces.
We only mention the core essentials:
The Hölder inequality holds \cite[Sec.~9.2]{hiai_lectures_2021}: If $r^{-1}=p^{-1}+q^{-1}$, then
\begin{equation}
    \norm{hk}_{L^r(\M)}\le \norm h_{L^p(\M)}\norm k_{L^q(\M)}.
\end{equation}
In the case that $p^{-1}+q^{-1}=1$, the products $hk,kh\in L^1(\M)$ have the same trace $\tr(hk)=\tr(kh)$ \cite[Prop.~9.22]{hiai_lectures_2021}.
In fact, the $L^p(\M)$ and $L^q(\M)$ are dual spaces with the dual pairing $(h,k)\mapsto \tr(hk)$ if $1<p<\oo$ and $p^{-1}+q^{-1}=1$.
This defines an inner product on $L^2(\M)$, which turns it into a Hilbert space. 
The action of $\M$ by left multiplication on $L^2(\M)$ is a standard form of $\M$ with $J$ and $\P$ given by
\begin{equation}\label{eq:haagerup L2 std form}
    J\xi = \xi^*, \qquad \P = L^2(\M)^+.
\end{equation}
Another neat feature of the theory of Haagerup  $L^p$-spaces is that it lets us describe the modular objects directly through the operators $h_\psi$.
For instance, it holds that if $\psi\in\M_*^+$ is faithful, then \cite[Lem.~10.5]{hiai_lectures_2021}:
\begin{equation}
    \sigma^\psi_t(a) = h_\psi^{it} a h_\psi^{-it},\qquad a\in\M.
\end{equation}
Moreover, if $\phi$ is another faithful element $\varphi\in\M_*^+$, then Connes cocycle is given by \cite[Lem.~10.13]{hiai_lectures_2021}
\begin{equation}
    (D\psi:D\phi)_t = h_\psi^{it}h_\phi^{-it}.
\end{equation}
Another nice application is to the definition of \emph{commutativity} of positive linear functionals $\psi,\phi\in\M_*^+$.
Given the Haagerup $L^1$-space, the obvious definition is that $h_\psi$ and $h_\phi$ commute and, indeed, this is equivalent to the traditional definitions in terms of modular theory such as $\phi\circ\sigma^\psi_t=\phi$ for all $t\in\RR$ (see \cite[Prop.~10.18]{hiai_lectures_2021} for details).

Before we move on, we discuss how the Haagerup $L^p$-spaces relate to the $L^p$-spaces relative to a fixed trace on a semifinite von Neumann algebra.
Recall that we can write the core $c(\M)$ of a semifinite von Neumann algebra $\M$ with \nsf trace $\tau_0$ as $(\M\barox L^\oo(\RR),\id\ox\,\theta,\tau_0\ox\int_{-\oo}^\oo\placeholder e^{-s}ds)$ above (see \eqref{eq:semifinite core}).
Since the flow acts trivially on $\M\ox 1$, the operators $h\aff\Tilde\M$ that have grade $p^{-1}$ are precisely the operators of the form $a\ox e^{-t/p}$ (where $t$ is the "dummy variable" for functions on $\RR$).
Thus, we may write the Haagerup $L^p$-space of the semifinite von Neumann algebra $\M$ as:
\begin{equation}\label{eq:semifinite haagerup Lp}
    L^p(\M) = L^p(\M,\tau_0) \ox e^{-t/p}.
\end{equation}
If $\psi\in\M_*^+$ and if $\rho_\psi = d\psi/d\tau_0$ is its density operator in $L^1(\M,\tau_0)$, then the corresponding operator $h_\psi$ in the canonical $L^1$-space $L^1(\M)$ is given by
\begin{equation}\label{eq:semifinite h}
    h_\psi = \rho_\psi \ox e^{-t}.
\end{equation}
While the standard $L^p$-spaces relative to a trace are, of course, much more accessible, the abstract approach via Haagerup $L^p$-spaces has the structural benefit that it is canonical, i.e., does not require the choice of a trace.
The measure theory analog of this are the canonical $L^p$-spaces of nonsingular spaces, i.e., standard Borel spaces equipped with an equivalence class of, which do not depend on the choice of a measure $\sigma$-finite measures (see \cite[Sec.~4.3]{arano_ergodic_2021}).

\subsubsection{The flow of weights and the Haagerup-Størmer spectral state}\label{sec:spectral states}

The flow of weights is a canonical classical dynamical system that can be assigned to a von Neumann algebra $\M$.
Working in the category of von Neumann algebras, we describe a classical dynamical system by a flow $\theta:\RR\acts\A$ on an abelian von Neumann algebra $\A$.
The flow of weights is the central restriction of the trace-scaling flow on the core of the von Neumann algebra, i.e., if $c(\M)=(\Tilde\M,\tilde\theta,\tau)$ is the core of $\M$, then the flow of weights is:
\begin{equation}
    \A = Z(\Tilde\M), \qquad \theta_t = \tilde\theta_t \restriction Z(\Tilde\M).
\end{equation}
By \cref{thm:core}, the flow of weights is ergodic if and only if the von Neumann algebra $\M$ is a factor.
By \eqref{eq:semifinite core}, the flow of weights of all semifinite factors is given by the translations on $\RR$.
For type $\III$ factors, however, the flow of weights captures the subtype (see \cite[Sec.~XII.1]{takesaki2}):
\begin{itemize}
    \item Type $\III_0$: The flow of weights is aperiodic.\footnote{Type $\III_0$ factors are the only type $\III$ factors with aperiodic flow weights. 
    Among all factors, type $\III_0$ factors are singled out by having a properly ergodic flow of weights.} 
    In this case, the flow of weights is $\A= L^\oo(X,\mu)$ and $\theta_t(f)(x)=f(F_t(x))$ for some properly ergodic flow $(F_t)$ on $(X,\mu)$.\footnote{An ergodic flow $(F_t)$ on $(X,\mu)$ is properly ergodic if $X$ is  not equal to the orbit of a single point (up to null sets) \cite[Sec.~XII.3]{takesaki2}. For instance, the irrational line flow on the torus is properly ergodic.}
    \item Type $\III_\lambda$, $0<\lambda<1$: The flow of weights is periodic with minimal period $T=-\log \lambda$. In this case, $\A \cong L^\oo(\RR/T\ZZ)$ with $\theta$ corresponding to the periodic shift.
    \item Type $\III_1$: The flow of weights is trivial $\theta\equiv\id$ and $\A=\CC$.
\end{itemize}

\null

In the following, we review a construction due Haagerup and Størmer \cite{haagerup_equivalence_1990}, that associates with a normal positive linear function $\psi$ on $\M$ a positive linear functional $\hat\psi$ on the flow of weights.
If $\psi$ is a normal semifinite weight on $\M$, we denote the spectral projection of the operator $h_\psi= d\tilde\psi/d\tau$ onto the interval $(1,\oo)$ by $e_\psi$, i.e., $e_\psi = 1_{(1,\oo)}(h_\psi)$.

\begin{lemma}[{\cite[Lem.~3.1]{haagerup_equivalence_1990}}]\label{lem:normalization}
    Let $\psi$ be a normal semifinite weight on $\M$ and let $(\Tilde\M,\tilde\theta,\tau)$ be the core of $\M$.
    Then 
    \begin{equation}
        \tau(e_\psi) = \psi(1).
    \end{equation}
\end{lemma}

Because of its central importance, we sketch the proof given in \cite{haagerup_equivalence_1990}:

\begin{proof}
    The identity $\int_{-\oo}^\oo e^t s^{-1} 1_{(1,\oo)}(e^{-t})\,dt = 1_{(0,\oo)}(s)$ shows $\int_{-\oo}^\oo e^t h_\psi^{-1} 1_{(1,\oo)}(e^{-t}h_\psi)\,dt = \supp(h_\psi)=\supp \psi$, where we used item \ref{it:haagerup L13} of \cref{thm:haagerup L1}.
    Hence, using $\tilde \theta_t(h_\psi) = e^{-t}h_\psi$, we find
    \begin{align*}
        \int_{-\oo}^\oo \tilde\theta_t(e_\psi h_\psi^{-1})\,dt = \int_{-\oo}^\oo e^t h_\psi^{-1} 1_{(1,\oo)}(e^{-t}h_\psi) \,dt =\supp(\psi).
    \end{align*}
    Since $\tilde\psi$ is $\tilde\theta$-invariant, we have 
    \begin{align*}
        \tau(a e_\psi)= \tilde\psi( a e_\psi h_\psi^{-1})  
        &= \psi\bigg( \int_{-\oo}^\oo \tilde\theta_t(ae_\psi h_\psi^{-1})\,dt\bigg)
        = \psi\bigg(a \int_{-\oo}^\oo \tilde\theta_t(e_\psi h_\psi^{-1})\,dt \bigg)= \psi(a)
    \end{align*}
    for $a\in\M$,
    where we used $1_{(0,\oo)}(h_\psi)=\supp(h_\psi)=\supp(\psi)$.
    Setting $a=1$ implies the claim.
\end{proof}

\begin{definition}
    Let $\M$ be a von Neumann algebra with flow of weights $\theta:\RR\acts\A$.
    The \emph{spectral functional} of $\psi\in\M_*^+$ is the normal positive linear functional $\hat\psi\in\A_*^+$ defined by
    \begin{equation}
        \hat\psi(a) = \tau(e_\psi a),\qquad a\in \A=Z(\Tilde\M).
    \end{equation}
\end{definition}

If $\psi$ is a state, we call $\hat \psi$ the \emph{spectral state} of $\psi$.
We list a few basic properties of the map $\psi\mapsto\hat\psi$.
By \cref{lem:normalization}, this map is norm-preserving.
Moreover, it is monotonic, i.e.,
\begin{equation}
    \psi\le \phi \implies \hat\psi\le\hat\phi.
\end{equation}
The flow of the spectral state is
\begin{equation}\label{eq:flow of the spectral state}
    \hat\psi \circ\theta_t = e^{-t}\,\Hat{e^t\psi}, \qquad t\in\RR.
\end{equation}
Monotonicity and \eqref{eq:flow of the spectral state} imply 
\begin{equation}\label{eq:theta-regular}
    \hat\psi\circ\theta_t \ge e^{-t}\hat\psi, \qquad t>0.
\end{equation}
We refer to this property as \emph{$\theta$-regularity}.
We denote the $\theta$-regular state space on $\A$ as 
\begin{equation}\label{eq:theta-regular state space}
    \states_*^\theta(\A) = \{ \chi \in \A_*^+ : \chi\circ\theta_t\ge e^{-t}\chi,\ t>0\}.
\end{equation}
The following is the main theorem of \cite{haagerup_equivalence_1990}:

\begin{theorem}[{\cite{haagerup_equivalence_1990}}]\label{thm:distance spectral states}
    The distance of the unitary orbits of $\psi,\phi\in \M_*^+$ is 
    \begin{equation}
        \inf_{u\in\U(\M)}\, \norm{\psi-u\phi u^*} = \norm{\hat\psi - \hat\phi}.
    \end{equation}
    If $\M$ is type $\III$, then $\nstates(\M)\ni\psi\mapsto\hat\psi\in\states_*^\theta(\A)$ is surjective.
\end{theorem}

In particular, \cref{thm:distance spectral states} states that two normal states $\psi,\phi$ on $\M$ are approximately unitarily equivalent if and only if $\hat\psi=\hat\phi$, where approximate unitary equivalence is defined as
\begin{equation}
    \psi\sim\phi \quad :\iff\quad \inf_{u\in\U(\M)} \,\norm{\psi-u\phi u^*}=0.
\end{equation}
Comparing \cref{thm:distance spectral states} with \cref{thm:majorization} shows that the spectral states on general von Neumann algebras play a similar role to that of the spectral scale and distribution function in semifinite factors.
In fact, the distribution function is just a special representation of the spectral states (see below).

Another important consequence of \cref{thm:distance spectral states} is that it bridges the gap between the classification of type $\III$ factors in terms of the state space diameter (see \cref{sec:basics}) and the textbook version of the classification in terms of the flow of weights, as explained at the beginning of this subsection.
Indeed, \cref{thm:distance spectral states} shows that the state space diameter (see \eqref{eq:state space diam}) of a type $\III$ factor is
\begin{equation}\label{eq:state space diam of HS}
    \diam \nstates(\M)/\!\sim\ = \diam \states_*^\theta(\A),
\end{equation}
where the diameter of the $\theta$-regular state space is taken with respect to the norm distance.
The right-hand side (RHS) of \eqref{eq:state space diam} can be computed from the flow of weights.
It is given by $2(1-e^{-T/2})/(1+e^{-T/2})$, where $T\in\bar\RR^+$ is the minimal period of $\theta$.
Thus, \eqref{eq:state space diam} follows from $\lambda=e^{-T}$ with $e^{-\oo}=0$.

To better understand the Haagerup-Størmer spectral states, we consider the case of a semifinite factor $\M$ with a \nsf trace $\tau_0$.
Recall the description of the core from \eqref{eq:semifinite core} and the description of the Haagerup $L^1$-space from \eqref{eq:semifinite haagerup Lp} and \eqref{eq:semifinite h}.
If $\psi\in\M_*^+$, then \eqref{eq:semifinite h} implies $e_\psi = 1_{(1,\oo)}(\int^\oplus_\RR \rho_\psi e^{-t}\,dt) = \int_{\RR}^\oplus 1_{(e^{t},\oo)}(\rho_\psi)\,dt$ (where we changed from "dummy variable notation" to direct integral notation, i.e., $\rho_\psi\ox e^t\equiv \int^\oplus_\RR e^t \rho_\psi\,dt$).
Thus, if $f\in L^\oo(\RR)=Z(\Tilde\M)$, then
\begin{equation*}
    \hat\psi(f) = \tau(f e_\psi) 
    = \int_{-\oo}^\oo \tau_0(1_{(e^t,\oo)}(\rho_\psi)) f(t)\, e^{-t}dt 
    = \int_0^\oo D_\psi(t) f(\log t)\,dt,
\end{equation*}
where $D_\psi$ denotes the distribution function of $\psi$ relative to $\tau_0$ (see \cref{sec:spectral states}).
Therefore, if we use the map $f\mapsto f\circ\log$ to represent the flow of weights $\A = L^\oo(\RR^+)$ with dilations $\theta_t(f)(s)=f(se^{-t})$, then the density of the normal positive linear functional $\hat \psi$ with respect to the Lebesgue measure $dt$, is precisely the distribution function $D_\psi$.
Thus, \cref{thm:distance spectral states} reduces to the statement $\inf_{u\in\U(\M)}\norm{u\psi u^*-\phi}=\norm{D_\psi-D_\phi}_{L^1(\RR^+)}$ that we already know from \cref{thm:majorization}.

\subsection{Useful Lemmas and their consequences}\label{sec:useful}

In this section, we present a collection of useful tools for our later analysis of quantum entanglement and von Neumann algebraic quantum information theory.

\subsubsection{Partial isometries vs.\ unitaries}

In the following, we collect results summarizing the interplay of contractions, partial isometries, and unitaries with regard to their action on vectors and states.
Recall that we denote by $\U(\M)$ the unitary group and by $\V(\M)$ the set of partial isometries in a von Neumann algebra $\M$.
The set of contractions in $\M$, i.e., the unit ball, is denoted $B(\M)$.
We say that an operator $a$ on a Hilbert space acts isometrically on a vector $\Phi$ if $\norm{a\Phi}=\norm\Phi$.

\begin{lemma}\label{lem:isometric-contraction}
        Let $s\in \M$ be a contraction with polar decomposition $s=v\abs s$ and let $\Phi\in \H$ be a unit vector. Then
        \begin{enumerate}[(i)]
            \item\label{it:isometric-contraction1}
            If $s$ acts almost isometrically on $\Phi$, then $s\Phi\approx v\Phi$. Quantitatively, 
            \begin{equation}\label{eq:isometric-contraction1}
                \norm{s\Phi-v\Phi} \le \big( 2-2\norm{s\Phi}^2 \big)^{1/2}.
            \end{equation}
            If $s$ acts isometrically on $\Phi$, then it acts isometrically on $a'\Phi$ and $sa'\Phi=va'\Phi$ for all $a'\in\M'$.
            \item\label{it:isometric-contraction2}
            If $s$ acts almost isometrically on $\Phi$, it can be approximated with unitaries:
            If $\eps>0$ is such that $\norm{s\Phi}\ge 1-\eps$, there is a $u\in\U(\M)$ with
            \begin{equation}
                \norm{s\Phi-u\Phi}\le 3\eps^{1/2}.
            \end{equation}
        \end{enumerate}
\end{lemma}
\begin{proof}
    \ref{it:isometric-contraction1}: We set $\delta= 1-\norm{s\Phi}^2$.
    Then $1-\delta = \norm{s\Phi}^2\le  \norm{\abs s\Phi}^2 =\ip\Phi{\abs s^2\Phi} \le \ip\Phi{\abs s\Phi}$.
    Thus, we have $\norm{s\Phi-v\Phi}^2 \le \norm{v(\abs s-1)\Phi}^2 \le \norm{\abs s\Phi-\Phi}^2 \le 1+\norm{\abs s\Phi}^2-2\ip\Phi{\abs s\Phi} \le 2 (1- \ip\Phi{\abs s\Phi}) \le 2\delta$. Rearranging shows \eqref{eq:isometric-contraction1}.
    If $s$ acts isometrically and if $a'\in\M'$, we have $sa'\Phi=a's\Phi = a'v\Phi$, which implies $\norm{a'v\Phi}^2 = \ip{\Phi}{a'^*a'v^*v\Phi} = \ip\Phi{a'^*a'\Phi}=\norm{a'\Phi}$.

    \ref{it:isometric-contraction2}:
    The argument is based on \cite[Lem.~2.4]{haagerup_equivalence_1990}.
    Let $\delta=\eps/2$. 
    By the Russo-Dye Theorem \cite{russo1966note} (see also \cite{kadison_means_1985}), we have $(1-\delta)s=\sum_i p_i u_i$ for a finite probability distribution $(p_i)$ and unitaries $u_i$.
    Then
    \begin{align*}
        \sum p_i\norm{u_i\Phi - (1-\delta)s\Phi}^2 
        &= \sum p_i \big( 1 + (1-\delta)^2 \norm{s\Phi}^2  - 2(1-\delta) \Re \ip{u_i \Phi}{s\Phi} \big)\\
        &= 1 + (1-\delta)^2 \norm{s\Phi}^2  - 2(1-\delta) \Re \ip{(1-\delta)s \Phi}{s\Phi} \\
        &= 1- (1-\delta)^2 \norm{s\Phi}^2
         \le 1- (1-\eps/2)^2 (1-\eps)^2 \le 4\eps.
    \end{align*}
    Thus, there is an index $i$ such that $u=u_i$ satisfies $\norm{u\Phi-(1-\delta)s\Phi} \le 2\eps^{1/2} $ and, hence,
    \begin{equation*}
        \norm{u\Phi - s\Phi}^2 \le (\norm{u\Phi-(1-\delta)s\Phi} + \delta)^2 \le (2\eps^{1/2}+\eps/2)^2 \le 9\eps.\qedhere
    \end{equation*}
\end{proof}

As a consequence, we have:

\begin{corollary}\label{cor:isometrically acting pi's}
    \begin{enumerate}
        \item 
            If $v$ is a partial isometry in a von Neumann algebra $\M$ on $\H$ such that $v$ acts isometrically on a vector $\Phi\in\H$, then for every $\eps>0$, there exists a unitary $u\in\M$ such that $\norm{u\Phi-v\Phi}<\eps$.
        \item
            If $v$ is a partial isometry and $\phi\in\M_*^+$ such that $\phi(v^*v)=\phi(1)$ ($v\phi v^*$ and $\phi$ have the same norm), then, for every $\eps>0$ there exists a unitary $u\in\M$ such that $\norm{u\phi u^*-v\phi v^*}<\eps$.
    \end{enumerate}
\end{corollary}
\begin{proof}
    The first item is clear from \cref{it:isometric-contraction2} of \cref{lem:isometric-contraction}.
    The second item follows from the first item by considering the representing vector $\Phi\in\P$ in a standard form $(\H,J,\P)$ of $\M$ (or any other purification of $\phi$): 
    The assumption $\phi(v^*v)=\phi(1)$ implies $v^*v\ge \supp(\phi)$, which implies $\norm{v\Phi}=\norm{\Phi}$.
    Thus, there exists a unitary $u$ such that $\norm{(v-u)\Phi}< \eps/(2\norm\Phi)$.
    Consequently, we have  $\norm{v\phi v^*-u\phi u^*}\le \norm{(v-u)\Phi}\norm{(v+u)\Phi} < \eps$.
\end{proof}

The next Lemma shows that every isometrically acting isometry can be replaced by a unitary on a larger Hilbert space:

\begin{lemma}\label{lem:unitary extension}
    Let $v\in \M$ be a partial isometry.
    If $\K$ is an infinite-dimensional Hilbert space and $\Omega\in \K$ is a unit vector, then there is a unitary $u\in \M\barox\B(\K)$ such that
    \begin{equation}
        v\ox \ket\Omega = u (v^*v\ox \ket{\Omega}).
    \end{equation}
\end{lemma}

\begin{proof}
    By disintegration, we may assume that $\M$ is a factor.
    Recall that a partial isometry $w$ can be extended to a unitary if and only if $1-ww^*$ and $1-w^*w$ are Murray-von Neumann equivalent.
    We set $w=v\ox \kettbra\Omega$. Then $1-w^*w = 1-v^*v\ox \kettbra\Omega$ and $1-ww^*=1-vv^*\ox \kettbra\Omega$ are both infinite projections in a factor and, hence, Murray-von Neumann equivalent.
    Hence, there is a unitary $u\in\M\barox\B(\K)$ such that $w = w^*w u = uw^*w$.
    Thus, we have $(v\ox \ket\Omega)\Psi = w( \Psi\ox\Omega) = uw^*w(\Psi\ox\Omega) = u((vv^*\Psi)\ox\Omega) = u(v^*v\ox\ket\Omega)\Psi$.
\end{proof}

To state our next Lemma, we introduce a bit of notation that will be useful in later sections.
For elements $x,y$ of a metric space $(X,d)$, we write  $x\approx_\eps y$ if $d(x,y)<\eps$.
We will use this notation freely in cases where the metric is clear from the context.
For instance, for normal states $\phi,\psi$ on a von Neumann algebra $\M$, $\psi\approx_\eps\phi$ means $\norm{\psi-\phi}<\eps$.

\begin{lemma}[{\cite[Thm.~2.2 \& Lem.~2.4]{haagerup_equivalence_1990}}]\label{lem:approximate unitary equivalence}
    Let $\phi,\psi\in\nstates(\M)$. The following are equivalent:
    \begin{enumerate}[(a)]
        \item approximate unitary equivalence: for all $\eps>0$, there exist unitaries $u\in\M$ with $u\psi u^*\approx_\eps \phi$;
        \item for all $\eps>0$, there is a $v\in\V(\M)$ with $vv^*=\supp(\phi)$, $v^*v=\supp(\psi)$, and $v\psi v^* \approx_\eps \phi$;
        \item for all $\eps>0$, there is a contraction $s\in \M$ such that $s\psi s^* \approx_\eps \phi$.
    \end{enumerate}
\end{lemma}

We also note the following fact about the topology and the completions of the unitary group:

\begin{lemma}[{\cite[Par.~I.3.2.9]{blackadar_operator_2006}}]\label{lem:topologies-on-the-unitary-group}
    Let $\M$ be a von Neumann algebra on $\H$.
    The weak, strong, and strong* topologies coincide on the unitary group $\U(\M)$.
    The unitary group $\U(\M)$ is a strong*-closed subspace of $\B(\H)$.
    Depending on $\M$, it can fail to be weakly-closed or strongly-closed.
\end{lemma}

\subsubsection{Purifications}\label{sec:purifications}

A \emph{purification} of a normal positive linear functional $\phi$ on a von Neumann algebra $\M$, is a triple $(\pi,\H,\Phi)$ of a representation $(\pi,\H)$ of $\M$ and a vector $\Phi\in\H$ such that
\begin{equation*}
    \phi(a) = \ip\Phi{\pi(a)\Phi}, \qquad a\in\M.
\end{equation*}
If $(\pi,\H)$ is a fixed representation then a \emph{purification in $\H$} is a vector $\Phi\in\H$ such that $(\pi,\H,\Phi)$ is a purification.
Every normal positive linear functional $\phi$ on $\M$ has a purification.
In particular, it has a \emph{canonical purification}, namely the triple $(\pi,L^2(\M),\Omega_\phi)$ of the standard representation $\pi:\M\to \B(L^2(\M))$ with $\Omega_\phi\in L^2(\M)^+$ the unique implementing vector of $\phi$ in the positive cone.
In addition to the canonical purification, every $\phi\in\M_*^+$ has a unique minimal purification, namely its GNS representation (see below).

The following simple Lemma will become an indispensable tool for our analysis of pure state entanglement:

\begin{lemma}\label{lem:M-linear partial isometry}
    Let $(\pi_i,\H_i,\Phi_i)$, $i=1,2$, purifications of $\phi\in\M_*^+$.
    Then, there is a unique partial isometry $v':\H_1\to\H_2$ with the following properties: $v'\Phi_1=\Phi_2$, $v'^*v'=[\pi_1(\M)\Phi_1]$, $v'v'^*=[\pi_2(\M)\Phi_2]$, and 
    \begin{equation}
        \pi_2(a)v' = v'\pi_1(a), \qquad a\in\M.
    \end{equation}
\end{lemma}
\begin{proof}
    We define $v'$ on $\overline{\pi_1(\M)\Phi_1}$ as $v' \pi_1(a)\Phi_1= \pi_2(a)\Phi_2$, $a\in\M$.
    This is well-defined and isometric since $\norm{\pi_2(a)\Phi_2} = \ip{\pi_2(a)\Phi_2}{\pi_2(a)\Phi_2}^{1/2} = \phi(a^*a)^{1/2} = \norm{\pi_1(a)\Phi_1}$, $a\in\M$.
    We extend $v'$ trivially onto the orthogonal complement of $\overline{\pi_1(\M)\Phi_1}$.
    By construction, $v'$ is a partial isometry $\H_1\to\H_2$ satisfying the specified properties.
    If $u$ is another partial isometry with the specified properties and $\Phi=\pi_1(a)\Phi_1 + \Phi^\perp$, where $\Phi^\perp \perp [\pi_1(\M)\Phi_1]$, then $u\Phi=  u\pi_1(a) \Phi_1 = \pi_2(a)u\Phi_1 = \pi_2(a)\Phi_2 = u\pi_2(a)\Phi_1 = v'\Phi$, $a\in \M$.
\end{proof}

We note three important consequences of this Lemma.
The first one states that vectors in $\H$ have a canonical polar decomposition relative to a von Neumann algebra $\M$ on $\H$:

\begin{corollary}[Polar decomposition]\label{cor:vector polar decomp}
    Let $\M$ be a von Neumann algebra of operators on $\H$.
    Let $\Phi\in \H$ be a vector and let $\phi\in\M_*^+$ be the induced positive linear functional. 
    Let $\pi$ be the standard representation of $\M$ on $L^2(\M)$.
    There exists a unique partial isometry  $v':L^2(\M) \to \H$ such that
    \begin{enumerate}[(i)]
        \item $v'$ maps the canonical purification $\Omega_\phi$ to $\Phi$, i.e., $v\Omega_\phi = \Phi$;
        \item $v'$ is $\M$-linear, i.e., $v'\pi(a)= av'$ for all $a\in\M$;
        \item $v'v'^* = [\M\Phi]$, $v'^*v'=[\pi(\M)\Omega_\phi] = J\pi(\supp(\phi))J$.
    \end{enumerate}
\end{corollary}
\begin{proof}
    Apply \cref{lem:M-linear partial isometry} to the canonical purification $(\pi,L^2(\M),\Omega_\phi)$ and $(\id,\H,\Phi)$.
\end{proof}

Next, we consider the structure of vectors on a Hilbert space inducing the same functional on a fixed von Neumann algebra:

\begin{corollary}\label{cor:all-purifications}
    Let $\M$ be a von Neumann algebra on $\H$.
    If $\Phi_1,\Phi_2\in \H$ are vectors of the same norm. 
    The following are equivalent:
    \begin{enumerate}[(a)]
        \item $\Phi_1$ and $\Phi_2$ induce the same functional $\phi$ on $\M$;
        \item $\Phi_1 = v'\Phi_2$ for a partial isometry $v'\in\M'$;
        \item for all $\eps>0$, there is a unitary $u'\in\M'$ such that $\Phi_1 \approx_\eps u'\Phi_2$.
    \end{enumerate}
    Therefore, if $\Phi\in\H$ is a purification of $\phi\in\M_*^+$, then the set of all purifications in $\H$ is
    \begin{equation}
        \{ v' \Phi : v'\in \V(\M'),\ \norm{v'\Phi}=\norm\Phi\} = \overline{\{u' \Phi: u'\in\U(\M')\}}.
    \end{equation}
\end{corollary}
\begin{proof}
    Equivalence of the first two items is direct from \cref{lem:M-linear partial isometry}. Equivalence with the third item follows from \cref{cor:isometrically acting pi's}.
\end{proof}

We emphasize that the partial isometry $v'$ appearing in \cref{cor:all-purifications} necessarily acts isometrically on $\Phi_2$.
Finally, we note that among all purifications, there is a unique minimal one:

\begin{corollary}[GNS representation]\label{cor:GNS purification}
    Up to unitary equivalence, a normal positive linear functional $\phi\in\M_*^+$ has a unique purification $(\pi,\H,\Phi)$, which is minimal in the sense that $\H = \overline{\pi(\M)\Phi}$.
\end{corollary}
\begin{proof}
    \emph{Uniqueness}: \Cref{lem:M-linear partial isometry} implies that any two minimal purifications are unitarily equivalent.
    \emph{Existence}: We restrict the canonical purification to the invariant subspace $\overline{\pi(\M)\Phi}$.
\end{proof}

The unique purification in \cref{cor:GNS purification} is the GNS representation \cite[Thm.~I.9.14]{takesaki1}.

\subsubsection{\texorpdfstring{$\M$}{M}-operators and inner completely positive maps}\label{sec:M-operators}

In this subsection, we consider inner completely positive maps between different amplifications of a given representation of a von Neumann algebra $\M$.
To do so, we consider properties and characterizations of the space $\M\barox \B(\K_1,\K_2)$.

A \emph{completely positive} (cp) map $T:\M\to\N$ between von Neumann algebras $\M$ and $\N$ is a linear map such that 
\begin{equation}
    T\ox \id : \M\barox \B(\K)\to \N\barox\B(\K)
\end{equation}
is positive for all (finite-dimensional) Hilbert spaces $\K$ \cite{paulsen_completely_2003}.
Equivalently, $T$ is cp if for finite collections $a_1,\ldots, a_n\in \N$, $b_1,\ldots, b_n\in \M$, we have $\sum_{ij=1}^n a_i^*T(b_i^* b_j)a_j\ge0$, see \cite[Sec.~4.3]{takesaki1}.
It is \emph{normal} if it is continuous for the respective $\sigma$-weak operator topologies.
This is the case if and only if there is a predual map $T_*:\N_*\to\M_*$.
A cp map is \emph{subunital} if $T(1)\le 1$. Every bounded cp map is subunital up to a scalar multiple, namely its norm.
In the following, we only consider cp maps between type $\I$ factors, which is the only case that we need in this work.
Every normal subunital cp map $T:\B(\H_2)\to\B(\H_1)$ has a \emph{Kraus decomposition} 
\begin{equation}\label{eq:kraus-form-math}
    T = \sum_x k_x^*(\placeholder)k_x, \qquad \sum_x k_x^*k_x \le 1, 
\end{equation}
for a collection $\{k_x\}$ of so-called \emph{Kraus operators} $k_x\in\B(\H_1,\H_2)$ \cite{kraus_general_1971}.
Conversely, \eqref{eq:kraus-form-math} defines a subunital normal cp map, which is unital if and only if equality holds on the RHS. 
Moreover, every normal cp map has a \emph{Stinespring dilation}.
The Stinespring dilation for unital cp maps together with two of its important properties is summarized in the following Lemma (see, e.g., \cite{westerbaan_paschke_2017}).
To state it, we need the \emph{cp-order}. If $T,S$ are cp maps, we say that $T$ cp-dominates $S$, written $S\le_{cp} T$, if $T-S$ is cp.
The cp-order interval $[0,T]$ is the convex set of cp maps $S\le_{cp} T$.

\begin{lemma}[Stinespring dilation]\label{lem:stinespring}
    Let $T:\B(\H_1)\to\B(\H_2)$ be normal ucp map.
    \begin{enumerate}
        \item 
        There is a Hilbert space $\K$ and an isometry $v:\H_2\to\H_1\ox\K$ such that $T = v^*(\placeholder\ox1)v$ and $\overline{(\B(\H_1)\ox1)v\H_2}=\H\ox\K$. These properties determine $(\K,v)$ up to a unitary on $\K$.
        \item
        For every isometry $v_1:\H_2\to\H_1\ox\K_1$ such that $T = v_1^*(\placeholder\ox 1)v_1$, there is an isometry $w:\K\to\K_1$ such that $v_1=(1\ox w)v$.
        \item 
        The map $q \mapsto S= v^*(\placeholder \ox q)v$ is an affine bijection between the order interval $[0,1]_{\B(\K)}$ and the cp-order interval $[0,T]_{cp}$ such that $q\le q' \iff S\le_{cp} S'$.
    \end{enumerate}
\end{lemma}

We refer to isometries $v:\H_1\to\H_2\ox\K$ with $T = v^*(\placeholder\ox1)v$ as \emph{Stinespring isometries} and to the unique isometry in the first item as the \emph{minimal} Stinespring isometry
As a consequence of the cp-order isomorphism property, a Kraus decomposition (see \eqref{eq:kraus-form-math}) is in one-to-one correspondence with a family of one-dimensional projections $q_x = \kettbra{\Omega_x}$ such that $\sum_x q_x=1$ and
\begin{equation}\label{eq:kraus-from-stinespring}
    k_x= (1\ox \bra{\Omega_x})v.
\end{equation}
Moreover, the second item of \cref{lem:stinespring} shows that for every pair $v_j:\H_1\to\H_2\ox \K_j$, $j=1,2$, of Stinespring isometries, there is a partial isometry $w:\K_1\to\K_2$ such that
\begin{equation}
    v_2 = (1\ox w)v_1.
\end{equation}

For the remainder of this section, we consider a fixed von Neumann algebra $\M$ on a Hilbert space $\H$.
A cp map $T$ on $\B(\H)$ is said to be $\M$-inner if it has a Kraus representation $T = \sum_x k_x^*(\placeholder) k_x$ with Kraus operators $k_x$ in $\M$.
For our applications in \cref{sec:vNQI}, we need to generalize this notion to cp maps between Hilbert spaces of the form $\H\ox \K_j$, $j=1,2$.
To do so, we consider the strongly closed linear subspace
\begin{equation}
    \M \barox \B(\K_1,\K_2) \subset \B(\H\ox\K_1,\H\ox\K_2)
\end{equation}
generated by operators of the form $a\ox b$, $a\in\M$, $b\in \B(\K_1,\K_2)$.
We refer to the elements of $\M\barox\B(\K_1,\K_2)$ as \emph{$\M$-operators}.
Identifying, $\H$ with $\H\ox\CC$, we obtain a notion of $\M$-operators $\H\to\H\ox\K$ and $\H\ox\K\to\H$.
The following Lemma collects equivalent characterizations of the set of $\M$-operators.

\begin{lemma}\label{lem:M-operators}
    Let $a\in\B(\H\ox\K_1,\H\ox\K_2)$.
    The following are equivalent:
    \begin{enumerate}[(a)]
        \item\label{it:M-operators0} $a\in \M\barox\B(\K_1,\K_2)$.
        \item\label{it:M-operators1} $(1\ox\bra{\Omega_2}) a (1\ox\ket{\Omega_1}) \in \M$ for all $\Omega_1\in\K_1$, $\Omega_2\in\K_2$.
        \item\label{it:M-operators2} $(a'\ox 1_{\K_2})a=a(a'\ox 1_{\K_1})$ for all $a'\in\M'$.
        \item\label{it:M-operators3} $a$ can be written as an $\M$-valued matrix, i.e., if $\{\ket{i}_t\}_i$ are ONBs of $\K_t$, $t=1,2$, then 
        \begin{equation*}
            a = \sum_{ij} a_{ij} \ox \ket{i}_2\bra j_1, \qquad a_{ij}\in \M.
        \end{equation*}
        \item\label{it:M-operators4} 
            $a^*(a'\ox 1_{\K_2})a= a^*a(a'\ox 1_{\K_1})$ for all $a'\in\M'$.
    \end{enumerate}
\end{lemma}
\begin{proof}
    For notational ease, we suppress the indices indicating the Hilbert spaces $\K_1,\K_2$. 
    The implications \ref{it:M-operators0} $\Rightarrow$ \ref{it:M-operators1} $\Rightarrow$ \ref{it:M-operators3} $\Rightarrow$ \ref{it:M-operators0}, and \ref{it:M-operators3} $\Rightarrow$ \ref{it:M-operators2} $\Rightarrow$ \ref{it:M-operators4} are clear.

    \ref{it:M-operators2} $\Rightarrow$ \ref{it:M-operators1}: 
    This follows from the bicommutant theorem: Let $a'\in\M'$, then
    \begin{align*}
        a' (1\ox \bra{\Omega_2})a(1\ox\ket{\Omega_1}) 
        &= (1\ox \bra{\Omega_2}) (a'\ox 1)a(1\ox\ket{\Omega_1})\\
        &= (1\ox \bra{\Omega_2}) a(a'\ox 1)(1\ox\ket{\Omega_1})
        = (1\ox \bra{\Omega_2})a(1\ox\ket{\Omega_1}) a'.
    \end{align*}

    \ref{it:M-operators4} $\Rightarrow$ \ref{it:M-operators2}:
    By applying the assumption to self-adjoint elements, it follows that $a^*a\in \M'\ox 1$.
    Let $a'\in\M'$. We have to show that $x=(a'\ox1)a-a(a'\ox1)$ vanishes. This follows from
    \begin{align*}
        x^*x
        &= a^*(a'^*a'\ox1)a- a^*(a'^*\ox1)a(a'\ox1) \\
        &\phantom= \ -(a'^*\ox1)a^*(a'\ox1)a +(a'^*\ox1)a^*a(a'\ox1) =0.\qedhere
    \end{align*}
\end{proof}

As a consequence of \cref{lem:M-operators}, we note the following properties of $\M\barox\B(\K_1,\K_2)$:
\begin{itemize}
    \item $\M\barox\B(\K_1,\K_2)$ is closed in the weak, ($\sigma$-)strong and ($\sigma$-)strong* topologies,
    \item $a\in \M\barox\B(\K_1,\K_2)$ if and only if $a^*\in \M\barox\B(\K_2,\K_1)$ 
    \item the class of $\M$-operators is closed under products, whenever they make sense.
\end{itemize}

We return to discussing cp maps.
We say that a cp map $T:\B(\H\ox\K_2)\to\B(\H\ox\K_1)$ is \emph{$\M$-inner} if it has a Kraus decomposition with $\M$-operators, i.e., if
\begin{equation}\label{eq:M-inner-cp}
    T= \sum_x k_x^*(\placeholder)k_x, \qquad \{k_x\}\subset \M \barox\B(\K_1,\K_2).
\end{equation}

\begin{lemma}\label{lem:winner}
    The following are equivalent for a normal ucp map $T:\B(\H\ox\K_2)\to \B(\H\ox \K_1)$:
    \begin{enumerate}[(a)]
        \item\label{it:winner1} $T$ is $\M$-inner,
        \item\label{it:winner2} $T(a'\ox 1_{\K_2}) = (a'\ox 1_{\K_1})$ for all $a'\in\M'$,
        \item\label{it:winner3} The minimal Stinespring isometry $v$ of $T$ is an $\M$-operator,
    \end{enumerate}
    In this case, all Stinespring isometries $w$ and all Kraus operators $k_x\in \B(\H\ox\K_1,\H\ox\K_2)$ for $T$ are $\M$-operators.
    Moreover, $T$ satisfies
    \begin{equation}\label{eq:winner2}
        T\big(a (a'\ox1_{\K_2})\big) = T(a)(a'\ox1_{\K_1}), \qquad a\in \M\barox\B(\K_2),\ a'\in\M'.
    \end{equation}
\end{lemma}

\begin{proof}
    The implication \ref{it:winner1} $\Rightarrow$ \ref{it:winner2} is clear.
    \ref{it:winner2} $\Rightarrow$ \ref{it:winner3}: If $w$ is a Stinespring isometry, then \cref{it:M-operators4} of \cref{lem:M-operators} holds and, thus, the same Lemma implies that $w$ is an $\M$-operator.

    \ref{it:winner3} $\Rightarrow$ \ref{it:winner1}: Let $v:\H\ox\K_1\to \H\ox\K_2\ox\K$ be the minimal Stinespring isometry and let $\Omega_x\in\K$ be vectors such that $\sum_x\kettbra{\Omega_x}=1$.
    Set $k_x = (1\ox\bra{\Omega_x})v$. This is by definition an $\M$-operator, and we have
    $\sum_x k_x^*(\placeholder)k_x = \sum_x v^*(\placeholder\ox \kettbra{\Omega_x})v= v^*(\placeholder\ox1)v=T$.

    If $\{k_x\}$ is a family of Kraus operators for $T$ then $w = \sum_x k_x \ox \ket{x}: \H\ox\K_1\to\H\ox\K_2\ox\ell^2_X$ is a Stinespring isometry and, hence, an $\M$-operator by the argument in \ref{it:winner2} $\Rightarrow$ \ref{it:winner3} above. Thus, each $k_x=(1\ox\bra x)w\in\M$.
    \Cref{eq:winner2} follows directly from items \ref{it:winner1} or \ref{it:winner3}.
\end{proof}

For our purposes, we need \emph{unitary} dilations in addition to the Stinespring dilation.
A unitary dilation of a normal ucp map $T:\B(\H_2)\to\B(\H_1)$ is pair $(u,\Omega)$ of a unitary $u : \H_1\ox \K_1 \to \H_2\ox\K_2$ for suitable Hilbert spaces $\K_1,\K_2$, together with a unit vector $\Omega\in\K_1$ such that 
\begin{equation}\label{eq:unitary-dilation-math}
    T= (1\ox\bra\Omega)u^* (\placeholder\ox 1)u(1\ox\ket\Omega).
\end{equation}

\begin{lemma}\label{lem:unitary-dilation}
    Let $T:\B(\H\ox\K_1)\to\B(\H\ox\K_2)$ be a normal ucp map.
    Then $T$ has a unitary dilation $(u,\Omega)$ with an $\M$-unitary $u$.
    If $\K_1=\K_2=\CC$, then $(u,\Omega)$ is of the form $\Omega\in \K$, $u\in \M\barox\B(\K)$.
\end{lemma}

\begin{proof}
    For notational simplicity, we write $\H_1\H_2$ for Hilbert space tensor products $\H_1\ox\H_2$ in this proof.
    Let $v:\H\K_1\to \H\K_2\K$ be the minimal Stinespring isometry of $T$, which is $\M$-inner by \cref{lem:winner}.
    Pick unit vectors $\Psi_j \in \K_j$, $j=1,2$, $\Psi\in\K$, and denote the first canonical basis vector of $\ell^2=\ell^2(\NN_0)$ by $\ket0$.
    Consider the partial isometry $ w = (v\ox \ket{\Psi_1}) (1_{\H\K_1} \ox \bra{\Psi_2\ox\Psi}) \in \M\barox\B(\K_1\K_2\K)$ (with the obvious implicit reordering of the tensor factors).
    By \cref{lem:unitary extension}, there is a unitary $u \in \M\barox\B(\K_1\K_2\K \ell^2)$ such that $w\ox \ket0 = u(ww^*\ox \ket0)$. 
    Thus, we have
    \begin{equation}\label{eq:unitary-dilation-help}
        v\ox \ket{\Psi_1}\ox\ket0 
        = (w\ox \ket0) \big(1_{\H\K_1}\ox\ket{\Psi_2\ox\Psi}\big)
        = u\big(1_{\H\K_1}\ox \ket{\Psi_2}\ox\ket\Psi \ox \ket0\big).
    \end{equation}
    Set $\tilde\K_1=\K_2\K\ell^2$, $\tilde\K_2=\K_1\K\ell^2$, $\Omega = \Psi_2\ox\Psi\ox \ket0\in \tilde\K_1$. Then $u$ is an operator $\H\K_1\tilde\K_1\to\H\K_2\tilde\K_2$ and \eqref{eq:unitary-dilation-help} implies \eqref{eq:unitary-dilation-math}.
\end{proof}

\begin{lemma}\label{lem:inner-distillation}
    The following are equivalent for a normal subunital cp map $T:\B(\K_2)\to \B(\H\ox\K_1)$:
    \begin{enumerate}[(a)]
        \item\label{it:inner-distillation1} there is an $\M$-inner subunital cp map $\hat T:\B(\H\ox\K_2)\to \B(\H\ox\K_1)$ with $T= \hat T(1_\H\ox \placeholder)$;
        \item\label{it:inner-distillation2} the range of $T$ is contained in $\M\barox\B(\K_1)$;
        \item\label{it:inner-distillation3} for some/every basis $\{\ket i\}$ of $\K_2$, there is a collection of operators $k_{x,i}\in \M\barox\B(\K_1)$ such that $T(a) = \sum_{xij} a_{ij}\, k_{x,i}^*k_{x,j}$, where $a\in\B(\K_2)$ and $a_{ij}=\bra i a \ket j$.
    \end{enumerate}
    Moreover, $T$ is unital if and only if $\hat T$ is unital if and only if $\sum_{xij} k_{x,i}^*k_{x,j}=1$.
\end{lemma}
\begin{proof}
    \ref{it:inner-distillation1} $\Rightarrow$ \ref{it:inner-distillation2} is trivial. 
    (\ref{it:inner-distillation3}, "some") $\Rightarrow$ \ref{it:inner-distillation1}:
    We set $k_x = \sum_i k_{x,i}\ox \ket i \in \M \barox\B(\K_1,\K_1\ox\K_2)$.
    Subunitality of $T$ gives $\sum_x k_x^*k_x=\sum_{x,i}k_{x,i}^*k_{x,i}\le1$ with equality if and only if $T$ is unital.
    We set 
    \begin{equation*}
        \hat T = \sum_{x} k_x^*(\placeholder \ox 1_{\K_1}\ox \placeholder)k_x: \B(\H\ox\K_2)\to \B(\H\ox\K_1).
    \end{equation*}
    By construction, $\hat T$ is an $\M$-inner subunital cp map with $\hat T(1_\H\ox \ketbra ij) = \sum_x k_x^*(1_{\H\K_1}\ox \ketbra ij)k_x = k_{x,i}^* k_{x,j}= T(\ketbra ij)$.
    Hence, the claim holds.

    \ref{it:inner-distillation2} $\Rightarrow$ (\ref{it:inner-distillation3}, "every"):
    Let $\{\ket i\}$ be a basis for $\K_2$ and let $\lambda_i>0$ be such that $\sum_i \lambda_i=1$.
    We use the Choi isomorphism.
    Consider the contraction $t=\sum_{ij} \lambda_i\lambda_j T(\ketbra ij)\ox\ketbra ij$ on $\H\ox\K_1\ox\K_2$.
    Our assumption implies $t\in\M\ox\B(\K_1\ox\K_2)$ and complete positivity of $T$ implies that is positive.
    Thus, $t$ is of the form $t = \sum_x l_{x,i}^*l_{x,j}\ox \ketbra ij$ for a collection of operators $l_{x,i} \in \M\barox\B(\K_1)$.
    Set $k_{x,i}= \lambda_i^{-1}l_{x,i}$.
    It follows that $T(\ketbra ij) = \sum_x k_{x,i}^*k_{x,j}$, which implies the claim.
\end{proof}

\section{Basics of von Neumann algebraic quantum information theory}\label{sec:vNQI}

\localtableofcontents

\null

In this section, we describe and discuss the framework, which we use to study quantum systems with infinitely many degrees of freedom from an information-theoretic point of view.
The technical results of this section are mostly known (a notable exception is the axiomatic approach in \cref{sec:axioms}), but the operational point of view that we adopt is new.

Two of the fundamental questions in quantum information theory are \cite{phyjou}
\begin{center}\it
    What information can we write into and read from quantum systems?
\end{center}
Depending on the context, these questions involve constraints that must be fulfilled when interacting with quantum systems.
We like to think of a hypothetical agent that has access to operations that are compatible with these constraints.
The prototypical example is a \emph{subsystem constraint}, which means that the agent can only act on a subsystem of the full system. 
We study how well agents with partial access to a full system perform at \emph{operational tasks}, e.g., distinguishing a given pair of states on the full system.
Under suitable assumptions, these agents are modeled by von Neumann algebras of operators on the full system's Hilbert space.
In particular, this assumption applies to the case where an agent has access to the operations on a subsystem.

Our motivation for studying von Neumann algebraic quantum information theory is not mere curiosity about quantum information in a most general setting, in which quantitatively and qualitatively new effects become possible.
We are interested in information-theoretic properties of concrete physical models of systems with infinitely many degrees of freedom.
These are naturally described by von Neumann algebras \cite{naaijkensQuantumSpinSystems2017,haag_local_1996,gabbiani_operator_1993}, and it is often possible to understand their mathematical properties explicitly, e.g., the type classification \cite{haag_local_1996,gabbiani_operator_1993,naaijkensAnyonsInfiniteQuantum2012,jones_local_2025,van_luijk_critical_2025} or the index of some subfactor inclusion \cite{kawahigashi_subfactor_2005,kawahigashi_subfactor_2005,naaijkens_kosaki-longo_2013,longo_index_1989,longo_index_1990}.
A good understanding of the interplay of algebraic properties on one hand and operational ones on the other hand yields novel insights about physical models \cite{summers_vacuum_1985,van_luijk_relativistic_2024,van_luijk_critical_2025}.



\subsection{Operations on full quantum systems}\label{sec:operations-full-sys}

In this subsection, we recall the mathematical formulation of the operational building blocks of Quantum Theory, which remains untouched since the early days of \cite{ludwig_versuch_1964,ludwig_deutung_1970,kraus_general_1971,hellwig_pure_1969,hellwig_operations_1970,davies_operational_1970}.

Let us start with the usual opening words: A full quantum system is described by a Hilbert space $\H$.
We have added the word `full' to this standard textbook sentence, as we will later consider more general quantum systems that cannot be described in the purely Hilbert space-based formulation of quantum mechanics.
However, each of these more general quantum systems arises as a subsystem of a full quantum system, hence the name.

The \emph{states} of a full quantum system are described either by density operators $\rho$ on the Hilbert space $\H$ or by normal states $\phi:\B(\H)\to\CC$ on the von Neumann algebra $\B(\H)$.
These points of view are equivalent via $\phi = \tr \rho (\placeholder)$.
We adopt the latter as it suits our purposes better, and write $\nstates(\H)$ for the normal state space of $\B(\H)$.
Binary measurements, i.e., measurements with two outcomes, labeled "yes" and "no", are described by operators $0\le f \le 1$ on $\H$ called \emph{effects}.
If a binary measurement with effect $f$ is measured when the system's state is $\phi$, the success probability, i.e., the probability of the outcome "yes", is
\begin{equation}\label{eq:probability}
    \phi(f) \in [0,1].
\end{equation}
The effect $1-f$ measures the probability of the outcome "no". Indeed, $\phi(1-f)= \phi(1)-\phi(f) = 1-\phi(f)$.
There are two trivial effects: $f=0$ and $f=1$ corresponding to the trivial measurements that output "no" and "yes" deterministically for all states $\phi$.\footnote{More generally, a scalar effect corresponds to flipping a biased coin. These are precisely the effects for which the success probability does not depend on the state of the system.}

We now discuss general measurements, but restrict ourselves to the case where the outcome space $X$ is discrete (and countable).
A measurement $m$ is described by a collection $(m_x)_{x\in X}$ of effects $m_x$ describing the binary measurement "has the outcome $x$ been measured?".
Linearity of \eqref{eq:probability} in $f$ shows that the sum $\sum_{x\in A} m_x$ is the effect corresponding to whether the outcome is in the set $A$.
Since every application of the measurement results in some output, we have $\sum_x m_x=1$.
The resulting object $(m_x)$ is called a positive operator-valued measure (POVM).
In this formulation, a binary measurement is described by the POVM $(f,1-f)$.
As a consequence of \eqref{eq:probability}, the outcomes of the measurement with POVM $(m_x)_{x\in X}$ are distributed according to the probability distribution $(p_x)_{x\in X}$, $p_x = \phi(m_x)$, where $\phi$ is the state of the system.

Quantum channels on the system are described by normal unital cp maps $T$ on $\B(\H)$ (Heisenberg picture) or by normalization-preserving cp maps $T_*$ on the predual $\B_{*}\hspace{-1pt}(\H)$ (Schrö\-dinger picture).\footnote{Under the identification of $\predualB(\H)$ with the trace class $\T(\H)$, these become precisely the trace-preserving cp maps.}
The two pictures are related by the equation
\begin{equation}\label{eq:schrödinger picture}
   T_*(\phi)=\phi\circ T.
\end{equation}
For our purposes, it will be more natural to work in the Heisenberg picture; we will use the term "quantum channel" synonymously with "normal ucp map".
Every quantum channel has a unitary dilation: It can be implemented by preparing an ancillary quantum system $\K$ in a pure state $\Omega\in\K$, applying a unitary $u$ on the joint system $\H\ox\K$, and then discarding the ancilla (see \cref{lem:unitary-dilation}), i.e.,
\begin{equation}\label{eq:unitary-dilation}
    T = (1\ox \bra\Omega)u^*(\placeholder\ox1)u(1\ox\ket\Omega).
\end{equation}
Setting $v = u(1\ox\ket\Omega):\H\to\H\ox\K$, one obtains a, not necessarily minimal, Stinespring dilation $T = v^*(\placeholder\ox1)v$ (see \cref{sec:M-operators}).\footnote{While it is true that every channels has a unitary dilation, it is not guaranteed that the minimal Stinespring dilation arises this way. As an example, consider $T= v^*(\placeholder)v$ with $v$ a proper isometry, e.g., the right shift on $\ell^2$.}
Moreover, every quantum channel has a Kraus decomposition
\begin{equation}\label{eq:kraus form}
    T = \sum_x k_x^*(\placeholder) k_x,\qquad \sum_x k_x^*k_x =1,
\end{equation}
with Kraus operators $\{k_x\}$ (see \cref{sec:M-operators}).
The smallest number of Kraus operators is the \emph{Kraus rank}. Equivalently, the Kraus rank is the dimension of the Hilbert space $\K$ appearing in the minimal Stinespring dilation (see \cref{lem:stinespring}) of $T$.


Consider again a quantum channel $T$ with unitary dilation $u\in \U(\H\ox\K)$, $\Omega\in \K$ (cp.\ \eqref{eq:unitary-dilation}).
We now consider what happens when an agent performs a measurement $(n_x)$ on the ancillary system instead of discarding it.
Since the measurement outcome depends on the state of the main system, this defines a measurement $(m_x)$ on $\H$, which is readily checked to be given by
\begin{equation}\label{eq:measurement-dilation}
    m_x = v^*(1\ox n_x)v, \qquad v= u(1\ox\ket\Omega).
\end{equation}
In addition to a measurement on $\H$, this setup includes a description of the \emph{post-measurement states}: 
If $\phi$ is the initial state of the system and if $x$ is the observed outcome, then the post-measurement state $\phi_x'$ is 
\begin{equation}\label{eq:post-measurement}
    \phi_x' = \frac1{p_x} (T_x)_*(\phi),
\end{equation}
where $p_x = \phi(m_x)$ is the probability of the outcome $x$, and $(T_x)_*$ is the Schrödinger picture version of the subunital cp map
\begin{equation}\label{eq:dilation of Tx}
    T_x(a) = v^*(a\ox n_x)v, \qquad a\in\B(\H).
\end{equation}
The maps $T_x$ satisfy $\sum_x T_x=T$ and include a description of the measurement $m_x$ via $m_x=T_x(1)$.
This structure is known as an \emph{instrument}: An instrument is a collection of cp maps $(T_x)_{x\in X}$ such that $T=\sum_x T_x$ is a quantum channel.\footnote{We have assumed that the post-measurement states are states of the given system. More generally, they could be states of other systems. In particular, measurements are instruments where the output system is trivial.}
It describes a general measurement apparatus, including a description of post-measurement states.
Importantly, every instrument is of the above form: If $T = v^*(\placeholder\ox1)v$ is a Stinespring dilation of $T$, then there exists a POVM $(n_x)$ on the ancilla such that \eqref{eq:dilation of Tx} holds.
Every instrument $(T_x)$ has a Kraus representation: There exist Kraus operators $\{k_{y|x}\}$ such that
\begin{equation}\label{eq:instrument Kraus}
    T_x = \sum_{y} k_{y|x}^*(\placeholder)k_{y|x}, \qquad \sum_{y}k_{y|x}^*k_{y|x}=T_x(1).
\end{equation}
In the case that each $T_x$ has Kraus rank $1$, i.e., $T_x = k_x^*(\placeholder)k_x$ for some $k_x\in\B(\H)$, we call $(T_x)$ an \emph{instrument with Kraus rank 1}.
These instruments have the property that the post-measurement states of all pure states are pure.
By \eqref{eq:instrument Kraus}, every instrument arises from an instrument with Kraus rank 1 via coarse-graining the output space.

Finally, we discuss topologies on the set of operations.
Let us begin with effects: A net $(f_n)$ converges to an effect $f$ in expectation values, i.e., $\lim_n \phi(f_n)=\phi(f)$ for all $\phi\in\nstates(\H)$, if and only if $f_n\to f$ in the $\sigma$-weak operator topology.
Indeed, this follows because $\predualB(\H)$ is spanned by the normal state space $\nstates(\H)$.
Similarly, a net $(m_{n;x})$ of POVMs with outcome space $X$ (independent of $n$), converges to a POVM $(m_x)$ in expectation values if and only if $m_{m;x}\to m_x$ in $\sigma$-weak operator topology for all $x\in X$.
On the set of channels, we have two natural topologies:
\begin{itemize}
    \item the \emph{topology of convergence in expectation values}, where a net $(T_n)$ converges to a channel $T$ if $\lim_n \phi(T_n(a))=\phi(T(a))$ for all $a\in\B(\H)$, $\phi\in\nstates(\H)$, and 
    \item the \emph{topology of norm convergence on states}, where a net $(T_n)$ converges to a channel $T$ if $\lim_n \norm{(T_n)_*(\phi)-T_*(\phi)}=0$ for all $\phi\in\nstates(\H)$.
\end{itemize}
Clearly, the former is a stronger topology than the latter.
By the above, we identify the topology of convergence in expectation values as the \emph{point-$\sigma$-weak topology} (see \cref{sec:hilbert-spaces}).
Likewise, the topology of norm convergence on states, is the point-norm topology of the predual maps $T_*$.
Both of these topologies natural extend to the set of instruments with a fixed outcome space $X$.

\subsection{What is a subsystem and why is it a von Neumann algebra?}\label{sec:agents}

The systems that are considered in quantum (information) theory are often effective systems rather than fundamental ones.
Examples range from effective qubit systems, e.g., transmon qubits \cite{transmon}, to the code spaces in quantum error correction \cite{kitaevFaulttolerantQuantumComputation2003,kitaevAnyonsExactlySolved2006} or quasi-particles like phonons or anyons.
In the mathematical analysis, it does not matter how a system is realized physically.
An electron's spin degree of freedom is just as much a "qubit" as every other effective two-dimensional system.
Whenever emergent systems arise, they arise as subsystems of larger quantum systems that are composed of fundamental systems.

We aim for an operational treatment of quantum systems that are more general than those described by a Hilbert space alone, which we call \emph{full} quantum systems.
Here, we consider quantum systems that arise as \emph{subsystems} of full quantum systems.
As we will see later, this is the most general case.
But what qualifies as a subsystem? --- As discussed in the previous paragraph, this does not have a unique answer. 
In the following, we adopt a purely operational point of view by considering hypothetical \emph{agents with partial access to a full quantum system}.
Such agents are specified by the collection $\O$ of operations (measurements, channels, and instruments) that they can implement on the full system.
We identify a subsystem with an agent that can perform those operations on the full subsystem that only act on the subsystem.
If the set of operations of the agent behaves like the set of operations on a subsystem (this will be made precise), we regard the agent's operations as those that only act on a (virtual) subsystem of the full system.

In von Neumann algebraic quantum information theory, we consider agents with partial access to a full quantum system whose implementable operations are captured by a von Neumann algebra, called the \emph{observable algebra}, on the Hilbert space of the full quantum system.
Before we discuss the justification and limitations of this assumption, particularly in regard to describing subsystems, we describe the resulting set of implementable operations.

Let us denote by $\H$, the Hilbert space of the full quantum system, and let $\M\subset \B(\H)$ be the observable algebra associated with an agent with partial access to the full system.
We begin by specifying the agent's instruments as follows:
\begin{center}\it
    The agent can implement instruments by preparing an ancillary system with Hilbert \\ space $\K$, applying a unitary in $\U(\M \barox \B(\K))$, and then measuring the ancilla.
\end{center}
The implementable measurements and channels are derived from this: 
The agent can implement channels/measurements by discarding the classical/quantum output of an instrument they can implement.
Let us repeat the above in mathematical terms: The agent can implement all instruments $(T_x)$ on $\H$ that are of the form 
\begin{equation}\label{eq:implementable instruments}
    T_x = (1\ox \bra\Omega)u^* ( \placeholder\ox n_x)u(1\ox\ket\Omega),
\end{equation}
for some Hilbert space $\K$, a unit vector $\Omega\in \K$, a POVM $(n_x)$ on $\K$ and a unitary $u\in\M\barox\B(\K)$.
Here, we have assumed without loss of generality that the ancilla is prepared in a pure state. Indeed, if the state of the ancilla were mixed, one could pass to a purification to obtain a pure state-unitary dilation of the instrument.
The repeatable measurements and channels are then given by \eqref{eq:measurement-dilation} and \eqref{eq:unitary-dilation}, respectively.

As a consequence of the results in \cref{sec:M-operators}, in particular, \cref{lem:unitary-dilation}, we will show that the agent can implement precisely the set $\O_\M$ of \emph{$\M$-inner} operations.
We have already discussed $\M$-inner channels (under the name `normal ucp maps') in \cref{sec:M-operators}.
An instrument $(T_x)$ is $\M$-inner if each $T_x$ is an $\M$-inner cp map.
This means that the instrument has a Kraus decomposition with Kraus operators $\{k_{y|x}\}$ in $\M$ (see \eqref{eq:instrument Kraus}).
A POVM $(m_x)$ on $\H$ is $\M$-inner if $m_x\in\M$ for all $x$.

\begin{proposition}\label{prop:implementable}
    Consider an agent with observable algebra $\M$.
    Then the agent's operations, as defined above, are exactly the $\M$-inner ones.
\end{proposition}
\begin{proof}
    By definition, all operations of the agent are $\M$-inner. We have to show the converse.
    The Kraus decomposition tells us that $\M$-inner channels arise from discarding the classical output of an $\M$-inner instrument.
    If $(m_x)$ is a $\M$-inner POVM, then $T_x= m_x^{1/2}(\placeholder)m_x^{1/2}$ is an $\M$-inner instrument from which $(m_x)$ arises by discarding the quantum output.
    It remains to show that every $\M$-inner instrument $(T_x)$ arises as in \eqref{eq:implementable instruments}.
    Let $v:\H\to\H\ox\K$ be the minimal Stinespring isometry of $T=\sum_x T_x$.
    By \cref{lem:stinespring}, there is a POVM $(n_x)$ on $\K$ such that $T_x= v^*(\placeholder\ox n_x)v$ for all $x$.
    By \cref{lem:unitary-dilation}, the $\M$-inner channel $T$ has a unitary dilation $T= (1\ox\bra\Omega)u^*(\placeholder\ox1)u(1\ox \ket\Omega)$ for some unit vector $\Omega\in\tilde\K$, and $u\in\U(\N\ox\B(\tilde\K))$.
    Since $u(1\ox\ket\Omega)$ is a Stinespring isometry for $T$, \cref{lem:stinespring} implies an isometry $w:\K\to\tilde\K$ with $(1\ox w)v = u(1\ox\ket\Omega)$.
    Let $(p_x)$ be a probability distribution and set $\tilde n_x= w q_x w^*+p_x (1-w^*w)$.
    Then $(\tilde n_x)$ is a POVM on $\K$ and $T_x = (1\ox \bra\Omega)u^*(\placeholder\ox \tilde n_x)u(1\ox\ket\Omega)$.
\end{proof}

We highlight some consequences of \cref{prop:implementable} before we come to the discussion of the assumption that the agent's operations are described by a von Neumann algebra.
The name `observable algebra' for the von Neumann algebra $\M$ associated with an agent is due to the following:
Since the agent's measurements are described by the POVMs in $\M$, the set of effects that the agent can implement is precisely the unit interval
\begin{equation}
    [0,1]_\M = \{ 0\le f \le 1\ :\ f\in\M\}.
\end{equation}
Hence, $\M$ is generated by the agent's observables (= measurements).
In \cref{sec:M-operators}, we discussed equivalent characterization of $\M$-inner channels.
A notable one is the following: A channel $T$ on the full system is $\M$-inner if and only if it acts trivially on the commutant $\M'$. Hence, \cref{prop:implementable} shows
\begin{equation}
    \text{$T$ is implementable} \quad \iff \quad T \restriction \M' = \id.
\end{equation}
Moreover, we see that one agent is more powerful (i.e., can implement more operations) than another one, if and only if the other agent's observable algebra is a subalgebra:
\begin{equation}\label{eq:inclusion of observable algs}
    \M\subset \N \quad \iff \quad \O_\M\subset\O_\N.
\end{equation}
Thus, the study of inclusions of von Neumann algebras corresponds to the study of a pair of agents where one is more powerful than the other.
Let us also mention that the set of $\M$-inner operations is closed in the topology of convergence in expectation values, and, hence, closed in the topology of norm convergence on states.
Indeed, if $T_n$ is a net of $\M$-inner channels that converges to $T$ in expectation values, then $T_n \restriction \M'=\id$ implies $T\restriction \M'=\id$, which shows that is $\M$-inner as well (see \cref{lem:winner}). 
The analogous statement holds for instruments.
If $(f_n)$ is a net of effects in $\M$ and $f$ is an effect such that $f_n$ converges to $f$ in expectation values, i.e., $\lim_n\phi(f_n) =\phi(f)$ for all $\phi\in\nstates(\H)$, i.e., $f_n\to f$ in the $\sigma$-weak operator topology, then $f\in \M$ since von Neumann algebras are $\sigma$-weakly closed.
Applying this to all effects in a POVM, we see that the set of measurements is closed in the topology of convergence in expectation values.

\null

Having discussed the mathematical consequences, we now come to the motivation and justification of the assumption that the operations of an agent with partial access to a full quantum system are described by a von Neumann algebra as above.
In the following, we focus on a few key points. A purely operational characterization of the assumption will be shown in \cref{sec:axioms}.

Let us begin by saying that we neither claim that every von Neumann algebra corresponds to a physically relevant setup, nor that all settings in which some agent has partial access to a quantum system should be modeled with a von Neumann algebra.\footnote{The latter is obviously false for an agent that can only choose between two measurements.}
However, we believe that many relevant settings exist in which the above description is appropriate, particularly as a description of subsystems.

Perhaps, the strongest motivation is an abundance of examples.
The prime source of nontrivial examples comes from studying sectors of quantum systems with infinitely many degrees of freedom.
The choice of sector determines a Hilbert space $\H$, describing the `full system', and subsystems, e.g., field excitations in certain localization regions, are described by von Neumann algebras on $\H$ \cite{haag_local_1996,naaijkensQuantumSpinSystems2017,gabbiani_operator_1993}.
We discuss this class of examples in detail in the settings of quantum many-body systems and quantum field theory in \cref{sec:many-body,sec:qft}, respectively.
In \cref{sec:idealized}, we study how the idealization of having infinitely many entangled states in quantum information theory can be studied rigorously in the von Neumann algebraic setting.

The existence of many examples does not rule out the possibility that the set of $\M$-inner operations might have certain undesirable features for sufficiently weird von Neumann algebras.
Moreover, it does not suddenly turn the mathematical assumption that an agent with partial access to a full quantum system is modeled by a von Neu\-mann algebra into an operational one.
We address these points by showing that this mathematical assumption is, in fact, equivalent to a set of purely operational axioms.
As the arguments require some space, we devote a separate subsection to this (see \cref{sec:axioms}).
The upshot is that we can arrive at our assumption, and, hence, at the definition of a von Neumann algebra in purely operational terms.


Next, we show that the assumption that an agent has access to the set of $\M$-inner operations for some von Neumann algebra $\M$ can always be regarded as an idealization.

\begin{proposition}\label{prop:idealization}
    For every set $\O$ of operations on a quantum system with Hilbert space $\H$, there is a unique smallest von Neumann algebra $\M$ such that every operation in $\O$ is $\M$-inner.
\end{proposition}
\begin{proof}
    Let $\M$ be the von Neumann algebra generated by the Kraus operators of channels or instruments in $\O$ and by the effects of measurements in $\O$.
    Clearly, this makes every operation in $\O$ an $\M$-inner one.
    Conversely, every von Neumann algebra $\N$ such that $\O\subseteq\O_\N$ must contain these elements (see \cref{sec:M-operators}). Thus, we have $\M\subset\N$.
\end{proof}

A setting in which a description via von Neumann algebras arises naturally is that of \emph{symmetries}, where the implementable operations are those that commute with a symmetry group $G$ of unitaries on $\H$.
This point of view arises in various contexts in physics, e.g., in the setting of quantum reference frames \cite{bartlett_reference_2007,fewster_quantum_2024,wigner1995messung,yanase1961optimal,araki_measurement_1960}.
The set of operations that have a unitary implementation which is covariant with respect to the trivial group representation on the ancilla system, is precisely the set of $\M$-inner operations with $\M$ the commutant of the symmetry group:

\begin{lemma}\label{lem:covariance}
    Let $G$ be a group and let $g\mapsto v_g$ be a unitary representation on $\H$. The following are equivalent for an instrument $(T_x)$:
    \begin{enumerate}[(a)]
        \item\label{it:covariance1} $(T_x)$ has a covariant unitary dilation: There is a Hilbert space $\K$, a unitary $u\in\U(\H\ox\K)$ commuting with $\{v_g\ox1\}_{g\in G}$, a POVM $(n_x)$ on $\K$, and a vector $\Omega\in\K$ such that \eqref{eq:implementable instruments} holds.
        \item\label{it:covariance2} $(T_x)$ is $\M$-inner with $\M = \{v_g : g\in G\}'$.
    \end{enumerate}
\end{lemma}
\begin{proof}
    Since $\{v_g\ox 1: g\in G\}= \M\barox\B(\K)$, the commutativity assumption in \ref{it:covariance1} is equivalent to $u \in \M \barox \B(\H)$. Thus, the statement follows from \cref{prop:implementable}.
\end{proof}

We now consider a general setup of approximately implementable operations.
This is an abstraction of what happens in the examples discussed in \cref{sec:many-body,sec:idealized}.
We consider a sequence (or a net) of agents $A_n$ with partial access to a full quantum system with Hilbert space $\H$, e.g., each agent could have access to a finite-dimensional subsystem.
We assume the agents to be of increasing capability, i.e., if $n\ge m$, then the agent $A_n$ can implement all operations of the agent $A_m$.
By \eqref{eq:inclusion of observable algs}, this corresponds to an increasing sequence (net) $(\M_n)$ of von Neumann algebras on $\H$.

We say that a quantum channel $T$ can be \emph{implemented approximately} by the agents $(A_n)_n$ in a topology on the set of channels, if there is a net $(T_\alpha)$ of channels such that each $T_\alpha$ is $\M_{n_\alpha}$-inner for some $n_\alpha$ (and hence $\M_k$-inner or all $k\ge n_\alpha$) and such that $T_\alpha$ converges to $T$ in the chosen topology.

\begin{lemma}\label{lem:approx-inner}
    Let $T$ be a quantum channel.
    The following are equivalent:
    \begin{enumerate}[(a)]
        \item\label{it:approx-inner1} $T$ can be implemented approximately by the agents $(A_n)$ in the topology of convergence in expectation values (= the point-$\sigma$-weak topology),
        \item\label{it:approx-inner2} $T$ can be implemented approximately by the agents $(A_n)$ in the topology of norm convergence on states,
        \item\label{it:approx-inner3} $T$ is $\M$-inner, where $\M = \bigvee_n \M_n$ the von Neumann algebra generated by the sequence (net) $(\M_n)$.
    \end{enumerate}
\end{lemma}
\begin{proof}
    \ref{it:approx-inner2} $\Rightarrow$ \ref{it:approx-inner1} is clear. \ref{it:approx-inner1} $\Rightarrow$ \ref{it:approx-inner3} follows from the point-$\sigma$-weak closure of the set of $\M$-inner channels (and the fact that all $\M_n$-inner operations are $\M$-inner).
    It remains to show \ref{it:approx-inner3} $\Rightarrow$ \ref{it:approx-inner2}:
    We let $\M_0$ denote the unital *-algebra generated by the von Neumann algebras $(\M_n)$.
    The point-$\sigma$-weak closure of the set of $\M$-inner channels implies the "only if" statement.
    To see the converse, consider a unitary dilation $T = (1\ox \bra\Omega)u^*(\placeholder\ox1)u(1\ox\ket\Omega)$ with $u\in\M\barox\B(\K)$, $\Omega\in\K$.
    By Kaplansky's density theorem \cite[Thm.~II.4.8]{takesaki2}, we can approximate $u$ with a net of contraction $(u_\alpha)$ in $\M_0\odot\B(\K)$.
    We consider the subunital cp maps $S_\alpha = (1\ox\bra\Omega)u_\alpha^*(\placeholder\ox1)u_\alpha(1\ox\ket\Omega)$ and $\omega = \ip\Omega{(\placeholder)\Omega}$.
    Then 
    \begin{equation*}
        \norm{T_*(\phi)-(S_\alpha)_*(\phi)} 
        \le \norm{u(\phi\ox\omega)u-u_\alpha(\phi\ox\omega)u_\alpha^*} \to 0
    \end{equation*}
    for all $\phi\in\nstates(\H)$.
    By construction each $S_\alpha$ is $\M_{n_\alpha}$-inner for some $n_\alpha$.
    However, $S_\alpha$ is not a channel since it might be subunital.
    We can, however, turn it into a subunital channel by considering $T_\alpha = S_\alpha + k_\alpha$, $k_\alpha = (1-S_\alpha(1))^{1/2} \in \M_{n_\alpha}$.
    By construction, $T_\alpha$ is an $\M_{n_\alpha}$-inner channel.
    Since $k_\alpha = [(1\ox \bra\Omega)(1-u_\alpha^*u_\alpha)(1\ox \ket\Omega)]^{1/2}$, the $\sigma$-strong* convergence $u_\alpha\to u$ implies that $k_\alpha\to 0$ $\sigma$-strongly. 
    Indeed, this follows from the strong continuity of the functional calculus (see \cite[Thm.~II.4.7]{takesaki1}) and the fact that the strong and $\sigma$-strong operator topologies coincide on the unit ball (see \cite[Lem.~II.2.5]{takesaki1}).
    Thus, we have $\norm{(S_\alpha)_*(\phi)- (T_\alpha)_*(\phi)} \to 0$ for all $\phi\in\nstates(\H)$, which shows that $T_\alpha$ converges to $T$ as specified.
\end{proof}

The assertion in \cref{lem:approx-inner} can be extended to instruments. 
In fact, the same proof applies if one accounts for the measurements $(n_x)$ on the ancillary system.
As a consequence, we obtain an approximation statement for measurements:
A measurement $(m_x)$ of the full system can be approximated by measurements of the agents $(A_n)$ in the topology of convergence in expectation values if and only if it is $\M$-inner, i.e., $m_x\in\M$ for all $x$.

\subsection{Axiomatic approach}\label{sec:axioms}

In this subsection, we consider a set of operational axioms describing what an agent with partial access to a full quantum system can do.
We show that these axioms are equivalent to the assumption that the agent's operations are captured by a von Neumann algebra as discussed in \cref{sec:agents}.
Our axioms are formulated in terms of a collection of instruments the agent can implement.
Quantum channels are identified with instruments that have a single classical outcome.
x
Our axioms allow the agent to freely act on ancillary quantum systems, including preparing and discarding ancillas. 
This means that they can implement operations whose input and output systems have Hilbert spaces $\H_1=\H\ox\K_1$ and $\H_2=\H\ox\K_2$, where we identify $\H\ox\CC$ with $\H$.
In this more general setting, a quantum instrument $(T_x)$ is a collection of subunital cp maps $T_x:\B(\H_2)\to\B(\H_1)$ and a quantum channel is a normal ucp map $T:\B(\H_2)\to\B(\H_1)$.

Our axioms are as follows:
\begin{enumerate}[({A}1)]
    \setcounter{enumi}{-1}
    \item\label{axiom:consistency}
        Consistency with the statistical interpretation:
        The set of operations is closed under composition.

    \item\label{axiom:limits}
    Limits: The set of implementable channels is closed in the topology of convergence in expectation values.

    \item\label{axiom:ancilla} 
    Full control over ancillary systems: 
    The agent can prepare and discard ancillary systems freely, and they can perform arbitrary operations between ancillary systems.

    \item\label{axiom:reverse} Reversibility: The agent can implement the inverse of every implementable unitary quantum channel.
    
    \item\label{axiom:dilation} Dilations: 
    Every implementable quantum channel can be implemented by preparing an ancillary system, applying a unitary quantum channel, and discarding the ancillary system.

    
    \item\label{axiom:completeness} Selection-completeness: Let $T\up y= (T_{x}\up y)$, $y\in Y$, be a finite collection of implementable instruments. Write $P_\phi(x|y)$ for the probability of the outcome $x$ if the measurement $T\up y$ was applied to the state $\phi$.
    For every "selection" map $y\mapsto x_y$ such that
    \begin{equation}\label{eq:axiom completness}
        \sum_y P_\phi(x_y | y) \le 1 \qquad \text{for all states $\phi$},
    \end{equation}
    there is an implementable instrument $(S_y)$ with outcome space $Y\cup\{\perp\}$ such that, for every input state, the probability and post-measurement state of an outcome $y\in Y$ of $(S_y)$ are the same as for the outcome $x_y$ of $T\up y$.
\end{enumerate}

Some remarks are in order. 
The convex combinations mentioned in \ref{axiom:consistency} correspond to statistical mixtures.
This assumption could be dropped here as it can be deduced from \ref{axiom:completeness}.
Axioms \ref{axiom:dilation} and \ref{axiom:reverse} both encapsulate the idea that the only source of irreversibility is the loss of information (see, e.g., \cite{bennett_demons_1987,gregoratti_quantum_2003}).
Roughly speaking, axiom \ref{axiom:completeness} says that any operation that can be implemented by means of postselection can also be implemented at once.

To state the main theorem, we need to extend the notion of $\M$-inner operations to the setting with ancillas.
An instrument $(T_x)$ with $T_x:\B(\H\ox\K_2)\to\B(\H\ox\K_1)$ is $\M$-inner if it has a Kraus decomposition \eqref{eq:instrument Kraus} with Kraus operators $k_{y|x}$ in $\M\barox\B(\K_1,\K_2)$.
In particular, this applies to quantum channels, which are identified with instruments with a single output.

\begin{theorem}\label{thm:axioms}
    The operational axioms \ref{axiom:consistency} to \ref{axiom:completeness} hold if and only if there is a von Neumann algebra $\M$ such that the agent can implement exactly those operations that are $\M$-inner.
\end{theorem}

We say that an isometry $v:\H\ox\K_1\to\H\ox\K_2$ is implementable by the agent if the quantum channel $T=\Ad(v^*)$ is implementable.
In particular, we use this language for unitaries. By \ref{axiom:reverse}, the set of implementable unitaries is closed under adjoints.

We emphasize that axiom \ref{axiom:dilation} also makes sense (and is assumed) in the presence of ancillas.
To be precise, the mathematical statement of \ref{axiom:dilation} in the presence of ancillary systems is: 
\begin{enumerate}
    \item[(A4)] 
    For every implementable channel $T:\B(\H\ox\K_2)\to\B(\H\ox\K_1)$, there exist Hilbert spaces $\tilde\K_1,\tilde\K_2$, a unit vector $\Omega\in\tilde\K_1$ and an implementable unitary $u:\H\ox\K_1\ox\tilde\K_1\to\H\ox\K_2\ox\tilde\K_2$ such that $T = (1\ox\bra\Omega)u^*(\placeholder\ox1)u(1\ox\ket\Omega)$.%
    \footnote{This formulation, in which the input and output and output ancillary systems $\tilde\K_1$ and $\tilde\K_2$ may be different, is equivalent to the same statement with $\tilde\K_1=\tilde\K_2$ if $\K_1=\K_2$.}
\end{enumerate}
For clarity, we formulate \ref{axiom:completeness} in explicit mathematical language:
\begin{enumerate}
    \item[(A5)] 
    If $T\up y =(T_x\up y)$, $y\in Y$, is a finite collection of implementable instruments and $y\mapsto x_y$ is a selection of outcomes with $\sum_y T_{x_y}\up y$ subunital, there is an implementable instrument $(S_y)$ with outcome space $Y\cup\{\perp\}$ and $S_y=T_{x_y}\up y$.
\end{enumerate}
This formulation is, of course, equivalent to the operational formulation presented above.

We are now starting the proof of the "only if" statement in \cref{thm:axioms}. Afterward, we show the converse, which is essentially contained in \cref{prop:implementable}.
Until further notice, we consider an agent with partial access to a full system with Hilbert space $\H$, and assume that the axioms \ref{axiom:consistency} to \ref{axiom:completeness} hold.
In the following, the Hilbert spaces of ancillary systems are denoted $\K,\K_1,\K_2,\ldots$.
We say that the agent can implement an isometry $v:\H\ox\K_1\to\H\ox\K_2$ if they can implement the quantum channel $\Ad(v^*): \B(\H\ox\K_2)\to\B(\H\ox\K_1)$.
In particular, this applies to unitaries.

\begin{lemma}\label{lem:implement identity}
    The agent can implement the identity channel on $\H$
\end{lemma}
\begin{proof}
    By \ref{axiom:dilation}, the agent can implement a unitary on a larger system and by \ref{axiom:reverse}, they can implement the inverse of this unitary.
    Thus, by \ref{axiom:consistency}, they can implement the identity on a larger system. 
    In combination with \ref{axiom:ancilla}, this implies that they can implement the identity on $\H$. 
\end{proof}

Since the agent has full control over ancillary systems (see axiom \ref{axiom:ancilla}) and since they can implement the identity on $\H$, axiom \ref{axiom:consistency} ensures that they can implement all quantum channels of the form
\begin{equation*}
    T = \id\ox\;\! S: \B(\H\ox\K_2)\to\B(\H\ox\K_1), \qquad S:\B(\K_2)\to\B(\K_1).
\end{equation*}

\begin{lemma}\label{lem:implement stinespring}
    The agent can implement the minimal Stinespring isometry of every channel they can implement.
\end{lemma}

\begin{proof}
    Let $T : \B(\H\ox\K_2)\to\B(\H\ox\K_1)$ be an implementable channel.
    By \ref{axiom:dilation}, there are Hilbert spaces $\tilde\K_1,\tilde\K_2$ and an implementable unitary $u:\H\ox\K_1\ox\tilde\K_1\to\H\ox\K_2\ox\tilde\K_2$ such that $T = s^*(\placeholder\ox 1_{\tilde\K_2})s$, where $s= u (1\ox\ket\Omega)$ for some unit vector $\Omega\in\tilde\K_1$.
    By \ref{axiom:consistency}, $s$ is an implementable isometry.
    Let $v:\H\ox\K_1 \to \H_2\ox\K_2\ox\K$ be the minimal Stinespring isometry of $T$.
    By \cref{lem:stinespring}, there is an isometry $w:\K \to \tilde\K_2$ such that $s = (1\ox w)v$.
    Hence, we have $v = (1\ox w^*)s$.
    Let $S: \B(\tilde \K_2) \to \B(\K)$ be some quantum channel such that $S \ge w^*(\placeholder)w$.
    By \ref{axiom:ancilla}, the agent can implement the channel $S$.
    Consequently, they can implement the product $\Ad(s^*)\circ (\id\ox S) = \Ad(s^*)\circ
    \Ad(1\ox w)= \Ad(v^*)$, i.e., they can implement the Stinespring isometry $v$.
\end{proof}


\begin{lemma}\label{lem:implement-all-instruments}
    If the agent can implement a quantum channel $T$, they can implement every instrument $(T_x)$ with $\sum_x T_x=T$.
\end{lemma}
\begin{proof}
    This follows from \ref{axiom:ancilla}, \cref{lem:implement stinespring}, and the cp-order interval property of the Stinespring dilation (see \cref{lem:stinespring}), which implies that every instrument $(T_x)$ with $\sum_x T_x=T$ arises by measuring the ancilla system in the Stinespring dilation of $T$.
\end{proof}

\begin{lemma}\label{lem:compactness}
    The quantum channels $T:\B(\H\ox\K_1)\to\B(\H\ox\K_2)$ that the agent can implement form a compact set in the point-$\sigma$-weak topology.
\end{lemma}
\begin{proof}
    Recall that the unit ball of a von Neumann algebra is $\sigma$-weakly compact (this follows from the Banach-Alaoglu theorem).
    This implies the well-known fact that the set of normal ucp maps on $\B(\H)$ is compact in the point-$\sigma$-weak topology (indeed, this follows from Tychonov's theorem).
    By \ref{axiom:limits}, the set of implementable maps is a closed subset and, hence, itself compact.
\end{proof}

We say that an operator $k:\H_1\to\H_2$ is a Kraus operator of a quantum channel $T:\B(\H_2)\to\B(\H_1)$, if $T$ has a Kraus decomposition $T = \sum_x k_x^*(\placeholder)k_x$ with $k_x=k$ for some index.
The existence of a Kraus decomposition for general cp maps (see \cite{kraus_general_1971}) implies that $k$ is a Kraus operator of $T$ if and only if $k^*(\placeholder)k\le_{cp}T$.

We denote by $K(\K_1,\K_2)$ the set of contractions $k:\H\ox\K_1\to\H\ox\K_2$ that are Kraus operators of some quantum channel $T:\B(\H\ox\K_2)\to\B(\H\ox\K_1)$ that the agent can implement.
Identifying $\H$ with $\H\ox\CC$, the symbol $K(\CC,\CC)$ denotes the possible Kraus operators of implementable quantum channels that only act on the system. We introduce the shorthand $K:= K(\CC,\CC)$.

\begin{lemma}\label{lem:K-basics}
    \begin{enumerate}[(i)]
        \item\label{it:K-basics1} $K(\K_1,\K_2)$ contains all contractions of the form $1\ox k_0$, $k_0:\K_1\to\K_2$. 
        In particular, $K(\K,\K)$ contains all scalars with absolute value $\le1$.
        \item\label{it:K-basics2} If $k\in K(\K_1,\K_2)$ and $h\in K(\K_2,\K_3)$, then $hk\in K(\K_1,\K_3)$.
        \item\label{it:K-basics3} A quantum channel $T:\B(\H\ox\K_2)\to\B(\H\ox\K_1)$ can be implemented by the agent if and only if it has a Kraus decomposition with Kraus operators in $K(\K_1,\K_2)$.
        \item\label{it:K-basics4}
        If $k\in K(\K_1,\K_2)$ then $k^*\in K(\K_2,\K_1)$.
        \item\label{it:K-basics5}
        $K(\K_1,\K_2)$ is closed in the strong operator topology.
        \item\label{it:K-basics6}
        If $k_1,\ldots k_n\in K(\K_1,\K_2)$, $\lambda_2,\ldots\lambda_n\in \CC$ satisfy $\norm{\sum_i\lambda_ik_i}\le 1$, then $\sum_i\lambda_i k_i \in K$.
    \end{enumerate}
\end{lemma}
\begin{proof}
    \ref{it:K-basics1} is clear from \ref{axiom:ancilla}, and \ref{it:K-basics2} is clear from \ref{axiom:consistency} (consider the composition of channels that have $k$ and $h$ as their Kraus operators).

    \ref{it:K-basics3}: The only "only if" implication holds by definition of the set $K(\K_1,\K_2)$. 
    We show the converse using \cref{lem:implement-all-instruments} and \ref{axiom:completeness}.
    Let $(k_y)$, $k_y\in K(\K_1,\K_2)$, be a collection of Kraus operators for $T$.
    Then, by definition of $K(\K_1,\K_2)$ there exist quantum channels $S\up y$, $y\in Y$, such that $k_y$ is a Kraus operator of $S\up y$.
    By \cref{lem:implement-all-instruments}, the agent can implement the two-outcome instrument $(S_i\up y)_{i=1,2}$ with $S_1\up y = k_y^*(\placeholder)k_y$ and $S_2\up y=S\up y-S_1\up y$.
    Since $\sum_y S_1\up y(1)=\sum_y k_y^*k_y=1$, \ref{axiom:completeness} implies that there is an implementable instrument $(T_y)$ with $T_y = S_1\up y = k_y^*(\placeholder)k_y$.
    The implementability of $(T_y)$ implies that $T=\sum_y T_y$ is implementable.

    \ref{it:K-basics4}:
    Let $0\ne k\in K(\K_1,\K_2)$ and let $T$ be an implementable quantum channel with Kraus operators $k_1,k_2,k_3,\ldots$ and $k_1=k$.
    Our goal is to construct an implementable channel that has $k^*$ as one of its Kraus operators. 
    Let $v:\H\ox\K_1\to\H\ox\K_2\ox\K$ be the mnimal Stinespring isometry of $T$.
    By the cp-order interval property of the Stinespring dilation, we can find rank-1 operators $q_x = \kettbra{\Omega_x}$ with $\sum_x q_x=1$ and $k_x^*(\placeholder)k_x = v^*(\placeholder\ox q_x)$.
    Adjusting the phases of the $\Omega_x$, we have $k_x = (1\ox \bra{\Omega_x})v$.
    Now \ref{axiom:dilation} implies that there are Hilbert spaces $\tilde \K_1,\tilde\K_2$, a vector $\Omega\in \tilde\K_1$ and an implementable unitary $u:\H\ox\K_1\ox\tilde\K_1\to \H\ox\K_2\ox\tilde\K_2$ such that $T = (1\ox \bra\Omega)u^*(\placeholder) u(1\ox \ket\Omega)$.
    By \cref{lem:stinespring}, there is an isometry $w:\K \to \tilde\K_2$ such that $(1\ox w)v = u(1\ox \ket\Omega)$.
    Setting $\Omega_x' = w \Omega_x\in\tilde\K_2$, we find $k_x = (1\ox \bra{\Omega_x}) u (1\ox \ket\Omega)$.
    We set $\Phi = \norm{\Omega_1'}^{-1}\Omega_1'\in\tilde\K_2$.
    Let $\{\Psi_y\}_{y=1}^\oo$, be a collection of vectors in $\tilde\K_1$ with $\Psi_1=\norm{\Omega_1'}\Omega$ and $\sum_y \kettbra{\Psi_y}=1$.
    We set $m_y = (1\ox\bra{\Psi_y})u^*(1\ox\ket\Phi)$ and note $m_1=k_1^*=k^*$.
    By construction, 
    \begin{equation*}
        S = \sum_y m_y^*(\placeholder)m_y = \sum_y (1\ox \bra\Phi)u(\placeholder\ox \kettbra{\Psi_y})u^*(1\ox \ket\Phi) 
        =(1\ox \bra\Phi)u(\placeholder\ox 1)u^*(1\ox \ket\Phi)
    \end{equation*}
    is a quantum channel.
    The channel $S$ is implementable since the unitary $u^{-1}$ is, which is ensured by \ref{axiom:reverse}. Since $k^*$ is a Kraus operator of $S$, the claim follows.

    \ref{it:K-basics5}:
    Let $(k_n)$ be a strongly convergent net of operators in $K(\K_1,\K_2)$ and let $k$ denote its limit.
    Let $(T_n):\B(\H\ox\K_2)\to\B(\H\ox\K_1)$ be a net of implementable quantum channels such that $k_n$ is a Kraus operator of $T_n$.
    By \cref{lem:compactness}, there is a point-$\sigma$-weakly convergent subnet $(T_{n_m})$ and, by \ref{axiom:limits}, the agent can implement its limit $T$.
    Then $T_n \ge_{cp} k_n^*(\placeholder)k_n$ for all $n$ implies $T\ge_{cp} k^*(\placeholder)k$.
    Thus, $k$ is a Kraus operator of $T$ and, thus, an element of $K(\K_1,\K_2)$.

    \ref{it:K-basics6}:
    We start by showing the following two claims:
    \begin{enumerate}[({vi.}1), left=.3cm]
        \item\label{it:claim1} If $\sum_i \abs{\lambda_i}^2=1$ and $\sum_i k_i^*k_i\le 1$, then $\sum_i \lambda_i k_i\in K$.
        \item\label{it:claim2} If $k \in K(\K_1,\K_2)$, $\lambda \in \CC$ and $\norm{\lambda k}\le 1$, then $\lambda k\in K$.
    \end{enumerate}
    \ref{it:claim1}: 
    We can find a unitary $n\times n$ matrix $(u_{ij})$ with $u_{1j}=\lambda_j$.
    By \cref{it:K-basics3}, there is an implementable channel $T = \sum_i k_i^*(\placeholder)k_i + T_\perp$. Set $l_i = \sum_j u_{ij} k_j$. Then $T=\sum_i l_i^*(\placeholder)l_i + T_\perp $ shows $l_i\in K(\K_1,\K_2)$, and we have $l_1 = \sum_i \lambda_i k_i$.
\\
    \ref{it:claim2}:
    By \ref{it:claim1}, $\norm k\le 2^{-1/2}$ implies $2^{1/2} k \in K$ (take $(\lambda_i)=(2^{-1/2},2^{-1/2},0,0\ldots)$).
    By \ref{it:K-basics1} and \ref{it:K-basics2}, we have $\lambda k\in K(\K_1,\K_2)$ for all $k\in K(\K_1,\K_2)$ and $\lambda\in\CC$ with $\abs\lambda\le 1$.
    Together, these two observations imply the claim since every complex number $\lambda$ is of the form $2^{n/2}\mu$ with, $n\in\NN$, $\mu\in\CC$, and $\abs\mu\le 1$.
    \\
    We \ref{it:claim1} and \ref{it:claim2} proven, we are now in the position to show \ref{it:K-basics6}. Set $\hat k_i = rk_i$ for sufficiently small $r >0$ such that $\sum_i \hat k_i^*\hat k_i\le 1$, and set $\hat \lambda_i = \norm{(\lambda_i)}_{\ell^2}^{-1}\lambda_i$.
    Using \ref{it:claim1}, we have $\sum_i \hat\lambda_i \hat k_i \in K(\K_1,\K_2)$.
    Thus, \ref{it:claim2} implies that $\sum_i \lambda_i k_i = (\norm{(\lambda_i)}_{\ell^2}/r) \sum_i \hat\lambda_i \hat k_i \in K$.
\end{proof}

\begin{lemma}\label{lem:vNA-unit-ball}
    Let $K_0 \subset \B(\H)$ be a set of contractions with $1\in K_0$.
    Suppose that 
    \begin{itemize}
        \item $K$ is closed under products and adjoints;
        \item if $k,h\in K_0$, $\lambda,\mu\in\CC$ and $\norm{\lambda k+\mu h}\le 1$, then $\lambda k+\mu h\in K_0$;
    \end{itemize}
    Then the positive scalar multiples $\M=\RR^+\cdot K_0$ form a unital *-algebra of bounded operators whose unit ball is $K_0$.
    If $K_0$ is closed in one of the locally convex operator topologies, then $\M$ is a von Neumann algebra.
\end{lemma}
\begin{proof}
    The properties of a *-algebra are easily checked.
    That $K_0$ is exactly the unit ball of $\M$ follows from the second item.
    The second claim follows from Kaplansky's density theorem \cite[Thm.~II.4.8]{takesaki1}, which states that the unit ball of a *-algebra is $\sigma$-strongly* dense in the unit ball of the von Neumann algebra it generates. 
\end{proof}

\begin{corollary}\label{cor:its-a-vNa}
    $\M = \RR^+ \cdot K$ is a von Neumann algebra on $\H$ and the following hold:
    \begin{enumerate}[(i)]
        \item\label{it:its-a-vNa1} $K(\K_1,\K_2)$ is precisely the set of $\M$-contractions $k:\H\ox\K_1\to\H\ox\K_2$.
        \item\label{it:its-a-vNa2} The agent can implement an instrument $(T_x)$ with $T_x:\B(\H\ox\K_2)\to \B(\H\ox\K_1)$ if and only if it is $\M$-inner.
    \end{enumerate}
\end{corollary}
\begin{proof}
    \Cref{lem:vNA-unit-ball} shows that the properties of $K$ that are listed in \cref{lem:K-basics} make $\M = \RR^+\cdot K$ a von Neumann algebra whose unit ball is precisely $K$.

    \ref{it:its-a-vNa1}:
    The case $\K_1=\K_2=\CC$ is simply the statement above that $K$ is the unit ball of $\M$.
    By \cref{it:K-basics1} of \cref{lem:K-basics}, we have $1\ox \bra{\Omega}\in K(\K,\CC)$ and $1\ox\ket{\Omega}\in K(\CC,\K)$ for all unit vectors $\Omega\in \K$.
    Thus \cref{it:K-basics2} of \cref{lem:K-basics} imply that if $k\in K(\K_1,\K_2)$, then $(1\ox \bra{\Omega_2})k(1\ox\ket{\Omega_1})\in \M$ for all $\Omega_j\in\K_j$, $j=1,2$.
    This shows that $K(\K_1,\K_2)$ is a subset of the set of $\M$-operators.
    Conversely, let $k$ be an $\M$-operator.
    Picking bases $\{\ket{x}_j\}$ of $\K_j$, $j=1,2$, we can write $k= \sum k_{ij}\ox\ket{i}\bra{j}$ with the sum being strongly convergent.
    Thus, we can approximate $k$ with linear combinations of operators of the form $k_0\ox \ket i\bra j$ with $k_0\in \M$ a contraction.
    By \cref{it:K-basics5,it:K-basics6} of \cref{lem:K-basics}, it suffices to show that these operators are in $K(\K_1,\K_2)$.
    Let $T_0$ be a channel on $T$ which has $k_0$ as one of its Kraus operators.
    Let $S:\B(\K_2) \to \B(\K_1)$ be a quantum channel, which has $\ket{i}\bra {j}$ as one of its Kraus operators, e.g., $S = \bra{i}(\placeholder)\ket{i} 1_{\K_2}$.
    Since $k_0\ox \ketbra{\Omega_1}{\Omega_2}$ is a Kraus operator of the channel $T_0\ox S$, which is implementable by \ref{axiom:ancilla}, we have $k\in K(\K_1,\K_2)$.

    \ref{it:its-a-vNa2}:
    This follows from \cref{it:K-basics3} of \cref{lem:K-basics}.
\end{proof}

\begin{proof}[Proof of \cref{thm:axioms}]
    \Cref{cor:its-a-vNa} shows the "only if" implication of \cref{thm:axioms}.

    It remains to show the converse:
    We consider an agent who can implement precisely the $\M$-inner operations $\B(\H\ox\K_2)\to\B(\H\ox\K_1)$ for arbitrary $\K_1,\K_2$, and show that the axioms \ref{axiom:consistency} to \ref{axiom:completeness} hold:
    Axioms \ref{axiom:consistency} holds since the class of $\M$-operators is closed under products.
    Axiom \ref{axiom:limits} is checked using \cref{lem:winner}: If $T_n:\B(\H\ox\K_2)\to \B(\H\ox\K_1)$ is a net of $\M$-inner channels converging to $T$ in the topology of convergence in expectation values (the point-$\sigma$-weak topology), then $T_n(a'\ox1_{\K_2})= a'\ox 1_{\K_1}$ for all $a'\in\M'$ implies $T(a'\ox1_{\K_2})= a'\ox 1_{\K_1}$.
    Hence, $T$ is $\M$-inner by \cref{lem:winner}.
    Axiom \ref{axiom:ancilla} is clear since all instruments $(T_x)$ of the form $T_x = \id_{\B(\H)}\ox S_x$, for an instrument $(S_x)$ with $S_x:\B(\K_2)\to\B(\K_1)$, are $\M$-inner.
    Axiom \ref{axiom:reverse} is clear: If $u\in\M\barox\B(\K_1,\K_2)$ is a unitary, then its inverse $u^{-1}=u^*$ is in $\M\barox\B(\K_2,\K_1)$.
    Axiom \ref{axiom:dilation} is shown in \cref{lem:unitary-dilation}.
    It remains to check axiom \ref{axiom:completeness}.
    Let $(T_x\up y)$ be $\M$-inner instruments with $T_x\up y :\B(\H\ox\K_2)\to\B(\H\ox\K_1)$. If \eqref{eq:axiom completness} holds for a selection map $y\mapsto x_y$, then the $\M$-inner cp map $S_Y = \sum_y T_{x_y}\up y$ is subunital.
    Let $L:\B(\H\ox\K_2)\to\B(\H\ox\K_1)$ be some $\M$-inner quantum channel, e.g., $\id_\H \ox \omega(\placeholder)1$ for some $\omega\in\nstates(\K_2)$, and set $S_\perp = (1-S_Y(1))^{1/2} L(\placeholder)(1-S_Y(1))^{1/2}$.
    Then $(S_y)_{y\in\tilde Y}$ with $\tilde Y=Y\cup\{\perp\}$ is indeed an $\M$-inner instrument.
\end{proof}

\subsection{No information without disturbance}\label{sec:information-disturbance}

We consider an agent $A$ and a referee $R$, both of whom have partial access to a quantum system with Hilbert space $\H$.
We assume that the state of the full quantum system is known to the referee but unknown to the agent.
We ask:
\begin{center}\it
    What can the agent measure that cannot be detected by the referee?
\end{center}
If the referee has access to the full quantum system, this question reduces to the classic question \emph{What can be measured without disturbing the system?} with the well-known answer \emph{Nothing} \cite{busch_no_2009,kretschmann_continuity_2008,kretschmann_information-disturbance_2008}.
The assumption that the state of the full system is unknown to the agent is crucial. If the agent has partial knowledge about the system's state, e.g., knows that the state commutes with certain observables, the agent can perform undetectable operations that reveal further information.\footnote{Assume, for instance, that both the agent and the referee have full access to the system and assume that the agent knows that the system's state $\phi$ commutes with a self-adjoint operator $h$ on $\H$.
If $p$ is a spectral projection of $h$, consider the instrument $(T_1,T_2)=(p(\placeholder)p,p^\perp(\placeholder)p^\perp))$. 
Since $p$ commutes with $\phi$, we have $T_*(\phi)=\phi$, i.e., the agent can measure without disturbing the system.
}
We will characterize the information that can be extracted secretly by an agent $A$ as what is detectable with measurements in the von Neumann algebra
\begin{equation}\label{eq:information-disturbance intersection}
    \M_A\cap \M_R',
\end{equation}
where $\M_A$ and $\M_R$ are the von Neumann algebras associated with the agent and the referee, respectively.
This has two important consequences for us:

\begin{itemize}
    \item We understand the \emph{center} $Z(\M_A)=\M_A\cap\M_A'$ of the observable algebra $\M_A$ of the agent $A$ as describing the information that can be extracted without disturbance relative to the agent (see also \cite{kuramochi_accessible_2018}).
    Thus, the observable algebra $\M_A$ is a \emph{factor} if and only if the agent cannot extract information without disturbing the state, relative to the agent's observables.

    \item In \emph{cryptographic} settings, it is necessary to understand what an eavesdropper is able to learn about a given subsystem in secret.
    Here, the roles of the agent and the referee are swapped:
    We assume that the full system is in a state $\phi$, known to the referee (i.e., the referee knows the marginal state $\phi|_{\M_R}$), and that the referee encodes classical information $y$ by choosing an operation, e.g., a unitary $u_y\in\M_R$.
    The agent, who plays the role of an eavesdropper, wishes to learn the information $y$ and performs an instrument $(T_x)$.
    This attack cannot be detected by the referee precisely when $(T_x)$ is $\M_R'$-inner.
\end{itemize}

The precise setup is as follows:
The full system is in a state $\phi$, unknown to the agent, but known to the referee (i.e., the referee knows the marginal state $\phi|_{\M_R}$).
The referee's goal is to decide whether the agent performs an instrument $(T_x)$ or not.
Since the referee would not know the agent's measurement outcome, applying the instrument would, from the referee's point of view, result in applying the quantum channel $T=\sum_x T_x$.
Thus, the referee's job is to decide whether the system's state is $\phi$ or $T_*(\phi)$.
This is possible if and only if the referee can implement a measurement whose outcomes are distributed differently for these two states.
(We consider quantitative notions of state distinguishability in \cref{sec:fidelity}.)
Since the referee can implement any POVM in $\M_R$, this is equivalent to $\phi|_{\M_R}\ne T_*(\phi)|_{\M_R}$.
Let us note this key fact as a Lemma:

\begin{lemma}\label{lem:distinguishing states with partial access}
    A pair of normal states $\phi,\psi$ on a full quantum system can be distinguished by an agent with observable algebra $\M$ if and only if $\phi|_\M\ne \psi|_\M$.
\end{lemma}

Returning to our problem, we conclude that $\phi|_{\M_R}\ne T_*(\phi)|_{\M_R}$ must hold for all normal states $\phi\in\nstates(\H)$. Otherwise, the agent, who does not know the state of the system, cannot have a guarantee that the referee does not know of his actions.
This property is much easier to state in the Heisenberg picture, where it simply takes the simple form $T \restriction \M_R = \id$.
By \cref{lem:implement-all-instruments}, the requirement $T\restriction \M_R$ is equivalent to $\M_R'$-innerness of $(T_x)$.
Thus, we have
\begin{equation}\label{eq:cannot-detect}
    \text{$(T_x)$ cannot be detected} \quad \iff \quad (T_x)\ \text{is $\M_R'$-inner.}
\end{equation}
If the referee has full access, i.e., $\M_R=\B(\H)$, this implies the usual statement: If an instrument $(T_x)$ does not disturb the system in the sense that $\sum_x T_x=\id$, then $(T_x)$ is inner for $\B(\H)'=\CC1$, and, therefore, the measurement that it implements cannot reveal information about the state of the system since $T_x(1)\propto 1$ for all $x$.

In the general case, where both the agent and the referee have partial access to the full system, \eqref{eq:cannot-detect}, tells us that the instruments which the agent can implement without detection from the referee are the $(\M_A\cap\M_R')$-inner ones, i.e., instruments of the form
\begin{equation}
    T_x = \sum_{y} k_{y|x}^*(\placeholder) k_{y|x}, \qquad k_{y|x} \in \M_A\cap \M_R'.
\end{equation}
If the referee and the agent have access to the same observables, i.e., if $\M_A=\M_R$, then the application of an instrument $(T_x)$ is undetectable to the agent if and only if the instrument is \emph{central} in the sense that the Kraus operators $k_{y|x}$ are contained in the center $Z(\M_A)=\M_A\cap \M_A'$.

For the case that the referee and the agent have access to the same observables, an \emph{infor\-mation-disturbance trade-off}, i.e., a quantitative version of the above statement, was obtained in \cite{kretschmann_continuity_2008,kretschmann_information-disturbance_2008}.
It would be interesting to see whether this can be generalized to the case of a referee with partial access to the full system.

\subsection{Bipartite systems, product states and local state preparation}\label{sec:independent-agents}

We want to understand when two subsystems can be regarded as independent in the sense that measurements on one subsystem can be performed jointly with measurements on the other subsystem.
This is the case if and only if all measurements on one system commute with all measurements on the other system \cite{busch_quantum_2016}.
More generally, we say that two agents $A$ and $B$, called Alice and Bob, which both have constrained access to the same full quantum system, are \emph{independent} if all measurements of Alice commute with all measurements of Bob.
Clearly, this is the case if and only if their observable algebras $\M_A$ and $\M_B$ commute, and it follows that
\begin{equation}
    T_A \circ T_B = T_B \circ T_A
\end{equation}
for all channels $T_{A/B}$ the agents can implement, respectively.
Of course, the analogous statement holds for instruments.
Independent agents can perform \emph{correlation experiments} on the full quantum system:
If $\phi$ is a state of the full system, and Alice and Bob perform measurements $(m_x^A)$, $(m_y^B)$, respectively, then the outcomes are distributed according to the probability distribution
\begin{equation}\label{eq:correlation function}
    p(x,y)=\phi(m_x^A\,m_y^B).
\end{equation}
If we fix the agent Alice, then the most powerful agent $B$ that is independent from Alice is the one whose observable algebra is the commutant $\M_B=\M_A'$ of Alice's observable algebra (see \eqref{eq:inclusion of observable algs}).
Equivalently, independence of a pair of agents can be defined by saying that agent neither of the two agents can detect whether the other agent has applied an operation.
Indeed, the equivalence immediately follows from the results in \cref{sec:information-disturbance} (see \eqref{eq:cannot-detect}).

A common paradigm in quantum information theory is that of \emph{separated labs}, which describes the idea that Alice and Bob are spatially separated and are free to do whatever they want in their respective labs.
For instance, it is generally assumed that the local agents can discard their share of an entangled state and locally prepare a new state, which results in a product state of the full system.
The notion of independence given above is generally not sufficient to allow for such seemingly trivial maneuvers.
In fact, there are \emph{no product states} in general:

\begin{lemma}[{\cite{summers_independence_1990}}]\label{lem:no product states}
    Let $\M_A$ and $\M_B$ be commuting factors on a Hilbert space $\H$. 
    The following are equivalent:
    \begin{enumerate}[(a)]
        \item\label{it:split1} Product states: For some pair (resp.\ all pairs) of normal states $\phi_{A/B}$ on $\M_{A/B}$, there is a state $\phi$ of the full system, such that
            \begin{equation}
                \phi(ab) = \phi_A(a)\phi_B(b), \qquad (a,b)\in\M_A\times\M_B.
            \end{equation}
        \item\label{it:split2} Tensor product: There are Hilbert spaces $\H_{A/B}$ with factors $\N_{A/B}\subset\B(\H_{A/B})$ such that 
            \begin{equation}
                (\H,\M_A,\M_B) \cong (\H_A\ox\H_B,\N_A\ox1,1\ox\N_B).
            \end{equation}
        \item\label{it:split3} Split property: here is a type $\I$ factor $\N$ such that $\M_A\subset \N\subset \M_B'$.
    \end{enumerate}
\end{lemma}
\begin{proof}
    The implications \ref{it:split1} $\Leftarrow$ \ref{it:split2} $\Leftrightarrow$ \ref{it:split3} are clear.
    \ref{it:split1} $\Rightarrow$ \ref{it:split3}: By \cite[Ex.~IV.5.1.(b)]{takesaki1}, $\M_A\vee\M_B\cong\M_A\barox\M_B$. By \cite[Thm.~3.9]{summers_independence_1990}, this implies \ref{it:split3}.
\end{proof}

This tells us that product states only exist if and only if the observable algebras of the two agents live on distinct tensor factors of the underlying Hilbert space $\H$.
Even if product states exist, they cannot be prepared locally relative to all observables of the respective agents, unless the von Neumann algebras are type $\I$ factors.
To understand this, we begin with the following result, shown by Werner in the context of quantum field theory in \cite{werner_local_1987}, which characterizes the milder requirement of preparing a state only for a subalgebra:

\begin{lemma}[{\cite{werner_local_1987}}]\label{prop:werner state prep}
    Let $\M_0\subset \M$ be factors. The following are equivalent:
    \begin{enumerate}[(a)]
        \item\label{it:state prep1} For some/every normal state $\phi_0$ on $\M_0$ there is an inner ucp map $T:\M\to\M$ such that $T |_{\M_0}= \phi_0(\placeholder)1$
        \item\label{it:state prep2} Split property: There is a type $\I$ factor $\N$ such that $\M_0\subset\N\subset\M$.
    \end{enumerate}
\end{lemma}
\begin{proof}
    (\ref{it:state prep1}, "some") $\Rightarrow$ \ref{it:state prep2}:
    Let $\M$ act on a Hilbert space $\H$ and extend $T$ to a normal ucp map on $\B(\H)$ by using a Kraus decomposition in $\M$. 
    Now let $\psi$ be a normal state on $\B(\H)$ and let $\phi = T_*(\psi)$ ($=\psi\circ T$).
    Then $\phi(a_0 a') = \psi(T(a_0a') = \psi(T(a_0)a') = \phi_0(a_0) \psi(a')$, $(a_0,a')\in\M_0\times \M'$ shows that $\phi$ is a product state for the commuting algebras $\M_0$ and $\M'$.
    By \cref{lem:no product states}, the split property holds.

    \ref{it:state prep2} $\Rightarrow$ (\ref{it:state prep1}, "every"):
    Given a normal state $\phi_0$ on $\M_0$, pick a normal state $\phi$ on $\N$ extending $\phi_0$.
    Since $\N$ is a type $\I$ factor, there are Kraus operators $\{k_x\}$ in $\N$ with $T(a):=\sum_x k_x^*(a) k_x=\phi(a)1$ for all $a\in\N$ and, hence, $T(a) =\phi_0(a)1$ for $a\in\M_0$.
    Clearly, $T$ is $\M$-inner and unital.
\end{proof}

As a consequence, we learn that it is impossible to prepare states with inner operations on the full von Neumann algebra, unless the latter is a type $\I$ factor:

\begin{lemma}\label{lem:local state preparability}
    Let $\M$ be a von Neumann algebra and assume that there is an inner ucp map $T:\M\to\M$ of the form $T=\phi(\placeholder)1$ for a normal state $\phi$ on $\M$.
    Then $\M$ is a type $\I$ factor.
\end{lemma}
\begin{proof}
    Since $T$ is inner, we have $T\restriction\M' = \id$, which implies that $T\restriction Z(\M) = \id$.
    Thus, we have $Z(\M) = T(Z(\M)) = \phi(Z(\M))1 = \CC1$, which shows that the center $Z(\M)$ is trivial, i.e., $\M$ is a factor.
    By \cref{prop:werner state prep} (applied to $\M_0=\M$), it follows that $\M$ is type $\I$.
\end{proof}

We return to independent agents with partial access to a full quantum system. 
The following result characterizes when correlation experiments fully determine the states of the full quantum system.
In particular, this is only possible if both $\M_A$ and $\M_B$ are factors:

\begin{lemma}\label{lem:irreducible}
    Let $\M_A$ and $\M_B$ be commuting von Neumann algebras on $\H$.
    The following are equivalent:\footnote{The Lemma and its proof generalize verbatim to any finite number of pairwise commuting von Neumann algebras.}
    \begin{enumerate}[(a)]
        \item\label{it:irreducible1} A state of the full quantum system is completely determined through correlation experiments of $A$ and $B$, i.e., if $\phi$ and $\psi$ are states of the full system, then $\psi=\phi$ if and only if 
        \begin{equation}\label{eq:irreducible1}
            \phi(m^A_x\, m^B_y) = \psi(m^A_x\, m^B_y)
        \end{equation}
        for all measurements $(m_x^A)$, $(m_y^B)$ of the respective agents.
        \item\label{it:irreducible2} The only measurements that commute with all measurements of $A$ and $B$ are the trivial ones. I.e., $\M_A'\cap\M_B'=\CC1$ or, what is equivalent,
        \begin{equation}\label{eq:irreducible2}
            \M_A \vee \M_B=\B(\H).
        \end{equation}
        \item\label{it:irreducible3} $\M_A$ and $\M_B$ are irreducible, i.e., there is no proper closed subspace of $\H$ that is both $\M_A$- and $\M_B$-invariant.
    \end{enumerate}
    In this case, $\M_A$ and $\M_B$ are necessarily factors.
\end{lemma}
\begin{proof}
    \ref{it:irreducible2} $\Leftrightarrow$ \ref{it:irreducible3} is clear because the projections onto jointly invariant closed subspaces are precisely the projections of the von Neumann algebra $\M_A'\cap\M_B'= (\M_A\vee\M_B)'$.

    \ref{it:irreducible1} $\Rightarrow$ \ref{it:irreducible2}:
    Let $u\in (\M_A\vee \M_B)'$ be a unitary and let $\phi$ be a normal state on $\B(\H)$. 
    Then $\phi$ and $u\phi u^*$ cannot be distinguished with correlation experiments, hence \ref{it:irreducible1} forces $u=1$. Thus, $(\M_A\vee \M_B)' =\M_A'\cap\M_B'=\CC1$.

    \ref{it:irreducible2} $\Rightarrow$ \ref{it:irreducible1}:
    Let $\psi,\phi$ be normal states on $\B(\H)$.
    By assumption the *-algebra $\A$ generated by $\M_A$ and $\M_B$ is dense.
    Every element of $\A$ is a linear combination of products of the form $ab$ for effects $a\in[0,1]_{\M_A}$ and $b\in[0,1]_{\M_B}$.
    Thus, \eqref{eq:irreducible1} is equivalent to the statement that $\psi$ and $\phi$ are equal on a $\sigma$-weakly dense *-subalgebra. By normality, this is equivalent to $\psi=\phi$.

    The fact that $\M_A$ and $\M_B$ are factors follows from \eqref{eq:irreducible2} because $Z(\M_A) \subset Z(\M_A\vee \M_B) = Z(\B(\H))=\CC1$.
\end{proof}

Based on the above, we now discuss what it means to have a \emph{bipartite system}.
Of course, a precise formalization is not achievable without a proper definition of what constitutes a subsystem.
In all cases, the observables that are localized on the respective subsystem will be commuting. 
In our framework, this results in commuting von Neumann algebras $\M_A$, $\M_B$ on $\H$.
Apart from commutativity, the following two requirements are natural:
\begin{enumerate}
    \item The full system should be composed of the two subsystems, in the sense of the following equivalent properties (see \cref{lem:irreducible}):
        \begin{enumerate}[(a)]
            \item states of the full system are fully determined by correlation experiments;
            \item measurements that commute with all measurements of both subsystems are trivial.
        \end{enumerate}
    \item The two subsystems should be "complementary" in the sense that one subsystem describes precisely those observables of the full system that commute with the other system, i.e., $\M_A$ and $\M_B$ are commutants:
    \begin{equation}\label{eq:haag duality}
        \M_A=\M_B'.
    \end{equation}
\end{enumerate}
We refer to the first property as \emph{irreducibility}, which is justified by \cref{lem:irreducible}.  
The second property is known as \emph{Haag duality} \cite{keyl_entanglement_2006}, named after a similar property with the same name in quantum field theory \cite{haag_local_1996} (see also \cref{sec:qft}).
Neither of these two properties implies the other: Irreducibility requires that $\M_A$ and $\M_B$ are factors, which does not follow from Haag duality.
On the other hand, Haag duality is not automatic from irreducibility.
If Haag duality holds, then $Z(\M_A)=Z(\M_B)$, and irreducibility is equivalent to factoriality of $\M_A$ and, therefore, also equivalent to factoriality of $\M_B$.
If the first property holds, then $\M_A\subset\M_B'$ and $\M_B\subset\M_A'$ are so-called irreducible subfactor inclusions, which means that the relative commutant $\M_A'\cap\M_B'=\CC1$ is trivial (see \cref{lem:irreducible}) \cite{jones_introduction_1997,kawahigashi_subfactor_2005,jones_index_1983}.

Haag duality is sometimes described as a "technical assumption" \cite{keyl_entanglement_2006,bols_category_2025,naaijkensQuantumSpinSystems2017}.
It is the opinion of the author that this falls somewhat short.
There are cases where the failure of Haag duality is understood on physical grounds (see, e.g., \cite{naaijkens_kosaki-longo_2013,naaijkens_subfactors_2018,longo_index_1989,longo_index_1990}).
In \cref{sec:haag}, we show that Haag duality has an operational interpretation that only refers to local operations:\footnote{This is not the case for the formulation above. Indeed, suppose that $\M_A\subset\M'_B$ is a proper inclusion. Then the above refers to those operations that are not in $\M_B$.}\textsuperscript{,}\footnote{This operational interpretation does not apply to the weaker version of Haag duality studied in \cite{ogata_classification_2021,naaijkens_split_2022,jones_dhr_2024}.}
Haag duality holds if and only if pure states with the same marginals on one subsystem are related by unitaries of the other subsystem.

We analyze bipartite systems and, in particular, the resulting entanglement properties, in \cref{sec:bipartite systems,sec:locc,sec:strong-entanglement} in great detail.
Even though we regard both of the assumptions above as natural, we often assume Haag duality without factoriality because many of our theorems will not need this assumption.

We close this subsection by combining what we learned in the above and \cref{sec:information-disturbance} about the information-theoretic properties that correspond to factoriality of the observable algebra:

\begin{corollary}
    Let $\M_A$ be a von Neumann algebra on $\H$.
    The following are equivalent:
    \begin{enumerate}[(a)]
        \item $\M_A$ is a factor;
        \item there is von Neumann algebra $\M_B$ that commutes with $\M_A$ such that correlation experiments of the pair $(\M_A,\M_B)$ (see \eqref{eq:correlation function}) provide full information about normal states on $\B(\H)$;%
        \footnote{While the choice $\M_B=\M_A'$ always is an example if $\M_A$ is a factor, it is not the unique choice for $\M_B$.
        Indeed, every irreducible subfactor $\M_B\subset\M_A'$ is just as good.}
        \item there is a von Neumann algebra $\M_B$ that commutes with $\M_A$ such that the full system does not have nontrivial operations that commute with the operations of $\M_A$ and $\M_B$;
        \item it is not possible to extract information about the state of the full system with $\M_A$-inner operations in a way that cannot be detected with measurements in $\M_A$.
    \end{enumerate}
\end{corollary}

\subsection{Entanglement and LOCC}\label{sec:entanglement-basics}

In the following, we go through the basics of entanglement theory in the von Neumann algebraic setting. 
An in-depth analysis of bipartite entanglement under the assumption of Haag duality will be carried out in \cref{sec:bipartite systems,sec:locc,sec:strong-entanglement}.
Entanglement is nowadays understood as the resource that LOCC cannot create \cite{horodecki_quantum_2009,chitambar_quantum_2019}.
Thus, in order to talk about entanglement, we need to make sense of LOCC in our settings.

We consider a full quantum system with Hilbert space $\H$ together with a finite number of agents $A_1,\ldots A_N$, each of which has partial access to the full system.
We assume the agents to be pairwise independent, so that the corresponding von Neumann algebras $\M_j=\M_{A_j}\subset \B(\H)$ are pairwise commuting.

A general LOCC protocol is of the following form:
In the first round, one of the agents performs an instrument and then communicates whatever outcome they measure to the other agents.
In each subsequent round, some agent chooses an instrument based on the outcomes of the previous rounds, applies it, and communicates the results to the other agents. 
The protocol terminates after a finite number $r$ of rounds.
The protocol itself is the algorithm that determines the order in which the agents act and which operations they apply based on which messages they receive.
Thus, in each run and for each initial state $\psi$, an LOCC protocol produces a string $x= x_1\cdots x_r$ of outcomes and an output state $\psi_x'$.
Therefore, the protocol defines an instrument $(T_x)$ whose outcome space $X$ is the set of all possible strings $x_1\cdots x_r$.

If $\psi$ and $\phi$ are states of the full system, we say that \emph{$\psi$ can be transformed to $\phi$ with LOCC by the of agents $A_1,\ldots A_N$} if there exists an LOCC protocol with instrument $\{T_x\}_{x\in X}$ such that, on average, the instrument takes $\psi$ to $\phi$, i.e., $T_*(\psi)=\phi$, where $T=\sum_x T_x$.
We denote this by
\begin{equation*}
    \psi\locc\phi.
\end{equation*}
Note that we only ask that $\psi$ is taken to $\phi$ on average.
However, in the case that $\phi$ is a pure state, it follows that each of the post-measurement states $\psi_x'$ is equal to $\phi$.
More precisely, $\psi_x'=\phi$ if $p_x>0$ (otherwise, the post-measurement states are not defined).

As we said above, entanglement is defined as the resource that LOCC cannot create.
This defines comparative entanglement theory: A state is more entangled than a given state if the former can be reached from the latter with LOCC, at least approximately. 
We can, of course, just take this definition and apply it to a collection of pairwise independent agents with constrained access to the full system.
It is a matter of taste whether one feels comfortable in calling the resulting relation entanglement, even in cases where the agents do not correspond to independent subsystems of the whole system.
In any case, it does not really matter as the concept of LOCC transitions is meaningful even in the general setting of pairwise independent agents with constrained access to the full system.

\subsection{State distinguishability}\label{sec:fidelity}

We discuss two ways to operationally quantify the distance of quantum states relative to an agent with partial access to the full system.
These two quantifications are the \emph{norm distance} and the \emph{fidelity}.
To make sense of the latter, we discuss Uhlmann's Theorem.

\subsubsection*{Norm distance}

In order to distinguish a pair $\psi,\phi$ of states of the full system, the agent performs measurements, which result in outcome distributions that are as different as possible for the two states $\psi$ and $\phi$.
Since the agents wish to distinguish between two hypotheses, they only need to consider binary measurements.
The simplest form to do so is to find an effect $f$ for which the two success probabilities $\psi(f)$ and $\phi(f)$ are as different as possible.
Mathematically, this amounts to optimizing $\abs{\psi(f)-\phi(f)}$ over all effects that are accessible to the agent.
Using the map $f\mapsto 2f-1$, which is an affine bijection between $[0,1]_\M$ and $[-1,1]_\M$ (the set of self-adjoint contractions in $\M$), we recognize this as half the norm distance of the induced states on $\M$:
\begin{equation}
     \sup_{f\in [0,1]_\M} \abs{\psi(f)-\phi(f)} = \frac12 \sup_{f\in [-1,1]_\M} \abs{(\psi-\phi)(f)} = \frac12 \norm{\psi|_\M-\phi|_\M}.
\end{equation}
This tells us that the norm distance of states on $\M$ has an operational interpretation in terms of state distinguishability.
If the agent has full access, i.e., if $\M=\B(\H)$, this corresponds precisely to the operational interpretation of the trace distance \cite[Sec.~9.2.1]{nielsen_quantum_2010}.

A central property of the norm distance is that it is monotone under inner quantum channels. 
In fact, for normal ucp maps $T:\N\to\M$ between von Neumann algebras $\M$ and $\N$, which we call outer quantum channels (see \cref{sec:intrinsic}), it holds
\begin{equation}\label{eq:norm DPI}
    \norm{T_*(\omega_1) - T_*(\omega_2)} \le \norm{\omega_1-\omega_2} , \qquad \omega_1,\omega_2\in\nstates(\M).
\end{equation}
This is called the \emph{data processing inequality} (for the norm distance) in the quantum information theory literature \cite{nielsen_quantum_2010}.
This reflects the fact that a quantum channel cannot increase the distinguishability of two states.
Mathematically, \eqref{eq:norm DPI} is a direct consequence of the fact that ucp maps, and hence their preduals, are norm contractions.
If $\Psi$ and $\Phi$ are unit vectors and if $\psi,\phi\in\nstates(\H)$ denote the corresponding normal pure states, then
\begin{equation}\label{eq:general vector state distance}
    \norm{\psi|_\M-\phi|_\M} \le 2 \norm{\Psi-\Phi}.
\end{equation}

\subsubsection*{Fidelity}

We begin by recalling how the fidelity is defined in the usual framework of quantum information theory. 
We then follow Hiai's definition of the fidelity for normal states on a von Neumann algebra via Haagerup $L^p$-spaces \cite{hiai_quantum_2021}.
Afterward, we generalize Uhlmann's Theorem, which connects the intrinsically defined fidelity with any given representation of the von Neumann algebra.
As a consequence of Uhlmann's theorem, the equivalence with Uhlmann's original definition of the fidelity \cite{uhlmann_transition_1976} follows.
Moreover, we prove the Fuchs-van de Graaf inequalities, which we use to show that convergence in fidelity is equivalent to norm convergence.
Most, if not all, of these results here are well-known in the literature. 

The fidelity of a pair $\psi,\phi$ of normal states on $\B(\H)$ is defined as $F(\psi,\phi)=\norm{\rho_\psi^{1/2}\rho_\phi^{1/2}}_1^2$, where $\rho_\phi,\rho_\psi$ denote the corresponding density operators in $\T(\H)$ \cite{nielsen_quantum_2010}.
In the case that $\psi$ and $\phi$ are pure states with implementing vectors $\Psi,\Phi$, the Fidelity becomes
\begin{equation}
    F(\psi,\phi) = \abs{\ip\Psi\Phi}^2.
\end{equation}
This quantity has an operational interpretation in terms of state distinguishability. 
Indeed, this follows from the operational interpretation of the norm distance, here applied to pure states, since $\norm{\psi-\phi}=\norm{\kettbra\Psi-\kettbra\Phi}_1 = 2 \sqrt{1-\abs{\ip\Psi\Phi}^2} = 2\sqrt{1-F(\psi,\phi)}$.
The operational interpretation of the fidelity for pure states is extended to arbitrary normal states on $\B(\H)$ by Uhlmann's theorem and its consequences \cite{nielsen_quantum_2010,dodd_simple_2002}.

Following Hiai \cite{hiai_quantum_2021}, the fidelity of a pair $\psi,\phi$ of normal states on a von Neumann algebra $\M$ is defined through the framework of Haagerup $L^p$-spaces (see \cref{sec:haagerup Lp}) 
\begin{equation}\label{eq:def fidelity}
    F(\psi,\phi) = \norm{h_\psi^{1/2}h_\phi^{1/2}}_{L^1(\M)}^2.
\end{equation}
Since the Haagerup $L^p$-spaces are isomorphic to the Schatten classes for $\M=\B(\H)$, \eqref{eq:def fidelity} reduces to the standard definition in this case.
The fidelity is 1 if and only if the states are equal:
    $F(\psi,\phi)=1\iff \psi=\phi$.
The \emph{Fuchs-van de Graaf inequalities} \cite{fuchs_cryptographic_1999} 
\begin{equation}\label{eq:FvdG}
    1- F(\psi,\phi)^{1/2} \le \frac12 \norm{\psi-\phi} \le (1- F(\psi,\phi))^{1/2},
\end{equation}
$\psi,\phi \in \nstates(\M)$, show that convergence in fidelity is equivalent to convergence in norm. We will prove \eqref{eq:FvdG} as a consequence of Uhlmann's Theorem below.
Like the norm distance, the fidelity satisfies a data processing inequality: If $T:\N\to\M$ is a normal ucp map, then \cite[Rem.~3.15 \& Thm.~3.16]{hiai_quantum_2021}:
\begin{eqnarray}\label{eq:fidelity dpi}
    F(T_*(\psi),T_*(\phi)) \ge F(\psi,\phi)
\end{eqnarray}
for all pairs of normal states on $\M$.
In particular, the fidelity of states on $\M$ is monotone under $\M$-inner quantum channels.

We will prove two versions of Uhlmann's theorem:

\begin{theorem}\label{thm:uhlmann extensions}
    Let $\M$ be a von Neumann algebra on $\H$ and let $\psi,\phi$ be normal states on $\M$.
    Then 
    \begin{equation}\label{eq:uhlmann extension}
        F(\psi,\phi) = \sup_{\tilde\psi,\tilde\phi}\, F(\tilde\psi,\tilde\phi),
    \end{equation}
    where the optimization is over normal states $\tilde\psi,\tilde\phi$ on $\B(\H)$ that extend $\psi$ and $\phi$, respectively.
    In fact, for every fixed extension $\tilde\psi$ of $\psi$, we have $F(\psi,\phi)=\sup_{\tilde\phi} F(\tilde\psi,\tilde\phi)$.
\end{theorem}

A purification in $\H$ of a state $\psi$ on $\M$ is a vector $\Psi\in\H$ such that $\psi(a)=\ip\Psi{a\Psi}$, $a\in\M$.
If the states $\psi$ and $\phi$ admit purifications in $\H$, we have:

\begin{theorem}\label{thm:uhlmann}
    Let $\M$ be a von Neumann algebra on $\H$.
    Let $\psi,\phi$ be states on $\M$ that have purifications in $\H$.
    Then
    \begin{equation}\label{eq:uhlmann}
        F(\psi,\phi) = \sup_{\Psi,\Phi}\,  \abs{\ip\Psi\Phi}^2,
    \end{equation}
    where the optimization is over all purifications $\Psi$ of $\psi$ and $\Phi$ of $\phi$ in $\H$.
    If $\Psi$ is a fixed purification, we have $F(\psi,\phi)=\sup_\Phi\abs{\ip\Psi\Phi}$.
\end{theorem}

We emphasize that the RHSs of \cref{eq:uhlmann extension,eq:uhlmann} depend on the representation $\H$, but that the LHSs are defined intrinsically.
Uhlmann's theorem extends the operational meaning of the fidelity of normal states on $\B(\H)$ (see above) to states on a von Neumann algebra.

\Cref{thm:uhlmann} is essentially contained in the literature:
It follows from first applying \cite[Cor.~2]{alberti_note_1983} of Alberti first to the standard representation, where the equivalence of \eqref{eq:def fidelity} and the RHS of \eqref{eq:uhlmann} can be checked explicitly, and then to a general representation.%
\footnote{It might seem that \cref{thm:uhlmann} follows directly from \cite[Cor.~2]{alberti_note_1983}. Note, however, that the definition of the fidelity in \eqref{eq:uhlmann} is different from Uhlmann's definition \cite{uhlmann_transition_1976}, which is used by Alberti in \cite{alberti_note_1983}.}
In the following, we give a self-contained proof of \cref{thm:uhlmann,thm:uhlmann extensions}.

\begin{lemma}\label{lem:unit-ball-opt}
    Under the assumptions of \cref{thm:uhlmann}, let $\Psi_0,\Phi_0$ be a pair of purifications.
    Then 
    \begin{equation}\label{eq:unit ball opt}
        \sup_{\Psi,\Phi}\, \abs{\ip\Psi\Phi} 
        =\sup_{u'\in \U(\M')} \abs{\ip{\Psi_0}{u'\Phi_0}} = \sup_{s'\in B(\M')} \abs{\ip{\Psi_0}{s'\Phi_0}},
    \end{equation}
    where the optimization on the LHS is as in \cref{thm:uhlmann}.
\end{lemma}
\begin{proof}
    The first equality follows from \cref{cor:all-purifications}.
    To see the second equality, we use the Russo-Dye theorem \cite{russo1966note} to replace the optimization over $B(\M)$ on the RHS by an optimization over the weakly dense subset $\conv(\U(\M'))$.
    By convexity of the function that we are optimizing, we can drop the convex hull, leaving us with an optimization over $\U(\M')$, as desired.
\end{proof}

\begin{proof}[Proof of \cref{thm:uhlmann}]
    For the sake of this proof, we denote the RHS of \eqref{eq:uhlmann} by $F_\H(\psi,\phi)$.

    \emph{Step 1.} We assume that $\M$ acts by left multiplication on the Haagerup $L^2$-space $\H=L^2(\M)$.
    Then $h_\psi^{1/2}$ and $h_\phi^{1/2}$ are purifications of $\psi$ and $\phi$, respectively, so that \cref{lem:unit-ball-opt} implies
    \begin{align*}
        F_{L^2(\M)}(\psi,\phi)
        = \sup_{s'\in B(\M)} \abs{\ip{h_\psi^{1/2}}{s'h_\phi^{1/2}}_{L^2(\M)}}^2
        = \sup_{s\in B(\M)} \abs{\tr ( h_\psi^{1/2} h_\phi^{1/2} s)}^2 
        = \norm{h_\psi^{1/2}h_\phi^{1/2}}_{L^1(\M)}^2,
    \end{align*}
    where we used the following facts: the commutant of $\M$ is the action of $\M$ on $L^2(\M)$ by right multiplication \cite[Thm.~9.29]{hiai_lectures_2021}, and the $L^1$-norm of an operator $h\in L^1(\M)$ equals $\sup_{s\in B(\M)}\abs{\tr(hs)}$ \cite[Lem.~9.23]{hiai_lectures_2021}.

    \emph{Step 2.} We show that $F_\H$ does not depend on $\H$.
    Combined with the first step, this shows \eqref{eq:uhlmann}.
    Let $(\pi_i,\H_i)$, $i=1,2$, be representations of $\M$ and let $\Psi_i,\Phi_i$ be purifications of $\psi,\phi$, respectively, in $\H_i$, $i=1,2$.
    Let $u,v$ be the $\pi_1$-$\pi_2$-intertwining partial isometries from \cref{lem:M-linear partial isometry} such that $u\Psi_1 =\Psi_2$, $v\Phi_1=\Phi_2$.
    Since $uv^*$ is a contraction in the commutant of $\pi_2(\M)$, we have
    \begin{equation}
        \abs{\ip{\Psi_1}{\Phi_1}} = \abs{\ip{\Psi_2}{uv^*\Phi_2}} 
        \le \sup_{s'\in B(\pi_2(\M)')} \abs{\ip{\Psi_2}{s'\Phi_2}} = F_{\H_2}(\psi,\phi)^{1/2}.
    \end{equation}
    Thus, $F_{\H_1}(\psi,\phi) = \sup_{\Psi_1,\Phi_1} \abs{\ip{\Psi_1}{\Phi_1}}^2 \le F_{\H_2}(\psi,\phi)$.
    The converse inequality follows from the same arguments with exchanged indices $1\leftrightarrow2$.
    The last claim follows from \cref{lem:unit-ball-opt}.
\end{proof}

When $\M_0\subset\M$ is a von Neumann subalgebra and $\psi,\phi$ are normal states on $\M$, Uhlmann's Theorem implies the following special case of \eqref{eq:fidelity dpi}
\begin{equation}\label{eq:fidelity poormans dpi}
    F(\psi|_{\M_0},\phi|_{\M_0}) \ge F(\psi,\phi).
\end{equation}
Indeed, this follows directly from the fact that purifications of functional $\psi\in\M_*^+$ are, in particular, purifications of its restriction $\psi|_{\M_0}$ to a subalgebra $\M_0\subset\M$.
As a consequence, we observe:

\begin{corollary}\label{cor:uhlmann extensions}
    Let $\psi$ and $\phi$ be normal states on $\M$ and let $\N\supset \M$ be an extension of $\M$.
    For every extension $\tilde\psi$ of $\psi$ and every $\eps>0$, there exists an extension $\tilde \phi$ of $\phi$ such that
    \begin{equation}
        F(\psi,\phi) \approx_\eps F(\tilde \psi,\tilde\phi).
    \end{equation}
\end{corollary}
\begin{proof}
    We may take $\N$ to be in standard representation and let $\Psi$ be a purification of $\tilde\psi$.
    By \cref{thm:uhlmann}, there is a purification $\Phi$ of $\phi$ such that $F(\psi,\phi)\approx_\eps\abs{\ip\Psi\Phi}^2$.
    We let $\tilde\phi$ be denote the state that $\Phi$ induces on $\N$.
    Then \eqref{eq:fidelity poormans dpi} implies
    \begin{equation*}
        F(\psi,\phi) \approx_\eps F(\Psi,\Phi) \ge F(\tilde \psi,\tilde\phi) \ge F(\psi,\phi).  \qedhere
    \end{equation*}
\end{proof}

In particular, \cref{thm:uhlmann extensions} follows from \cref{cor:uhlmann extensions}.
We conclude by proving the Fuchs-van de Graaf inequality based on Uhlmann's Theorem.

\begin{proof}[Proof of \eqref{eq:FvdG}]
    Let $(\H,J,\P)$ be a standard form of $\M$.
    Let $\eps>0$. By \cref{thm:uhlmann}, there are purifications $\Psi,\Phi\in \H$ of the states $\psi,\phi$ such that $F(\psi, \phi) \le \abs{\ip\Psi\Phi}^2+\eps$.
    Let $\omega_{\Psi},\omega_{\Phi}\in\predualB(\H)$ denote the vector states induced by $\Psi,\Phi$, respectively.
    Then
    \begin{equation*}
        \norm{\psi-\phi}\le \norm{\omega_\Psi-\omega_\Phi} 
        = 2(1-\abs{\ip{\Psi}{\Phi}}^2)^{1/2} 
        \le 2(1-F(\psi,\phi)+\eps)^{1/2}.
    \end{equation*}
    Thus, the upper bound follows by taking $\eps\to0$.
    To see the lower bound, we consider the canonical purifications $\Omega_{\psi},\Omega_{\phi}\in\P$. 
    Using the inequality $\norm{\Omega_\psi-\Omega_\phi}^2 \le \norm{\psi-\phi}$ (see \eqref{eq:state-vector-estimate}) and \cref{thm:uhlmann}, we have
    \begin{equation*}
        1-F(\psi,\phi)^{1/2} \le 1-\ip{\Omega_\psi}{\Omega_\phi} = \frac12 \norm{\Omega_\psi-\Omega_\phi}^2 \le \frac12\norm{\psi-\phi}. \qedhere
    \end{equation*}
\end{proof}

In addition to the fidelity and the norm distance, there are many more, typically asymmetric, distance measures such as the relative entropy and other quantum divergences, most of which have been considered in the context of von Neumann algebras, see \cref{sec:notes-vNQI} or \cite{ohya_quantum_1993,hiai_quantum_2021,berta_smooth_2015,berta_renyi_2018}.

\subsection{Intrinsic description of quantum systems and outer quantum channels}\label{sec:intrinsic}

In \cref{sec:agents}, we considered subsystems of the full quantum systems and described them through their observable algebras.
Since every subsystem is a system in its own right, there should be an intrinsic description that bypasses the surrounding full quantum system.
Mathematically, it seems obvious what this description is:
We take the observable algebra $\M$ of the system and forget the Hilbert space that it acted on.
Then the states of the system are the normal states on $\M$, measurements are $\M$-valued POVMs, and operations on the quantum system are the $\M$-inner ones.
This makes sense because \cref{prop:implementable} characterizes the implementable operations in a way that is intrinsic to $\M$.
Moreover, it does not yield a larger class of systems because all abstract von Neumann algebras have faithful representations.


In order to justify such an intrinsic description of quantum systems, we need to argue that the observable algebra $\M$ of a subsystem of a full quantum system indeed captures all operational properties that are intrinsic to the subsystem.
Although this seems obvious, we give an explicit argument.
Let us consider an abstract von Neumann algebra $\M$ and a pair of full quantum systems with Hilbert spaces $\H_1$ and $\H_2$.
We assume that both of these have subsystems with observable algebras $\pi_1(\M)$ and $\pi_2(\M)$, where $\pi_j:\M\to\B(\H)$, $j=1,2$, are faithful representations of $\M$.
As a first step, we ask whether the possible reduced states of the two setups are different.
This is not the case: Whenever a von Neumann algebra $\M$ acts on a Hilbert space $\H$, then all of its normal states are induced by normal states on the surrounding $\B(\H)$.
As a consequence, the two environments cannot be distinguished without further knowledge of the states of the respective full quantum systems.
Indeed, let $(m_x)$ be a measurement that the agent performs in the two setups, i.e., they perform $(\pi_j(m_x))$ on the system with Hilbert space $\H_j$, $j=1,2$.
It is clear that if $\phi_1,\phi_2$ are states of the respective full systems such that $\phi_1\circ\pi_1 = \phi_2 \circ\pi_2$, then the output distributions on both systems are the same
\begin{eqnarray}
    p_x = \phi_1(\pi_1(m_x)) = \phi_2(\pi_2(m_x))\qquad \forall x.
\end{eqnarray}
Thus, all measurements on system 1 can be reproduced by measurements on system 2 and vice versa. 
We conclude that the agent cannot learn anything about the environment by measuring the system.

The point of view that a von Neumann algebra describes an abstract quantum system is, of course, not new. 
In fact, it dates back all the way to von Neumann himself, for whom finding the most general mathematical description of quantum systems was one of the initial motivations for studying von Neumann algebras \cite{redei}.

We may now supplement the sentence "a full quantum system is described by a Hilbert space $\H$" with \emph{a quantum system is described by a von Neumann algebra $\M$}.
We emphasize again that we do not claim (nor believe) that every von Neumann algebra arises as the observable algebra of an actual quantum system. 
Nonetheless, there is no clear separation between those von Neumann algebras that qualify as models of physical systems and those that do not.
This makes it even more interesting to understand the interplay between operational properties and von Neumann algebraic properties of the von Neumann algebra.


We caution the reader that, even though abstract von Neumann algebras can be used to describe classical systems or quantum-classical hybrids \cite{dammeier_quantum-classical_2023}, this is only partially covered by our framework.
The reason is that in our framework, only inner quantum channels, which always act trivially on the center $Z(\M)$, are implementable operations.
However, in a truly classical system, the observable algebra would be abelian so that all inner operations would be trivial.
To understand this, we recall that our conclusion that only inner operations can be applied was deduced from the assumption that we are dealing with subsystems of full quantum systems (or, more generally, agents with partial access).
This assumption does not apply to classical systems.
Instead, it applies to cases where the center of the von Neumann algebras describes \emph{superselection rules} such as the parity superselection rule in fermionic systems.\footnote{The observable algebra associated with a finite number of fermionic modes is the \emph{even} part of CAR algebra. 
When considering a fixed sector of a systems with infinitely many modes,
the observable algebra is the von Neumann algebra generated by the even part of the CAR algebra in the sector under consideration.}
If one wants to describe a quantum-classical hybrid system, one has to specify which operations should be regarded as implementable since the observable algebra itself does not yet know what kind of system it describes.%
\footnote{Another aspect to be cautious about when discussing quantum-classical hybrids is that a direct integral $\int_X \M_x\,d\mu(x)$ with fibers $\M_x\cong \M$ does not carry the same information as a tensor product $L^\oo(X,\mu)\barox\M$. In both cases, it is clear what the classical subsystem is, but only in the latter case do we know what the quantum subsystem is (assuming $\M$ is a factor). Identifying the two requires choosing coordinates for every $x\in X$, and this choice is unique up to a unitary section $x_\mapsto u_x$.}

In \cref{sec:agents}, we identified the operations of an agent with observable algebra $\M$ as the $\M$-inner.
It is sometimes natural to consider a larger class of quantum channels that we call \emph{outer quantum channels}, as opposed to inner ones.
In particular, these allow us to go from one system to another.
Consider a pair of von Neumann algebras $\M$ and $\N$.
Then, in the Heisenberg picture, an outer quantum channel from is a normal ucp map
\begin{equation*}
    T: \N \to \M.
\end{equation*}
The corresponding Schrödinger picture version is that of a normalization-preserving cp map $T_*: \M_*\to \N_*$. 
The two are connected via $T_*(\phi)=\phi\circ T$, $\phi\in \M_*$.
Outer quantum channels describe abstract state transformations that respect stochastic mixtures, i.e., convex combinations of states, and are compatible with innocent bystanders.\footnote{This means that $T\ox\id$ is still an outer channel if $\id$ denotes the identity on some system $\B(\K)$, the innocent bystander, that does not take part in the transformation.}
Outer quantum channels have a dilation theory known as Paschke dilations \cite{westerbaan_paschke_2017}.
Importantly, many natural constructions of new (outer) quantum channels from old ones preserve the subset of inner channels, e.g., the Petz recovery map \cite{ohya_quantum_1993,hiai_quantum_2021} of an inner channel is easily seen to be an inner channel and the the Paschke dilation reduces to the inner Stinespring dislation for inner channels (see \cref{lem:winner}).

We believe that the implementability of outer channels $T:\N\to\M$ by an agent whose observable algebra is $\M$ (or $\N$) is highly questionable in general, even if $\M=\N$.
Implementability is, however, not strictly necessary for a quantum channel to describe operationally meaningful state transformations of a quantum system.
For instance, the channel could describe a coarse-graining procedure, or it could a form of noise that affects states if the agent does nothing for an extended period of time.
To illustrate this, we consider a class of examples where outer channels arise from symmetries:
Consider a subsystem with observable algebra $\M$ of a full quantum system with Hilbert space $\H$.
Let $U$ be unitary on $\H$ with the property that $U^*\M U=\M$.
Then $T = \Ad(U)\restriction\M$ is an automorphism of $\M$ and, hence, an outer quantum channel, which is, in general, not $\M$-inner.
For $T$ to be inner, it is necessary and sufficient that $U$ decomposes as a product $U=uu'$ of unitaries $u\up\prime\in\M\up\prime$.
A concrete example in the context of quantum field theory, where $U$ implements a Lorentz transformation, is discussed in \cref{sec:qft}.

We close our discussion of outer channels with remarks on cases, where channels between different von Neumann algebras can be regarded as implementable.
We fix an agent with observable algebra $\M$.
If $\M$ and $\N$ are type $\I$ factor, every outer quantum channel $T:\M\to\M$ is inner \cite[Thm.~3.3]{kraus_general_1971}.
In \cref{sec:axioms}, we have considered implementable channels $T:\M\barox\B(\K_2)\to\M\barox\B(\K_1)$ in the presence of ancillary systems, and identified them as the $\M$-inner channels.
If $\N\subset \M$ is a von Neumann subalgebra of an agent's observable algebra $\M$, then the inclusion map $j:\N\hookrightarrow\M$ corresponds to discarding all operations that are not $\N$-inner.
Since the agent can always decide to discard some of their operations, the inclusion map implementable.
Afterward, the agent's observable algebra is the subalgebra $\N$.
If the agent wishes to discard the full system, they choose $\N=\CC1$.
However, in general, the agent cannot apply an inclusion map $\M \hookrightarrow \tilde \N$ into a larger algebra, as this would give them access to new observables.
The following result shows that every channel whose output system (in the Schrödinger picture) is a full quantum system can be implemented:

\begin{lemma}
    An outer quantum channel $T:\B(\K)\to \M$ (more generally, $T:\B(\K_2)\to \M\barox\B(\K)$) can be implementable by an agent who has access to the set of $\M$-inner operations by applying an $\M$-inner operation and then discarding the system, instead of the ancilla.
\end{lemma}
\begin{proof}
    By \cref{lem:inner-distillation} there is an $\M$-inner quantum channel $\hat T:\M\barox\B(\K_2)\to\M\barox\B(\K_1)$ such that $T = \hat T(1_\H\ox \placeholder)$.
\end{proof}

\subsection{Examples of von Neumann algebraic quantum systems}\label{sec:examples}

In the following, we consider three physical scenarios in which quantum systems with infinitely many degrees of freedom are modelled with von Neumann algebras.

\subsubsection{Ground state sectors of quantum many-body systems}\label{sec:many-body}

We consider a quantum many-body system of spins localized on the sites of a lattice $\Gamma\subset \RR^d$, e.g., $\Gamma=\ZZ^d$.
To every finite region $A\subset \Gamma$, we associate the Hilbert space $\H_A=\bigotimes_{x\in A}\CC^d$ and the unital *-algebra $\A_A = \B(\H_A)$.
If $B\subset A$ is a smaller region, we use the product form $\H_A = \H_B \ox\H_{A\setminus B}$ to identify $\A_B$ with the subalgebra $\A(B)\ox 1_{A\setminus B}$ of $\A_A$.
In this way, we obtain a directed family of unital *-algebras $(\A_A)_{A\Subset \Gamma}$ (where $\Subset$ indicates a finite subset), and, upon taking their union, we get a unital *-algebra
\begin{equation}
    \A_\Gamma = \bigcup_{A\Subset \Gamma} \A_A
\end{equation}
which describes the collection of all operators with finite support.
This *-algebra can be equipped naturally with a $C^*$-norm via $\norm a = \norm{a}_{\B(\otimes_{x\in A}\CC^d)}$ for some/all $A\Subset\Gamma$ with $a\in A_A$, and if we take the completion $\bar\A_\Gamma$ what we obtain the so-called quasi-local ($C^*$-) algebra \cite{bratteli_robinson2}.\footnote{In fact, we have just described its standard construction.}
By construction, the unital *-algebra $\A_\Gamma$ and the quasi-local $C^*$-algebra $\bar\A_\Gamma$ have the same algebraic states and the same representation theory.

The algebraic structure of finitely localized operators, which only depends on the dimension $d$ of the local spins and the lattice $\Gamma$, is the same for many distinct physical models.
For instance, both the Heisenberg antiferromagnet and a collection of uncorrelated spins are described by $\Gamma=\ZZ$ and $d=2$.
Different physical models are specified by their Hamiltonians, which are formally written as
\begin{equation}
    H = \sum_{A\Subset \Gamma} \Phi(A), \qquad \Phi(A)=\Phi(A)^*\in \A_A,
\end{equation}
where the interaction terms $\Phi(A)$ decay sufficiently fast when $A$ becomes large.
To make this rigorous, one associates to every finite region $A\Subset \Gamma$, the Hamiltonian $H_A= \sum_{A'\Subset \Gamma} \Phi(A')$ and the dynamics $\tau_t^A(a) = e^{itH_A}(\placeholder)e^{-itH_A}$.
If the interaction strength $\norm{\Phi(A)}$ decays sufficiently fast at large distances, the limits $\lim_{A \uparrow\Gamma} \tau_t^A(a)$ (along the directed set of finite regions) exists in $\bar\A_\Gamma$ for all $a\in\A_\Gamma$ and define a strongly continuous dynamics $\tau : \RR \curvearrowright \bar\A_\Gamma$ (see \cite{bratteli_robinson2, nachtergaele_quasi-locality_2019, nachtergaele_quasi-locality_2022, van_luijk_convergence_2024}).
In the following, we make the additional assumption that the infinitesimal generator $\delta$ of $\tau$ contains $\A_\Gamma$ as a core on which it is given by the norm-convergent series $\delta(a) =\sum_{A\Subset\Gamma} i[\Phi(A),a] =\text{``}i[H,a]\text{''}$.\footnote{These assumptions are quite mild (see \cite{bratteli_robinson2,nachtergaele_quasi-locality_2019,nachtergaele_quasi-locality_2022}).
For instance, they hold if the interaction has finite range, i.e., $\Phi(A) = 0$ whenever the diameter of $A$ exceeds some fixed number $\ell<\oo$.}
A \emph{ground state} $\omega$ is a state $\omega$ on $\A_\Gamma$ such that 
\begin{equation}\label{eq:ground state}
    -i\omega(a^*\delta(a)) \ge 0, \qquad a\in \A_\Gamma
\end{equation}
(see \cite[Def.~1.9]{ogata_classification_2021} or \cite{bratteli_robinson2}).
The ground space (meaning the set of ground states) is a convex $w^*$-compact set.
If we assume the ground state to be unique, it is a pure state on $\A_\Gamma$.

In this work, we will mean by a \emph{sector} of the many-body system a Hilbert space $\H$ carrying an irreducible representation of $\A_\Gamma$.%
\footnote{Every representation of $\A_\Gamma$ is faithful since $\A_\Gamma$ is a simple unital *-algebra, i.e., it does not have nontrivial ideals. We also remark that every representation of the unital *-algebra $\A_\Gamma$ (by bounded operators) extends to a representation of the quasi-local algebra $\bar\A_\Gamma$.}\textsuperscript,\footnote{Two irreducible representations describe the same sector of the many-body system if they are unitarily equivalent.}
In the following, we suppress the representation and regard $\A_\Gamma$ as a subalgebra of $\B(\H)$.
A \emph{ground state sector} is a sector that arises via the GNS representation with respect to a pure ground state $\omega$ \cite{bratteli_robinson2}.
The normal state space $\nstates(\H)$ describes those states of the many-body system that are finite excitations above the ground state.
More precisely, $\nstates(\H)$ is the set of states of the many-body system that can be approximated by acting on the ground state with operations of finite support, where we say that an operation has support in a finite region $A\Subset\Gamma$ if it can be written with Kraus operators in $\A_A$.

From now on, we consider a fixed sector $\H$. 
For a region $A\subset \Gamma$, we think of an agent, also denoted $A$, that can implement operations with \emph{approximately finite support} in $A$.
By this, we mean operations that are finitely localized in $A$, and any limits thereof in the topology of pointwise norm-convergence on states or the topology of convergence in expectation values.
Based on \cref{lem:approx-inner}, both topologies lead to the same class of operations, namely those that are $\M_A$-inner, where $\M_A$ is the von Neumann algebra 
\begin{equation}\label{eq:vNa in many-body}
    \M_A = \A_A''.
\end{equation}
With this definition, we have $\M_\Gamma=\B(\H)$ and $\M_\emptyset =\CC1$.
We note the above as a Lemma:

\begin{lemma}\label{lem:many-body}
    An operation on the Hilbert space $\H$ of the sector is approximately finitely localized in a region $A\subset \Gamma$ if and only if it is $\M_A$-inner.
    In particular, all operations are approximately localized in $\Gamma$.
\end{lemma}

The irreducibility of the action of $\A_\Gamma$ on $\H$ implies that $\M_A\vee \M_{A^c} = \B(\H)$ for all $A\subset\Gamma$.
By \cref{lem:irreducible}, this has the consequence that the von Neumann algebras $\M_A$ are, in fact, factors for all regions $A\subset\Gamma$.
Moreover, the factors $\M_A$ are, by construction, AFD.
The fact that finitely localized operators in disjoint regions commute implies that the factors $\M_A$ and $\M_B$ commute if $A\cap B=\emptyset$.

\begin{figure}[ht!]
    \centering
    \def\svgwidth{.4\textwidth}
    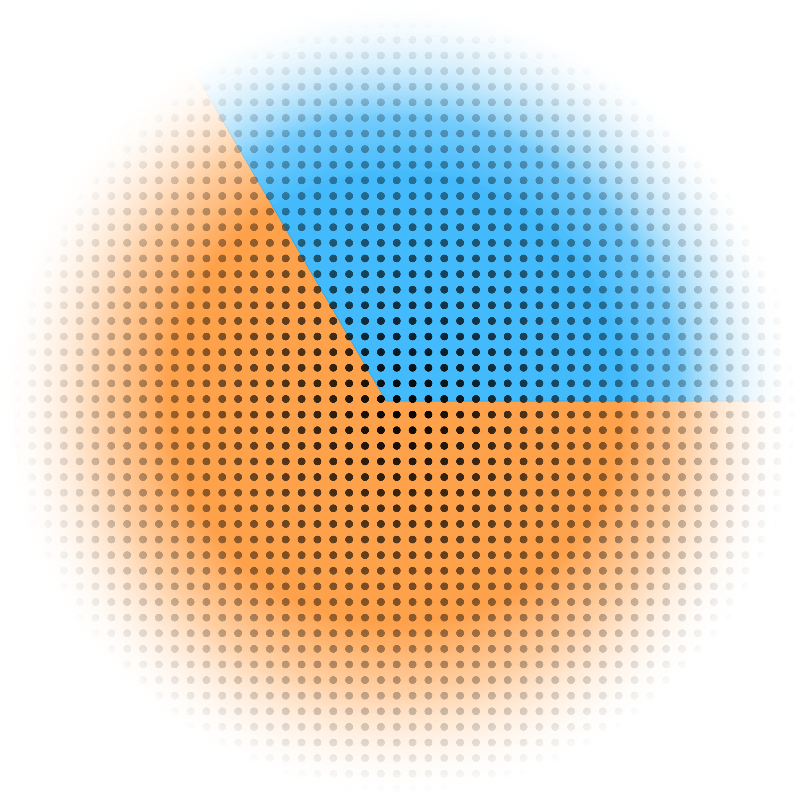
    \caption{Bipartition of an infinite quantum spin system into a cone and its complement.}
    \label{fig:many-body}
\end{figure}

Let us now discuss entanglement theory in the sector $\H$.
To do this, we fix an $N$-partition $\Gamma = A_1\cup \ldots A_N$ of the lattice, which gives us an $N$-tuple of pairwise commuting von Neumann algebras on $\H$.
We consider a pair of normal states $\psi,\phi\in\nstates(\H)$.
By \cref{lem:many-body}, the von Neumann algebraic definition of 
\begin{equation}\label{eq:many-body locc}
    \psi\locc\phi
\end{equation} 
in the sense of \cref{sec:entanglement-basics} is physically meaningful.
It encodes that $\psi$ can be reached from $\phi$ with classical operations and local operations of approximately finite support.
In this sense, the entanglement theory in the ground state is fully captured by the von Neumann algebras in \eqref{eq:vNa in many-body}.


We conclude with a brief summary of what is known about Haag duality and the type of the factors $\M_A$, $A\subset \Gamma$, in some concrete models.
%
\paragraph{\it Spin chains.}
For spin chains, the lattice is given by $\Gamma=\ZZ$, and we pick the bipartition $R=\NN$, $L=\ZZ\setminus\NN$.
We begin by considering the ground state sector of translation-invariant Hamiltonians with a uniform energy gap.
Building on the work of Hastings \cite{hastings_area_2007}, Matsui showed in \cite{matsui_boundedness_2013} that the factors $\M_L$ and $\M_R$ are of type $\I$.
In particular, this shows that Haag duality holds, i.e., $\M_L=\M_R'$.
We now turn to general, non-gapped, translation-invariant spin chains.
In \cite{keyl_entanglement_2006}, it was conjectured that the only possible types for $\M_L$ and $\M_R$ are type $\I$ and type $\III_1$. 
So far, no counterexamples have been found and no models violating Haag duality for the left-right bipartition are known.
The critical XX model, and the transverse-field Ising models are known to have type $\III_1$ factors $\M_{L/R}$ and to satisfy Haag duality (see \cite{keyl_entanglement_2006} and \cite{van_luijk_critical_2025}, respectively).
In both cases, this was shown by using a Jordan-Wigner isomorphism to reduce the problem to a problem about free fermion chains (see the next paragraph).

\paragraph{\it Critical fermion chains.}
Following \cite{van_luijk_critical_2025}, we consider the ground state sector of a 1D fermion chain with a quadratic translation-invariant Hamiltonian
\begin{equation}
    H = \sum_{x,y} h(x-y) \,a_x^\dagger a_y.
\end{equation}
We assume here that the interactions are not too long-ranged in the sense that $\sum_{n\in\ZZ} h(n)<\oo$, and assume that the model is critical in the sense that the energy band $\hat h(k)$, i.e., the Fourier transform $\hat h(k)=\sum_n e^{ink}h(n)$, takes on both signs.
Moreover, for simplicity, let us assume that the model has a unique ground state.
Then the fermionic observable algebras $\M_L$ and $\M_R$ associated with a left-right bipartition of the chain are of type $\III_1$ and satisfy Haag duality (see \cite{van_luijk_critical_2025} for details).

\paragraph{\it Gapped systems in 2D.}

We consider ground state sectors of translation invariant gapped Hamiltonians on a 2D lattice $\Gamma = \ZZ^2$.
For simplicity, we assume the ground state to be unique.
In the following, we list some known type properties and how they depend on topological order \cite{zeng_introduction_2019}, which are compiled from the works \cite{naaijkensAnyonsInfiniteQuantum2012,fiedlerHaagDualityKitaevs2015,naaijkensQuantumSpinSystems2017,ogata_type_2024,naaijkens_split_2022,jones_local_2025,chuah_boundary_2024,bhardwaj_superselection_2025,van_luijk_large-scale_2025}.

\begin{itemize}
    \item If a region $A$ is "properly infinite" in the sense that it contains balls of arbitrary diameter, then the corresponding factor $\M_A$ is properly infinite, i.e., of type $\I_\oo$, $\II_\oo$, or $\III$ (see \cite[App.~D]{van_luijk_large-scale_2025}, which uses an argument from \cite{naaijkensAnyonsInfiniteQuantum2012,keyl_entanglement_2006}).

    \item In topologically trivial systems, all types that were not excluded in the previous item occur:
    In \cite{van_luijk_large-scale_2025}, a family of models is constructed in which the factors $\M_A$ of all infinite regions $A\subset\Gamma$ have the same type, which can be chosen to be $\I_\oo$, $\II_\oo$, or $\III_\lambda$ with $0<\lambda\le1$ by adjusting the parameters of the model. Moreover, these models satisfy Haag duality in arbitrary bipartitions.

    \item In topologically ordered systems, the factors $\M_C$ associated with cones are not of type $\I$ \cite{naaijkens_split_2022}.

    \item In Kitaev's quantum double model for abelian groups, the factors associated with cones have been shown to be of $\II_\oo$ \cite{ogata_type_2024}.
    For cones, Haag duality is shown in \cite{fiedlerHaagDualityKitaevs2015}.
    For regions $A=C_1\cup C_2$ that are unions of distanced cones, Haag duality is known to fail. 
    In this case, the failure can be understood explicitly: Ribbon operators creating pairs of anyonic excitations in the respective cones are unitaries that commute with $\M_{A^c}$ but are not elements of $\M_{C_1}\vee \M_{C_2}$ \cite{fiedler_jones_2017}.

    \item In Levin-Wen models \cite{levinStringnetCondensationPhysical2005}, the observable algebras associated with cones are of type $\II_\oo$ if the quantum dimensions of the simple objects of the unitary fusion category are all $1$ \cite{jones_local_2025}.
    In all other cases, they are of type $\III$ and the subtypes can be computed from the fusion rules of the unitary fusion category (a proof due to Izumi can be found in the appendix of \cite{jones_local_2025}).
    Haag duality for cones is shown in \cite{ogata_haag_2025}.
\end{itemize}

\subsubsection{Relativistic quantum fields}\label{sec:qft}

In algebraic quantum field theory~\cite{haag_local_1996}, a fixed sector of a quantum field theory on a spacetime $\mathbb{M}$ is described by a net
\begin{equation*}
    \O \mapsto \A(\O)
\end{equation*}
of von Neumann algebras $\A(\O)$, defined for open regions $\O\subset \mathbb{M}$, on the sector's Hilbert space $\H$.
The von Neumann algebras $\A(\O)$ describe the collection of smeared quantum fields in $\O$. 
The net $\O\mapsto \A(\O)$ has to satisfy certain axioms which reflect physical properties.
For instance, the \emph{causality} axiom states that $\A(\O_1)$ and $\A(\O_2)$ commute whenever $\O_1$ and $\O_2$ are causally separated.
The axiom of \emph{isotony} states that observables in a smaller region $\O_1\subset\O_2$ are, in particular, also observables of the larger region: $\A(\O_1)\subset\A(\O_2)$.
In general, there are a variety of further axioms, depending on the context.
We briefly discuss the most common ones:
\emph{Causal completeness} states that the observable algebras $\A(\O)$ and $\A(\O'')$ of a region $\O$ and its causal closure $\O''$ are equal ($\O'$ denotes the causal complement and, hence, $\O''$ is the causal closure).
\emph{Additivity} states that the observable algebra $\A(\O_1\cup\O_2)$ of a union is $\A(\O_1)\vee\A(\O_2)$.
\emph{Haag duality} states that the observable algebra $\A(\O')$ of the causal complement of a region $\O$ is given by the commutant of $\A(\O)$:
\begin{equation}\label{eq:HD in qft}
    \A(\O)'=\A(\O').
\end{equation}
Note that Haag duality implies both causality and causal completeness. 
These axioms (and a few more that we do not discuss here) imply that the von Neumann algebras $\A(\O)$ are AFD factors of type $\III_1$ factors~\cite{buchholz_universal_1987,baumgaertel1995oam,yngvason_role_2005}.

Due to the causality axiom, a pair $\O_A,\O_B\subset \mathbb{M}$ of causally separated regions gives rise to a pair of commuting von Neumann algebras $\M_A=\A(\O_A)$, $\M_B=\A(\O_B)$ on $\H$.
If the Haag duality axiom holds and if the regions $\O_A$ and $\O_B$ are causal complements of each other, i.e., $\O_A=\O_B'$ (see \cref{fig:qft}), then the von Neumann algebraic bipartite system $(\M_A,\M_B,\H)$ satisfies Haag duality.

\begin{figure}[ht!]
    \centering
    \def\svgwidth{.28\textwidth}
    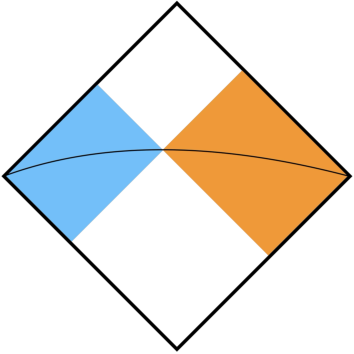
    \caption{Penrose diagram of Minkowski spacetime with one spatial dimension, partitioned into a pair $\O_A,\O_B$ of causally separated regions with $\O_A=\O_B'$.}
    \label{fig:qft}
\end{figure}

Let us now consider the simple example of the vacuum sector of a free scalar field theory of mass $m$ on the Minkowski spacetime $\RR^{1,3}$.
The Lorentz group is implemented by a unitary representation $g\mapsto U_g$ with respect to which the local von Neumann algebras are covariant $\A(g\O)=U_g\A(\O) U_g^*$.
We now consider a region $\O$ in $\RR^{1,3}$ together with a Lorentz transformation $g$ that leaves $\O$ invariant, e.g., $g$ may be a spatial rotation and $\O$ a causal diamond over a rotation-symmetric base, or a Lorentz boost and $\O=W$ a suitable Rindler wedge with $gW=W$. 
Invariance of $\O$ under $g$ then implies $\A(\O)=U_g\A(\O)U_g^*$.
Thus, the restriction of $\Ad(U_g)$ defines an automorphism on $\M=\A(\O)$, which is, in general, non-inner.%
\footnote{As an example, consider as the region $\O$ a Rindler wedge $W$. Then the modular automorphism group of the vacuum state is known to coincide with a one-parameter group of Lorentz boosts $g_t$ with $g_tW=W$ \cite{bisognano_duality_1975}. Since the factors are type $\III_1$, the modular flow (= the Lorentz boost) is not inner \cite{takesaki2}.}
In the language of \cref{sec:intrinsic}, this provides examples of outer quantum channel that are not inner.

We caution the reader that the QFT setting described above does not fit perfectly into the framework developed in this chapter, in particular, \cref{sec:agents}.
Although the von Neumann algebras $\A(\O)$ are interpreted as the observable algebras of the fields localized in the region $\O$, we cannot assume that all $\A(\O)$-inner operations are physically implementable by an agent localized in $\O$.
This assumption necessarily leads to violations of causality as pointed out in \cite{Sorkin1993-SORIMO} (see also \cite{much_superluminal_2023}).
For instance, it implies that the outcomes of a local measurement depend on the application of a unitary in a causally separated region \cite{papageorgiou_eliminating_2024,fewster_measurement_2025}.
We refer to the works \cite{fewster2020local_measurements,fewster_measurement_2025,bostelmann_impossible_2021} for a recently proposed solution to the problem of finding a mathematical description of local operations and measurements.
Roughly speaking, their idea is to define local operations as those that can be constructed by introducing another field theory (the "probe"), allowing a localized interaction, followed by a measurement of the probe.
Note that this follows the same philosophy as our definition of implementable instruments via unitary dilations in \cref{sec:agents}, but gives a constructive description of the allowed unitaries $u$, which is in accordance with principles of relativistic field theory.
Nonetheless, the assumption that all $\A(\O)$-inner operations are implementable is often made (sometimes implicitly) \cite{haag_local_1996,summers_vacuum_1985,summers_independence_1990,summers_maximal_1988,summers_maximal_1987,summers_bells_1995,verch_distillability_2005,van_luijk_relativistic_2024,hollands_entanglement_2018}. In this case, the "operational" statements that follow should be taken with a grain of salt.
In addition to the above, there is the problem that classical communication, which the notion of LOCC and, hence, our operational understanding of entanglement, is built upon, for causally separated systems does not make any sense physically (see \cite{verch_distillability_2005} for a nice discussion of this point).

None of the problems above apply to non-relativistic quantum field theories, where the notion of localization is purely spatial (as opposed to the spatio-temporal localization in relativistic field theory).
As an example, consider conformal field theories on the circle, which describe critical many-body systems in the thermodynamic systems.

\subsubsection{Infinitely many entangled pure states}\label{sec:idealized}


We consider a setting in which a finite number of agents $A_1,\ldots A_N$ share an infinite number of finite-dimensional $N$-partite systems $\H_j = \H_{1|j}\ox\ldots \H_{N|j}$ in pure states 
\begin{equation}
    \Psi_j \in \H_j, \qquad j\in \NN.
\end{equation}
The full system is an $N$-partite system with infinitely many degrees of freedom, see \cref{fig:itpfi} for a visualization with $N=3$. 
We assume that agents can only act on finitely many of the infinite $N$-partite systems.
This assumption allows for a description of the full system in terms of a Hilbert space, namely the restricted infinite tensor product of the Hilbert spaces $\H_j$ relative to the sequence $\Psi_j$ \cite[Sec.~XIV.1]{takesaki3}:
\begin{equation}
    \H = \bigotimes_{j\in\NN}\, (\H_j ; \Psi_j).
\end{equation}
This Hilbert space contains the vector $\Psi= \otimes_{j\in\NN} \Psi_j$. In fact, it is the subspace of the unrestricted infinite tensor product \cite{von_Neumann_1939_infinite} that can be approximated by finite perturbations of $\Psi$.
Here, a vector $\Phi$ is a finite perturbation of $\Psi$ if $\Phi=a\Psi$ for an operator $a$ that acts trivially on all but finitely many of the multipartite systems.
With each agent $A_x$, we associate an AFD factor $\M_x$ on $\H$ via
\begin{equation}
    \M_x = \bigotimes_{j\in\NN}\, \big( \M_{x|j},\Psi_j),
\end{equation}
where $\M_{x|j}=1^{\ox j-1}\ox \B(\H_{x|j})\ox 1^{\ox N-j}$ is the observable algebra of the agent $A_x$ on the $j$th $N$-partite system (see \cite[Sec.~XIV.1]{takesaki3}).
More generally, for every collection of agents $\prod_{x\in I} A_x$, $I \subset \{1,\ldots N\}$, 
\begin{equation}
    \M_{I} = \bigvee_{x\in I}\M_{x|j} = \bigotimes_{j\in \NN} \,(\M_{I|j},\Psi_j),
\end{equation}
defines an AFD factor on $\H$, where $\M_{I|j}= \vee_{x \in I} \M_{x|j}\subset \B(\H_j)$.
Clearly, we have $\M_{\{x\}}=\M_x$, $\M_I\subset \M_J$ if $I\subset J$, and $\M_{\{1,\ldots,\}}=\B(\H)$.
The commutation theorem for infinite tensor products \cite[Thm.~XIV.1.9]{takesaki3} implies that Haag duality holds in any bipartition
\begin{equation}
    \M_I' = \M_{I^c}, \qquad I \subseteq \{1,\ldots,N\}.
\end{equation}

\begin{figure}[ht!]
    \centering
    \def\svgwidth{.6\textwidth}
    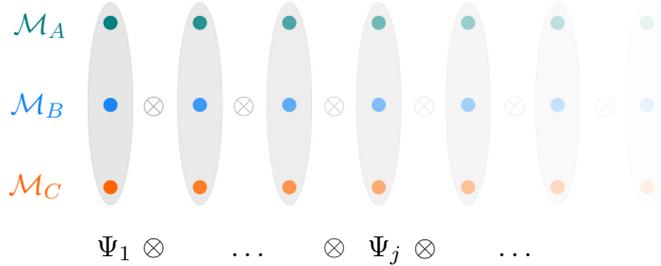
    \caption{Three parties sharing infinitely many entangled states.}
    \label{fig:itpfi}
\end{figure}

Let us consider the class of implementable operations.
As said above, we assume that the agents can only act on finitely many systems.
By \cref{lem:approx-inner}, the operations that can be implemented approximately by the agent are precisely the $\M_x$-inner ones.

The factors that we constructed above are so-called ITPFI (infinite tensor product of finite type $\I$) factors, and we may refer to $N$-partite systems constructed above as \emph{ITPFI multipartite systems}.
These form a mathematically well-understood pool of nontrivial examples, in particular, for the bipartite case where $N=2$.
Indeed, if $N=2$, the algebraic properties of $\M_1$, $\M_2$ are completely understood thanks to the work of Araki and Woods \cite{araki_classification_1968,araki_complete_1966}.
Let us consider a well-known example. If $\H_j = (\CC^2)^{\ox 2}$ and 
\begin{equation}
    \Psi_j = (1+\lambda)^{-1/2}\big(\lambda^{1/2} \ket0\ket0 +\ket1\ket1\big), \quad j\in\NN,
\end{equation}
then $\M_1$ and $\M_2=\M_1'$ are factors in standard representation.
They are of type $\I$ if $\lambda=0$, of type $\III_\lambda$ if $0<\lambda<1$, of type $\II_1$ if $\lambda=1$ and of type $\III_{\lambda^{-1}}$ if $\lambda>1$.
Even in the case of more than two parties, ITPFI systems will be very useful as they allow us to construct multipartite embezzling states (see \cref{sec:strong-entanglement}).

\subsection{Notes}\label{sec:notes-vNQI}

In the following, we briefly review some aspects that have to be mentioned, but that did not receive a detailed exposition due to the constraints of this work.
There is a vast literature that is concerned with properties of von Neumann algebras inspired by information theory.
In the following, we stick to a selection of cases that have a direct operational interpretation.
Due to the extensive literature on the subject, this list remains incomplete.

\paragraph{\it Bell's inequalities.}
Arguably, the first work in von Neumann algebraic quantum information theory is Summers and Werner's 1985 paper \cite{summers_vacuum_1985} on the topic of Bell's inequalities \cite{bell_einstein_1964} in quantum field theory.
A series of papers \cite{summers_vacuum_1985,summers_maximal_1987,summers_bells_1987,summers_bells_1995} followed in which they relate the maximal violation of Bell's inequalities to algebraic properties of the observable algebras in quantum field theory (see \cref{sec:qft}).\footnote{In passing, they discovered self-testing (see, e.g., \cite[Lem.~2.2]{summers_maximal_1988}), which was later rediscovered in finite-dimensional quantum information theory \cite{supic_self-testing_2020}.}
In their works, Summers and Werner assumed that all measurements in the local observable algebras are, in principle, physically implementable.
While we know today that this assumption leads to contradictions with causality (see the last paragraph of \cref{sec:qft}), this assumption does place their work in the context of von Neumann algebraic quantum information theory.
Their results characterize the AFD type $\II_1$ factor and its commutant in standard representation as the universal resource for Bell nonlocality, see \cite[Thm.~2.5]{summers_maximal_1988}.

\paragraph{\it Relative entropy and quantum divergences.}

In \cref{sec:fidelity}, we considered the norm distance and the fidelity, which operationally quantify the distance of states.
Both of these are symmetric, i.e., the distance from $\psi$ to $\phi$ is the same as the distance from $\phi$ to $\psi$.
However, as in classical information theory, asymmetric distance measures become necessary in hypothesis testing, where false negatives can be much worse than false positives \cite{wang_resource_2019}.
By far the most popular among these distance measures is the quantum relative entropy.
There is, however, a whole zoo of so-called quantum divergences, with operational meanings in various settings, e.g., quantum cryptography, \cite{wang_resource_2019}.
Many of these have been extended to (or were originally defined in) the von Neumann algebraic setting \cite{hiai_quantum_2021,berta_renyi_2018,berta_smooth_2015,furrer_position-momentum_2014,hiai_alpha-z-renyi_2024,ohya_quantum_1993}.
We especially recommend the monograph \cite{hiai_quantum_2021} of Hiai, which, among other things, includes a nice description of the theory of recovery maps.
Applications to explicit quantum information theoretic problems with systems described by von Neumann algebras have been considered, e.g., in \cite{berta_smooth_2015,naaijkens_subfactors_2018,berta_renyi_2018,ohya_quantum_1993}.
Let us also mention \cite{fawzi_asymptotic_2025}, where the asymptotic equipartition theorem is generalized to the von Neumann algebraic setting.

\paragraph{\it Relative entropy of entanglement.} 
The relative entropy of entanglement is a standard entanglement monotone for mixed states.
It is defined as the smallest relative entropy between a given state and the set of separable, i.e., classically correlated, states.
The relative entropy of entanglement has been considered and estimated in the context of von Neumann algebraic systems arising in quantum field theory in \cite{hollands_entanglement_2018}, where explicit estimates for the entanglement entropies in concrete models can be found.
However, it only makes sense in bipartite systems that admit normal product states, which requires the split property to hold (see \cref{lem:no product states}).
This rules out the relative entropy of entanglement as a sensible entanglement measure in the settings we are interested in, where the split property generally fails (see \cref{sec:entanglement-basics}).

\paragraph{\it Tsirelson's problem.}

Tsirelson's problem \cite{scholz_tsirelsons_2008} is a deep problem about correlation experiments, which was solved in \cite{ji_mipre_2022}.
It has sparked lots of research, especially on the intersection of operator algebras and quantum information theory in the last 20 years (see, e.g., \cite{scholz_tsirelsons_2008,junge_connes_2011,ji_mipre_2022,van_luijk_schmidt_2024,slofstra_set_2019,vidick_operator_2019,slofstra_tsirelsons_2020,liu_non-local_2022,fritz_tsirelsons_2012,ozawa_tsirelsons_2013}).
The setup is that of a bipartite system in which Alice and Bob are given access to a collection of measurements with a common outcome space $X$. 
We label Alice's measurements by $i$ and Bob's measurements by $j$.
Depending on the state of the full system and the respective choices of measurements, they observe a probability distribution $p(x,y|i,j)$ on the joint outcomes $(x,y)$.
Tsirelson's problem asks whether the correlation functions of finite-dimensional bipartite systems can approximate the correlation functions in the "commuting operator framework", which allows Alice and Bob access to pairwise commuting observables on a joint, possibly infinite-dimensional, Hilbert space.
In \cite{ji_mipre_2022}, this was settled in the negative.
Nonetheless, an explicit counterexample is not known.

Given the above and the setup of von Neumann algebraic quantum information theory, the natural question is which algebraic properties a pair of commuting von Neumann algebras $(\M_A,\M_B)$ on a Hilbert space $\H$ must have, such that it allows for counterexamples to Tsirelson's problem.
It is known that counterexamples cannot be found if $\M_A$ and $\M_B$ are AFD \cite{scholz_tsirelsons_2008} (see \cite[Prop.~52]{van_luijk_schmidt_2024} for a detailed proof).
In fact, counterexamples can only be found if $\M_A$ and $\M_B$ are not Connes embeddable \cite{junge_connes_2011}.

\paragraph{\it Subfactors and index.}

A subfactor is an inclusion of factors $\N\subset \M$ \cite{jones_index_1983}.
In his seminal paper \cite{jones_index_1983}, Jones defined an index $[\M:\N]$ for type $\II_1$ subfactor inclusions, which Kosaki and Longo generalized to subfactors of arbitrary type \cite{kosaki_extension_1986,longo_index_1989}.
A subfactor is irreducible if $\N'\cap\M=\CC1$.
For non-irreducible subfactors, the index depends on the choice of a conditional expectation $\M\to\N$.
Roughly speaking, this means an outer quantum channel that describes how to extend states on $\N$ to states on $\M$.
For subfactors of type $\II_1$, there is a canonical conditional expectation, and for irreducible subfactors, there is a unique one (if it exists).
The theory of subfactors and their index developed into one of the main research areas in operator algebras in the 1990s (see \cite{jones_introduction_1997} and the references therein).
In von Neumann algebraic quantum information theory, subfactors, or more generally, inclusions of von Neumann algebras, occur when we consider two agents with partial access to a full quantum system such that one agent is strictly more powerful than the other one (see \eqref{eq:inclusion of observable algs}).
The index can then be shown to be an operational measure of how much more the powerful agent can do \cite{naaijkens_subfactors_2018}. 
In \cite{fiedler_jones_2017} this idea was considered in the context of secret sharing, where the more powerful agent tries to hide information from the less powerful one.
In \cite{gao_relative_2020}, the index is characterized using relative entropies and other quantum divergences between states on $\N$ and $\M$ with the help of a conditional expectation, although this work does not apply to the type $\III$ case.
The index has also appeared in work on quantum teleportation and dense coding in \cite{conlon_quantum_2023} (see also \cite{huang_dense_2019}).

We mention an open problem, also discussed in \cite{van_luijk_schmidt_2024}:
If $(\M_A,\M_B)$ is an irreducible pair of commuting factors on a Hilbert space $\H$ (see \cref{lem:irreducible}), does the index of the irreducible subfactor $\M_A\subset\M_B'$ have an operational significance in entanglement theory?
We can answer this affirmatively for the case where the index is trivial, i.e., $[\M_B':\M_A]=1$.
Indeed, a trivial index is the same as Haag duality $\M_A=\M_B'$, which we show to have an operational interpretation in \cref{sec:haag}.

\paragraph{\it Sufficiency.}
Further following the idea that an inclusion $\N\subset\M$ of von Neumann algebras corresponds to a setting with two agents, where one agent has access to more operations than the other one, it makes sense to ask whether the weak agent has access to sufficiently many observables in order to distinguish a given pair $(\omega_0,\omega_2)$ of states as good as the strong agent.
This is precisely the question, whether the weaker agent's observable algebra $\N$ is sufficient in the sense(s) of Umegaki, Petz, Jencova and others (see \cite{ohya_quantum_1993,hiai_quantum_2021,jencova_sufficiency_2006,petz_sufficiency_1988,kuramochi_minimal_2017,luczak_quantum_2014,luczak_aspects_2021} and the references therein).
There is a unique minimal sufficient subalgebra. Petz showed that this is precisely the von Neumann algebra generated by the Connes cocycles $(D\omega_0:D\omega_1)_t$, $t\in\RR$.
Instead of asking for the sufficiency of subalgebras, one can more generally ask for the sufficiency of a quantum channel, meaning that the channel is reversible relative to the given states.
We refer to the references above for details.
In the case of abelian von Neumann algebras, these notions of sufficiency reduces to the notions of sufficient statistics in classical probability theory \cite{jencova_sufficiency_2006}.

\section{Catalytic states on von Neumann algebras}\label{sec:catalytic states}

\localtableofcontents

\null

In this section, we consider approximate unitary equivalence of states.
We define a notion of catalytic states on von Neumann algebras, and we see in \cref{sec:strong-entanglement} that catalytic states are precisely the marginal states of embezzling states on bipartite quantum systems.
Many of our results on embezzlement of entanglement will be proved by reducing them to statements about catalytic states.
The results in this section are based on \cite{van_luijk_embezzlement_2024} and rely on insights from \cite{haagerup_equivalence_1990,haagerup_classification_2007,mbz_draft}.

\null

Approximate unitary equivalence of a pair of normal states $\omega_1,\omega_2$ on a von Neumann algebra $\N$ is defined as
\begin{equation}
    \omega_1\sim\omega_2 \quad :\iff\quad \inf_{u\in\U(\N)}\ \norm{\omega_1-u\omega_2 u^*}=0.
\end{equation}
If $\N$ is an observable algebra of a quantum system, approximate unitary equivalence means that, up to an arbitrarily small error, the given states can be transformed into each other with invertible quantum channels.
In finite type $\I$ factors, approximate unitary equivalence implies exact unitary equivalence. This is in general false, even for factors of type $\I_\oo$.

To decide approximate unitary equivalence, we need \emph{spectral} information of the two states.
Indeed, in type $\I$ factors, a pair of states is approximately unitarily equivalent if and only if they have the same spectrum (with equal multiplicities).
This generalizes to general von Neumann algebras if we replace the spectrum by the Haagerup-Størmer spectral state (see \cref{sec:spectral states}).
We say that a nonzero positive linear functional $\psi$ on a von Neumann algebra $\M$ is a \emph{catalyst} for a pair of normal states $\omega_1,\omega_2$ on a von Neumann algebra $\N$ if $\psi\ox\omega_1\sim\psi\ox\omega_2$ but $\omega_1\not\sim\omega_2$.
We refer to \cite{lipka-bartosik_catalysis_2024} for a recent review of catalysis in quantum information theory (where catalysis is considered in a variety of resource theories).
In the language of \cite{lipka-bartosik_catalysis_2024}, we could refer to the above notion of catalysis also as "embezzlement of (im)purity".

\begin{definition}\label{def:universal catalyst}
    A normal state $\psi$ on a von Neumann algebra $\M$ is \emph{universally catalytic} if 
    \begin{equation}\label{eq:universal catalyst}
        \psi\ox\omega_1 \sim\psi\ox\omega_2
    \end{equation}
    for all pairs of normal states $\omega_1,\omega_2$ on a type $\I_\oo$ factor.
\end{definition}

We will see later that universally catalytic states satisfy \eqref{eq:universal catalyst} for all pairs of normal states on an arbitrary AFD factor, even though the definition only asks for type $\I_\oo$ factors.
We quantify how close a given normal positive linear functional $\psi\in\M_*^+$ is to being universally catalytic by considering the worst-case error:
\begin{equation}\label{eq:def kappa}
    \kappa(\psi):= \adjustlimits \sup_{\omega_1,\omega_2} \inf_{u} \norm{\psi\ox \omega_1 - u(\psi\ox\omega_2)u^*},
\end{equation}
where the optimizations are over normal states on a type $\I_\oo$ factor $\B(\K)$ and unitaries in $\M\barox\B(\K)$, respectively.

The main theorem of this section states that $\kappa(\psi)$ measures the deviation of the Haagerup-Størmer spectral functional $\hat\psi$ from being invariant under the flow of weights:

\begin{theorem}\label{thm:kappa}
    Let $\psi$ be a normal positive linear functional on a von Neumann algebra $\M$. Then
    \begin{equation}
        \kappa(\psi) = \sup_{t>0} \,\norm{\hat\psi-\hat\psi\circ\theta_t}.
    \end{equation}
\end{theorem}

We briefly sketch the proof, which will be given in detail in \cref{sec:proof of thm:kappa} below:
\Cref{thm:distance spectral states}, lets us write the RHS of \eqref{eq:def kappa} as $\sup_{\omega_1,\omega_2} \norm{(\psi\ox\omega_1)^\wedge-(\psi\ox\omega_2)^\wedge}$.
The flow of weights of $\M\barox \B(\K)$ can be identified with the flow of weights of $\M$ itself (see \cref{prop:semifinite amp fow}).
We will show that under this identification $(\psi\ox\omega)^\wedge$ becomes a smeared out version of $\hat\psi$ over its orbit:
\begin{equation}\label{eq:smearing}
    (\psi\ox\omega)^\wedge = \sum \lambda_i\, \hat\psi\circ\theta_{\log \lambda_i},
\end{equation}
where $(\lambda_i)$ are the eigenvalues of the density operator $\rho$ associated with $\omega$.
This readily implies that $\kappa(\psi)\le \sup_{t>0}\,\norm{\hat\psi-\hat\psi\circ\theta_t}$:
\begin{equation*}
    \norm{(\psi\ox\omega_1)^\wedge-(\psi\ox\omega_2)^\wedge} = \Big\|\sum \lambda_i\lambda_j (\hat\psi\circ\theta_{\log\lambda_i}-\hat\psi\circ\theta_{\log\lambda_j})\Big\|
    \le \sup_{t>0}\  \norm{\hat\psi-\hat\psi\circ\theta_t}.
\end{equation*}
For the converse, consider states of the form $\omega_1={n}^{-1} \tr(p(\placeholder))$ and $\omega_2=m^{-1}\tr(q(\placeholder))$ for finite-dimensional projection $p,q$ on $\K$ of dimensions $n$ and $m$, respectively. Then 
\begin{equation*}
    \norm{(\psi\ox\omega_1)^\wedge-(\psi\ox\omega_2)^\wedge} = \norm{\hat\psi\circ\theta_{\log n}-\hat\phi\circ \theta_{\log m}} = \norm{\hat\psi -\hat\psi \circ\theta_{\log t}}, \qquad t=\frac{m}{n}.
\end{equation*}
Since $\log \QQ^+$ is dense in $\RR$, this implies the converse inequality $\kappa(\psi)\ge \sup_{t>0}\, \norm{\hat\psi-\hat\psi\circ\theta_{\log t}}$.

In particular, \cref{thm:kappa} tells us that 
\begin{center}\it
    $\psi$ is universally catalytic if and only if $\hat\psi$ is invariant under the flow of weights.
\end{center}
This generalizes \cite[Prop.~3.2]{haagerup_classification_2007} stating that a $\hat\psi$ is invariant if and only if $\psi\ox \bra0 \placeholder\ket0 \sim \psi \ox \frac1n \tr_n$, where $\tr_n$ denotes the standard trace on $M_n(\CC)$.
It is clear that, if $\psi$ is a universally catalytic state on a von Neumann $\M$ and $\phi\sim\psi$, then $\phi$ is also universally catalytic.
If $\M$ is a factor, the converse holds as well: 

\begin{corollary}\label{cor:catalysts are unitarily equivalent}
    Let $\M$ be a factor and let $\psi,\phi$ be normal states on $\M$.
    If $\psi,\phi$ are universally catalytic, then $\psi\sim\phi$.
\end{corollary}

\Cref{cor:catalysts are unitarily equivalent} follows from the fact that an ergodic flow admits at most one invariant normal state. 
Indeed, this implies that all universally catalytic states have the same spectral state, which implies the claim by \cref{thm:distance spectral states}.

Next, we use optimize $\kappa$ over the normal state space in order to obtain algebraic invariants:
\begin{equation}\label{eq:kappa max and min}
    \kappa_{\max}(\M) = \sup_{\psi\in \states_*(\M)}\, \kappa(\psi),
    \qquad \kappa_{\min}(\M) = \inf_{\psi\in\states_*(\M)}\, \kappa(\psi).
\end{equation}
In \cref{sec:ergodic flows}, we completely determine these algebraic invariants by examining the corresponding invariants for ergodic flows on abelian von Neumann algebras.
The results are summarized in the following:

\begin{theorem}\label{thm:value of kappa min}
    Let $\M$ be a von Neumann algebra. 
    Then $\kappa_{\min}(\M)$ is either $0$ or $2$.
    The following are equivalent:
    \begin{enumerate}[(a)]
        \item $\M$ admits a universally catalytic state;
        \item $\kappa_{\min}(\M)=0$;
        \item the flow of weights of $\M$ admits an invariant normal state.
    \end{enumerate}
    In particular, $\kappa_{\min}(\M)=2$ if $\M$ is semifinite and $\kappa_{\min}(\M)=0$ if $\M$ is a factor of type $\III_\lambda$, $0<\lambda\le1$.
\end{theorem}

In \cite{haagerup_classification_2007}, the class of AFD factors whose flow of weights admits normal invariant states is characterized as the class of AFD factors with the property that they admit a type $\III_1$ subfactor with conditional expectation.

We see that universally catalytic states only exist on factors of type $\III$.
Apart from some type $\III_0$ factors, all type $\III$ factors admit universally catalytic states.

\begin{theorem}\label{thm:value of kappa max}
    Let $\M$ be a factor. Then
    \begin{equation}
        \kappa_{\max}(\M) = \diam \nstates(\M)/_\sim
    \end{equation}
    is exactly the diameter of the state space modulo unitaries (see \cref{eq:state space diam}).
    Hence, if $\M$ has type $\III$, then it has subtype $\III_\lambda$ with 
    \begin{equation}
        \lambda = \bigg(\frac{2-\kappa_{\max}(\M)}{2+\kappa_{\max}(\M)} \bigg)^2.
    \end{equation}
\end{theorem}

We see that in factors of type $\III_\lambda$ with $\lambda\approx 1$, we have $\kappa(\psi)\approx0$ for all normal states $\psi$.
A special role is played by type $\III_1$ factors, which are precisely those factors with the property that all states are universally catalytic.
In fact, this holds for algebras:

\begin{proposition}\label{prop:universal-catalyst}
    Let $\M$ be a von Neumann algebra. Then $\M$ is of type $\III_1$, i.e., a direct integral of type $\III_1$ factors, if and only if every normal state is universally catalytic.
\end{proposition}

\subsection{Haagerup-Størmer spectral states of product states}\label{sec:proof of thm:kappa}

\Cref{prop:semifinite amp fow} shows that we may identify the flow of weights of $\M\barox\N$ with the flow of weights of $\M$ itself.
Note, however, that this identification depends on the choice of a \nsf trace $\tau_0$ on $\N$.

\begin{proposition}\label{prop:convolution formula}
    Let $\M$ be a von Neumann algebra and let $\N$ be a semifinite factor with \nsf trace $\tau_0$.
    Then \eqref{eq:semifinite amp fow} identifies the flow of weights of $\M$ and $\M\barox\N$.
    This identification has the property that
    \begin{equation}\label{eq:convolution formula}
        (\psi\ox \omega)^\wedge 
        = \int_0^\oo \lambda_\omega(t)\,\hat\psi\circ\theta_{\log\lambda_\omega(t)}\,dt,
    \end{equation}
    for  $\psi\in\M_*^+$ and $\omega\in\N_*^+$, where $\lambda_\omega(t)$ is the spectral scale of $\omega$ relative to $\tau_0$ (see \cref{sec:spectral states}).
\end{proposition}

\begin{proof}
    We recall that \cref{prop:semifinite amp fow} identifies the core of $\M\barox\N$ as $(\Tilde\M\ox\N,\tilde\theta\ox\id,\tau\ox\tau_0)$, where $(\Tilde\M,\tilde\theta,\tau)$ is a core of $\M$.
    Let $\rho_\omega = d\omega/d\tau_0\in L^1(\N,\tau_0)$ be the density operator $\omega$.
    Notice that the dual weight of $\psi\ox\omega$ is simply $\tilde\psi\ox\omega$. Thus,
    \begin{equation*}
        h_{\psi\ox\omega} = \frac{d(\psi\ox\omega)^\sim}{d(\tau\ox\tau_0)} 
        = \frac{d\tilde\psi}{d\tau}\ox \frac{d\omega}{d\tau_0}
        = h_\psi \ox \rho_\omega.
    \end{equation*}
    Using the spectral measure decomposition $\rho_\omega=\int_{\RR^+} t \,dp_\omega(t)$, we obtain
    \begin{equation*}
        h_\psi \ox \rho_\omega
        = \int_{\RR^+} t \,h_\psi \ox dp_\omega(t)
        =\int_{\RR^+} h_{t\psi} \ox dp_\omega(t).
    \end{equation*}
    Thus, the spectral projection $e_{\psi\ox\omega}=1_{(1,\oo)}(h_{\psi\ox\omega})$ of $h_{\psi\ox\omega}$ is 
    \begin{equation*}
        e_{\psi\ox\omega} 
        =\int_{\RR^+} 1_{(1,\oo)}(h_{t\psi}) \ox dp_\omega(t) 
        = \int_{\RR^+} e_{t\psi} \ox dp_\omega(t).
    \end{equation*}
    This and \eqref{eq:flow of the spectral state} imply
    \begin{align*}
        (\psi\ox\omega)^\wedge
        = (\tau\ox\tau_0) \big((\placeholder) e_{\psi\ox\omega}\big)
        &= \int_{\RR^+} \tau((\placeholder)e_{t\psi})\,\tau_0(dp_\omega(t))\\
        &= \int_{\RR^+} \Hat{t\psi} \,\tau_0(dp_\omega(t))
        = \int_{\RR^+} t\,\hat\psi\circ\theta_{\log t} \,\tau_0(dp_\omega(t)).
    \end{align*}
    By \cref{eq:integration with spectral scale}, the measure $\tau_0(dp_\omega(t))$ is the push-forward measure (see \cite[Ex.~1.4.38]{tao2011introduction}) $\mu$ of the Lebesgue measure along the measurable function $\lambda_\omega$.
    Indeed, $\mu(A)=\abs{\lambda_\omega^{-1}(A)} = \int_0^\oo 1_A(\lambda_\omega(t))\,dt = \tau_0(1_A(\rho_\omega))$ for measurable $A\subset\RR^+$.
    This and \eqref{eq:flow of the spectral state}, imply  
    \begin{equation*}
        (\psi\ox\omega)^\wedge = \int_{\RR^+} \Hat{\lambda_\omega(t)\psi} \,dt 
        = \int_{\RR^+} \lambda_\omega(t)\,\hat\psi \circ\theta_{\log \lambda_\omega(t)}\,dt.
    \end{equation*}
\end{proof}

\Cref{eq:convolution formula} can be expressed as a convolution of $\hat\psi$ by a measure $\nu_\omega$ on the multiplicative group $(0,\oo)$ with respect to the action $\theta_{\log}$ of $(0,\oo)$ on the flow of weights.
Indeed, the proof shows that we may write \eqref{eq:convolution formula} as 
\begin{equation}\label{eq:actual convolution}
    (\psi\ox\omega)^\wedge = \hat\psi *_\theta \nu_\omega := \int_{(0,\oo)} \hat\psi \circ\theta_{\log t}\, d\nu_\omega(t),
\end{equation}
where $\nu_\omega$ is the measure defined by $d\nu_\omega(t)= t \tau_0(dp_\omega(t))$ with $p_\omega$ denoting the spectral measure of $\rho_\omega$.
Note that $\nu_\omega$ is a positive measure with normalization $\nu_\omega((0,\oo))=\omega(1)$.
We can rewrite it in the concise form
\begin{equation*}
    \nu_\omega(A) = \tau_0( \rho_\omega|_{A}),
\end{equation*}
where $\rho_\omega|_A$ is the restriction of $\rho_\omega$ to the part of the spectrum contained in the Borel set $A\subset (0,\oo)$, i.e.,  $\rho_\omega|_A= \rho_\omega\, 1_A(\rho_\omega)$.

In the case that $\N$ is a type $\I$ factor with $\tau_0$ the standard trace, \eqref{eq:convolution formula} reduces to \eqref{eq:smearing}.
For states with a flat spectrum, we get:

\begin{corollary}\label{cor:convolution formula projection}
    In the setting of \cref{prop:convolution formula}, let $\omega= \tau_0(p)^{-1}\tau_0(p(\placeholder))$ for a finite projection.
    Then 
    \begin{equation}\label{eq:convolution formula for projection state}
        (\psi\ox\omega)^\wedge = \hat\psi\circ\theta_{\log \tau_0(p)}.
    \end{equation}
\end{corollary}

In particular, we have $(\psi\ox\omega)^\wedge=\hat\psi$ if the projection $p$ has trace one.
This is, in particular, the case if $\omega$ is a pure state on a type $\I$ factor with the standard trace.

\begin{lemma}\label{lem:kappa upper bound}
    Let $\M$ be a von Neumann algebra and let $\N$ be a semifinite factor with \nsf trace $\tau_0$. 
    We identify the flow of weights of $\M\barox\N$ with that of $\M$ as in \cref{prop:convolution formula}.
    If $\omega_1,\omega_2$ are normal states on $\N$, then
    \begin{equation}
        \norm{(\psi\ox\omega_1)^\wedge-(\psi\ox\omega_2)^\wedge} \le \sup_{t>0} \ \norm{\hat\psi-\hat\psi\circ\theta_t}.
    \end{equation}
\end{lemma}
\begin{proof}
    Let $\lambda_i(t)$ denote the spectral scales of $\omega_i$, $i=1,2$, relative to $\tau_0$.
    Then
    \begin{align*}
        \norm{(\psi\ox\omega_1)^\wedge-(\psi\ox\omega_2)^\wedge}
        &= \bigg\|  \int_0^\oo \lambda_1(t)\,\hat\psi\circ \theta_{\log\lambda_1(t)}\,dt-\int_0^\oo \lambda_2(s)\,\hat\psi\circ\theta_{\log\lambda_2(s)}\,ds\bigg\| \\
        &= \bigg\| \int_0^\oo\! \int_0^\oo \lambda_1(t)\lambda_2(s)\,\big(\hat\psi\circ \theta_{\log\lambda_1(t)}-\hat\psi\circ\theta_{\log\lambda_2(s)}\big)\,dt\,ds\bigg\| \\
        &\le \int_0^\oo\!\int_0^\oo \lambda_1(t)\lambda_2(s) \,\big\|\hat\psi\circ\theta_{\log\lambda_1(t)}-\hat\psi\circ\theta_{\log\lambda_2(s)}\big\| \,dt\,ds\\
        &\le  \sup_{t,s>0} \norm{\hat\psi\circ\theta_{\log\lambda_1(t)}-\hat\psi\circ\theta_{\log \lambda_2(s)}}\ \!\int_0^\oo \lambda_1(t)\,dt \int_0^\oo\lambda_2(s)\,ds\\
        &\le \sup_{t>0} \ \norm{\hat\psi-\hat\psi\circ\theta_{\log t}}.\qedhere
    \end{align*}
\end{proof}

\begin{lemma}\label{lem:kappa lower bound}
    Let $\M$ be a von Neumann algebra and let $\N$ be a semifinite factor with \nsf trace $\tau_0$. 
    We identify the flow of weights of $\M\barox\N$ with that of $\M$ as in \cref{prop:convolution formula}.
    If $\N$ is not a finite type $\I$ factor, then
    \begin{equation}
        \sup_{\omega_1,\omega_2} \norm{(\psi\ox\omega_1)^\wedge - (\psi\ox\omega_2)^\wedge} \ge \ \sup_{t>0} \ \norm{\hat\psi-\hat\psi\circ\theta_{t}},
    \end{equation}
    where the optimization is over pairs of normal states on $\N$.
\end{lemma}
\begin{proof}
    Consider finite projections $p$ and $q$ in $\N$ and let $n=\tau_0(p)$ and $m=\tau_0(q)$.
    Consider the states $\omega_1=n^{-1}\tau(p(\placeholder))$ and $\omega_2 = m^{-1}\tau(q(\placeholder))$.
    Then, by \cref{cor:convolution formula projection},
    \begin{align*}
        \norm{(\psi\ox\omega_1)^\wedge-(\psi\ox\omega_2)^\wedge} 
        = \norm{\hat\psi\circ\theta_{\log n} -\hat\psi\circ\theta_{\log m}} 
        = \norm{\hat\psi-\hat\psi\circ\theta_{\log(n/m)}}.
    \end{align*}
    The possible values of $n/m$ are exactly the ratios of numbers that occur as the dimensions of finite projections.
    If $\N$ is a type $\II$ factor, then the possible values of $n/m$ are $\RR>0$ (see \cref{sec:semifinite vNas}), so that the claim follows.
    If $\N$ is a type $\I_\oo$ factor, then the possible values $n/m$ are the positive rational numbers (this is true even if $\tau_0$ is not the standard trace).
    If $t>0$ is irrational, we may pick a sequence of rational numbers $r_n\to e^t$. 
    It then follows that $\norm{\hat\psi-\hat\psi\circ\theta_t} = \lim_n \norm{\hat\psi-\hat\psi\circ\theta_{\log r_n}} \le \sup_{\omega_1,\omega_2}\ \norm{(\psi\ox\omega_1)^\wedge-(\psi\ox\omega_2)^\wedge}.$
\end{proof}

As a consequence of \cref{prop:convolution formula}, we deduce the following generalization of \cref{thm:kappa}:

\begin{theorem}\label{thm:kappa semifinite}
    Let $\psi$ be a normal positive linear functional on a von Neumann algebra $\M$, and let $\N$ be a semifinite factor not of finite type $\I$.
    Then
    \begin{equation}\label{eq:kappa semifinite}
        \adjustlimits \sup_{\omega_1,\omega_2}\inf_u\ \norm{\psi\ox\omega_1-u(\psi\ox\omega_2)u^*} = \sup_{t>0}\ \norm{\hat\psi-\hat\psi\circ\theta_t},
    \end{equation}
    where the first supremum is over normal states on $\N$ and the infimum is over unitaries in $\M\barox\N$.
\end{theorem}

\begin{proof}
    We use \cref{thm:distance spectral states} to write the LHS of \eqref{eq:kappa semifinite} in terms of the flow of weights of $\M\barox\N$:
    \begin{equation}\label{eq:kappa via spectral states}
        \adjustlimits \sup_{\omega_1,\omega_2}\inf_u\ \norm{\psi\ox\omega_1-u(\psi\ox\omega_2)u^*} = \sup_{\omega_1,\omega_2}\ \norm{(\psi\ox\omega_1)^\wedge-(\psi\ox\omega_2)^\wedge},
    \end{equation}
    Therefore, \cref{lem:kappa upper bound} shows the inequality $\le$, and \cref{lem:kappa lower bound} shows the inequality $\ge$.
\end{proof}

Before passing on, we note another consequence of \cref{prop:convolution formula} that will be helpful later on.

\begin{lemma}\label{lem:distance of unitary orbits and amplification}
    Let $\M$ be a von Neumann algebra with normal positive linear functionals $\psi,\phi\in \M_*^+$.
    Let $\N$ be a semifinite factor with \nsf trace $\tau_0$ and let $0\ne p\in\N$ be a finite projection.
    Consider the state $\omega = \tau_0(p)^{-1}\tau_0(p(\placeholder))$.
    Then
    \begin{equation}
        \inf_{u\in \U(\M)} \norm{u \phi u^*-\psi} = \inf_{u \in \U(\M\barox\N)} \norm{u(\psi\ox\omega)u^* - \phi\ox\omega}
    \end{equation}
\end{lemma}
\begin{proof}
    We identify the flow of weights of $\M\barox\N$ with the flow of weights of $\M$ based on \eqref{eq:semifinite amp fow}.
    Using \cref{thm:distance spectral states}, we may rewrite the LHS as $\norm{\hat\psi-\hat\psi}$ and the RHS as $\norm{(\psi\ox\omega)^\wedge-(\phi\ox\omega)^\wedge}$.
    The claim now follows from \cref{cor:convolution formula projection}:
    \begin{equation*}
        \norm{(\psi\ox\omega)^\wedge-(\phi\ox\omega)^\wedge}
        =\norm{(\hat\psi-\hat\phi)\circ \theta_{\log \tau_0(p)}}
        =\norm{\hat\psi-\hat\psi}. \qedhere
    \end{equation*}
\end{proof}

\subsection{Ergodic flows on abelian von Neumann algebras}\label{sec:ergodic flows}

In this subsection, we consider ergodic flows on abelian von Neumann algebras.
Based on \cref{thm:distance spectral states,thm:kappa}, the results that we derive here have direct consequences for catalytic states on von Neumann algebras.

If $\theta:\RR\acts\A$ is a point-$\sigma$-weakly continuous flow on a general abelian von Neumann algebra $\A$, we define the \emph{$\theta$-regular state space} $\states_*^\theta(\A)$ as in \eqref{eq:theta-regular state space}:
\begin{equation*}
    \states_*^\theta(\A):= \big\{\chi\in \nstates(\A) : \chi\circ\theta_t\ge e^{-t}\chi,\ t>0 \big\}.
\end{equation*}
Clearly, this is a convex space.
Its diameter $\diam \states_*^\theta(\A)$ is defined with respect to the norm distance. 
We measure how much normal positive linear functionals on $\A$ deviate under the flow $\theta$ with the lower semicontinuous function
\begin{equation}\label{eq:kappa on the fow}
    \kappa_\theta:\A_*^+\to \RR^+, \qquad 
    \kappa_\theta(\chi) = \sup_{t>0} \,\norm{\chi-\chi\circ\theta_t}.
\end{equation}
Note that $\kappa_\theta(\chi)=0$ if and only if $\chi$ is $\theta$-invariant, and that every $\theta$-invariant state is $\theta$-regular. 
If $\theta:\RR\acts\A$ is the flow of weights of a von Neumann algebra $\M$, then \cref{thm:kappa} implies
\begin{equation}\label{eq:kappas}
    \kappa_\theta(\hat\psi)=\kappa(\psi), \qquad\psi\in\M_*^+.
\end{equation}
The main result of this subsection is the following:

\begin{theorem}\label{thm:state space diam via flow}
    Let $\theta:\RR\acts\A$ be an ergodic point-$\sigma$-weakly continuous flow on an abelian von Neumann algebra.
    Let $T\in \bar \RR^+$ be the minimal period of $\theta$ ($T=\oo$ if $\theta$ is aperiodic), then
    \begin{equation}\label{eq:state space diam via flow}
        \diam \states_*^\theta(\A)= \sup_{\chi\in \states_*^\theta(\A)} \kappa_\theta(\chi)= 2\frac{1-e^{-T/2}}{1+e^{-T/2}}.
    \end{equation}
    If $T<\oo$, then $\A$ admits both a $\theta$-invariant normal state, and it admits a $\theta$-regular normal state that is as non-$\theta$-invariant as possible in the sense that $\kappa_\theta(\chi)=\diam\states_*^\theta(\A)$.
    If $T=\oo$, exactly one of the following is true:
    \begin{itemize}
        \item There exists an invariant normal state on $\A$ and $\kappa_\theta(\chi)<2$ for  faithful $\chi\in\nstates(\A)$.
        \item There exists no invariant normal state on $\A$ and $\kappa_\theta(\chi)=2$ for all $\chi\in\nstates(\A)$.
    \end{itemize}
\end{theorem}

We begin by following \cite[Sec.~6]{haagerup_equivalence_1990} in noting that
\begin{equation}
    \A_\theta = \{a\in \A : \lim_{t\to0}\,\norm{a-\theta_t(a)}=0\}
\end{equation}
is a $\sigma$-weakly dense $C^*$-subalgebra.
Thus, $\nstates(\A)\subset \states(\A_\theta)$ is weak*-dense, where $\states(\A_\theta)$ denotes the state space of the $C^*$-algebra $\A_\theta$.
For $\omega\in\states(\A_\theta)$, we set
\begin{equation}
    \chi_\omega = \int_0^\oo e^{-t} \,\omega\circ\theta_{-t}\,dt \in \states(\A_\theta).
\end{equation}
Note that $\chi_\omega\in\nstates(\A)$ if $\omega\in\nstates(\A)$. 

\begin{lemma}[{\cite[Sec.~6]{haagerup_equivalence_1990}}]\label{lem:HS sec6}
    \begin{enumerate}[(i)]
        \item 
        $\omega \mapsto \chi_\omega$ is an affine bijection of the convex sets $\states(\A_\theta)$ and  $\states^\theta(\A_\theta) := \{\chi\in\states(\A_\theta) : \chi\circ\theta_t\ge e^{-t}\chi,\ t>0\}$.
        \item
        $\A_*$ is linearly spanned by $\states_*^\theta(\A)$.
    \end{enumerate}
\end{lemma}

We briefly sketch the easy part of the proof. Namely, we check that $\chi_\omega\in\states^\theta(\A_\theta)$ for all $\omega\in\states(\A_\theta)$:
If $t>0$, then
\begin{equation*}
    e^t \,\chi_\omega \circ\theta_t =\int_0^\oo\! e^{t-s} \omega\circ\theta_{t-s} \,ds 
    =  \int_{-t}^\oo\! e^{-s} \omega\circ\theta_{-s} \,ds
    \ge \int_0^\oo\! e^{-s}\omega\circ\theta_{-s} \,ds= \chi_\omega.
\end{equation*}

\begin{lemma}\label{lem:sup kappa C*}
    $\sup_{\chi\in\states_*^\theta(\A)} \kappa_\theta(\chi) = \sup_{\chi\in \states^\theta(\A_\theta)} \kappa_\theta(\chi)$.
\end{lemma}

\begin{proof}
    The inequality "$\le$" is trivial.
    Given $\eps>0$, let $t>0$, $\chi\in\states^\theta(\A_\theta)$ be such that $\norm{\chi-\chi\circ\theta_t}$ is $\eps$-close to the supremum on the RHS. 
    By \cref{lem:HS sec6}, $\chi=\chi_\omega$ for a unique $\omega\in\states(\A_\theta)$.
    Let $(\omega_\alpha)$ be a net of normal states on $\A$ that weak*-converges to $\omega$.
    Then $\chi_{\omega_\alpha}\in \states^\theta_*(\A)$. This and weak*-lower semicontinuity of the norm on $\A_\theta^*$ imply
    \begin{equation*}
        \RHS -\eps \le \norm{\chi_\omega-\chi_\omega\circ\theta_t} \le \liminf_\alpha \norm{\chi_{\omega_\alpha}-\chi_{\omega_\alpha}\circ\theta_t} \le \LHS.
        \qedhere
    \end{equation*}
\end{proof}

\Cref{lem:sup kappa C*} reduces much of the analysis that we have to do to the point-norm continuous flow on $\A_\theta$.
In general, if $\theta:\RR\acts\B$ is a point-norm continuous action on a $C^*$-algebra $\B$, then the map $\B^*\times\RR\ni (\omega,t)\mapsto \omega\circ\theta_t\in\B^*$ is jointly continuous when the dual space is equipped with the weak*-topology.%
\footnote{
This follows from the triangle inequality: If $(\omega_\alpha,t_\alpha)\to(\omega,t)$ then $\abs{(\omega\circ\theta_t- \omega_\alpha\circ\theta_{t_\alpha})(a)}
    \le \abs{(\omega-\omega_\alpha)(\theta_t(a))} + \norm\omega \norm{(\theta_t-\theta_{t_\alpha})(a)}\to 0$.

}


\begin{lemma}\label{lem:per is lsc}
    Let $\B$ be a $C^*$-algebra and let $\theta:\RR\acts\B$ be point-norm continuous $\RR$-action.
    For a state $\omega$ on $\B$, denote by 
    \begin{equation}
        \per_\theta(\omega) = \inf\{t>0 : \omega\circ\theta_t =\omega\} \qquad (\text{with}\ \inf \emptyset=\oo)
    \end{equation}
    its minimal period under $\theta$.
    Then $\per_\theta:\states(\B)\to\bar\RR^+$ is weak*-lower semicontinuous.
\end{lemma}
\begin{proof}
    We show that the sublevel sets of $\per_\theta$ are weak*-closed.
    Let $T>0$ and let $(\omega_\alpha)$ be a weak*-Cauchy net of states on $\B$ with minimal period $\le T$.
    For each $\alpha$, let $t_\alpha$ be a period for $\omega_\alpha$ in the interval $(T/2,T]$.
    To see that these exist, take any period $s_\alpha\in (0,T]$ of $\omega_\alpha$, which exist since $\per_\theta(\omega_\alpha)\le T$, and let $t_\alpha = 2^{n_\alpha}s_\alpha$, where $n_\alpha$ is the smallest integer $\ge\log_2(T/s_\alpha)-1$.
    We can pass to a subnet of $(\omega_\alpha)$ such that $t = \lim_\alpha t_\alpha$ exists since $[T/2,T]$ is compact.
    We claim that $t$ is a period for $\omega=\lim_\alpha\omega_\alpha$. 
    Indeed, joint continuity of $(s,\nu)\mapsto \nu\circ\theta_s$ implies
    \begin{equation}
        \omega\circ \theta_t = \lim_\alpha \omega_\alpha\circ \theta_{t_\alpha} = \lim_\alpha \omega_\alpha = \omega.
    \end{equation}
    Since $t\le T$ is a period for $\omega$, the minimal period of $\omega$ is bounded by $T$.
\end{proof}

\begin{lemma}\label{lem:aperiodic pure state}
    Let $\theta:\RR\acts\B$ be a point-norm continuous flow on an abelian $C^*$-algebra $\B$.
    Let $\omega$ be a pure state that is aperiodic under $\theta$.
    Then $\chi=\int_0^\oo e^{-t}\omega\circ\theta_{-t}dt \in\states^\theta(\B)$ satisfies
    \begin{equation}
        \lim_{t\to\oo}\,\norm{\chi-\chi\circ\theta_t} =2.
    \end{equation}    
\end{lemma}
\begin{proof}
    We follow ideas from \cite[Secs.~6 \& 10]{haagerup_equivalence_1990}.
    If $t>0$, we have
    \begin{align*}
        \norm{\chi-\chi\circ\theta_t} 
        &= \norm{\int_0^{\oo}e^{-s}\omega\circ\theta_{-s} ds-\int_{-t}^\oo e^{-s-t}\omega\circ\theta_{-s} ds}\\
        &= \norm{\int_0^\oo (1-e^{-t})e^{-s} \omega\circ\theta_{-s} ds - \int_{-t}^0 e^{-s-t} \omega\circ\theta_{-s} ds}.
    \intertext{The assumption that $\omega$ is aperiodic and the fact that distinct pure states on abelian $C^*$-algebras are orthogonal imply that this is the norm difference of a pair of orthogonal positive linear functionals on $\B$. Therefore, we can continue the calculation as}
        &= \norm{\int_0^\oo (1-e^{-t})e^{-s} \omega\circ\theta_{-s} ds} + \norm{\int_{-t}^0 e^{-s-t} \omega\circ\theta_{-s} ds}\\
        &= (1-e^{-t})\, \norm{\int_{-t}^0 e^{-s} \omega\circ\theta_{-s} ds} + e^{-t}\,\norm{\int_{-t}^0 e^{-s}\omega\circ\theta_{-s} ds}\\
        &= 1-e^{-t} + e^{-t}\int_0^t e^{s} ds = 1-e^{-t} +e^{-t} (e^t-1) =2 - 2e^{-t} \to 2.\qedhere
    \end{align*}
\end{proof}

Next, we need a dichotomy for $\RR$-actions on abelian von Neumann algebras based on \cite[Sec.~4.5]{arano_ergodic_2021}.
The next Proposition and its proof were communicated to the author by Amine Marrakchi:

\begin{proposition}\label{prop:amine}
    Let $\theta:\RR\acts\A$ be a (not necessarily ergodic) $\RR$-action on an abelian von Neumann algebra.
    The following are equivalent:
    \begin{enumerate}[(a)]
        \item\label{it:amine1} there exists no $\theta$-invariant normal state on $\A$;
        \item\label{it:amine2} there is a faithful normal state $\chi$ on $\A$ such that $\kappa_\theta(\chi)=2$;
        \item\label{it:amine3} for every normal state $\chi$ on $\A$ it holds that $\kappa_\theta(\chi)=2$.
    \end{enumerate}
\end{proposition}
\begin{proof}
    Let $(X,\mu)$ be a standard measure space such that $\A=L^\oo(X,\mu)$.
    Let $\sigma:\RR\acts X$ be the $\mu$-nonsingular action such that $\theta_t(f)(x) = f( \sigma_t(x))$.
    We denote by $\sigma_t(\mu)$ the push-foward measure, i.e., $\sigma_t(\mu)(A) = \mu(\sigma_t^{-1}(A))=\mu(\sigma_{-t}(A))$, $t\in\RR$.
    While $\sigma_t$ does not preserve $\mu$, it does preserve the measure class of $\mu$, i.e.,  $\sigma_t(\mu)\sim \mu$, $t\in\RR$.
    We define the positive measurable function $m_t = \sqrt{d\sigma_t(\mu)/d\mu}: X \to \RR^+$.
    The operator $U_tf(x) = m_t \cdot (f\circ \sigma_{-t})$ is a unitary operator on $L^2(X,\mu)$, which implements the action $\theta:\RR\acts\A$ in the sense that
    \begin{equation}
        (U_t^* f U_t g) = \theta_t(f)\cdot g, \qquad t\in\RR
    \end{equation}
    (this is the so-called Koopmann representation of the dynamical system $\sigma:\RR\acts X$).
    For every $\chi\in\A_*^+$, there is a unique positive $\mu$-absolutely continuous measure $\nu_\chi$ on $X$ such that $\chi(f) = \int f d\nu_\chi$, $f\in L^\oo(X,\mu)=\A$.
    Note that the $1$-norm distance of measures is the norm distance on $\A_*$, i.e., $\norm{\chi-\chi'} =\norm{\nu_\chi-\nu_{\chi'}}_1$, and that the measure corresponding to $\chi\circ \theta_t$ is the push-forward measure $\sigma_t(\nu_\chi)$.\footnote{The $1$-norm is defined on signed measures as is twice the total variation norm, i.e., $\norm\mu_1 = 2 \sup_A \abs{\mu(A)}$, where the optimization is over measurable subsets $A\subset X$.}
    Thus, we have $\frac12 \kappa_\theta(\chi) = \sup_{t\in\RR} \norm{\nu_\chi - \sigma_t(\nu_\chi)}$.
    We set $g_\chi = (d\nu_\chi/d\mu)^{1/2} \in L^2(X,\mu)^+$, so that $\nu_\chi(f) = \ip{g_\chi}{fg_\chi}_{L^2}$.
    It follows that $U_tg_\chi = g_{\chi\circ\theta_t}$.
    Then \cref{lem:the-standard-estimate} below implies that 
    \begin{equation}\label{eq:equivalent-orthogonality}
        \inf_{t\in\RR}\, \ip{g_\chi}{U_tg_\chi}_{L^2} = 0  \iff \kappa_\theta(\chi) = \sup_{t\in\RR} \, \norm{\nu_\chi- \sigma_t(\nu_\chi)}_{1} = 2.
    \end{equation}
    Thus, \ref{it:amine1} $\Leftrightarrow$ \ref{it:amine2} follows from \cite[Prop.~4.5]{arano_ergodic_2021}.
    Next, we show \ref{it:amine3} $\Rightarrow$ \ref{it:amine2} based on an argument from \cite[Lem.~4.4]{arano_ergodic_2021}: 
    Let $t_n$ be such that $\ip{g_\chi}{U_{t_n}g_\chi}_{L^2}\to 0$.
    Since $\chi$ is faithful, $g_\chi\in L^2(X,\mu)^+$ has full support.
    For $f_1,f_2\in L^\oo(X,\mu)^+$, set $g_j= f_jg_\chi\in L^2(X,\mu)^+$.
    Then $g_j\le \norm{f_j}g_\chi$ implies $0\le \ip{g_1}{U_{t_n} g_2}_{L^2} \le \norm{f_1}\norm{f_2} \ip{g_\chi}{U_{t_n}g_\chi}\to0$.
    If we decompose non-positive functions $f_1,f_2\in L^\oo(X,\mu)$ into linear combinations of positive functions, it follows that $\ip{g_1}{U_{t_n}g_2}\to 0$ for all $g_1,g_2\in L^\oo(X,\mu)g_\chi$.
    Since $g_\chi$ has full support, the latter is a dense subspace of $L^2(X,\mu)$. Thus, $U_{t_n}\to 0$ in the weak operator topology.
    In particular, we get $\ip{g_{\chi'}}{U_{t_n}g_{\chi'}}\to 0$ for all $\chi'\in\nstates(\A)$.
    The claim follows from \eqref{eq:equivalent-orthogonality}.
\end{proof}


\begin{lemma}\label{lem:the-standard-estimate}
    Let $(X,\mu)$ be a $\sigma$-finite measure space and let $\nu,\nu'$ be $\mu$-absolutely continuous probability measures. Let $g\up\prime = (d\nu\up\prime/d\mu)^{1/2} \in L^2(X,\mu)^+$.
    Then
    \begin{equation}
        1-\delta \le \ip g{g'}_{L^2}\le (1-\delta^2)^{1/2}, \qquad \delta =\frac12\norm{\nu-\nu'}_1.
    \end{equation}
\end{lemma}
\begin{proof}
    The lower bound follows from the standard estimate $\norm{g-g'}_{L^2}^2 \le \norm{\nu-\nu'}_1$: 
       $1-\delta \le 1-\frac12\norm{g-g'}_{L^2}^2 = 1- (1-\ip{g}{g'}_{L^2}) = \ip{g}{g'}_{L^2}$.
    To obtain the upper bound, consider the density operators $\rho\up\prime=\kettbra{g\up\prime}$ on $L^2(X,\mu)$.
    Then $(\nu-\nu')(f) = \tr(\rho-\rho')f$, $f\in L^\oo(X,\mu)$, implies
    \begin{align*}
        \delta=2\norm{\nu-\nu'} 
        \le 2\norm{\kettbra g-\kettbra{g'}}_1 
        = (1-\ip g{g'}^2)^{1/2}.
    \end{align*}
    Rearranging shows the upper bound.
\end{proof}


\begin{proof}[Proof of \cref{thm:state space diam via flow} for periodic flows]
    The case where $T=0$ is trivial.
    Let $0<T<\oo$ be the minimal period of $\theta$. 
    By \cite[Lem.~10.2 \& 10.3]{haagerup_equivalence_1990}, $\diam\states_*^\theta(\A) = \diam\states(\A_\theta) = 2 (1-e^{-T/2})/(1+e^{-T/2})$.
    It is clear that $\kappa_\theta(\chi)\le \diam\states_*^\theta(\A)$ for all $\chi\in\states_*^\theta$.
    It remains to show that there exists a normal state $\chi$ on $\A$ such that $\kappa_\theta(\chi)=2 (1-e^{-T/2})/(1+e^{-T/2})$.
    By \cite[Lem.~10.2]{haagerup_equivalence_1990}, we may assume $\A_\theta = C(\RR/T\ZZ)$ with $\theta_t(f)=f\circ\sigma_t$, $\sigma_t(s)=s-t$ with addition mod $T$.
    Periodicity of the flow implies that there is a unique $\theta$-invariant normal state $\omega$ on $\A$, which is necessarily faithful.\footnote{
    Indeed, take any normal state $\omega_0$ on $\A$ and set $\omega= T^{-1}\int_0^T \omega_0\circ\theta_tdt$.
    Invariance is clear, faithfulness holds because $\supp(\omega)$ is a $\theta$-invariant nonzero projection, hence the identity, and uniqueness holds because $d\omega'/d\omega$ is $\theta$-invariant, hence constant, for every other $\theta$-invariant normal state $\omega'$.}
    It follows that $\omega(f) = T^{-1}\int_0^T f(t)dt$, $f\in C(\RR/T\ZZ)$ because the only $\sigma$-invariant Borel probability measure on $\RR/T\ZZ$ is the uniform distribution.
    Faithfulness of $\omega$ on $\A$ and $\sigma$-weak-density of $\A_\theta\subset \A$ lets us identify $\A$ with $L^\oo(\RR/T\ZZ)$.
    Let $\chi$ be the normal state on $\A=L^\oo(\RR/T\ZZ)$ with
    \begin{equation*}
        \chi(f) = \int_0^T p(s)f(s)\,ds, \qquad p(s) = (1-e^{-T})^{-1} e^{-s}.
    \end{equation*}
    We claim that the state $\chi\circ\theta_t$ has the density $g_t(s) p(s)$, where 
    \begin{equation*}
        g_t = e^{-t}1_{[0,T-t)} + e^{T-t} 1_{[T-t,T)}, \qquad t\in\RR,
    \end{equation*}
    with addition modulo $T$.
    Indeed, 
    \begin{align*}
        \chi( \theta_t(f)) 
        &= (1-e^{-T})^{-1} \bigg( \int_t^T f(s-t) e^{-s}ds + \int_0^t f(s+T-t) e^{-s}ds\bigg)\\
        &= (1-e^{-T})^{-1} \bigg( \int_0^{T-t} f(s) e^{-s-t}ds + \int_{T-t}^T f(s) e^{-s-t+T}ds\bigg)
        = \int_0^T f(s)g_t(s) p(s)\,ds.
    \end{align*}
    Since $g_t \ge e^{-t}$, it follows that $\chi$ is $\theta$-regular.
    Next, we compute the norm distance of $\chi$ and $\chi\circ\theta_{T/2}$:
    \begin{align*}
        \norm{\nu-\nu\circ\theta_{T/2}} &= \norm{p- g_tp}_{L^1}
        = (1-e^{-T})^{-1}\int_0^T \abs{1-g_t(s)} e^{-s}ds \\
        &= (1-e^{-T})^{-1}\bigg(\int_0^{T/2} \abs{1-e^{-T/2}} e^{-s}ds +  \int_{T/2}^T \abs{1-e^{T/2}} e^{-s}ds \bigg)\\
        &= (1-e^{-T})^{-1}\Big( (1-e^{-T/2})(1-e^{-T/2})+ (e^{T/2}-1)(e^{-T/2}-e^{-T}) \Big)\\
        &=2 (1-e^{-T/2})/(1-e^{T/2}).
    \end{align*}
    It follows that
    \begin{equation*}
        2 \frac{1-e^{-T/2}}{1-e^{T/2}} = \diam\states_*^\theta(\A) \ge \kappa_\theta(\chi) \ge \norm{\nu-\nu\circ\theta_{T/2}} =2 \frac{1-e^{-T/2}}{1-e^{T/2}}.
    \end{equation*}
    Thus, we have equality everywhere, which finishes the proof.
\end{proof}

\begin{proof}[Proof of \cref{thm:state space diam via flow} for aperiodic flows]
    Clearly, we have $2\ge \diam\states_*^\theta(\A) \ge \sup_{\chi\in\states_*^\theta(\A)} \kappa(\chi)$.
    We have to show that the RHS equals $2$.

    \emph{Case 1: There exists an invariant normal state.}
    The claim that $\kappa_\theta(\chi)<2$ for all faithful normal states $\chi$ on $\A$ follows from \cref{prop:amine}.
    It remains to show $\sup_{\chi\in\states_*^\theta(\A)} \kappa(\chi)=2$.
    By \cref{lem:sup kappa C*}, we can equivalently optimize over states $\chi\in\states(\A_\theta)$.
    We claim that $\A_\theta$ necessarily admits aperiodic pure states. 
    In this case, our claim is implied by \cref{lem:aperiodic pure state}.
    To see that $\A_\theta$ has aperiodic pure states, we identify $\A_\theta = C(X)$ with $X$ being the set of pure states equipped with the weak*-topology.
    On the set of pure states, the flow acts as $\omega\mapsto\omega\circ\theta_t$, which defines a continuous flow $\sigma:\RR\acts X$.
    This flow is ergodic since the only $\sigma$-invariant continuous functions on $X$ are the constant functions (because $(\A_\theta)^\theta\subset\A^\theta = \CC1$).
    Let $\omega$ be an invariant normal state on $\A$.
    Ergodicity implies that $\omega$ is faithful since $\supp(\omega)$ is a $\theta$-invariant nonzero projection.
    Let $\mu$ be the $\sigma$-invariant Borel measure on $X$ such that $\omega(f)=\int_Xf\,d\mu$, $f\in C(X)=\A_\theta$.
    Since $\omega$ is faithful and $\A_\theta$ is $\sigma$-weakly dense in $\A$, we may identify $\A=L^\oo(X,\mu)$.
    By \cref{lem:per is lsc}, the minimal periodicity $\per_\sigma:X\to \bar\RR^+$, $\per(x)=\inf\{t>0: \sigma_t(x)=x\}$ is lower semicontinuous, hence Borel-measurable.
    Since $\per_\sigma$ is a $\sigma$-invariant function, i.e., $\per_\sigma(\sigma_t(x))=\per_\sigma(x)$ for all $t\in\RR$, ergodicity implies that it is constant $\mu$-almost everywhere, i.e., $\mu$-almost every point $x\in X$ has period $T$ for some fixed $T\in\bar\RR^+$.
    If $T<\oo$, then $\theta$ would be $T$-periodic on $\A$, which contradicts our assumption.
    Therefore, $\mu$-almost every pure state $\omega$ on $\A_\theta$ is aperiodic.
    
    \emph{Case 2: There exists no invariant normal state.}
    \cref{prop:amine} implies that every normal state $\chi$ on $\A$ has the property $\kappa_\theta(\chi)=2$.
    Since $2\ge \diam\states_*^\theta(\A) \ge \kappa(\chi)=2$ for every $\chi\in\states_*^\theta(\A)$, this implies $\diam\states_*^\theta(\A)=2$.
\end{proof}

We now prove \cref{thm:value of kappa min,thm:value of kappa max,prop:universal-catalyst} based on the above:

\begin{proof}[Proof of \cref{thm:value of kappa min}]
    Let $\theta:\RR\acts\A$ be the flow of weights of $\M$.
    \Cref{thm:kappa} implies $\kappa_{\min}(\M) = \inf \kappa_\theta(\chi)$, where the infimum ranges over $\theta$-regular normal states $\chi$ on $\A$.
    Thus, the equivalence between the existence of a $\theta$-invariant normal state and $\kappa_{\min}(\M)=0$ follows from \cref{prop:amine}.
    Since every $\theta$-invariant normal state is $\theta$-regular, it is the spectral state of a state on $\M$, which by \cref{thm:kappa} is universally catalytic.
    This finishes the proof.
\end{proof}

\begin{proof}[Proof of \cref{thm:value of kappa max}]
    Let $\theta:\RR\acts\A$ be the flow of weights of $\M$.
    The statement in \cref{thm:value of kappa max} directly follow from those in \cref{thm:state space diam via flow} because of the following:
    \begin{itemize}
        \item $\psi\mapsto\hat\psi$ is an isometric bijection $\nstates(\M)/_\sim\, \to \states_*^\theta(\A)$ (see \cref{thm:distance spectral states}),
        \item $\kappa(\psi) = \kappa_\theta(\hat\psi)$ and $\psi$ is a universalö catalytic if and only if $\hat\psi$ is $\theta$-invariant (see \cref{thm:kappa}),
        \item the flow of weights is periodic with minimal period $0\le T<\oo$ if and only if $\M$ is a type $\III_\lambda$ factor with $\lambda = e^{-T}>0$. \qedhere
    \end{itemize}
\end{proof}

\begin{proof}[Proof of \cref{prop:universal-catalyst}]
    The fact that the flow of weights of a direct integral is the direct integral of the corresponding flows of weights (see \cref{lem:core direct int}) implies that a von Neumann algebra $\M$ is of type $\III_1$ if and only if its flow of weights $\theta:\RR\acts\A$ is trivial, i.e., $\theta\equiv\id$ (in this case, $\A = Z(\M)$).
    If the flow of weights is trivial, then \cref{thm:kappa} implies $\kappa(\psi)=0$ for all $\psi\in\M_*^+$.
    Conversely, if $\kappa(\psi)=0$ for all $\psi\in\M_*^+$, then \cref{thm:kappa} implies that $\kappa_\theta(\chi)=0$ for all $\chi\in\states_*^\theta(\A)$.
    By \cref{lem:HS sec6}, $\A_*$ is spanned by $\states_*^\theta(\A)$.
    Thus, $\chi\circ\theta_t=\chi$ for all $\chi\in\A_*$, which implies $\theta_t=\id$, for all $t\in\RR$.
\end{proof}

\subsection{Catalysis of state transitions on general factors}\label{sec:catalysis general factors}

\begin{proposition}\label{prop:transitions on AFD factors}
    Let $\psi$ be a normal positive linear functional on a von Neumann algebra $\M$, and let $\N$ be an AFD factor.
    If $\omega_1,\omega_2$ are normal states on $\N$, then
    \begin{equation}
        \inf_{u \in \U(\N)} \ \norm{u(\psi\ox\omega_1)u^* - \psi\ox\omega_2} \le \kappa(\psi).
    \end{equation}
\end{proposition}

We prove \cref{prop:transitions on AFD factors} by reducing the general case to the semifinite case via a martingale Lemma inspired by \cite[Thm.~8.1]{haagerup_equivalence_1990}.
To state it, we need the concept of a conditional expectation:
A conditional expectation onto a von Neumann subalgebra $\N$ of a von Neumann algebra $\M$ is a surjective ucp map $E:\M\to\N$ such that $E^2=E$.

\begin{lemma}\label{lem:martingale}
    Let $\N$ be a factor and let $\N_\alpha$ be an increasing net of subfactors, indexed by a directed set $I$ such that $\bigcup_\alpha \N_\alpha$ is $\sigma$-weakly dense in $\N$.
    Let $E_\alpha$, $\alpha\in I$, be conditional expectations onto $\N_\alpha$, such that $E_\alpha E_\beta = E_\alpha$ for all $\alpha\le\beta$, and such that
    \begin{equation*}
        \lim_\alpha\ \norm{\psi\circ (\id-E_\alpha)} = 0, \qquad \psi\in\M_*.
    \end{equation*}
    If \cref{prop:transitions on AFD factors} holds for all $\N_\alpha$, then it also holds for $\N$.
\end{lemma}
\begin{proof}
    We apply the results of \cite[Sec.~7]{haagerup_equivalence_1990} to the conditional expectations $P_\alpha = \id\ox E_\alpha$ onto the subalgebras $\M\ox\N_\alpha\subset\M\ox\N$.
    In particular, we use that
    \begin{equation}
        \norm{\hat\phi_1 - \hat \phi_2} = \lim_\alpha\, \norm{\hat\phi_{1,\alpha}-\hat\phi_{2,\alpha}}, \qquad \phi_1,\phi_2\in (\M\ox\N)_*^+,
    \end{equation}
    where $\hat \phi_{i,\alpha}$, $i=1,2$, denotes the spectral state of $\phi_{i,\alpha}=\phi_i\restriction\M\ox\N_\alpha$ (note that this is a normal positive linear functional on the flow of weights of $\M\ox\N_\alpha$).
    If we apply this to $\phi_i=\psi\ox\omega_i$, $i=1,2$, for normal states $\omega_1,\omega_2\in \nstates(\N)$ and $\psi\in\M_*^+$, \cref{thm:kappa semifinite} implies
    \begin{align*}
        \norm{(\psi\ox\omega_1)^\wedge-(\psi\ox\omega_2)^\wedge} 
        = \lim_\alpha\ \norm{(\psi\ox\omega_{1,\alpha})^\wedge-(\psi\ox\omega_{2,\alpha})^\wedge} \le \kappa(\psi),
    \end{align*}
    where $\omega_{i,\alpha}=\omega_i \restriction \N_\alpha$, $i=1,2$.
    Together with \cref{thm:distance spectral states}, this implies the claim.
\end{proof}

\begin{proof}[Proof of \cref{prop:transitions on AFD factors}]
    By \cref{thm:kappa semifinite}, the claim holds for semifinite $\N$.
    It remains to show the claim for type $\III$ factors.
    If $\N$ is an AFD factor of type $\III_0$, a net of conditional expectations $E_\alpha$ onto an increasing net of semifinite subfactors $\N_\alpha$ as in \cref{lem:martingale}, exists according to~\cite[Prop.~8.3]{haagerup_equivalence_1990}.
    Thus, the claim follows from \cref{lem:martingale}.
    If $\N$ is an AFD factor of type $\III_\lambda$, $0<\lambda\le1$, then it is an ITPFI factor \cite{elliott1976afd,connes1980_non_itpfi,takesaki3}.
    Thus, we may write $\N$ as $\N = \barox_{n\in\NN} (\B(\H_n),\omega_n)$, where $\dim\H_n<\oo$, and where each $\omega_n$ is a faithful state on $\B(\H_n)$.
    For $\alpha\in\NN$, we set $\N_\alpha = \barox_{n=1}^\alpha \B(\H_n)$ and $E_\alpha = \id \ox \,(\ox_{n>\alpha}\omega_n) : \N \to \N_\alpha$.
    Then $E_\alpha$ is a sequence of conditional expectations that satisfies the properties in \cref{lem:martingale}.
    Since each $\N_\alpha$ is semifinite, \cref{lem:martingale} implies the claim.
\end{proof}

\begin{remark}
    We expect that \cref{prop:transitions on AFD factors} holds in full generality without the assumption that the factor $\N$ is AFD.
    To prove this stronger statement, one needs to relate the spectral states of $\psi\ox\omega_i$, $i=1,2$, to the spectral states of $\psi$ and $\omega_i$, which requires, in particular, an understanding of the flow of weights of $\M\barox\N$ in terms of the flow of weights of $\M$ and $\N$.
    This is fully understood for semifinite factors (see \cref{prop:convolution formula} in \cref{sec:proof of thm:kappa}).
    However, our argument for type $\III$ factors in the proof of \cref{prop:transitions on AFD factors} was a reduction to the semifinite case via a martingale-type argument, which does not generalize to the general non-AFD type $\III$ case.
\end{remark}

\section{Bipartite quantum systems}\label{sec:bipartite systems}

\localtableofcontents

\null

In this section, we initiate our systematic study of bipartite systems and lay the groundwork for our discussion of bipartite entanglement theory.
The material of this section is partly based on \cite{van_luijk_embezzlement_2024,van_luijk_pure_2024}, but also contains new results such as the operational characterization of Haag duality in \cref{sec:haag}.


\subsection{Definition and basic facts}\label{sec:bipartite systems_basics}

Let us begin with a formal definition.

\begin{definition}
    A \emph{bipartite system} is a triple $(\M_A,\M_B,\H)$ of a Hilbert space $\H$ and a pair of commuting von Neumann algebras $\M_{A/B}$ on $\H$ such that
    \begin{equation*}
        \M_A=\M_B'.
    \end{equation*}
\end{definition}

Although we are mostly interested in the case where both $\M_A$ and $\M_B$ are factors, we use the above definition that only assumes Haag duality (the property $\M_A=\M_B'$, see \cref{sec:independent-agents}) because most of our results apply to the non-factorial case.
We refer to \cref{sec:independent-agents} for a discussion of factoriality and Haag duality from an operational point of view.
In \cref{sec:haag}, we prove an operational interpretation of Haag duality in the context of pure state local entanglement theory and local operations.

From a structural point of view, our definition of a bipartite system is equivalent to concretely represented von Neumann algebras $(\M,\H)$ via the identification $(\M,\H) \equiv (\M,\M',\H)$.
Conceptually, however, our definition places the two subsystems on the same hierarchical level, which reflects how we think about physical bipartite systems.
We import some definitions from von Neumann algebras to bipartite systems.

\begin{definition}\label{def:bipartite system}
    A bipartite system $(\M_A,\M_B,\H)$ is:
    \begin{itemize}
        \item \emph{factorial} if $\M_A$ and, hence, $\M_B$, are factors,
        \item \emph{semifinite} if $\M_A$ and, hence, $\M_B$ is a semifinite von Neumann algebra,
        \item of \emph{type $\rm X$} if $\M_A$ and $\M_B$ are of type $\rm X$, where $\rm X = \I$, $\I_n$ ($n\in\NN$), $\II$, $\II_1$, etc.
    \end{itemize}
    If $(\M_A,\M_B,\H)$ is factorial, we define the \emph{minimal type} as the minimum of the types of $\M_A$ and $\M_B$, relative to the naive ordering \eqref{eq:naive type order}.
\end{definition}

If $(\M_A,\M_B,\H)$, $(\N_A,\N_B,\K)$ are bipartite systems, then the tensor product and direct sum
\begin{equation}
    (\M_A\ox\N_A,\M_B\ox\N_B,\H\ox\K),\qquad (\M_A\oplus\N_A,\M_B\oplus\N_B,\H\oplus\K)
\end{equation}
are again bipartite systems.
Moreover, every bipartite system uniquely decomposes into a direct sum
\begin{equation}\label{eq:type decomp of bipartite systems}
    (\M_A,\M_B,\H) = (\M_A^\I,\M_B^\I,\H^\I) \oplus (\M_A^\II,\M_B^\II,\H^\II) \oplus (\M_A^\III,\M_B^\III,\H^\III)
\end{equation}
of bipartite systems of types $\I$, $\II$, and $\III$, respectively.
The disintegration of von Neumann algebras induces a disintegration of bipartite systems: 
Every bipartite system is of the form
\begin{equation}
    (\M_A,\M_B,\H) = \int_{X}^\oplus (\M_{A_x},\M_{B_x},\H_x)\,d\mu
\end{equation}
for a measure space $(X,\mu)$ such that $\M_A\cap\M_B \cong L^\oo(X,\mu)$, and a measurable field of factorial bipartite systems $(\M_{A_x},\M_{B_x},\H_x)$, $x\in X$.\footnote{This is really just the disintegration of von Neumann algebras since a bipartite system is nothing but a representation of a von Neumann algebra from a structural point of view. A "measurable field" of bipartite systems is just a measurable field of von Neumann algebras.}

The following fundamental fact will be useful on several occasions.
It states that ,, for all unit vectors $\Psi,\Phi\in\H$, the two parties agree on which of the two marginal states has a larger support projection.

\begin{lemma}[{\cite[Prop.~V.1.10]{takesaki1}}]\label{lem:mvn ordering}
    If $(\M_A,\M_B,\H)$ is a bipartite system on $\H$ and $\Psi,\Phi$ are unit vectors in $\H$, then%
    \footnote{In \eqref{eq:coupling-of-the-traces}, we write $[\M_A'\Psi]$ instead of $[\M_B\Psi]$ (which is the same thing since we assume Haag duality) to emphasize that the projection $[\M_A'\Psi]$ is intrinsic to the algebra $\M_A$. In fact, $[\M_A'\Psi]$ is exactly the support projection $\supp(\psi_A)$ of the state $\psi_A$ that $\Psi$ induces on $\M_A$; see \cref{eq:supp proj}.}
    \begin{equation}\label{eq:mvn ordering}
        [\M_A'\Psi] \succeq [\M_A'\Phi] \iff [\M_B'\Psi] \succeq [\M_B'\Phi].
    \end{equation}
\end{lemma}

By \eqref{eq:supp proj}, the statement \eqref{eq:mvn ordering} is the same as $\supp\psi_A \preceq \supp\phi_A \iff \supp\psi_B\preceq\supp\phi_B$, where $\psi_{A/B}$ and $\phi_{A/B}$ denote the respective marginal states.
We will see later that \cref{eq:mvn ordering} is moreover equivalent to the statement that $\Phi$ can be obtained from $\Psi$ with approximate stochastic LOCC (see \cref{sec:slocc}).
For factorial bipartite systems, \cref{lem:mvn ordering} was in the first paper of Murray and von Neumann \cite[Lem.~9.3.3]{rings_of_operators1}, where it is used to define the coupling constant for representations of type $\II_1$ factors.

\subsection{Standard bipartite systems}\label{sec:std bipartite systems}

In the following, we consider so-called standard bipartite systems.
Roughly speaking, these are bipartite systems, in which the two von Neumann algebras are equal in size.
We define a \emph{standard bipartite system} as a bipartite system $(\M_A,\M_B,\H)$ with the property that $\M_A$ and, therefore, $\M_B$ are in standard representation on $\H$.
As a consequence of the properties of standard representations (see \cref{sec:std form}), we will establish an operational characterization of the standard form in terms of the existence of arbitrary purification for both $\M_A$ and $\M_B$.

If $(\M_A,\M_B,\H)$ is a bipartite system, we say that a vector $\Omega\in\H$ is \emph{bicyclic} if it is cyclic for $\M_A$ and for $\M_B$ (see \cref{sec:weights and states}).
Since we assume Haag duality, this is equivalent to $\Omega$ being separating for both $\M_A$ and $\M_B$.\footnote{If Haag duality is not assumed, bicyclicity is stronger than the assumption of being cyclcic and separating for one of the two parties.}

\begin{proposition}\label{prop:std bipartite system}
    A bipartite system $(\M_A,\M_B,\H)$ is standard if and only if one (hence all) of the following equivalent properties holds:
    \begin{enumerate}[(a)]
        \item\label{it:std bipartite system1} Bicyclic vector: There exists a vector $\Omega$ that is bicyclic in the sense that it is cyclic for $\M_A$ and for $\M_B$ (equivalently, $\Omega$ is cyclic and separating for $\M_A$),
        \item\label{it:std bipartite system2} Exchange symmetry: There exists a conjugation $J$ such that $J\M_A J = \M_B'$ and $JzJ=z^*$ for $z\in \M_A\cap\M_B$,
        \item\label{it:std bipartite system3}  Purifications: For every state $\omega_A$ on $\M_A$ there is a vector $\Omega\in\H$ such that $\omega_A(a)=\ip{\Omega}{a\Omega}$, $a\in\M_A$, and $\M_B$ has the same property.
    \end{enumerate}
    In particular, if these equivalent properties hold, then $\M_A$ and $\M_B$ have the same type and subtype.
\end{proposition}

\begin{proof}
    The equivalence of standardness and \ref{it:std bipartite system1} is shown in \cref{lem:std rep} (note that $\Omega$ is bicyclic if and only if it is cyclic and separating for $\M_A$).
    Thus, \ref{it:std bipartite system1} $\Rightarrow$ \ref{it:std bipartite system2} and \ref{it:std bipartite system1} $\Rightarrow$ \ref{it:std bipartite system3} follow from \cref{thm:std form}.
    The implication \ref{it:std bipartite system2} $\Rightarrow$ \ref{it:std bipartite system3} is shown in \cite[Thm.~V.3.15]{takesaki1}.

    \ref{it:std bipartite system3} $\Rightarrow$ \ref{it:std bipartite system1}:
    In particular, faithful states on $\M_A$ and $\M_B=\M_A'$ are implemented by vectors in $\H$.
    Thus, $\M_A$ admits both a separating vector and a cyclic vector.
    By \cite[Cor.~V.1.14]{takesaki1} $\M_A$ also admits a cyclic separating vector. 
\end{proof}

As a consequence of \cref{thm:std form}, if $(\M_A,\M_B,\H)$ is a standard bipartite system, then one can pick a self-dual positive cone $\P\subset\H$ and a conjugation $J$ with $J\M_AJ=\M_B'$ such that $(\H,J,\P)$ is a standard form for both $\M_A$ and $\M_B$.

The following sufficient condition for standardness is often useful:

\begin{lemma}[{\cite[Thm.~III.2.6.16]{blackadar_operator_2006}}]\label{lem:std if infinite}
    Let $(\M_A,\M_B,\H)$ be a bipartite system.
    If $\M_A$ and $\M_B$ are both properly infinite (types $\I_\oo$, $\II_\oo$ or $\III$), then $(\M_A,\M_B,\H)$ is standard.
\end{lemma}

We consider standard subsystems induced by bipartite pure states.
We say that an anti-linear operator $J$ on $\H$ is a \emph{partial conjugation} if there is a closed subspace $\H_0\subset\H$ such that $J= J_0\oplus 0$ with respect to $\H=\H_0\oplus\H_0^\perp$, where $J_0$ is a conjugation on $\H_0$.

\begin{proposition}\label{prop:state induced truncation}
    Let $(\M_A,\M_B,\H)$ be a bipartite system.
    Let $\Psi\in\H$ be unit vector and let $\psi_{A/B}$ its marginal states on $\M_{A/B}$.
    We set $p_{A/B}=\supp(\psi_{A/B})$, $q=p_Ap_B$, $\H_\Psi=q\H$, and $\M_{\Psi,A/B}=q\M_{A/B}q$.
    Then:
    \begin{enumerate}[(1)]
        \item\label{it:state induced truncation1}
        $(\M_{\Psi,A},\M_{\Psi,B},\H_\Psi)$ is a standard bipartite system with $\Psi\in\H_\Psi$ a bicyclic vector.
        \item\label{it:state induced truncation2}
        There exists a partial conjugation $J$ on $\H$ such that
        \begin{enumerate}[(i)]
            \item $J\Psi=\Psi$ (and hence $\psi_A = \psi_B(J(\placeholder)^* J)$),
            \item $J^2 = q$,
            \item $J\M_{A}J = \M_{\Psi,B}$ and $J\M_{B}J = \M_{\Psi,A}$.
        \end{enumerate}
        \item\label{it:state induced truncation3}
        A normal state $\phi_{A/B}$ on $\M_{A/B}$ has a purification in $\H_\Psi$ if and only if $\supp(\phi_{A/B})\le \supp(\psi_{A/B})$. 
    \end{enumerate}
\end{proposition}

\begin{proof}
    \ref{it:state induced truncation1}:
    $\Psi$ is cyclic for both $\M_{\Psi,A/B}$ since these are precisely given by the projections onto the cyclic subspaces associated with $\M_{A/B}$.
    Haag duality $\M_{\Psi,A}=\M_{\Psi,B}'$ holds because taking commutants is compatible with cutdowns \cite[Prop.~5.5.6 \& Cor.~5.5.7]{kadison_ringrose1}.
    Standardness follows because $\Psi$ is bicyclic.

    \ref{it:state induced truncation2}:
    Let $J_0$ be the modular conjugation for $(\M_{\Psi,A},\M_{\Psi,B})$ induced by the bicyclic vector $\Psi\in\H_\Psi$.
    By definition, $\H_\Psi$ is a closed subspace of $\H$, and we obtain an operator $J$ with the desired properties via $J := J_0\oplus 0$ with respect to the direct sum decomposition $\H= \H_\Psi\oplus \H_\Psi^\perp$.

    \Cref{it:state induced truncation3} follows from \ref{it:state induced truncation1} and \cref{prop:std bipartite system}.
\end{proof}

\begin{definition}\label{def:truncated-system}
    If $(\M_A,\M_B,\H)$ is a bipartite system and $\Psi\in \H$, we refer to the bipartite system $(\M_{\Psi,A},\M_{\Psi,B},\H_\Psi)$ in \cref{prop:state induced truncation} as the \emph{truncated bipartite system} associated with $\Psi$, and to the operator $J$ in  as \emph{the partial modular conjugation induced by $\Psi$}, and we write $J = J_\Psi$.
\end{definition}

\subsection{Distance of pure states up to local unitaries}\label{sec:local distinguishability}

This subsection is devoted to showing that the local distinguishability of marginals of bipartite pure states, up to local unitaries, is as good for Alice as it is for Bob.
We will show this for both the fidelity and the norm distance.
In \cref{sec:haag}, we will show that a weaker version of this statement, the uniqueness of purifications up to local unitaries, characterizes Haag duality.

\begin{proposition}\label{prop:local distinguishability}
    Let $(\M_A,\M_B, \H)$ be a bipartite system. 
    Let $\Psi,\Phi\in\H$ be vectors and let by $\psi_{A/B},\phi_{A/B}\in (\M_{A/B})_*^+$ be the induced marginals on $\M_{A/B}$. 
    Then
    \begin{equation}\label{eq:local distinguishability norm}
        \inf_{u_A}\, \norm{u_A \phi_A u_A^* - \psi_A} = \inf_{u_B}\, \norm{u_B\phi_Bu_B^*-\psi_B},
    \end{equation}
    and
    \begin{equation}\label{eq:local distinguishability fid}
    \sup_{u_A}\, F(\psi_A,u_A\phi_A u_A^*) = \sup_{u_A,u_B} \abs{\ip{\Psi}{u_Au_B\Phi}}^2 = \sup_{u_B}\, F(\psi_B,u_B\phi u_B^*),
    \end{equation}
    where the optimizations are over unitaries $u_{A/B}\in\M_{A/B}$.
\end{proposition}

In \cref{eq:local distinguishability fid}, $F$ denotes the fidelity (see \cref{sec:fidelity}).
As a consequence of \eqref{eq:local distinguishability fid}, we note that a pair of vectors is approximately equivalent up to  local unitaries (LU) if and only if the respective marginals are  approximately unitarily equivalent:\footnote{This uses that convergence in Fidelity is equivalent to norm convergence, which follows from the Fuchs-van de Graaf inequalities, see \eqref{eq:FvdG}.}

\begin{corollary}\label{cor:local unitary equivalence}
    Let $(\M_A,\M_B,\H)$ be a bipartite system.
    Let $\Psi,\Phi\in\H$ be vectors and let by $\psi_{A/B},\phi_{A/B}\in (\M_{A/B})_*^+$ be the induced marginals on $\M_{A/B}$.
    The following are equivalent:
    \begin{enumerate}[(a)]
        \item $\Psi$ and $\Phi$ are approximately LU equivalent, i.e., for all $\eps>0$ there are unitaries $u_A\in\M_A$ and $u_B\in\M_B$ such that $u_Au_B\Psi \approx_\eps\Phi$;
        \item $\psi_A\sim\phi_A$, i.e., for all $\eps>0$, there exists a unitary $u_A\in\M_A$ such that $u_A\psi_Au_A^*\approx_\eps\phi_A$;
        \item $\psi_B\sim\phi_B$.
    \end{enumerate}
\end{corollary}

\Cref{cor:local unitary equivalence} will allow us to show that embezzlement of entanglement is equivalent to universal catalyticity, its single-system analog.
We now come to the proof of \cref{prop:local distinguishability}

\begin{proof}[Proof of \cref{prop:local distinguishability}]
    We begin by showing \cref{eq:local distinguishability norm}:

    \emph{Step 1.} 
    If $\Omega\in\H$ is a vector, we denote by $\mathrm{LU}(\Omega)$ its closed local unitary orbit, i.e., the closure of the set of vectors $u_Au_B\Omega$, where $u_{A/B}\in\M_{A/B}$ are unitaries.
    We note that both the LHS and the RHS of \eqref{eq:local distinguishability norm} remain unchanged if, instead of $\Psi, \Phi$, we take other vectors $\Psi'\in \mathrm{LU}(\Psi)$, $\Phi'\in \mathrm{LU}(\Phi)$.

    \emph{Step 2.} We assume that $(\M_A,\M_B,\H)$ is a standard bipartite system.
    Let $J$ be a conjugation and let $\P\subset\H$ be a $J$-invariant self-dual positive cone such that $(\H,J,\P)$ is a standard form for $\M_A$ (and hence $\M_B$), see \cref{sec:std form}.
    By step 1 and \cref{cor:all-purifications}, we may assume $\Psi,\Phi\in\P$.
    Then
    \begin{align*}
        \inf_{u_A} \norm{u_A\phi_A u_A^*-\psi_A} 
        &= \adjustlimits \inf_{u_A} \sup_a \abs{\ip \Psi{u_A^*au_A\Psi} - \ip\Phi{a\Phi}} \\
        &= \adjustlimits \inf_{u_A} \sup_a \abs{\ip \Psi{Ju_A^*J JaJ Ju_AJ\Psi} - \ip\Phi{JaJ\Phi}} \\
        &= \adjustlimits \inf_{u_B} \sup_b \abs{\ip \Psi{u_B^*bu_B\Psi} - \ip\Phi{b\Phi}} 
        = \inf_{u_b} \norm{u_B\phi_B u_B^* - \psi_B,}
    \end{align*}
    where the suprema are over the set of contractions in the respective algebra.

    \emph{Step 3.} We show the claim by reducing it to the case where the bipartite system is standard.
    Let $\K$ be an infinite-dimensional separable Hilbert space, let $\Omega\in\H$ be a unit vector, and let $\omega$ be the induced normal state on $\B(\K)$.
    By \cref{lem:std if infinite}, $(\M_A\ox \B(\K)\ox1,\M_B\ox1\ox\B(\K),\H\ox\K^{\ox2})$ is a standard bipartite system.
    The marginals of $\tilde\Psi=\Psi\ox\Omega^{\ox2}$ and $\tilde\Phi=\Phi\ox\Omega^{\ox2}$ are $\tilde\psi_{A/B}=\psi_{A/B}\ox \omega$ and $\tilde\phi_{A/B}=\phi_{A/B}\ox\omega$.
    \Cref{lem:distance of unitary orbits and amplification}, implies that
    \begin{equation*}
        \inf_{u_{A/B}}\,\norm{u_{A/B}\phi_{A/B} u_{A/B}^*-\psi_{A/B}} 
        =\inf_{\tilde u_{A/B}} \,\norm{\tilde u_{A/B}\tilde\phi_{A/B}\tilde u_{A/B}^* -\tilde\psi_{A/B}},
    \end{equation*}
    where $\tilde u_{A/B}$ ranges over unitaries in $\M_{A/B}\ox\B(\K)$.
    By step 2, we have 
    \begin{equation*}
        \inf_{\tilde u_A}\,\norm{\tilde u_A\tilde\phi_A\tilde u_A^*-\tilde\psi_A} = \inf_{\tilde u_B}\,\norm{\tilde u_B\tilde\phi_B\tilde u_B^*-\tilde\psi_B},
    \end{equation*}
    which finishes the proof.

    \Cref{eq:local distinguishability fid} follows from Uhlmann's Theorem (\cref{thm:uhlmann}) and \cref{cor:all-purifications}.
\end{proof}

\subsection{The coupling constant of semifinite bipartite systems}\label{sec:coupling}

In this subsection, we investigate factorial semifinite bipartite systems.
In particular, we discuss how traces on $\M_A$ and $\M_B$ are coupled to each other and how the coupling constant $c(\M_A:\M_B)$ measures the relative size of $\M_A$ and $\M_B$.
At the end of the subsection, we comment on the non-factorial case.
Mathematically, the material of this subsection is well-known and goes back all the way to Murray and von Neumann \cite{rings_of_operators1,rings_of_operators2}.
However, coming from bipartite pure state entanglement theory, we will look at these classical results from new perspectives.

Recall that a semifinite factor $\M$ (i.e., a type $\I$ or $\II$ factor) has a \nsf $\tau$ which is uniquely determined up to scaling.
In case of a factorial semifinite bipartite system $(\M_A,\M_B,\H)$, this means that we have, a priori, two \nsf traces $\tau_{A/B}$ that are only  determined up to scaling.
However, we can canonically fix their relative scaling.
This is a consequence of the following classic structure result on the representations of semifinite factors due to Murray-von Neumann \cite{rings_of_operators1,rings_of_operators2} (see also \cite[Thm.~V.3.8]{takesaki1}):

\begin{lemma}\label{lem:murray vN}
    Let $\M$ be a semifinite factor on $\H$. 
    For each pair of \nsf traces $\tau,\tau'$ on $\M$ and $\M'$, there exists a number $0<C<\oo$ such that $\tau[\M'\Psi] =C \cdot \tau'[\M\Psi]$ for all $\Psi\in\H$. 
    In particular, $[\M'\Psi]$ is a finite projection in $\M$ if and only if $[\M\Psi]$ is a finite projection in $\M'$.
\end{lemma}

We make the following convention, which we call the "trace coupling convention", regarding the joint scaling of the traces on a factorial semifinite bipartite system:

\begin{convention*}
    For a factorial semifinite bipartite system $(\M_A,\M_B,\H)$, the traces $\tau_A$ and $\tau_B$ are chosen such that 
    \begin{equation}\label{eq:coupling-of-the-traces}
        \tau_A([\M_A'\Psi])=\tau_B([\M_B'\Psi])\qquad \forall \Psi\in\H.
    \end{equation}
\end{convention*}

In the case of $\I$ factors $\M_A=\B(\H_A)\ox1$, $\M_B=1\ox\B(\H_B)$ on $\H=\H_A\ox\H_B$, our convention says that the projections onto the subspaces $V_A\ox \H_B$ and $\H_A\ox V_B$ have the same trace if and only if $\dim V_A=\dim V_B$.
In particular, the convention is satisfied by the standard traces (i.e., minimal projections have trace 1).
In the study of type $\II_1$ factors and their representation, the different convention that all traces are unital is usually used \cite{jones_index_1983}.

Next, we consider standard bipartite systems, where we show that the trace coupling convention agrees with the trace coupling induced by the modular conjugation:

\begin{lemma}\label{lem:std trace coupling}
    Let $(\M_A,\M_B,\H)$ be a factorial semifinite bipartite system.
    If the bipartite system is standard with conjugation $J$ (as in \cref{prop:std bipartite system}), then the trace coupling convention is equivalent to
    \begin{equation}\label{eq:std trace coupling}
        \tau_A (a) = \tau_B(Ja^*J), \qquad a\in\M_A.
    \end{equation}
\end{lemma}
\begin{proof}
    Let $(\H,J,\P)$ be a standard form of $\M_A$.
    Let $\omega_A$ be a state on $\M_A$ with a finite support projection, and let $\Omega\in\P$ be the purification.
    Note that $\supp(\omega_A) = [\M_A'\Omega] = [J \M_B'J\Omega] = [J\M_B'\Omega] = [J\M_B'\Omega]J = J \supp(\omega_B)J$, where $\omega_B$ is the marginal state of $\Omega$ on $\M_B$.
    This gives
    \begin{equation*}
        \tau_B([\M_B' \Omega]) =\tau_B(J[\M_A'\Omega]J).
    \end{equation*}
    Thus, the trace coupling convention (which sets the LHS equal to $\tau_A([\M_A'\Omega])$) is equivalent to \eqref{eq:std trace coupling} (which sets the RHS equal to the same value).
\end{proof}

A benefit of our convention is that it also makes sense when one or both of the factors is non-finite, which is not the case for the convention where all traces are taken to be unital (see above).

\begin{definition}[Coupling constant for bipartite systems]\label{def:coupling constant}
    The coupling constant of a factorial semifinite bipartite system $(\M_A,\M_B,\H)$ is defined as
    \begin{equation}\label{eq:coupling constant}
        c(\M_A:\M_B) := \frac{\tau_A(1)}{\tau_B(1)} \in [0,\oo],
    \end{equation}
    for some, hence all, choice of coupled traces $(\tau_A,\tau_B)$.
    For a factorial bipartite system of type $\III$, we set $c(\M_A:\M_B)=1$.
\end{definition}

In \eqref{eq:coupling constant} and below, we use the following conventions for degenerate ratios:
\begin{equation}\label{eq:ratio convention}
    \frac \oo r = \oo,\quad \frac r\oo=0, \qandq \frac\oo\oo=1, \qquad \text{ for $0<r<\oo$.}
\end{equation}
The coupling constant measures the size of $\M_A$ relative to $\M_B$.
It is related to the Murray-von Neumann coupling constant $c(\M_A,\H) \equiv \dim_{\M_A}\H$ (cp.\ \cite{rings_of_operators1, jones_index_1983}) via
\begin{equation}
    c(\M_A:\M_B) = \frac1{\dim_{\M_A}\H} = \dim_{\M_B}\H.
\end{equation}
It is a direct consequence of the definition that
\begin{equation}
    c(\M_B:\M_A) = c(\M_A:\M_B)^{-1},
\end{equation}
which also holds if the coupling constants are $0$ or $\oo$ if the convention \eqref{eq:ratio convention} is used.
In case of a tensor product bipartite system with Hilbert space $\H = \H_A\ox\H_B$, we have 
\begin{equation}
    c\big(\B(\H_A)\ox1:1\ox \B(\H_B)\big) = \frac{  \dim\H_A}{\dim\H_B}.
\end{equation}

\begin{lemma}[{\cite[Sec.~V.3]{takesaki1}}]\label{lem:coupling constant}
    Let $(\M_A,\M_B,\H)$ be a factorial bipartite system.
    \begin{enumerate}[(i)]
        \item\label{it:coupling constant1} $c(\M_A:\M_B)\ge 1$ (resp.\ $\le 1$) if and only if  $\M_A$ (resp.\ $\M_B$) admits a cyclic vector.

        \item\label{it:coupling constant2} $(\M_A,\M_B,\H)$ is standard if and only if $c(\M_A:\M_B)=1$.

        \item\label{it:coupling constant3} Suppose $(\N_A,\N_B,\K)$ is a bipartite system such that $\N_A\cong \M_B$. Then $(\M_A,\M_B,\H) \cong (\N_A,\N_B,\K)$ holds if and only if
            \begin{equation}
                c(\M_A:\M_B) = c(\N_A:\N_B).
            \end{equation}
    \end{enumerate}
\end{lemma}
\begin{proof}
    Clearly \ref{it:coupling constant1} $\Rightarrow$ \ref{it:coupling constant2}.
    \Cref{it:coupling constant1} holds trivially in the type $\III$ case.
    The semifinite case of the \cref{it:coupling constant1,it:coupling constant3} are shown in \cite[Thm.~V.3.13]{takesaki1} and \cite[Thm.~V.3.11]{takesaki1}, respectively.\footnote{To be precise, the results of \cite{takesaki1} are only shown under the assumption of finiteness. The reason for this is that \cite{takesaki1} only defines the coupling constant under this assumption, where it is defined via \cref{lem:murray vN} with the convention that the traces are unital. With our convention of defining the coupling constant, the proofs of \cite{takesaki1} apply to the non-finite case.}
    The type $\III$ case of the third item is shown in \cite[Cor.~V.3.2]{takesaki1}.
\end{proof}

To state our next result, we need conventions for traces on cut-downs:
If $\M_A$ is semifinite and $p_A\in\proj(\M_A)$ is a projection, then the cut-down $p_A \M_A p_A$ is a semifinite von Neumann algebra on $p_A\H$, which we equip with the restricted trace. 
The commutant of $p_A\M_Ap_A$, taken as a von Neumann algebra on $p_A\H$, is precisely $p_A\M_Bp_A$ \cite[Cor.~5.5.7]{kadison_ringrose1}.
If $p_B$ is a projection in $\M_B$, then $p_B\M_Ap_B$ is a von Neumann algebra on $p_B\H$ whose commutant is $p_B\M_Bp_B$ \cite[Prop.~5.5.6]{kadison_ringrose1}. 
In this case, we equip the cut-down $p_B\M_A p_B$ with the trace induced by this isomorphism $a\mapsto p_B a$.
Hence, if $q= p_Ap_B$ with $p_{A/B}$ as above, the commutant of $q\M_A q\subset \B(q\H)$ is $q\M_Bq$, so that $(q\M_Aq,q\M_Bq,q\H)$ is a bipartite system.
It can readily be checked that these conventions imply that $\tau_{q\M_j q}([(q\M_jq)'\Psi]) =\tau_{\M_j}([\M_j'\Psi])$, $j=A,B$, for all $\Psi\in q\H$, where we introduced the notation $\tau_{q\M_j q},\tau_{\M_j}$ to distinguish the original from the induced trace.

\begin{lemma}\label{lem:truncated traces}
     Let $(\M_A,\M_B,\H)$ be a factorial semifinite bipartite system and let $p_{A/B}\in\M_{A/B}$ be projections. Set $q=p_Ap_B$, $\K=q\H$, and $\N_{A/B}=q\M_{A/B}q$.
    Then the induced traces on $\N_A$ and $\N_B$ satisfy the trace coupling convention, and 
    \begin{equation}\label{eq:coupling constant truncated}
        c(\N_A:\N_B) = \frac{\tau_A(p_A)}{\tau_B(p_B)}.
    \end{equation}
    In particular, $(\N_A,\N_B,q\H)$ is standard if and only if $\tau_A(p_A)=\tau_B(p_B)$.
\end{lemma}
\begin{proof}
    As stated above, we have $\tau_{\N_j}([\N_j'\Psi])=\tau_{\M_j}([\M_j'\Psi])$, $j=A,B$, for $\psi\in \K$.
    Thus, the claim follows from the assumption that the trace coupling convention holds for the traces on $\M_{A/B}$.
    \Cref{eq:coupling constant truncated} follows directly from the definition of the coupling constant.
\end{proof}


\begin{lemma}\label{lem:rank condition for purification}
    Let $(\M_A,\M_B,\H)$ be a factorial semifinite bipartite system with coupled \nsf traces $(\tau_A,\tau_B)$ and let $\psi_A$ be a normal state on $\M_A$.
    Then $\psi_A$ has a purification in $\H$ if and only if 
    \begin{equation}\label{eq:rank condition for purification}
        \tau_A(\supp(\psi_A))\le \tau_B(1),
    \end{equation}
    where both sides may be infinite.
\end{lemma}
\begin{proof}
    If $\Psi$ is a purification in $\H$, then \eqref{eq:coupling-of-the-traces} implies that $\supp(\psi_{A})$ and $\supp(\psi_B)$ have the same trace. Both of them are obviously bounded by the trace of the identity operator.
    Now suppose that \eqref{eq:rank condition for purification} holds.
    Set $p_A=\supp(\psi_A)$ and let $p_B$ be a projection in $\M_B$ with $\tau_B(p_B)=\tau_A(p_A)$.
    Then \cref{lem:truncated traces} implies that the reduced bipartite system on $p_Ap_B\H$ is standard.
    Thus, \cref{prop:std bipartite system} implies that $\psi_A$ has a purification in the reduced bipartite system and, therefore, also in the original one.
\end{proof}

\begin{remark}[Non-factorial bipartite systems]\label{rem:non-factors1}
    The results presented above can be generalized to the case where $\M_A$ and $\M_B=\M_A'$ are general von Neumann algebras instead of factors.
    For this, traces and coupling constants have to be replaces by so-called center-valued traces and coupling functions (see \cite[Def.~V.2.33, Def.~V.3.9]{takesaki1}).
    This essentially amounts to a disintegration over the center $L^\oo(X,\mu)=Z(\M_A)=Z(\M_B)$ and then applying the above results to the factorials bipartite systems associated with the points $x\in X$.
\end{remark}

\subsection{Schmidt spectrum and entanglement entropies}\label{sec:schmidt spectrum}

We continue our discussion of semifinite factorial bipartite systems.
We use the trace coupling discussed in the previous subsection to generalize the notion of the Schmidt spectrum of bipartite pure states to the von Neumann algebraic setting. 
As an application, we can define entanglement monotones such as the Schmidt rank or the entanglement entropy.

Before discussing the generalization to the von Neumann algebraic case, let us briefly recall the Schmidt spectrum for bipartite systems of the form $\H=\H_A\ox\H_B$.
The story starts with the Schmidt decomposition:
If $\Psi$ is a state vector in $\H_A\ox\H_B$, then there exists a unique number $r(\Psi)\in \NN\cup\{\oo\}$, called the Schmidt rank (after Erhard Schmidt \cite{erhard_schmidt}), a unique decreasing $\ell^1$-sequence $(\lambda_n)_{n=1}^{r(\Psi)}$, called the Schmidt spectrum, and orthonormal families of vectors $\{\Phi_n^{A/B}\}_{n=1}^{r(\Psi)} \subset \H_{A/B}$ such that
\begin{equation}\label{eq:schmidt decomp}
    \Psi = \sum_{n=1}^{r(\Psi)} \lambda_n^{1/2} \,\Phi^A_n\ox\Phi^B_n.
\end{equation}
Entanglement monotones such as the entanglement entropy, which is defined as
\begin{equation}\label{eq:entanglement entropy}
    H(\Psi) := \sum_{n=1}^{r(\Psi)} \eta(\lambda_n) = \tr \eta(\rho_{A/B}), \qquad \eta(t) = -t\log t,
\end{equation}
only make sense because we know from the Schmidt decomposition that $\rho_A$ and $\rho_B$ have the same spectrum, namely the Schmidt spectrum of $\Psi$, which implies that \eqref{eq:entanglement entropy} is independent of the choice $A/B$.

In the following, we will see that the above generalizes to semifinite bipartite systems.
In particular, we will be able to generalize the usual entanglement monotones (e.g., the Schmidt rank and entanglement entropies) to the setting of  semifinite bipartite systems.
It will be shown in \cref{sec:locc} that these generalizations are indeed entanglement monotones.

Let $(\M_A,\M_B,\H)$ be a general bipartite system and let $\Psi\in\H$ be a unit vector. We denote by $\psi_{A/B}$ its marginal states.
By \cref{prop:state induced truncation}, $\Psi$ is a bicyclic vector for the truncated bipartite system $(\M_{\Psi,A},\M_{\Psi,B},\H_\Psi)$ associated with $\Psi$ (see \cref{def:truncated-system}).
The modular conjugation $J_\Psi$ induces an anti-isomorphism
\begin{equation*}
    j_\Psi : \M_{\Psi,A}\to\M_{\Psi,B}
\end{equation*}
such that 
\begin{equation}\label{eq:corresponding marginals}
    \psi_B(j_\Psi(a))= \psi_A(a^*),\qquad a\in \M_{\Psi,A}.
\end{equation}
Since the truncated bipartite system is standard, we may regard $\H_\Psi$ as a Haagerup $L^2$-space $L^2(\M_{\Psi,A})$ on which $\M_{\Psi,A}$ and $\M_{\Psi,B}=\M_{\Psi,A}'$ act by left/right multiplication and in which the vector $\Psi$ is a positive self-adjoint operator with full support.
This is the generalization of the Schmidt decomposition, which works for general bipartite systems.
In the case of a type $\I$ bipartite system on $\H=\H_A\ox\H_B$, it reduces to the Schmidt decomposition considered above, where $\H_\Psi$ becomes the subspace spanned by the ONB $\{\Phi_n^A\ox\Phi_m^B\}_{n,m=1}^{r(\Psi)}$ (see \eqref{eq:schmidt decomp}).

We now use the above generalization to define a notion of Schmidt spectrum.
We consider a factorial semifinite bipartite system $(\M_A,\M_B,\H)$ with coupled \nsf traces $\tau_A$ and $\tau_B$ (see \cref{sec:coupling}).
We define the \emph{Schmidt spectral scale} (or \emph{Schmidt scale} for short) of the unit vector $\Psi$ as
\begin{equation}\label{eq:Schmidt spectrum}
    \lambda_\Psi(t) := \lambda_{\psi_A}(t)=\lambda_{\psi_B}(t), \qquad t>0,
\end{equation}
where $\lambda_{\psi_{A/B}}$ denote the spectral scales of $\psi_{A/B}$ relative to the traces $\tau_{A/B}$ (see \cref{sec:semifinite vNas}).
The Schimdt scale depends only on the joint scaling of the coupled traces $(\tau_A,\tau_B)$.

\begin{lemma}
    The Schmidt scale in \eqref{eq:Schmidt spectrum} is well-defined.
\end{lemma}
\begin{proof}
    We consider the truncated bipartite system $(\M_{\Psi,A},\M_{\Psi,B},\H_\Psi)$ as above.
    \Cref{lem:truncated traces} implies that the truncated traces $\tau_{\Psi,A}$ and $\tau_{\Psi,B}$ satisfy the trace coupling convention.
    By \cref{prop:state induced truncation,lem:std trace coupling}, $j_\Psi$ identifies the truncated traces:
    \begin{equation*}
        \tau_{\Psi,A}(a^*) = \tau_{\Psi,B}(j_\Psi(a)), \qquad a\in\M_{\Psi,A}.
    \end{equation*}
    Together with \eqref{eq:corresponding marginals}, this shows $\lambda_{\psi_A}(t)=\lambda_{\psi_B}(t)$.
\end{proof}

By construction, we can compute arbitrary trace-functions of the marginal states of $\Psi$ in terms of the Schmidt scale. 
A trace-function of an operator $\rho$ is an expression of the form $\tau(f(\rho))$ for some Borel function $f:\RR\to\RR^+$:

\begin{corollary}\label{cor:equal trace functions}
    Let $f:\RR^+\to\RR^+$ be a Borel measurable function. Then
    \begin{equation}\label{eq:equal trace functions}
        \int_{\RR^+} f(\lambda_\Psi(t))\,dt = 
        \tau_A\bigg(f\bigg(\frac{d\psi_A}{d\tau_A}\bigg)\bigg) =  \tau_B\bigg(f\bigg(\frac{d\psi_B}{d\tau_B}\bigg)\bigg),
    \end{equation}
    where all terms may be infinite.
\end{corollary}

The above allows us to generalize the Schmidt rank and the entanglement entropies from finite-dimensional quantum information theory \cite{horodecki_quantum_2009,nielsen_quantum_2010,nielsen_conditions_1999}:
The Schmidt rank of a unit vector $\Psi\in \H$ is defined as
\begin{equation}\label{eq:schmidt rank}
    r(\Psi) := |\supp(\lambda_\Psi)\;\!|,
\end{equation}
where $\abs\placeholder$ denotes the Lebesgue measure. By \cref{cor:equal trace functions}, we have
\begin{equation}
    r(\Psi) = \tau_A\bigg(\supp\bigg(\frac{d\psi_A}{d\tau_A}\bigg)\!\bigg) 
    =\tau_B\bigg(\supp\bigg(\frac{d\psi_B}{d\tau_B}\bigg)\!\bigg).
\end{equation}
The \emph{entanglement entropy} of $\Psi$ is defined as
\begin{equation}
    H(\Psi) := \int_{\RR^+} \eta(\lambda(t))\,dt,
\end{equation}
The \emph{$\alpha$-Renyi entanglement entropy} $H_\alpha(\Psi)$, $\alpha\in (0,1)\cup(1,\oo)$ of a state vector $\Psi\in\H$ is defined as
\begin{equation}
    H_\alpha(\Psi) := \frac1{1-\alpha}\log \int_{\RR^+} \lambda_\Psi(t)^{\alpha}\,dt.
\end{equation}
In \cref{sec:locc}, these quantities will be shown to actually be entanglement monotones.

In general, there is no canonical choice for the joint scaling of the two traces. 
If we scale both traces by a factor $c>0$, the trace coupling conventions continues to hold, but the entanglement entropies are shifted by $\log c$
\begin{equation}
    H_\alpha(\Psi) \longrightarrow H_\alpha(\Psi) +  \log c.
\end{equation}
However, differences of entanglement entropies remain unaffected and have an intrinsic meaning\footnote{The story is similar for the Schmidt rank: Scaling the trace means scaling all Schmidt ranks, but ratios are unaffected. This fits into the entropy picture because the logarithm of the Schmidt rank is the max-entropy.}:
\begin{equation}
    H_\alpha(\Psi)-H_\alpha(\Phi) \longrightarrow H_\alpha(\Psi) -H_\alpha(\Phi).
\end{equation}
This is analogous to the role of entropy in classical statistical mechanics, where only entropy differences are meaningful (this can be traced back to the fact that the Lebesgue measure on phase space has no canonical normalization).
Another similarity with the entropies of continuous distributions on the phase space is that the entanglement entropies above can become negative.
However, as in classical statistical physics, the sign is irrelevant as the entropy is only meaningful up to an additive constant.

\subsection{An operational characterization of Haag duality}\label{sec:haag}

In this subsection, we provide an operational interpretation of Haag duality in terms of pure states and local operations.
In \cref{cor:all-purifications}, we saw that vectors in a Hilbert space $\H$ induce the same marginal state on an observable algebra $\M$ if and only if they are related by unitaries in the commutant $\M'$ (up to arbitrarily small error).
Here, we show the converse:

\begin{theorem}\label{thm:haag_duality}
    Let $\M_A,\M_B$ be commuting von Neumann algebras on a Hilbert space $\H$.
    The following are equivalent: 
    \begin{enumerate}[(a)]
        \item\label{it:hd1} Haag duality $\M_A=\M_B'$;
        \item\label{it:hd2} if $\Psi,\Phi\in\H$ are unit vectors with $\psi_A = \phi_A$, then $\Psi = v_B\Phi$ for a partial isometry $v_B\in\M_B$;
        \item\label{it:hd3} if $\Psi,\Phi\in\H$ are unit vectors with $\psi_A = \phi_A$, then, for all $\eps>0$, there exists a unitary $u_B\in\M_B$ such that $\Psi\approx_\eps u_B\Phi$.
    \end{enumerate}
\end{theorem}

\begin{lemma}\label{lem:hd lem}
    Let $\M_A,\M_B$ be commuting von Neumann algebras on a Hilbert space $\H$.
    The following are equivalent:
    \begin{enumerate}[(a)]
        \item\label{it:hd lem1} Haag duality $\M_A=\M_B'$;
        \item\label{it:hd lem2} $[\M_A\Omega] = [\M_B'\Omega]$ for all $\Omega\in\H$
        \item\label{it:hd lem3} $\overline{\U(\M_A)\Omega} \, =\, \overline{\U(\M_B')\Omega}$ for all $\Omega\in\H$.
    \end{enumerate}
\end{lemma}
\begin{proof}
    \ref{it:hd lem1} $\Rightarrow$   \ref{it:hd lem3} is trivial.
    \ref{it:hd lem3} $\Rightarrow$ \ref{it:hd lem2} follows from the fact that a von Neumann algebra is spanned by its unitary group.
    It remains to show \ref{it:hd lem2} $\Rightarrow$ \ref{it:hd lem1}:
    Recall if $p$ is projection and $\N$ is a von Neumann algebra on $\H$, then $p\in \N \iff \N'p\H \subset p\H$ ($p\H$ is invariant for $\N'$).
    We start by observing the following:
    \begin{equation}\label{eq:supp in B2}
        [\M_A\Omega]\in \M_B \quad\iff \quad [\M_A\Omega]=[\M_B'\Omega].
    \end{equation}
    Indeed, "$\Leftarrow$" is clear, and the converse follows from $\overline{\M_B'\Omega} \supseteq \overline{\M_A\Omega} \supseteq \M_B'\overline{\M_A\Omega} \supseteq \M_B'\Omega $.
    By \eqref{eq:supp in B2}, our assumption implies $[\M_A\Omega]\in \M_B$ for all $\Omega\in\H$.
    Thus, if $p\in \proj(\M_A')$ and $\{\Omega_n\}$ is a basis of $p\H$,
    then $p = [\cup_n \M_A\Omega_n] = \vee_n [\M_A\Omega_n] = \vee_n [\M_B'\Omega_n]$ shows that $p \in \M_B$.
    Since $p$ was arbitrary, this implies $\M_B=\M_A'$.
\end{proof}

\begin{proof}[Proof of \cref{thm:haag_duality}]
    \ref{it:hd1} $\Rightarrow$ \ref{it:hd2} is shown in \cref{cor:all-purifications}, and
    \ref{it:hd2} $\Rightarrow$ \ref{it:hd3} follows from \cref{cor:isometrically acting pi's}.

    \ref{it:hd3} $\Rightarrow$ \ref{it:hd1}:
    We check condition \ref{it:hd lem3} of \cref{lem:hd lem}.
    Let $\Psi\in \H$ be a unit vector.
    For any unitary $u_{A'}\in \U(\M_A')$, the vector $\Phi = u_{A'}\Psi$ is such that $\psi_A = \phi_A$.
    Using our assumption, we obtain $\Psi \in \overline{\U(\M_B)\Phi}$.
    Since $u_{A'}$ was an arbitrary unitary, this implies $\overline{\U(\M_A')\Phi}= \overline{\U(\M_B)\Phi}$.
\end{proof}

\begin{example}\label{exa:haag-duality-failure}
    We consider an example due to Naaijkens \cite{fiedler_jones_2017,naaijkens_subfactors_2018,naaijkens_kosaki-longo_2013} of a setup where Haag duality is known to fail.
    We consider the toric code on an infinite 2D spin lattice $\Gamma = \ZZ^2$ (see \cite{naaijkensAnyonsInfiniteQuantum2012} and \cref{sec:many-body}). Let $C_1,C_2$ be two disjoint cones.
    We define $A = C_1\cup C_2$ and $B= \Gamma\setminus A$.
    Then Haag duality fails $\M_A\subsetneq \M_B'$.
    Let $\Phi\in \H$ denote the ground state vector.
    \lvl{Let $\ell$ be a line from a point in $C_1$ to a point in $C_2$ and let $u_\ell$ be the corresponding line operator, see \cite{naaijkensAnyonsInfiniteQuantum2012}.
    Since the system is topologically ordered, the excited state $\Psi = u_\ell\Phi$ only depends on the endpoints of the line.
    With line operators in $\M_A$, these endpoints can be moved around freely.
    The topological excitations, however, cannot be removed with unitaries in $\M_A$.
    Therefore, $u_A\Psi$ remains orthogonal to $\Phi$ for all $u_A\in\M_A$.    
    In particular, this shows that $\Psi$ cannot be approximated by acting on $\Phi$ with unitaries from $\M_A$. 
    We claim that $\Psi$ and $\Phi$ have the same $B$-marginal.
    Let $b \in \M_B$ be an operator of finite support.
    Then the given geometry lets us pick a line $\ell'$ with the same start and end points as $\ell$ that does not pass through support of $b$.
    Then $u_{\ell'}$ commutes with $b$ and $\Psi = u_{\ell'}\Phi$, imply that $\ip\Psi{b\Psi} = \ip\Phi{b\Phi}$.
    Since this holds for all operators with finite support in $B$, we conclude that $\Psi$ and $\Phi$ indeed induce the same states on $\M_B$.
    }
\end{example}

\section{Bipartite pure state entanglement theory}\label{sec:locc}

\localtableofcontents

\null

In this section, we consider pure state entanglement theory in von Neumann algebraic bipartite systems.
We characterize both LOCC and stochastic LOCC transitions, and show the non-type $\III$ part of \cref{thm:minimal-type_i} in \cref{sec:main-results}.
The section is largely based on the article \cite{van_luijk_pure_2024}, and extends ideas from \cite{verch_distillability_2005,crann_state_2020,werner_local_1987}.


\noindent
Before going into details, we discuss a few generalities.

As long as we are considering LOCC protocols, we assume that all instruments involved have Kraus rank 1, i.e., that every instrument $(T_x)$ is of the form $T_x = k_x^*(\placeholder)k_x$ for Kraus operators $\{k_x\}$.
We can make this assumption without loss of generality because general instruments arise from coarse-graining of instruments with Kraus rank 1 (see \eqref{eq:instrument Kraus}). 
Therefore, instruments with Kraus rank 1 only yield more information, which the agents can then choose to ignore.
Since the set of instruments with Kraus rank 1 is closed under products, the Kraus rank 1 property is passed on to the total instrument $(T_x)$ of an LOCC protocol.

Consider a pair $\psi,\phi\in\nstates(\H)$ of pure states with implementing vectors $\Psi,\Phi\in\H$. 
In \cref{sec:entanglement-basics}, we defined what it means that $\psi$ can be transformed to $\phi$ with LOCC 
We write
\begin{equation*}
    \Psi\locc\Phi,
\end{equation*}
to denote that $\psi$ can be transformed to $\phi$ with LOCC, as defined in \cref{sec:entanglement-basics}.
Note that $T$ is a quantum channel on $\H$ and if $\{k_x\}$ is a collection of Kraus operators for $T$ such that $T_*(\psi)=\phi$, then $k_x\Psi$ is equal to $p_x^{1/2}\Phi$ up to phase, where $p_x= \psi(k_x^*k_x)=\norm{k_x\Psi}^2$, for all $x$.
By adapting the phase of the Kraus operators, we can assume $k_x\Psi=p_x^{1/2}\Phi$ for all $x$.
We thus have the following:

\begin{lemma}
    Let $\Psi,\Phi\in\H$ be unit vectors. Then $\Psi\locc\Phi$ if and only if there is an LOCC protocol whose total instrument has Kraus operators $\{k_x\}$ such that
    \begin{equation}
        k_x \Psi = p_x^{1/2}\Phi, \quad p_x=\norm{k_x\Psi}^2 \qquad \forall x.
    \end{equation}
\end{lemma}

Next, we will show that we may assume that an agent $A$ with von Neumann algebra $\M_A$ can apply subnormalized operations to states on which these are normalization preserving.
A subnormalized quantum channel on a system with Hilbert space $\H$ is a normal subunital cp map $T$ on $\B(\H)$ (Heisenberg picture) or a normalization-nonincreasing cp map $T_*$ on the predual $\predualB(\H)$ (Schrödinger picture) with the two versions related by \eqref{eq:schrödinger picture}.
A subnormalized channel on $\B(\H)$ is inner for a von Neumann algebra $\M$ if it admits a Kraus decomposition \eqref{eq:kraus form} with operators $\{k_x\}$ in $\M$ (of course, $\sum k_x^*k_x=1$ is then replaced by $\sum k_x^*k_x=T(1)$).

\begin{lemma}\label{lem:subnormalized channels}
    Let $\M$ be a von Neumann algebra on $\H$ and let $\psi,\phi\in\nstates(\H)$.
    If there is an $\M$-inner subnormalized channel $T$ on $\B(\H)$ such that $T_*(\psi)=\phi$, then there is an $\M$-inner channel $S$ with $S_*(\psi)=\phi$.
\end{lemma}

\begin{proof}
    We set $k_0= (1-T(1))^{1/2}\in \M$ and $S = T + k_0^*(\placeholder)k_0$.
    Clearly, $S$ is an inner ucp map with $S_*(\psi)= \phi + k_0 \psi k_0^*$.
    Since $S_*$ is normalization preserving, we have $\norm{k_0 \phi k_0^*} = k_0 \phi k_0^* (1) = S_*(\phi)(1) -T_*(\phi) =\phi(1)-\psi(1)=0$, which shows $S_*(\psi)=\phi$.
\end{proof}

Based on the Lemma, we will make statements like "the agent applies a partial isometry $u$ to the state $\Psi$". This is justified by the Lemma if $u$ acts isometrically on $\Psi$, i.e., if $u^*u\Psi=\Psi$, because the agent can always implement an extension $\tilde T$ of the subnormalized channel $u^*(\placeholder)u$, which gives the same output state.
The choice of this extension is irrelevant as long as we do not simultaneously consider other initial states.
The same reasoning can be applied to instruments.

\subsection{One-way LOCC}\label{sec:owlocc}

A one-way LOCC protocol is an LOCC protocol that with two rounds. 
Thus, it consists of an instrument of the first agent, who then communicates the outcome to the second agent, who performs an operation conditioned on what they received (it is irrelevant whether the second agent communicates the outcome since the first agent is not allowed to act afterward).

One of the fundamental results of pure state entanglement theory is that, if an LOCC protocol takes a pure state $\Psi\in\H$ to a pure state $\Phi\in\H$, then there exists a one-way LOCC protocol that does the same thing \cite{lo_concentrating_2001}.
An approximate version of this was generalized to the case of semifinite von Neumann algebraic bipartite systems in \cite[Thm.~4.11]{crann_state_2020}.
In this subsection, we show the following exact version of the statement for general bipartite systems:

\begin{proposition}\label{prop:owlocc}
    Let $(\M_A,\M_B,\H)$ be a bipartite system and let $\Psi,\Phi\in \H$ be unit vectors.
    If $\Psi\locc\Phi$, then there is a one-way LOCC protocol taking $\Psi$ to $\Phi$, in which Alice performs an instrument with Kraus operators $\{k_x\}$ and Bob performs a conditional partial isometry $u_x'$ such that
    \begin{equation}
        u_x'k_x\Psi = p_x^{1/2}\Phi.
    \end{equation}
    The partial isometry $u_x'$ acts isometrically on $k_x\Psi$.
\end{proposition}

Recall from \cref{lem:subnormalized channels} that it makes sense to say that Bob performs a partial isometry on a vector, provided that it acts isometrically.
For the proof, we use the following Lemma:

\begin{lemma}\label{lem:simulation}
    Let $\M$ be a von Neumann algebra on $\H$, and let $\pi$ be the representation of $\M$ that puts it in standard form on $L^2(\M)$.
    For a given pair $a,\Psi$ of an operator $a\in\M$ and a vector $\Psi\in\H$, 
    \begin{itemize}
        \item let $w':L^2(\M)\to \H$ be the partial isometry appearing in the polar decomposition $\Psi = w'\Omega_\psi$ of $\Psi$ (see \cref{sec:purifications}),
        \item let $v'\in \pi(\M')$ be the partial isometry appearing in the polar decomposition of $\pi(a)\Omega_\psi = v' \Omega_{a\psi a^*}$ of $\pi(a)\Omega_\psi \in L^2(\M)$,
        \item let $u\in \M$ be the partial isometry with $\pi(u^*)=j(v')$, and let $a' = w'j(\pi(u^*a)) w'^*=w'v'j(\pi(a))w'^* \in\M'$.
    \end{itemize}
    Then $u$ is isometric on $a\Psi$ and 
    \begin{equation}\label{eq:simulation}
        a\Psi = u a' \Psi.
    \end{equation}
\end{lemma}
\begin{proof}
    \emph{Step 1.} We assume $\H = L^2(\M)$, $\pi=\id$, and $\Psi=\Omega_\psi$.
    Since $\Omega_{a\psi a^*}$ is $J$-invariant, $v'^*a\Omega_\psi = \Omega_{a\psi a^*}$ implies $uj(a)\Omega_\psi = \Omega_{a\psi a^*}$.
    Since $v'^*$ acts isometrically on $a\Omega_\psi$, $u=j(v'^*)$ acts isometrically on $j(a)\Omega_\psi$ and hence on $a'\Omega_\psi = j(u^*a)\Omega_\psi$.
    Thus, we have
    \begin{equation*}
        a\Omega_\psi = v'\Omega_{a\psi a^*} = v'uj(a) \Omega_\psi =u j(u^*a)\Omega_\psi = ua'\Omega_\psi.
    \end{equation*}
    
    \noindent
    \emph{Step 2.}
    In the general case, step 1 implies that $\pi(a)\Omega_\psi = \pi(u)j(\pi(u^*a))\Omega_\psi$.
    We have
    \begin{equation*}
        a\Psi = w'\pi(a)\Omega_\psi 
        = w' \pi(u)j(\pi(u^*a))w'^*\Psi
        = u w'j(\pi(u^*a))w'^* \Psi =ua'\Psi.
    \end{equation*}
    Since $\pi(u)$ is isometric on $\pi(j(u^*a))\Omega_\psi$, $u$ is isometric on $a'\Psi = w'\pi(j(u^*a))\Omega_\psi$.
\end{proof}

\begin{corollary}\label{cor:simulation}
    Let $(\M_A,\M_B,\H)$ be a bipartite system and let $\Psi\in \H$ be a vector.
    For every collection $\{k_x\} \subset \M_A$ of Kraus operators of Alice, there exist partial isometries $u_x\in\M_A$ of Alice and a collection $\{k_x'\} \subset \M_B$ of Kraus operators of Bob such that
    \begin{equation}\label{eq:simulation x}
        k_x\Psi = u_x k_x'\Psi.
    \end{equation}
    The partial isometries $u_x$ act isometrically on $k_x'\Psi$.
\end{corollary}
\begin{proof}
    If we define $w',v'_x,u_x,k_x'$ as in \cref{lem:simulation} (applied to $a=k_x$), then \eqref{eq:simulation x} holds.
    We only have to check that $\sum k_x'^*k_x'\le 1$ ($\le$ is enough by \cref{lem:subnormalized channels}).
    We have
    \begin{align*}
        \sum_x k_x'^*k_x' = \sum_x w'j(\pi(k_x))^*v_x'^*w_x'^*w_x'v_x'j(\pi(k_x))w'^* 
        &\le \sum_x w'j(\pi(k_x^*k_x))w'^*  = w'w'^* \le 1,
    \end{align*}
    where we used $v_x'^*w_x'^*w_x'v_x'\le 1$.
\end{proof}

\begin{proof}[Proof of \cref{prop:owlocc}]
    We consider the LOCC protocol, which takes $\Psi$ to $\Phi$.
    Whenever it is Bob's turn, we use \cref{cor:simulation} to replace the instrument that Bob chooses in the given round by an instrument of Alice, followed by a conditional partial isometry of Alice without changing the post-measurement states of the initial state $\Psi$.
    This constructs a new instrument in which all communication goes from Alice to Bob and in which $B$ only acts with partial isometries.
    Each of these partial isometries acts isometrically on the post-measurement state of the previous round.
    Thus, we can define a new protocol in which Alice performs the product instrument of all her instruments in the first round and then sends the output to Bob, who performs the product $v_x'$ all these partial isometries in the second round.
    By construction, $v_x'$ is a contraction in $\M_B$ that acts isometrically on $\Psi$. Thus, we can replace it by an isometrically acting isometry $u_x$ \cref{lem:isometric-contraction}.
    By construction, this one-way protocol has the same post-measurement states as the initial one (on the input state $\Psi$).
\end{proof}

\subsection{Stochastic LOCC}\label{sec:slocc}

The idea of stochastic LOCC (SLOCC) is to consider a relaxed version of LOCC transitions in which it is only required that the state transitions are possible with some nonzero probability, see \cite{horodecki_quantum_2009} for details.
We begin our consideration of SLOCC in the multipartite case.
Afterward, we characterize the approximate SLOCC transitions in bipartite systems and show that the Schmidt rank is a complete monotone for approximate SLOCC transitions. 

Consider a quantum system with Hilbert space $\H$ together with a collection of agents $A_1,\ldots A_N$ that have partial access to the full quantum system with Hilbert space $\H$.
We assume that the operations that can be implemented by these agents are described by von Neumann algebras $\M_1,\ldots \M_N$ on $\H$.
Moreover, we assume that these agents are pairwise independent, which translates to pairwise commutativity of the von Neumann algebras $\M_j$, $j=1,\ldots N$.
If $\psi,\phi\in\nstates(\H)$ are states of the full system, we say that \emph{$\psi$ can be transformed to $\phi$ with SLOCC} by the of agents $A_1,\ldots A_N$, if there exists an LOCC protocol with overall instrument $(T_x)$ such that the probabilities $p_x = \psi(T_x(1))$ for which $\psi\circ T_x = p_x \phi$ add up to a nonzero number, called the overall transition probability.
In this case, we write
\begin{equation*}
    \psi\slocc \phi.
\end{equation*}
Moreover, we say that $\psi$ can be transformed to $\phi$ by means of \emph{approximate SLOCC}, denoted
\begin{equation*}
    \psi\barslocc\phi,
\end{equation*}
if, for every $\eps>0$, there is a state $\phi'\approx_\eps\phi$ such that $\psi\slocc\phi'$.
It is, at first, not clear that approximate SLOCC does not trivialize. After all, it includes cases where it is not possible to perform the desired state transition with any nonzero probability. 
If $\psi$ and $\phi$ are pure states, we use the same notations for the implementing vectors $\Psi,\Phi$.

Next, we show the following Lemma, which, in particular, states that a pure state SLOCC transition never requires communication between the agents:
\begin{lemma}\label{lem:slocc light}
    Let $\Psi,\Phi\in \H$ be unit vectors. Then 
    \begin{equation}
        \Psi\slocc \Phi \quad \iff \quad \exists a_j\in \M_j\  \text{s.t.}\ \Phi = a_1\cdots a_N \Psi.
    \end{equation}
\end{lemma}
\begin{proof}
    "$\Rightarrow$": 
    If $\Psi\slocc\Phi$, there must be at least one nonzero Kraus operator $k_x$ of a total instrument associated with an LOCC protocol such that $k_x\Psi \propto\Phi$.
    Since each Kraus operator of the total instrument is a product $k_x= a_{1,x}\cdots a_{N,x}$ of operators from the local algebras $a_{j,x}\in\M_j$, the claim follows.

    "$\Leftarrow$": Let $k_{1,j} \propto a_j$ be nonzero rescalings such that $\norm{k_{1,j}}\le 1$ for each $j$.
    Now consider the LOCC protocol where each agent applies the instrument with Kraus operators $(k_{1,j},k_{2,j})$ (no communication is needed).
    By construction, this LOCC protocol implements the state transition from $\Psi$ to $\Phi$ with nonzero probability.
\end{proof}

We now restrict our attention to bipartite systems. 

\begin{theorem}\label{thm:slocc}
    Let $(\M_A,\M_B,\H)$ be a bipartite system and let $\Psi,\Phi$ be unit vectors in $\H$. 
    The following are equivalent
    \begin{enumerate}[(a)]
        \item\label{it:slocc1} $\Psi \barslocc \Phi$;
        \item\label{it:slocc2} $\supp(\phi_{A/B})\preceq \supp(\psi_{A/B})$;
        \item\label{it:slocc3} for each $\eps>0$, there exists a unitary $u_A\in \M_A$ and $k_B\in\M'$ such that $u_Ak_B\Psi\approx_\eps\Phi$.
    \end{enumerate}
\end{theorem}

In \cref{it:slocc2}, "$\preceq$" denotes the Murray-von Neumann ordering of projections (see \cref{sec:von-neumann-algs}).
Before we come to the proof, we evaluate the above for factorial bipartite systems:

\begin{corollary}\label{cor:factor slocc}
    Let $(\M_A,\M_B,\H)$ be a factorial bipartite system and let $\Psi,\Phi\in\H$ be unit vectors.
    Then:
    \begin{enumerate}
        \item If $\M_A$ and, therefore, $\M_B$ are of type $\III$, then $\Psi\barslocc\Phi$.
        \item If $\M_A$ and, therefore $\M_B$ are semifinite, then
            \begin{equation}\label{eq:factor-slocc}
                \Psi\barslocc \Phi \quad \iff\quad r(\Psi)\ge r(\Phi),
            \end{equation}
            where $r$ denotes the Schmidt rank relative to some choice of coupled \nsf traces $\tau_{A/B}$ on $\M_{A/B}$ (see \cref{sec:schmidt spectrum}).
    \end{enumerate}
\end{corollary}

The results of \cref{sec:schmidt spectrum} imply that the RHS of \eqref{eq:factor-slocc} does not depend on the choice of the coupled traces.
We need the following Lemma, which is surely known but which we could not locate in the literature.

\begin{lemma}\label{lem:MvN lsc}
    Let $\M$ be a von Neumann algebra.
    If $\phi_n,\phi\in\M_*^+$ and if $\supp(\phi_n)\preceq p$ for a projection $p\in\M$, then $\phi_n\to\phi$ implies $\supp(\phi)\preceq p$.
\end{lemma}

\begin{proof}
    By disintegration, we may assume $\M$ to be a factor.
    If $\M$ is type $\III$, the statement is trivial because all nonzero projections are Murray-von Neumann equivalent.
    For semifinite $\M$, consider a \nsf trace $\tau$ on $\M$.
    If $\lambda_{\phi_n},\lambda_\phi$ denote the respective spectral scales relative of $\tau$ (see \cref{sec:semifinite vNas}), then \cref{thm:majorization} implies $\lambda_{\phi_n}\to\lambda_\phi$ in $L^1$-norm. We find
    \begin{equation}
        \liminf_n \tau(\supp(\phi_n)) = \liminf_n |\supp(\lambda_{\phi_n})| \ge |\supp(\lim_n \lambda_{\phi_n})| =|\supp(\lim_n\lambda_\phi)| = \tau(\supp(\phi)),
    \end{equation}
    where $\abs\placeholder$ denotes the Lebesgue measure on $\RR^+$ and where we used the lower-semicontinuity of the measure of the support on $L^1(\RR^+)$.
    Thus, $p \succeq \supp(\phi_n)$ implies $\tau(p) \ge \liminf_n \tau(\supp(\phi_n))\ge \tau(\supp(\phi))$, which implies $p\succeq \supp(\phi)$ because $\M$ is a factor.
\end{proof}

\begin{proof}[Proof of \cref{thm:slocc}]
    The two versions of \ref{it:slocc2} are equivalent by \cref{lem:mvn ordering}.
    \ref{it:slocc3} $\Rightarrow$ \ref{it:slocc1} follows from \cref{lem:slocc light}.
    \ref{it:slocc1} $\Rightarrow$ \ref{it:slocc3}: 
    Let $\eps>0$, then \cref{lem:slocc light} implies that there are operators $a_{A/B}$ such that $\Phi \approx_{\eps/2} a_Aa_B\Psi$.
    We may assume $\norm{a_A}\le 1$.
    Applying \cref{lem:simulation}, we get an operator $b_B$ and a partial isometry $v_A$, which acts isometrically on $b_B\Psi$, such that $v_Ab_B\Psi = a_B\Psi$.
    By \cref{cor:isometrically acting pi's}, there is a unitary $u_A$ such that $u_Ab_B\Psi \approx_{\eps/2} v_Ab_B\Psi$.
    Setting $k_B = b_B a_B$, we have $u_Ak_B\Psi \approx_\eps \Phi$.

    \ref{it:slocc3} $\Rightarrow$ (\ref{it:slocc2}, "$B$"):
    Let $u_{A,\eps},k_{B,\eps}$ be as in the statement and set $\Phi_\eps=u_{A,\eps}k_{B,\eps}$. 
    The $B$-marginal of $\Phi_\eps$ is $\phi_{B,\eps}=k_{B,\eps}\psi_Bk_{B,\eps}^*$, which converges in norm to $\phi$ as $\eps\to0$ and satisfies $\supp(\phi_{B,\eps}) \preceq \supp(\psi_B)$ for all $\eps>0$.
    By \cref{lem:MvN lsc}, we get $\supp(\phi_B)\preceq\supp(\psi_B)$.

    \ref{it:slocc2} $\Rightarrow$ \ref{it:slocc3}:
    Let $v_A$ be a partial isometry such that $v_Av_A^*=\supp(\psi_A)$, $v_A^*v_A\ge \supp(\phi_A)$.
    Then $[\M_B v_A\Psi] = [\M_A'v_A\Psi] = \supp(v_A\psi_Av_A^*) = v_A^*v_A \ge \supp(\phi_A) = [\M_A'\Phi]$ implies that $\Phi$ is in the closure of $\M_B v_A\Psi$.
    By \cref{cor:isometrically acting pi's}, we can approximate $v_A$ with unitaries in $\M_A$ acting on $\Psi$, which implies the claim.
\end{proof}

\subsection{Nielsen's theorem}\label{sec:nielsen}

Entanglement theory for bipartite mixed states is extremely hard \cite{horodecki_quantum_2009}.
Essentially, this follows from the problem that the class of LOCC protocols is intractable.
For pure states, the story is quite different as the reduction to one-way LOCC \cite{lo_concentrating_2001} shows.
As Nielsen noted in \cite{nielsen_conditions_1999} (see also \cite{nielsen_quantum_2010,crann_state_2020}), the latter allows for a complete characterization of pure state LOCC transitions in terms of the marginal states.
Nielsen's Theorem was generalized to semifinite factors in \cite{crann_state_2020}.
In the following, we begin by showing an exact version of Nielsen's theorem for general bipartite systems.
Afterward, we characterize approximate LOCC transitions in factorial bipartite systems in terms of majorization for the Schmidt scales.

\begin{theorem}[Nielsen's Theorem]\label{thm:nielsens thm1}
    Let $(\M_A,\M_B,\H)$ be a bipartite system and let $\Psi,\Phi$ be unit vectors in $\H$. 
    Then 
    \begin{equation}\label{eq:nielsens thm1}
        \Psi\locc \Phi \quad\iff\quad \psi_{A/B} = \sum_x p_x \,u_x\phi_{A/B} u_x^*
    \end{equation}
    for a probability distribution $(p_x)$ and a collection of partial isometries $u_x\in\M_{A/B}$.
\end{theorem}

Note that the equality on the RHS ensures that the partial isometries act isometrically on $\phi_{A/B}$ (see \cref{cor:isometrically acting pi's}).

\begin{lemma}\label{lem:nielsen thm1}
    If $\Psi\locc\Phi$, there are partial isometries $u_x\in \M_A$ and a probability distribution $(p_x)$ such that
    \begin{equation*}
        \psi_A = \sum_x p_x \,u_x\phi_A u_x^*.
    \end{equation*}
\end{lemma}
\begin{proof}
    By \cref{prop:owlocc} (applied with $A$ and $B$ swapped), there are partial isometries $u_x\in\M_A$ and Kraus operators $k_x\in\M_B$ such that $k_x\Psi = p_x^{1/2}u_x\Phi$.
    Therefore, we have
    \begin{equation*}
        \sum_x p_x\, u_x\phi_A u_x^* = \sum_x \ip{k_x\Psi}{(\placeholder)|_{\M_A}k_x\Psi} = \ip{\Psi}{(\placeholder)|_{\M_A}\sum_x k_x^*k_x\Psi} = \psi_A.\qedhere
    \end{equation*}
\end{proof}

\begin{lemma}\label{lem:nielsen std}
    Let $\phi$ be a normal state on a von Neumann algebra $\M$, let $(p_x)$ be a probability distribution, and let $u_x\in\M_A$ be a collection of partial isometries such that $u_x\phi_Au_x^*$ is normalized.
    Set $\psi = \sum_x p_x u_x\phi u_x^* = \sum_x \psi_x$, $\psi_x=p_x u_x\phi u_x^*$, and set $k_x = u_x^* (D\psi_x : D\psi)_{-i/2}$.
    If we consider $\M$ in standard form, then
    \begin{equation}\label{eq:nielsen std}
        k_x j(u_x^*)\Omega_\psi = p_x^{1/2}\Omega_\phi, \qquad \sum_x k_x^*k_x = \supp(\psi).
    \end{equation}
\end{lemma}

In the Lemma, $(D\psi_x : D\psi)_{-i/2}$ denotes the analytic continuation of the Connes cocycle.

\begin{proof}
    The analytic continuation of the Connes cocycle makes sense and satisfies
    \begin{equation*}
        (D\psi_x : D\psi)_{-i/2} \Omega_\psi = \Omega_{\psi_x}=\Omega_{p_x u_x\psi u_x^*}= p_x^{1/2} u_xj(u_x)\Omega_\psi.
    \end{equation*}
    This is shown in \cite[Lem.~A.58]{hiai_quantum_2021} in the context of Haagerup $L^p$-spaces (we apply the Lemma with $p=1$).
    Applying $u_x^*j(u_x^*)$ to both sides shows the first equality in \eqref{eq:nielsen std}.
    To see the second inequality, we consider $x',y'\in\M'$, we note that $\supp(k_x^*k_x)\le \supp(\psi)$.
    Now let $x',y'\in\M'$. Then
    \begin{align*}
        \ip{x'\Omega_\psi}{\sum_x k_x^* k_x y'\Omega_\psi} &= \sum_x \ip{x'k_x\Omega_\psi}{y'k_x\Omega_\psi} \\
        &= \sum_x p_x\ip{x'j(u_x)\Omega_\phi}{y'j(u_x)\Omega_\phi} \\
        &= \sum_x p_x\ip{x'u_xj(u_x)\Omega_\phi}{y'u_xj(u_x)\Omega_\phi} \\
        &= \sum_x p_x \ip{\Omega_{u_x\phi u_x}}{x'^*y'\Omega_{u_x\phi u_x^*}}\\
        &= \Big(\sum_x p_x\, u_x\phi u_x^*\Big)'(x'^*y') 
        = \psi'(x'^*y') = \ip{x'\Omega_\psi}{ y'\Omega_\psi},
    \end{align*}
    where we used the notation $\omega' = \ip{\Omega_\omega}{(\placeholder)\Omega_\omega}|_{\M'}$.
    Thus, we have $\sum_x k_x^*k_x=\supp(\psi)$.
\end{proof}

\begin{corollary}\label{cor:nielsens thm1}
    Under the assumptions of \cref{lem:nielsen std}, let $\M$ act on a Hilbert space $\H$ and let $\Psi,\Phi\in\H$ be purifications of $\psi,\phi$.
    Let $\Psi = w'\Omega_\psi$ and $\Phi = v'\Omega_\phi$ be the polar decompositions of $\Psi$ and $\Phi$ (see \cref{cor:vector polar decomp}).
    Let $u_x'=v'j(\pi(u_x^*))w'^*\in \M'$, where $\pi$ is the representation that puts $\M$ in standard form.
    Then 
    \begin{equation*}
        k_x u_x' \Psi = p_x^{1/2}\Phi.
    \end{equation*}
\end{corollary}

\begin{proof}[Proof of \cref{thm:nielsens thm1}]
    The two implications are shown in \cref{lem:nielsen thm1,cor:nielsens thm1}, respectively.
\end{proof}

We now introduce a stabilized version of pure state LOCC transitions.
Let $(\M_A,\M_B,\H)$ be a bipartite system.
For unit vectors $\Psi,\Phi\in\H$, we say that $\Psi$ can be transformed to $\Phi$ with \emph{approximate LOCC}, if for all $\eps>0$, there exists a $\Psi'\approx_\eps\Psi$ that can be reached from $\Psi$ with LOCC.
We denote this
\begin{equation}
    \Psi\barlocc\Phi.
\end{equation}

For states $\psi,\phi$ on an arbitrary factor $\M$, we define majorization $\psi\preceq\phi$ via $\psi \in \overline{\conv}\{u\phi u^*:u\in\U(\M)\}$.
By \cref{thm:majorization}, this is in accordance with the usual definition of majorization, whenever defined.

\begin{theorem}\label{thm:nielsens thm2}
    Let $(\M_A,\M_B,\H)$ be a factorial bipartite system and let $\Psi,\Phi$ be unit vectors in $\H$. 
    Then 
    \begin{equation}\label{eq:nielsens thm2 majorization}
        \Psi\barlocc \Phi \quad\iff\quad \psi_{A/B} \preceq \phi_{A/B}.
    \end{equation}
    In particular, if $\M_A$ and $\M_B$ are semifinite, this is equivalent to $\lambda_\Psi \preceq \lambda_\Phi$, where the Schmidt scales are taken with respect to a coupled pair of \nsf traces $(\tau_A,\tau_B)$.
\end{theorem}

The majorization relation for the Schmidt scales is the usual one, i.e., $\lambda_\Psi\preceq\lambda_\Phi$ is defined as $\int_0^t\lambda_\Psi(s)\,ds \le \int_0^t\lambda_\Phi(s)\,ds$ for all $t>0$.

\begin{proof}[Proof of \cref{thm:nielsens thm2}]
    If $(\M_A,\M_B,\H)$ is a general bipartite system, then the definition of approximate LOCC transitions and \cref{thm:nielsens thm1} directly imply
    \begin{equation*}
        \Psi\barlocc\Phi \quad \iff \quad \psi_{A/B} \in \overline{\conv} \{ u\phi_{A/B} u^* : u\in \U(\M_{A/B})\}.
    \end{equation*}
    In the factorial case, the RHS is our definition of $\psi_{A/B}\preceq\phi_{A/B}$.
    The statement about the Schmidt scales follows from \cref{thm:majorization} since $\lambda_\Psi=\lambda_{\psi_A}$ and $\lambda_\Phi=\lambda_{\phi_A}$ (see \cref{sec:schmidt spectrum}).
\end{proof}

\begin{remark}
    In \cite{crann_state_2020}, only approximate LOCC transitions are considered.
    However, the definition of approximate LOCC transitions above is similar but not the same as the one considered in \cite{crann_state_2020}. 
    Their definition, which is arguably more natural, only asks that for all $\eps>0$, there is an LOCC protocol taking the pure state $\omega_\Psi=\ip\Psi{(\placeholder)\Psi}$ to a (possibly mixed) state that is $\eps$-close to $\omega_\Phi$.
    The data processing inequality show that our notion of an approximate LOCC transition implies theirs notion, but the converse is less clear.
    However, for the case of semifinite factors, their version of Nielsen's theorem \cite[Thm.~5.3]{crann_state_2020} shows that both notions are, in fact, equivalent.
    In the type $\III$ case, our notion of approximate LOCC transitions trivializes (see \cref{sec:locc-types}) and, hence, so does their notion.
    Thus, the two notions are, in fact, equivalent for factorial bipartite systems.
\end{remark}

\subsection{Pure state entanglement monotones}\label{sec:entanglement monotones}

In this subsection, we note that the usual construction of pure state entanglement monotones based on Nielsen's theorem \cite{nielsen_conditions_1999} and monotones for majorization theory extends to the case of semifinite von Neumann algebraic bipartite systems.
The basics were already discussed in \cref{sec:schmidt spectrum}, where we defined the Schmidt scale $\lambda_\Psi$ in factorial semifinite bipartite systems. 

We begin by defining entanglement monotones: 
A pure state entanglement monotone for a bipartite system $(\M_A,\M_B,\H)$ is a scalar function $E$, defined on unit vectors $\Psi\in\H$, such that
\begin{equation*}
    \Psi \barlocc \Phi \quad \implies \quad E(\Psi)\ge E(\Phi).
\end{equation*}
In \cref{sec:slocc}, we have seen that the Schmidt rank $r(\Psi)$ of a factorial semifinite bipartite system is a monotone for approximate SLOCC transition. 
Thus, it is, in particular, an entanglement monotone.

\begin{proposition}\label{prop:entanglement monotones}
    Let $(\M_A,\M_B,\H)$ be a factorial semifinite bipartite system, and let $\tau_{A/B}$ be coupled \nsf traces on $\M_{A/B}$.
    For a convex non-decreasing function $f:\RR^+\to \RR^+$ and a unit vector $\Psi\in\H$, define 
    \begin{equation}\label{eq:entanglement monotone}
        E_f(\Psi) := \int_{\RR^+} f(\lambda_\Psi(t))\,dt,
    \end{equation}
    where $\lambda_\Psi$ is the Schmidt scale of $\Psi$ (see \eqref{eq:Schmidt spectrum}).
    Then the functions $E_f$ with $f$ as above form a complete family of entanglement monotones:
    \begin{equation}
        \Psi \barlocc \Phi \quad \iff \quad E_f(\Psi)\ge E_f(\Phi)\quad \forall f.
    \end{equation}
\end{proposition}
\begin{proof}
    By \cref{thm:nielsens thm2}, the approximate LOCC transition is equivalent to majorization $\lambda_\Psi \preceq \lambda_\Phi$ of the Schmidt scales.
    Thus, the claim follows from \cref{thm:majorization}.
\end{proof}

\begin{corollary}\label{cor:entanglement monotones are monotone}
    The entanglement entropy and its $\alpha$-Renyi variants (see \cref{sec:schmidt spectrum}) are entanglement monotones.
\end{corollary}

\subsection{LOCC and the type classification of factors}\label{sec:locc-types}

In this subsection, we show that the minimal type of semifinite bipartite systems is in one-to-one correspondence with certain operational properties in entanglement theory.

Let us fix a bipartite system $(\M_A,\M_B,\H)$. We make the following definitions:
\begin{itemize}
    \item A unit vector $\Phi$ in $\K_A\ox\K_B$ can be \emph{distilled} from a unit vector $\Psi\in\H$ if there is a unit vector $\tilde\Psi\in\H$ such that
        \begin{equation*}
            \Psi\ox \ket0\ket0 \barlocc \tilde\Psi \ox\Phi,
        \end{equation*}
        for some unit vectors $\ket0\in \K_{A/B}$.
        
        \item A unit vector $\Psi\in \H$ is \emph{maximally entangled} if
        \begin{equation*}\label{eq:maximally entangled}
            \Psi\barlocc\Phi \qquad\forall \Phi,
        \end{equation*}
        where $\Phi$ ranges over unit vectors in $\H$.
        
        \item 
            The \emph{one-shot entanglement} of a unit vector $\Psi$ is the maximal log-dimension at which a (finite-dimensional) maximally entangled state can be distilled from $\Psi$.
            I.e. the one-shot entanglement is $\log_2n$, where $n$ is the largest integer, such that $\Phi_n$ can be distilled from $\Psi$, where $\Phi_n = n^{-1/2}\sum_{j=1}^n\ket j\ket j \in \CC^n\ox\CC^n$.

    \item \emph{LOCC trivializes}, if $\Psi\barlocc\Phi$ for all pairs of unit vectors, or, equivalently, if all unit vectors are maximally entangled.
\end{itemize}

If the system is finite-dimensional, our definition of the one-shot entanglement reduces to the \emph{min-entanglement entropy} $E_\oo(\Psi)$, which is defined as the $-\log \lambda_{\max}(\Psi)$, where $\lambda_{\max}(\Psi)$ denotes the largest Schmidt coefficient of $\Psi$.\footnote{This is a direct consequence of Nielsen's theorem.}
The main Theorem of this subsection is the following:

\begin{theorem}\label{thm:semifinite classification}
    Let $(\M_A,\M_B,\H)$ be a factorial bipartite system.
    Then the operational entanglement properties defined above only depend on the minimal type and are given by:
    \\[-14pt]

    \begin{equation}
        \setlength{\tabcolsep}{7pt}
        \renewcommand\arraystretch{1.2}
        \begin{tabular}{@{\,} l cc cc ccc @{}}
            \toprule
            operational property/ minimal type
            & $\I_n$       & \ $\I_\infty$   & \ $\II_1$    & \ $\II_\infty$ & $\III$  \\ 
            \midrule
            one-shot entanglement           & \!\!\! $\le$$\log_2n$\!    & \!$<$$\infty$    & $\oo$&$\oo$ &$\oo$                                \\
            maximally entangled state &\yes & \no & \yes & \no & \yes\\
            LOCC trivializes & \no&\no     &\no&\no      & \yes        \\ 
            \bottomrule
        \end{tabular}
    \end{equation}
    \null
\end{theorem}


\begin{lemma}\label{lem:locc trivial in type III}
    Let $(\M_A,\M_B,\H)$ be a bipartite system of type $\III$ factors.
    If $\Psi,\Phi\in\H$ are unit vectors, then $\Psi$ can be transformed to $\Phi$ with approximate LOCC.
\end{lemma}
\begin{proof}
    By \cite[Lem.~4.3]{hiai_closed_1991} (or \cite[Lem.~9.3]{haagerup_equivalence_1990}), we have $\psi_A\in \overline{\conv}\{u \phi_A u^* : u\in\U(\M)\}$ for all pairs $\psi_A,\phi_A$ of normal states on $\M_A$.
    Thus, \cref{thm:nielsens thm2} implies the claim.
\end{proof}

\begin{lemma}\label{lem:one-shot-ent}
    Let $(\M_A,\M_B,\H)$ be a factorial bipartite system. If $\M_A$ and $\M_B$ are not of type $\I$, then every unit vector $\Psi\in\H$ has infinite one-shot distillable entanglement.
    If $\M_A$ and $\M_B$ are of type $\I$, no unit vector $\Psi\in\H$ has infinite one-shot entanglement.
\end{lemma}

\begin{proof}
    The type $\III$ case follows from \cref{lem:locc trivial in type III} and the fact that $\M_j \ox \B(\K_j)$ is of type $\III$ whenever $\M_j$ is, $j=A,B$.
    We consider the type $\II$ case.
    We set $\M_A=\M$, and pick coupled \nsf traces $(\tau,\tau')$ on $\M$ and $\M'$.
    Clearly, it is sufficient to show that a sequence of vectors with unbounded entanglement entropy can be distilled from any given vector $\Psi$.
    We let $\Phi_n$ denote the maximally entangled state in $\CC^n\ox\CC^n$.
    By \cref{thm:nielsens thm2}, we have to show that for each $n$, there exists a unit vector $\tilde\Psi$ such that
    \begin{equation*}
        \lambda_{\Psi\ox\ket0\ket0}=\lambda_\Psi \preceq \lambda_{\tilde\Psi\ox \Phi_n}.
    \end{equation*}
    Since the marginal state of $\Phi_n$ is the tracial state on $M_n(\CC)$, we have $\lambda_{\Phi_n}(t) = n^{-1} 1_{[0,n)}(t)$. 
    With this, it can readily be checked that
    \begin{equation*}
        \lambda_{\tilde\Psi\ox \Phi_n}(t) = \frac1n \lambda_{\tilde\Psi}\bigg(\frac tn\bigg).
    \end{equation*}
    We claim that, for each $n\in\NN$, there is a unit vector $\tilde\Psi$ with $\lambda_{\tilde\Psi}(t) = n \lambda_\Psi(nt)$.
    By the above, this will show the claim.
    Since $|\supp(n\lambda_\psi(n(\placeholder)))| \le |\supp(\lambda_\psi)|\le \tau(1)$, there is a normal state $\tilde\psi$ on $\M$ with spectral scale $\lambda_{\tilde\psi}= n \lambda_\Psi(nt)$ (see \cref{sec:spectral scales}).
    The trace coupling implies $\tau(\supp(\tilde\psi)) \le \tau(\supp(\psi)) = \tau(\supp(\psi'))$, where $\psi'$ denotes the $\M'$-marginal of $\Psi$).
    Thus, \cref{lem:rank condition for purification} implies that there is a purification $\tilde\Psi\in\H$ of $\tilde\psi$. Since $\lambda_{\tilde\Psi}= \lambda_{\tilde\psi}$, this finishes the proof.

    Finiteness of the one-shot distillable entanglement for type $\I$ was shown in \cite{keyl_infinitely_2003}. 
    In fact, \cite{keyl_entanglement_2006} shows the same statement also for mixed states and allows for a larger class of maps than LOCC.
\end{proof}

The proof above also shows the following:

\begin{corollary}\label{cor:locc not trivial}
    If $(\M_A,\M_B,\H)$ is a factorial semifinite bipartite system with $\M_{A/B}\ne \CC$, there exist unit vectors $\Psi,\Phi$ such that $\Psi$ can be transformed to $\Phi$ with approximate LOCC but not vice versa.
\end{corollary}

\begin{proof}[Proof of \cref{thm:semifinite classification}]
    We show the three rows of the table separately.

    One-shot entanglement: For non-type $\I$ factors, the statement is shown in \cref{lem:one-shot-ent}.
    For systems with minimal type $\I_n$, the most entangled state that can be distilled is $\Phi_n$. Hence, the one-shot entanglement of every unit vector is $\le \log_2n$.

    Maximally entangled states: 
    In the type $\III$ case, every unit vector is maximally entangled since LOCC trivializes.
    We now consider the semifinite case. Without loss of generality, let $\M_A$ be the "smaller" of the two factors (see \cref{sec:coupling}) so that minimal type equals the type of $\M_A$.
    Thus, every normal state on $\M_A$ has a purification (see \cref{sec:coupling}). 
    Therefore, \cref{thm:nielsens thm2} implies that a unit vector $\Psi\in\H$ is maximally entangled if and only if its $\M_A$-marginal state $\psi_A$ is maximally mixed in the sense that $\psi_A \preceq \phi_A$ for all $\phi_A\in\nstates(\M_A)$.
    By \cref{thm:majorization}, a semifinite factor has a maximally mixed state if and only if it is finite, i.e., of type $\I_n$, $n\in\NN$, or of type $\II_1$.
    This shows the claim.

    LOCC trivialization: This is shown in \cref{lem:locc trivial in type III}.
\end{proof}

Instead of the triviality of LOCC, one can use LOCC-embezzlement to distinguish types $\I$ and $\II$ from type $\III$.
This will be discussed in \cref{sec:mbz} (see \cref{prop:locc-mbz}).

\section{Strong forms of infinite entanglement}\label{sec:strong-entanglement}

\localtableofcontents

\null

In this section, we study particularly strong entanglement phenomena. 
The results we obtain here allow us to distinguish different von Neumann algebraic bipartite (or multipartite) systems with operational entanglement properties, even when both systems are infinitely entangled. 

\Cref{sec:mbz} is based on the article \cite{van_luijk_embezzlement_2024} and heavily relies on the results of \cref{sec:catalytic states} (except for \cref{prop:locc-mbz}, which is taken from \cite{van_luijk_pure_2024}).
\Cref{sec:lu-transitive} generalizes results from \cite{van_luijk_multipartite_2025}, and \cref{sec:universal-mbz} generalizes results from \cite{van_luijk_embezzlement_2024}.

\subsection{Embezzlement of entanglement}\label{sec:mbz}

Embezzlement of entanglement is a phenomenon in entanglement theory that was discovered by van Dam and Hayden in \cite{van_dam_universal_2003}.
The idea of embezzlement is that an entangled resource state admits the extraction of arbitrary entangled bipartite states with local unitaries so that the resource system returns to its initial state.
More precisely, let $\Psi$ be the resource state and let us consider a finite-dimensional bipartite system in a state $\Phi_1$ with little entanglement, e.g., a product state $\Phi_1=\ket0\ket0$.
The task is then to prepare a more entangled state $\Phi_2$ with local unitaries $u_A$, $u_B$ on the joint system and without disturbing the resource state $\Psi$, i.e.,
\begin{equation}
    u_A u_B (\Psi\ox \Phi_1) = \Psi \ox \Phi_2.
\end{equation}
This appears paradoxical, since local unitaries cannot generate entanglement.
Indeed, if $\Psi$ has finite one-shot entanglement, embezzlement is impossible.
However, if the resource state has infinite one-shot entanglement, it is less clear why embezzlement should be impossible.
Indeed, we know from logic paradoxes like Hilbert's hotel \cite{gamow1988one} that the impossible can become possible in the presence of infinite resources.

Although embezzlement of entanglement is impossible with finite-dimensional resource systems, it is possible in an approximate form, where the resource state depends on the error.
Indeed, the discovery of \cite{van_dam_universal_2003} was that the sequence of bipartite states
\begin{equation*}
    \Psi_n = c_n \sum_{\alpha=1}^n \alpha^{-1/2} \, \ket\alpha\ket\alpha \in \CC^n\ox\CC^n,
\end{equation*}
where $c_n= (\sum_{\alpha=1}^n \alpha^{-1} )^{-1/2}$, admits approximate embezzlement of bipartite pure states of Schmidt rank $d$ in the sense that
\begin{equation*}
    \norm{u_Au_B (\Psi_n\ox \ket0\ket 0)-  \Psi_n\ox\Phi} \le \eps_{d,n},\qquad
    \eps_{d,n}=\left(\frac{2\log d}{\log n}\right)^{1/2},
\end{equation*}
for suitable local unitaries $u_A,u_B$.\footnote{This follows from their estimate $F(u_Au_B(\Psi_n\ox\ket0\ket0),\Psi_n\ox\Phi)^{1/2} \ge 1-\frac{\log(d)}{\log(n)}$ and the Fuchs-van de Graaf inequalities \eqref{eq:FvdG}.}
Thus, if $\Phi_1,\Phi_2$ are bipartite pure states with Schmidt rank $d$, then
\begin{equation*}
    \norm{u_Au_B(\Psi_n\ox\Phi_1)-\Psi_n\ox\Phi_2} \le 2\eps_{d,n}
\end{equation*}
for suitable local unitaries.
A sequence of entangled states $(\Psi_n)$ with the above properties for some $\eps_{d,n}$ with $\lim_n\eps_{d,n}\to0$ for all $d\in\NN$ is known as a \emph{(universal) embezzling family} \cite{leung_characteristics_2014,leung_coherent_2013}.
In finite-dimensional systems, it is impossible to perform any sort of embezzlement with an error of $\eps=0$.
In fact, the same holds in infinite-dimensional bipartite systems $\H=\H_A\ox\H_B$ in the traditional, Hilbert space-based, framework of describing quantum systems \cite{cleve_perfect_2017}.
Exact embezzlement of specific entangled states $\Phi$ has been shown to be possible with resource systems described by C*-algebras and local operations described by *-automorphisms \cite{cleve_perfect_2017,liu_non-local_2022,cleve_constant_2022}.%
\footnote{In \cite{van_luijk_embezzlement_2024}, it is shown that separability of the Hilbert space renders exact embezzlement of all bipartite pure states impossible. 
Moreover, it is shown that exact embezzlement of all bipartite pure states is possible in certain non-separable settings.}

We consider here embezzlement of entanglement with resource systems described by von Neumann algebras.
This approach was suggested in unpublished notes of Haagerup, Scholz, and Werner \cite{mbz_draft}.
As above, we restrict our attention to the bipartite case for now.

\begin{definition}
    Let $(\M_A,\M_B,\H)$ be a bipartite system. A vector $\Psi\in\H$ is \emph{embezzling} if
    \begin{equation}
        \Psi \ox \Phi_1 \lu \Psi\ox \Phi_2
    \end{equation}
    for all Hilbert spaces $\K_A,\K_B$ and all unit vectors $\Phi_1,\Phi_2\in \K_A\ox\K_B$, where LU is considered relative to the bipartite system $(\M_A\ox\B(\K_A)\ox1,\M_B\ox1\ox\B(\K_B),\H\ox\K_A\ox\K_B)$.
\end{definition}

As a consequence of \cref{cor:local unitary equivalence}, we observe that a bipartite pure state $\Psi$ is embezzling if and only if one, hence both, of its marginal states are universally catalytic (see \cref{sec:catalytic states}).

\begin{proposition}\label{prop:mbz vs catalysis}
    Let $(\M_A,\M_B,\H)$ be a bipartite system and let $\Psi\in\H$. The following are equivalent:
    \begin{enumerate}[(a)]
        \item\label{it:mbz vs catalysis1} $\Psi$ is embezzling;
        \item\label{it:mbz vs catalysis2} $\psi_A$ is a universally catalytic state on $\M_A$;
        \item\label{it:mbz vs catalysis3} $\psi_B$ is a universally catalytic state on $\M_B$.
    \end{enumerate}
\end{proposition}

This observation, although based on a more pedestrian argument, was already made in \cite{mbz_draft}.
On a factor $\M$, all universally catalytic states are approximately unitarily equivalent (see \cref{cor:catalysts are unitarily equivalent}). This implies:

\begin{corollary}\label{cor:lu equivalence of mbz states}
    Let $(\M_A,\M_B,\H)$ be a bipartite system.
    If $\Psi_1\in\H$ is embezzling and $\Psi_2\lu\Psi_1$, then $\Psi_2$ is embezzling as well.
    If $\M_A,\M_B$ factors on $\H$, then the converse holds:
    If $\Psi_1,\Psi_2$ are embezzling vectors with the same norm, then $\Psi_1\lu\Psi_2$.
\end{corollary}

As a consequence of \cref{eq:local distinguishability norm}, we learn that the marginals $\psi_A$ and $\psi_B$ of a bipartite pure state $\Psi$ are equally good or bad at the task of universal catalysis:
\begin{equation}\label{eq:kappa entanglement}
    \kappa(\Psi):=\kappa(\psi_A) = \kappa(\psi_B), \qquad \Psi\in\H.
\end{equation}
By \cref{cor:lu equivalence of mbz states}, $\kappa(\Psi)$ quantifies the optimal performance of the state $\Psi$ at the task of embezzling entanglement from it (see \cref{rem:quantification with f} for a discussion of the analogous quantifier in fidelity).
This quantification is operational in terms of local distinguishability, i.e., $\kappa(\Psi)$ is precisely the worst-case error for the detection of whether the embezzlement of arbitrary states $\Phi_1,\Phi_2$ was successful.
As a consequence of \eqref{eq:kappa entanglement}, the worst worst-case error among all resource states of the bipartite system is $\kappa_{\min}(\M_A,\M_B,\H) := \inf \,\kappa(\Psi)$, where the infimum is over unit vectors $\Psi\in\H$, equals the corresponding local quantities:
\begin{equation}
     \kappa_{\min}(\M_A,\M_B,\H) = \kappa_{\min}(\M_A) = \kappa_{\min}(\M_B).
\end{equation}
Similarly, we have
\begin{equation}
    \kappa_{\max}(\M_A,\M_B,\H) = \kappa_{\max}(\M_A) = \kappa_{\max}(\M_B),
\end{equation}
where $\kappa_{\max}(\M_A,\M_B,\H)= \sup \kappa(\Psi)$, where the supremum is over unit vectors $\Psi\in\H$.

The above allows us to fully characterize the invariants $\kappa_{\min}$ and $\kappa_{\max}$ for bipartite systems of factors, based on \cref{thm:value of kappa min,thm:value of kappa max}:

\begin{theorem}
    Let $(\M_A,\M_B,\H)$ be a bipartite system of factors.
    \begin{enumerate}
        \item The value of $\kappa_{\min}(\M_A,\M_B,\H)$ is either $0$ or $2$. 
            The following are equivalent:
            \begin{enumerate}[(a)]
                \item there is an embezzling vector $0\ne \Psi\in\H$;
                \item $\kappa_{\min}(\M_A,\M_B,\H)=0$;
                \item the flows of weights of $\M_A$ and, therefore, $\M_B$ admit invariant normal states.
            \end{enumerate}
            These equivalent properties hold for all factors of type $\III_{\lambda}$, $0<\lambda\le 1$. They are false in the semifinite case.
    
        \item 
            The value of $\kappa_{\max}(\M_A,\M_B,\H)$ equals the diameter of the state space of $\M_A$ (and, therefore, $\M_B$) up to unitary equivalence.
            If $\M_A$ and, therefore, $\M_B$ are of type $\III$, this gives 
            \begin{equation}
                \kappa_{\max}(\M_A,\M_B,\H) = 2 \frac{1-\sqrt\lambda}{1+\sqrt\lambda}
            \end{equation}
            where  $\lambda$ labels the subtype (i.e., $\lambda$ is such that $\M_A$ and, hence, $\M_B$ are of type $\III_\lambda$).
    \end{enumerate}
    The first item also holds if $\M_A$ and, therefore, $\M_B$ are not factors but general von Neumann algebras.
\end{theorem}

Together with the results in \cref{sec:locc-types}, this proves \cref{thm:minimal-type_i} in \cref{sec:main-results}.

\begin{remark}[Quantiying embezzlement in fidelity]\label{rem:quantification with f}
    Let $(\M_A,\M_B,\H)$ be a bipartite system.
    In the above, we quantify the capability of a unit vector $\Psi\in\H$ at the task of embezzling entanglement in terms of the local distinguishability in norm.
    An alternative quantification, which is arguably more natural in the context of bipartite pure state entanglement theory, would be in terms of the fidelity 
    \begin{equation}
        \vartheta(\Psi) = \adjustlimits \inf_{\Phi_1,\Phi_2}\sup_{u_A,u_B} \,\big|\ip{\Psi\ox \Phi_1}{u_Au_B(\Psi\ox \Phi_2)}\big|^2,
    \end{equation}
    where $\Phi_1,\Phi_2$ are unit vectors in $\K_A\ox\K_B$ for infinite separable Hilbert spaces $\K_A,\K_B$, and where $u_{A/B}$ are local unitaries from the bipartite system $(\M_A\ox\B(\K_A)\ox1,\M_B\ox 1\ox\B(\K_B),\H_A\ox\K_A\ox\K_B)$.
    This quantification has the benefit that it measures directly how close the initial and final bipartite pure states are. 
    For a normal state $\psi$ on a von Neumann algebra $\M$, let us denote by $\vartheta(\psi)$ the analog of $\kappa(\psi)$ in fidelity:
    \begin{equation}
        \vartheta(\psi) = \adjustlimits \inf_{\phi_1,\phi_2} \sup_{u} F(u (\psi\ox\phi_1))u^* , \psi\ox \phi_2),
    \end{equation}
    where the infimum is over pairs of normal states on a type $\I_\oo$ factor and the supremum is over unitaries in $\M$.
    By \cref{prop:local distinguishability}, we then have 
    \begin{equation}
        \vartheta(\Psi) = \vartheta(\psi_A) = \vartheta(\psi_B).
    \end{equation}
    It is a consequence of the Fuchs-van de Graaf inequalities \eqref{eq:FvdG} that $\vartheta(\Psi)\approx 1$ if and only if $\kappa(\Psi)\approx 0$ is embezzling and that $\vartheta(\Psi)\approx 0$ if and only if $\kappa(\Psi)\approx 2$.
    Thus, the operational quantifications of embezzlement in terms of $\kappa$ and $\vartheta$ agree in the two extremal cases.
    As with $\kappa$, we can define $\vartheta_{\max}$ and $\vartheta_{\min}$, which only depend on the bipartite system.
    The above and our results on $\kappa$ imply that $\vartheta_{\max}$ is either $0$ or $1$ with the value $0$ (resp.\ $1$) occurring if and only if $\kappa_{\min}= 2$ (resp.\ $0$).
    It would be most interesting to work out the value of $\vartheta_{\min}$ for facts of type $\III_\lambda$.
    The problem that keeps us from applying our methods to $\vartheta_{\min}$ is that the machinery of Haagerup-Størmer spectral states is, at this point, only available for the norm distance of unitary orbits.
    Extending this machinery to the fidelity and possibly other quantum divergences (see \cref{sec:notes-vNQI}) would be of independent interest.
\end{remark}

Finally, we comment on a variant of embezzlement of entanglement, where the agents are allowed to communicate classically. 
We refer to this as \emph{LOCC-embezzlement}.
LOCC-\allowbreak embezzlement in finite-dimensional systems was studied in \cite{zanoni_complete_2024}.
We show that LOCC-\allowbreak embez\-zle\-ment is possible if and only if the bipartite system is type $\III$, in which case every unit vector is an LOCC-embezzling state:

\begin{proposition}\label{prop:locc-mbz}
    Let $(\M_A,\M_B,\H)$ be a factorial bipartite system, and let $\Psi\in\H$ be a unit vector. 
    The following are equivalent:
    \begin{enumerate}[(a)]
        \item For all unit vectors $\Phi_1,\Phi_2 \in \CC^n\ox\CC^n$, we have
        \begin{equation}\label{eq:locc-mbz}
            \Psi\ox\Phi_1 \barlocc \Psi\ox\Phi_2,
        \end{equation}
        \item $\M_A$ and $\M_B$ are of type $\III$.
    \end{enumerate}
\end{proposition}

\begin{proof}
    Suppose $\M_A$ and $\M_B$ are type $\III$.
    Then the same is true for $\M_A\ox M_n(\CC)$ and $\M_B\ox M_n(\CC)$.
    Hence, \cref{lem:locc trivial in type III} implies \eqref{eq:locc-mbz} for all $\Phi_1,\Phi_2$ if $\M_A$.
    
    To see the converse, we assume that $\M_A$ and $\M_B$ are semifinite and that \eqref{eq:locc-mbz} holds for all $\Phi_1,\Phi_2$.
    In particular, \eqref{eq:locc-mbz} holds for $\Phi_1=\ket 0\ket0$ and $\Phi_2 = n^{-1/2}\sum_{j=0}^{n-1}\ket j\ket j$.
    Following the arguments in the proof of \cref{lem:one-shot-ent}, this implies $\lambda_\Psi \preceq \lambda_{\Psi\ox \Phi_n}$,  $\lambda_{\Psi\ox \Phi_n}(t)= n^{-1} \lambda_\Psi(t/n)$ for all $n\in\NN$.
    Since the converse majorization holds trivially, this shows
    \begin{equation*}
        \lambda_\Psi(t) = \frac1n\lambda_\Psi\left(\frac tn\right).
    \end{equation*}
    Using this equation for two numbers $n,m\in\NN$, we get 
    \begin{equation*}
        \lambda_\Psi(t) = \frac mn\lambda_\Psi\left(\frac {mt}n\right).
    \end{equation*}
    Thus, right-continuity of the Schmidt spectral scale, which follows from right-continuity of the spectral scale of states semifinite factors (see \cref{sec:spectral scales}), implies $\lambda_\Psi(t) = c\lambda_\Psi(ct)$ for all $c>0$.
    Hence, $\lambda_\Psi(t) = t^{-1} \lambda_\Psi(1)$, which contradicts $\lambda_\Psi\in L^1(\RR^+)$, or $\lambda_\Psi =0$, which contradicts $\Psi \ne0$.
\end{proof}

\subsection{LU transitive multipartite systems}\label{sec:lu-transitive}

In the following, we consider multipartite systems in which all pure states are approximately equivalent up to local unitaries (approximately LU equivalent).
We first consider such systems abstractly and fully characterize the bipartite case.
Afterward, we use a multipartite embezzling family from \cite{leung_coherent_2013} to construct nontrivial examples.
Consider a collection $(\M_x)_{x=1}^N$, $N\in \NN$, of pairwise commuting von Neumann algebras on a Hilbert space $\H$.
We say that a pair of unit vectors $\Psi,\Phi\in \H$ is \emph{approximately LU equivalent}, denoted 
\begin{equation}
    \Psi\lu\Phi,
\end{equation}
if for all $\eps>0$, there exist unitaries $u_x\in\U(\M_x)$ such that $\norm{\Psi-\prod_x u_x \Phi}<\eps$.
We say that the collection $(\M_x)_{x=1}^N$ is \emph{LU transitive} if all pairs of unit vectors are approximately LU equivalent.
In the following, we consider a collection $(\M_x)_{x=1}^N$ as above, which we assume to be \emph{nontrivial} in the sense that $N\ge 2$ and $\M_x\ne \CC1$, $x=1,\ldots N$.

\begin{lemma}\label{lem:lu-transitive}
    If $(\M_x)_{x=1}^N$ is LU transitive, then  each $\M_x$ is a type $\III_1$ factor.
\end{lemma}
\begin{proof}
    Since the center $Z(\M_x)$ commutes with the unitary groups $\U(\M_y)$ for all $x,y$, all pairs of approximately LU equivalent unit vectors induce the same states on each of the centers $Z(\M_x)$.
    It is easy to see that if a von Neumann algebra $\M$ acts on a Hilbert space $\H$, then $\H$ contains unit vectors that induce different states on the center $Z(\M)$, unless the center is trivial.
    Thus, LU transitivity implies that each $\M_x$ is a factor.
    To see that each $\M_x$ is of type $\III_1$, we note that LU transitivity implies that all pairs of normal states $\psi,\phi$ on $\M_x$ that admit purifications in $\H$ are approximately unitarily equivalent.
    This is clearly false for semifinite factors other than $\M = \B(\H)$. If $\M_x=\B(\H)$, we have $\M_y=\CC1$ for $y\ne x$, which we excluded. Thus, each $\M$ is a type $\III$ factor.
    Since type $\III$ factors are always in standard representation (\cref{lem:std if infinite}), we know that all normal states on $\M$ have purifications in $\H$ (see \cref{sec:std form}).
    Thus, the assumption of LU transitivity implies that all normal states on $\M$ are approximately unitarily equivalent, which is the case if and only if $\M$ is a type $\III_1$ factor \cite{connes_homogeneity_1978} (see \cref{sec:von-neumann-algs}).
\end{proof}

\begin{proposition}\label{prop:lu-transitive}
    Let $(\M_A,\M_B)$ be a pair of commuting von Neumann algebras on $\H$.
    Assume that $\M_A\ne\CC$ and $\M_B\ne \CC$.
    The following are equivalent:
    \begin{enumerate}[(a)]
        \item\label{it:lu-transitive1} $(\M_A,\M_B)$ are LU transitive,
        \item\label{it:lu-transitive2} $\M_A$ and $\M_B$ are type $\III_1$ factors in Haag duality: $\M_A=\M_B'$.
    \end{enumerate}
\end{proposition}
\begin{proof}
    \ref{it:lu-transitive1} $\Rightarrow$ \ref{it:lu-transitive2}:
    Haag duality $\M_A=\M_B'$ follows from \cref{thm:haag_duality}.
    The type $\III_1$ factor property is shown in \cref{lem:lu-transitive}.
    \ref{it:lu-transitive2} $\Rightarrow$ \ref{it:lu-transitive1}:
    Let $\Psi,\Phi\in\H$ be unit vectors. By \cref{cor:local unitary equivalence}, $\Psi$ and $\Phi$ are approximately LU equivalent if and only if the marginal states $\psi_A$ and $\phi_A$ on $\M_A$ are approximately unitarily equivalent.
    This is the case since $\M_A$ is a type $\III_1$ factor \cite{connes_homogeneity_1978} (see \cref{sec:von-neumann-algs}).
\end{proof}

In the following, we give a construction of nontrivial LU transitive $N$-partite systems based on the embezzling family of Leung, Toner and Watrous (LTW) \cite{leung_coherent_2013}, which is the only known example of a multipartite embezzling family, and the idea of ITPFI multipartite systems discussed in \cref{sec:idealized}.

\begin{proposition}\label{prop:ltw-system}
    For each $2\le N\in \NN$, there exists an LU transitive collection $(\M_x)_{x=1}^N$ of $N$ pairwise commuting type $\III_1$ factors on a Hilbert space $\H$.
\end{proposition}

\begin{proof}
    We use the notation from \cite{van_luijk_multipartite_2025} to which we refer for details on the LTW construction.
    The construction yields a sequence of unit vectors $\Omega_k \in (\CC^{D_k})^{\ox N}$ in $N$-partite systems with local dimension $D_k$, $k\in \NN$, which have the following embezzlement property:
    For all $d$, unit vectors $\Psi,\Phi$ in $(\CC^d)^{\ox N}$ and $\eps>0$, there exists some index $k$ and a local unitary $u=\prod_{x=1}^Nu_x$ on $(\CC^{D_k})^{\ox N}$ such that
    \begin{equation}\label{eq:ltw-mbz-prop}
        \norm{\Omega_k \ox \Psi - u (\Omega_k\ox \Phi)}< \eps.
    \end{equation}
    The LTW family is the sequence of tensor products $\Omega\up n = \otimes_{k=1}^n \Omega_k$, $n\in\NN$.
    By \eqref{eq:ltw-mbz-prop}, the sequence $\Omega\up n$ can be used to embezzle arbitrary $N$-partite state up to arbitrarily small error.
    In the following, we will use no other property of the LTW family than the embezzling property \eqref{eq:ltw-mbz-prop}.
    The tensor product form of the sequence $\Omega\up n$ allows us to take the limit $n\to \oo$.
    As in \cref{sec:idealized}, we consider the restricted infinite Hilbert space tensor product 
    \begin{equation}\label{eq:ltw-itp}
        \H = \bigotimes_{n\in\NN} \,\big((\CC^{D_k})^{\ox N}; \Omega_k\big)
    \end{equation}
    on which we have $N$ commuting factors $\M_1,\ldots \M_N\subset \B(\H)$ corresponding to the $N$ agents.
    By construction, the vector $\Omega = \otimes_{k=1}^\oo \Omega_k$ is embezzling in the sense that for all $d\in\NN$, and all unit vectors $\Psi,\Phi\in (\CC^d)^{\ox N}$, we have
    \begin{equation}\label{eq:ltw-mbz-limit-prop}
        \Omega\ox\Psi \lu \Omega\ox\Phi,
    \end{equation}
    where the local unitaries are taken from the factors $\M_x \otimes (1^{\ox x-1}\ox M_{D_k}(\CC)\ox 1^{\ox N-x})$ on $\H\ox (\CC^d)^{\ox N}$.
    For each $k\in \NN$, we can consider the \emph{tail} Hilbert space $\H_{\ge k}$ formed as in \eqref{eq:ltw-itp} but restricting to integers $n\ge k$.
    On the tail Hilbert space, we have `tail' factors $\M_{x, \ge k}$ with $\M_{x} \cong M_{D_1}(\CC)\ox \ldots\ox M_{d_{k-1}}(\CC) \ox \M_{x, \ge k}$.
    Since \eqref{eq:ltw-mbz-prop} requires arbitrarily large indices for small errors $\eps>0$, the `tail' sequence $\{\Omega_l\}_{l\ge k}$ inherits the approximate embezzlement property.
    Thus, the unit vector $\Omega_{\ge k} = \otimes_{l\ge k} \Omega_l$ in $\H_{\ge k}$ is an embezzling vector in the sense of \eqref{eq:ltw-mbz-limit-prop}.
    We use this fact to show that the multipartite system $(\M_x)_{x=1}^N$ on $\H$ is LU transitive.
    Let $\Psi, \Phi\in \H$ and $\eps>0$ be given.
    Since $\H$ is a restricted infinite tensor product, there are an index $k$ and unit vectors $\Psi_{<k},\Phi_{<k}\in \bigotimes_{l=1}^{k-1} (\CC^{D_l})^{\ox N}$ such that 
    \begin{equation*}
        \Psi \approx_\eps \Psi_{<k}\ox \Omega_{\ge k}, \qquad \Phi \approx_\eps \Phi_{<k}\ox \Omega_{\ge k}.
    \end{equation*}
    The fact that $\Omega_{\ge k}$ is embezzling now implies 
    \begin{equation*}
        \Psi_{<k}\ox \Omega_{\ge k} \lu \Phi_{<k}\ox \Omega_{\ge k}.
    \end{equation*}
    Since this works for all $\eps>0$, it follows that $\Psi$ and $\Phi$ are approximately LU equivalent.
\end{proof}

\subsection{Universal embezzlers}\label{sec:universal-mbz}

A \emph{universal embezzler} is a system in which all pure states are embezzling states.
We begin by considering multipartite systems, showing that every LU transitive system is a multipartite embezzler.
Afterward, we focus on the bipartite case, where we show that a bipartite system is a universal embezzler if and only if it is of type $\III_1$.
As a consequence, we know that (bipartite) universal embezzlers exist in well-studied physical models, e.g., in conformal field theory \cite{gabbiani_operator_1993}, relativistic quantum field theory \cite{haag_local_1996,van_luijk_relativistic_2024}, critical spin chains \cite{keyl_entanglement_2006,van_luijk_critical_2025}, and critical fermion chains \cite{van_luijk_critical_2025}.

Consider a nontrivial (see \cref{sec:lu-transitive}) collection $(\M_x)_{x=1}^N$ of commuting von Neumann algebras on a Hilbert space $\H$.
A unit vector $\Psi\in \H$ is \emph{embezzling} if for all $d\in \NN$ and all pairs of unit vectors $\Phi_1,\Phi_2\in (\CC^d)^{\ox N}$, the vectors $\Psi\ox\Phi_1$ and $\Psi\ox\Phi_2$ are approximately LU equivalent with respect to joint the $N$-partite system on $\H\ox (\CC^d)^{\ox N}$.
The collection $(\M_x)_{x=1}^N$ is a \emph{universal embezzler} if every unit vector $\Psi\in\H$ is embezzling.

\begin{lemma}\label{lem:lu-transitive-mbz}
    If $(\M_x)_{x=1}^N$ is LU transitive and nontrivial, then it is a universal embezzler.
\end{lemma}
\begin{proof}
    By \cref{lem:lu-transitive}, each $\M_x$ is an infinite factor.
    Thus, for each $d\in \NN$, we can find an $\M_x$-unitary $u_x:\H\to\H\ox\CC^d$.
    Thus, if $(\N_x)_{x=1}^N$ denotes an finite dimensional $N$-partite system with Hilbert space $\K = (\CC^d)^{\ox N}$, then there is a local product $u = \prod_x u_x : \H \to \H \ox (\CC^d)^{\ox N}$ with $u\M_x u^* = \M_x\ox \N_x$.
    Thus, $(\M_x\ox\N_x)_{x=1}^N$ on $\H\ox(\CC^d)^{\ox N}$ is LU transitive as well.
    Consequently, if $\Phi_1,\Phi_2\in (\CC^d)^{\ox N}$ are unit vectors, then $\Psi\ox\Phi_1$ and $\Psi\ox\Phi_2$ are approximately LU equivalent.
\end{proof}

As a direct consequence of \cref{lem:lu-transitive-mbz,prop:ltw-system}, we obtain the existence of multipartite universal embezzlers:

\begin{corollary}
    For each $N\ge 2$, there exists a collection $(\M_x)_{x=1}^N$ of pairwise commuting factors, which is a universal embezzler. 
\end{corollary}

The converse to \cref{lem:lu-transitive-mbz} is false.
Indeed, LU transitivity requires that each $\M_x$ is a factor, which is not required for universal embezzlers.
This follows from the following structure result, which was conjectured in \cite{mbz_draft}, describing those bipartite systems (this assumes Haag duality) that are universal embezzlers.

\begin{proposition}\label{prop:bipartite-univ-mbz}
    A bipartite system $(\M_A,\M_B,\H)$ is a universal embezzler if and only if $\M_A$ and, therefore, $\M_B$ are von Neumann algebras of type $\III_1$.
\end{proposition}
\begin{proof}
    Recall that a unit vector $\Psi\in\H$ is embezzling if and only if $\kappa(\Psi)=0$ (see \cref{sec:mbz}). Thus, $(\M_A,\M_B)$ is a universal embezzler if and only if $\kappa(\Psi)=0$ for all $\Psi\in \H$.
    Thus, by \eqref{eq:kappa entanglement}, the assertion follows directly from \cref{prop:universal-catalyst}, which is the corresponding statement about catalytic states on a single von Neumann algebra.
\end{proof}

\Cref{prop:bipartite-univ-mbz,prop:lu-transitive} imply \cref{thm:univ-mbz_i} in \cref{sec:main-results}.

\begin{remark}
    In the bipartite case, we know that LU transitivity implies Haag duality (see \cref{prop:lu-transitive}).
    It is an interesting open question whether bipartite universal embezzlers necessarily satisfy Haag duality.
\end{remark}

\begin{remark}
    The proof of \cref{prop:ltw-system} shows the following partial converse to \cref{lem:lu-transitive-mbz}:
    If $(\M_x)_{x=1}^N$ is an ITPFI multipartite system (see \cref{sec:idealized}) on $\H$, which is a universal embezzler, then $(\M_x)_{x=1}^N$ is LU transitive.
    It would be pleasing to get rid of the ITPFI assumption.
\end{remark}

\clearpage
\phantomsection
\addcontentsline{toc}{section}{References}
\emergencystretch=1em
\printbibliography

\end{document}

%% file: preamble.tex
\usepackage[utf8]{inputenc}
\usepackage[T1]{fontenc}
\usepackage[english]{babel}
\usepackage{tikz,tikz-cd}
\usepackage{subcaption}
\usepackage{pgfplots}\pgfplotsset{compat=1.18}
\usepackage{amsmath,amssymb,amsthm}
\usepackage{mathrsfs} 
\usepackage{graphicx} 
\usepackage{ifthen}
\usepackage[a4paper,
            bindingoffset=5mm,
            margin=2.4cm,
            ]{geometry}
\usepackage[shortlabels]{enumitem}
\usepackage{tikz,tikz-cd}
\usepackage{mathtools}
\usepackage{apptools}
\usepackage{todonotes}
\usepackage{hang}
\usepackage{color}
\usepackage{etoc}
\usepackage{graphicx}
    \graphicspath{{figures/}}
\usepackage[bibstyle=numeric, backend=biber, sorting=none, citestyle=numeric-comp, giveninits, url=false, maxbibnames=5, isbn=false]{biblatex} 
\usepackage{csquotes}\MakeOuterQuote"
\usepackage{soul}\setstcolor{red}   
\usepackage{booktabs,multirow}
\usepackage{hyperref}
\usepackage{bookmark}
\usepackage{cleveref}
\usepackage{comment}

\newboolean{printversion}
\setboolean{printversion}{false}

\ifthenelse{\boolean{printversion}}{
    \hypersetup{allcolors=black,colorlinks,pdfencoding=auto,psdextra} 
    }{
    \hypersetup{colorlinks,citecolor=blue,linkcolor=blue,pdfencoding=auto,psdextra} 
    }

\ifthenelse{\boolean{printversion}}{
    \pretocmd{\section}{\cleardoublepage}{}{}
}{
    \pretocmd{\section}{\clearpage}{}{}
}

\ifthenelse{\boolean{printversion}}{
    \geometry{a4paper, bindingoffset=5mm, margin=2.5cm,bottom=2.65cm} 
}{
    \geometry{a4paper, bindingoffset=0mm, margin = 2.65cm}
}

\renewbibmacro{in:}{%
  \ifentrytype{article}{}{\printtext{\bibstring{in}\intitlepunct}}}
\AtEveryBibitem{\clearfield{month}}
\AtEveryBibitem{\clearfield{day}}

\newtheorem{theorem}{Theorem}
\newtheorem{lemma}[theorem]{Lemma}
\newtheorem{corollary}[theorem]{Corollary}
\newtheorem{proposition}[theorem]{Proposition}

\newtheorem{definition}[theorem]{Definition}
\newtheorem{introtheorem}{Theorem} 

\newtheorem*{convention*}{Convention}

\newtheorem*{problem*}{Problem}
\newtheorem*{assumption*}{Assumption}
\theoremstyle{definition}
\newtheorem{example}[theorem]{Example}
\newtheorem{remark}[theorem]{Remark}
\newtheorem*{warning*}{Warning}

\crefname{theorem}{Thm.}{Thms.}
\crefname{introtheorem}{Thm.}{Thms.}
\Crefname{theorem}{Theorem}{Theorems}
\Crefname{introtheorem}{Theorem}{Theorems}
\crefname{proposition}{Prop.}{Props.}
\Crefname{proposition}{Proposition}{Propositions}
\crefname{lemma}{Lem.}{Lemmas}
\Crefname{lemma}{Lemma}{Lemmas}
\crefname{definition}{Def.}{Defs.}
\Crefname{definition}{Definition}{Defs.}
\crefname{corollary}{Cor.}{Cors.}
\Crefname{corollary}{Corollary}{Corollaries}
\crefname{equation}{eq.}{eqs.}
\Crefname{equation}{Equation}{Equations}
\crefname{figure}{fig.}{figs.}
\Crefname{figure}{Figure}{Figures}
\crefname{section}{Sec.}{Secs.}
\Crefname{section}{Section}{Sections}

\counterwithin{theorem}{section}\counterwithin{equation}{section}

\newcommand{\ip}[2]{\langle #1,#2\rangle}

\newcommand{\ketbra}[2]{|#1\rangle\langle#2|}
\newcommand{\ket}[1]{|#1\rangle}
\newcommand{\kettbra}[1]{\ketbra{#1}{#1}}
\newcommand{\bra}[1]{\langle#1|}
\DeclareMathOperator{\supp}{supp}
\newcommand{\norm}[1]{\lVert #1\rVert}
\newcommand{\oo}{\infty}
\newcommand{\ox}{\otimes}
\newcommand{\barox}{\mathbin{\overline\ox}}
\newcommand{\mc}{\mathcal}
\newcommand{\eps}{\varepsilon}
\newcommand{\III}{{\mathrm{III}}}
\newcommand{\II}{{\mathrm{II}}}
\newcommand{\I}{{\mathrm{I}}}

\DeclareMathOperator{\Aut}{Aut}
\DeclareMathOperator{\Out}{Out}
\DeclareMathOperator{\Inn}{Inn}
\newcommand{\abs}[1]{\lvert #1 \rvert}
\newcommand{\up}[1]{^{(#1)}}

\DeclareMathOperator{\tr}{Tr}
\DeclareMathOperator{\Tr}{Tr} 

\renewcommand{\Tilde}{\widetilde}
\renewcommand{\Hat}{\widehat}

\newcommand{\hide}[1]{}

\DeclareUnicodeCharacter{0141}{\L{}}

\def\A{{\mc A}}

\def\B{{\mc B}}
\def\CC{{\mathbb C}}

\def\H{{\mc H}}
\def\K{{\mathcal K}}

\renewcommand{\L}{{\mathcal{L}}}
\def\M{{\mc M}}
\def\N{{\mc N}}
\def\O{{\mc O}}

\def\T{{\mc T}}
\def\U{{\mc U}}
\def\V{{\mc V}}

\def\RR{{\mathbb R}}

\def\NN{{\mathbb N}}
\renewcommand\P{\mc P}
\def\ZZ{{\mathbb Z}}
\def\QQ{{\mathbb Q}}
\newcommand{\R}{\mc R}

\newcommand\nsf{normal semifinite faithful\ }

\newcommand{\states}{\mathfrak S}
\newcommand\nstates{\mathfrak{S}_{*}\hspace{-.6pt}}
\newcommand\predualB{\B_*\hspace{-.6pt}}

\def\locc{\xrightarrow{\LOCC}}
\def\barlocc{\xrightarrow{\overline\LOCC}}
\def\slocc{\xrightarrow{\SLOCC}}
\def\barslocc{\xrightarrow{\overline\SLOCC}}
\def\aff{\ \tilde{\hspace{-.5pt}\in\hspace{.5pt}}\,\;\! }
\def\acts{\curvearrowright}

\DeclareMathOperator{\lin}{span}
\DeclareMathOperator{\id}{id}

\DeclareMathOperator{\Ad}{Ad}

\DeclareMathOperator{\conv}{conv}
\DeclareMathOperator{\proj}{Proj}
\DeclareMathOperator{\diam}{diam}

\newcommand{\LOCC}{\mathrm{LOCC}}
\newcommand{\SLOCC}{\mathrm{SLOCC}}

\newcommand{\placeholder}[0]{{\,\cdot\,}}

\renewcommand{\Re}{\mathrm{Re}}

\newcommand{\HS}{\mathrm{HS}}

\DeclareMathOperator{\LHS}{LHS}
\DeclareMathOperator{\RHS}{RHS}
\DeclareMathOperator{\per}{per}

\renewcommand{\liminf}{\varliminf}

\newcommand{\lvl}[1]{{\color{magenta}#1}}

\newcommand{\qandq}{\quad\text{and}\quad}

\def\locc{\xrightarrow{\LOCC}}
\def\lu{\mathrel{\,\xleftrightarrow{\ \overline{\mathrm{LU}}\ }}\,}
\def\barlocc{\xrightarrow{\overline\LOCC}}
\def\slocc{\xrightarrow{\SLOCC}}
\def\barslocc{\xrightarrow{\overline\SLOCC}}

\usepackage{tikz}
\newcommand{\no}{%
\tikz[scale=0.23] {
    \draw[line width=0.7,line cap=round] (0.0,0.05) to [bend left=4] (.9,1);
    \draw[line width=0.7,line cap=round] (0.1,0.95) to [bend right=2] (0.8,0.05);
}}
\newcommand{\yes}{%
\tikz[scale=0.23] {
    \draw[line width=0.7,line cap=round] (0.25,0) to [bend left=10] (1,1);
    \draw[line width=0.8,line cap=round] (0,0.35) to [bend right=1] (0.23,0);
}}

%% file: types.eps_tex
\begingroup%
  \makeatletter%
  \providecommand\color[2][]{%
    \errmessage{(Inkscape) Color is used for the text in Inkscape, but the package 'color.sty' is not loaded}%
    \renewcommand\color[2][]{}%
  }%
  \providecommand\transparent[1]{%
    \errmessage{(Inkscape) Transparency is used (non-zero) for the text in Inkscape, but the package 'transparent.sty' is not loaded}%
    \renewcommand\transparent[1]{}%
  }%
  \providecommand\rotatebox[2]{#2}%
  \newcommand*\fsize{\dimexpr\f@size pt\relax}%
  \newcommand*\lineheight[1]{\fontsize{\fsize}{#1\fsize}\selectfont}%
  \ifx\svgwidth\undefined%
    \setlength{\unitlength}{56.69291339bp}%
    \ifx\svgscale\undefined%
      \relax%
    \else%
      \setlength{\unitlength}{\unitlength * \real{\svgscale}}%
    \fi%
  \else%
    \setlength{\unitlength}{\svgwidth}%
  \fi%
  \global\let\svgwidth\undefined%
  \global\let\svgscale\undefined%
  \makeatother%
  \begin{picture}(1,0.54000001)%
    \lineheight{1}%
    \setlength\tabcolsep{0pt}%
    \put(0,0){\includegraphics[width=\unitlength]{types.eps}}%
    \put(0.40905142,0.45262697){\color[rgb]{0,0,0}\makebox(0,0)[lt]{\lineheight{1.25}\smash{\begin{tabular}[t]{l}$\I$\end{tabular}}}}%
    \put(0.75247002,0.45263419){\color[rgb]{0,0,0}\makebox(0,0)[lt]{\lineheight{1.25}\smash{\begin{tabular}[t]{l}$\I_n\quad\, (n\in\NN\cup\{\oo\})$ \end{tabular}}}}%
    \put(0.4049216,0.26409802){\color[rgb]{0,0,0}\makebox(0,0)[lt]{\lineheight{1.25}\smash{\begin{tabular}[t]{l}$\II$\end{tabular}}}}%
    \put(0.39493018,0.06895012){\color[rgb]{0,0,0}\makebox(0,0)[lt]{\lineheight{1.25}\smash{\begin{tabular}[t]{l}$\III$\end{tabular}}}}%
    \put(0.74219771,0.29131544){\color[rgb]{0,0,0}\makebox(0,0)[lt]{\lineheight{1.25}\smash{\begin{tabular}[t]{l}$\II_1$\end{tabular}}}}%
    \put(0.74219771,0.23310149){\color[rgb]{0,0,0}\makebox(0,0)[lt]{\lineheight{1.25}\smash{\begin{tabular}[t]{l}$\II_\oo$\end{tabular}}}}%
    \put(0.73500215,0.0729182){\color[rgb]{0,0,0}\makebox(0,0)[lt]{\lineheight{1.25}\smash{\begin{tabular}[t]{l}$\III_\lambda \ \quad  (0\le \lambda\le1)$\end{tabular}}}}%
    \put(0.02488971,0.26523007){\color[rgb]{0,0,0}\makebox(0,0)[lt]{\lineheight{1.25}\smash{\begin{tabular}[t]{l}$\M$\end{tabular}}}}%
  \end{picture}%
\endgroup%

%% file: qmb1.eps_tex
\begingroup%
  \makeatletter%
  \providecommand\color[2][]{%
    \errmessage{(Inkscape) Color is used for the text in Inkscape, but the package 'color.sty' is not loaded}%
    \renewcommand\color[2][]{}%
  }%
  \providecommand\transparent[1]{%
    \errmessage{(Inkscape) Transparency is used (non-zero) for the text in Inkscape, but the package 'transparent.sty' is not loaded}%
    \renewcommand\transparent[1]{}%
  }%
  \providecommand\rotatebox[2]{#2}%
  \newcommand*\fsize{\dimexpr\f@size pt\relax}%
  \newcommand*\lineheight[1]{\fontsize{\fsize}{#1\fsize}\selectfont}%
  \ifx\svgwidth\undefined%
    \setlength{\unitlength}{381.55379168bp}%
    \ifx\svgscale\undefined%
      \relax%
    \else%
      \setlength{\unitlength}{\unitlength * \real{\svgscale}}%
    \fi%
  \else%
    \setlength{\unitlength}{\svgwidth}%
  \fi%
  \global\let\svgwidth\undefined%
  \global\let\svgscale\undefined%
  \makeatother%
  \begin{picture}(1,1.00016086)%
    \lineheight{1}%
    \setlength\tabcolsep{0pt}%
    \put(0,0){\includegraphics[width=\unitlength]{qmb1.eps}}%
  \end{picture}%
\endgroup%

%% file: qft.eps_tex
\begingroup%
  \makeatletter%
  \providecommand\color[2][]{%
    \errmessage{(Inkscape) Color is used for the text in Inkscape, but the package 'color.sty' is not loaded}%
    \renewcommand\color[2][]{}%
  }%
  \providecommand\transparent[1]{%
    \errmessage{(Inkscape) Transparency is used (non-zero) for the text in Inkscape, but the package 'transparent.sty' is not loaded}%
    \renewcommand\transparent[1]{}%
  }%
  \providecommand\rotatebox[2]{#2}%
  \newcommand*\fsize{\dimexpr\f@size pt\relax}%
  \newcommand*\lineheight[1]{\fontsize{\fsize}{#1\fsize}\selectfont}%
  \ifx\svgwidth\undefined%
    \setlength{\unitlength}{170.07874016bp}%
    \ifx\svgscale\undefined%
      \relax%
    \else%
      \setlength{\unitlength}{\unitlength * \real{\svgscale}}%
    \fi%
  \else%
    \setlength{\unitlength}{\svgwidth}%
  \fi%
  \global\let\svgwidth\undefined%
  \global\let\svgscale\undefined%
  \makeatother%
  \begin{picture}(1,1)%
    \lineheight{1}%
    \setlength\tabcolsep{0pt}%
    \put(0,0){\includegraphics[width=\unitlength]{qft.eps}}%
    \put(0.19748938,0.60375124){\color[rgb]{0.97647059,0.97647059,0.97647059}\makebox(0,0)[lt]{\lineheight{1.25}\smash{\begin{tabular}[t]{l}$\mathcal{O}_A$\end{tabular}}}}%
    \put(0.62835353,0.62356747){\color[rgb]{0.97647059,0.97647059,0.97647059}\makebox(0,0)[lt]{\lineheight{1.25}\smash{\begin{tabular}[t]{l}$\mathcal{O}_B$\end{tabular}}}}%
  \end{picture}%
\endgroup%

%% file: itpfi.eps_tex
\begingroup%
  \makeatletter%
  \providecommand\color[2][]{%
    \errmessage{(Inkscape) Color is used for the text in Inkscape, but the package 'color.sty' is not loaded}%
    \renewcommand\color[2][]{}%
  }%
  \providecommand\transparent[1]{%
    \errmessage{(Inkscape) Transparency is used (non-zero) for the text in Inkscape, but the package 'transparent.sty' is not loaded}%
    \renewcommand\transparent[1]{}%
  }%
  \providecommand\rotatebox[2]{#2}%
  \newcommand*\fsize{\dimexpr\f@size pt\relax}%
  \newcommand*\lineheight[1]{\fontsize{\fsize}{#1\fsize}\selectfont}%
  \ifx\svgwidth\undefined%
    \setlength{\unitlength}{285.2844046bp}%
    \ifx\svgscale\undefined%
      \relax%
    \else%
      \setlength{\unitlength}{\unitlength * \real{\svgscale}}%
    \fi%
  \else%
    \setlength{\unitlength}{\svgwidth}%
  \fi%
  \global\let\svgwidth\undefined%
  \global\let\svgscale\undefined%
  \makeatother%
  \begin{picture}(1,0.42904554)%
    \lineheight{1}%
    \setlength\tabcolsep{0pt}%
    \put(0,0){\includegraphics[width=\unitlength]{itpfi.eps}}%
    \put(0.13179911,0.02320097){\color[rgb]{0,0,0}\makebox(0,0)[lt]{\lineheight{1.25}\smash{\begin{tabular}[t]{l}$\Psi_1$\end{tabular}}}}%
    \put(0.01568999,0.33600945){\color[rgb]{0,0.50196078,0.50196078}\makebox(0,0)[lt]{\lineheight{1.25}\smash{\begin{tabular}[t]{l}$\M_A$\end{tabular}}}}%
    \put(0.01189488,0.22387201){\color[rgb]{0.05882353,0.5254902,1}\makebox(0,0)[lt]{\lineheight{1.25}\smash{\begin{tabular}[t]{l}$\M_B$\end{tabular}}}}%
    \put(0.00926593,0.10980938){\color[rgb]{1,0.4,0}\makebox(0,0)[lt]{\lineheight{1.25}\smash{\begin{tabular}[t]{l}$\M_C$\end{tabular}}}}%
    \put(0.3193815,0.02223696){\color[rgb]{0,0,0}\makebox(0,0)[lt]{\lineheight{1.25}\smash{\begin{tabular}[t]{l}$\ldots$\end{tabular}}}}%
    \put(0.51036876,0.02320097){\color[rgb]{0,0,0}\makebox(0,0)[lt]{\lineheight{1.25}\smash{\begin{tabular}[t]{l}$\Psi_j$\end{tabular}}}}%
    \put(0.57515107,0.02320097){\color[rgb]{0,0,0}\makebox(0,0)[lt]{\lineheight{1.25}\smash{\begin{tabular}[t]{l}$\ox$\end{tabular}}}}%
    \put(0.19592177,0.22387201){\color[rgb]{0.6,0.6,0.6}\makebox(0,0)[lt]{\lineheight{1.25}\smash{\begin{tabular}[t]{l}$\ox$\end{tabular}}}}%
    \put(0.3221117,0.22387201){\color[rgb]{0.70196078,0.70196078,0.70196078}\makebox(0,0)[lt]{\lineheight{1.25}\smash{\begin{tabular}[t]{l}$\ox$\end{tabular}}}}%
    \put(0.44830157,0.22387201){\color[rgb]{0.8,0.8,0.8}\makebox(0,0)[lt]{\lineheight{1.25}\smash{\begin{tabular}[t]{l}$\ox$\end{tabular}}}}%
    \put(0.57449144,0.22387201){\color[rgb]{0.90196078,0.90196078,0.90196078}\makebox(0,0)[lt]{\lineheight{1.25}\smash{\begin{tabular}[t]{l}$\ox$\end{tabular}}}}%
    \put(0.70068127,0.22387201){\color[rgb]{0.9254902,0.9254902,0.9254902}\makebox(0,0)[lt]{\lineheight{1.25}\smash{\begin{tabular}[t]{l}$\ox$\end{tabular}}}}%
    \put(0.82687103,0.22387201){\color[rgb]{0.94901961,0.94901961,0.94901961}\makebox(0,0)[lt]{\lineheight{1.25}\smash{\begin{tabular}[t]{l}$\ox$\end{tabular}}}}%
    \put(0.44830157,0.02407123){\color[rgb]{0.10196078,0.10196078,0.10196078}\makebox(0,0)[lt]{\lineheight{1.25}\smash{\begin{tabular}[t]{l}$\ox$\end{tabular}}}}%
    \put(0.19592177,0.02407123){\color[rgb]{0,0,0}\makebox(0,0)[lt]{\lineheight{1.25}\smash{\begin{tabular}[t]{l}$\ox$\end{tabular}}}}%
    \put(0.69269323,0.02223696){\color[rgb]{0,0,0}\makebox(0,0)[lt]{\lineheight{1.25}\smash{\begin{tabular}[t]{l}$\ldots$\end{tabular}}}}%
  \end{picture}%
\endgroup%